\documentclass[11pt, oneside]{amsart}  	
\usepackage{geometry}                		
\pdfoutput=1
\usepackage[utf8]{inputenc}
\usepackage[english]{babel}
\usepackage{graphicx}				
								
\usepackage{amssymb}
\usepackage{amsmath}
\usepackage{amsthm}
\usepackage{appendix}
\usepackage{mathrsfs}
\usepackage{mathtools}
\usepackage[all,arc]{xy}
\usepackage{tikz}
\usepackage{qtree}
\usepackage{graphicx}
\usepackage{cancel}
\usepackage{slashed}
\usepackage{tensor}
\usepackage{tcolorbox}
\tcbuselibrary{breakable}

\usepackage{tkz-graph}
\usepackage{caption}
\usepackage{subcaption}

\usepackage{csquotes}
\usepackage[maxbibnames=4, sortcites=true]{biblatex}
\AtBeginBibliography{\small}

\tikzstyle{no_dec} = [rectangle]
\tikzstyle{intermediate} = [rectangle, text centered, draw=black!15]
\tikzstyle{finalterm} = [rectangle, text centered, draw=black!15, fill=blue!15]

\usepackage{hyperref}
\hypersetup{
    colorlinks,
    citecolor=black,
    filecolor=black,
    linkcolor=blue,
    urlcolor=black
}

\DeclarePairedDelimiter{\norm}{\lVert}{\rVert}
\DeclarePairedDelimiter{\abs}{\lvert}{\rvert}

\newcommand{\dist}{{\text{dist}}}
\renewcommand{\div}{\text{div}}
\newcommand{\curl}{\text{curl}}
\newcommand{\supp}{\text{supp}}
\newcommand{\tr}{\text{tr}}

\renewcommand{\d}{\partial}
\newcommand{\Id}{{\text{Id}}}

\newcommand{\eps}{\varepsilon}

\renewcommand{\hat}{\widehat}

\renewcommand{\sl}{\slashed}
\newcommand{\slg}{\slashed{g}}
\newcommand{\PsiSi}{{\Psi_{(S)}^{(i)}}}
\newcommand{\PsiBarSi}{{\underline\Psi_{(S)}^{(i)}}}
\newcommand{\PsiSiEps}{{\Psi_{(S),\eps}^{(i)}}}
\newcommand{\PsiBarSiEps}{{\underline\Psi_{(S),\eps}^{(i)}}}

\newcommand{\eqcount}{\stepcounter{equation}\tag{\theequation}}

\newcommand{\lie}{\mathcal{L}}

\newcommand{\R}{\mathbb{R}}
\newcommand{\C}{\mathbb{C}}
\renewcommand{\u}{\underline{u}}
\renewcommand{\C}{\underline{C}}
\renewcommand{\L}{\underline{L}}

\newcommand{\sld}{{\sl{\mathcal{D}}}}
\newcommand{\cale}{{\mathcal{E}}}
\newcommand{\calebar}{{\underline{\mathcal{E}}}}

\def\XXint#1#2#3{{\setbox0=\hbox{$#1{#2#3}{\int}$ }
\vcenter{\hbox{$#2#3$ }}\kern-.6\wd0}}

\theoremstyle{definition}
\newtheorem{thm}{Theorem}[section]

\theoremstyle{definition}
\newtheorem{lemma}{Lemma}[section]

\theoremstyle{definition}
\newtheorem{prop}{Proposition}[section]

\theoremstyle{definition}
\newtheorem{remark}{Remark}[section]

\theoremstyle{definition}
\newtheorem{defn}{Definition}[section]

\theoremstyle{definition}
\newtheorem{example}{Example}[section]

\theoremstyle{definition}
\newtheorem{corollary}{Corollary}[section]

\usepackage{enumerate}[shortlabel]

\title{Linear hyperbolic equations in a double null foliation}
\author{Christopher Stith}
\date{\today{}}

\begin{document}

\begin{abstract}
The Bianchi identities for the Weyl curvature tensor of a spacetime $(M, g)$ solving the vacuum Einstein equations in a double null foliation exhibit a hyperbolic structure, which can be used to obtain detailed nonlinear estimates on the null Weyl tensor components. The aim of this paper is twofold. First we discuss existence and uniqueness for solutions of first-order linear hyperbolic systems of equations in a double null foliation on an arbitrary spacetime, with initial data posed on a past null hypersurface $\underline C_0 \cup C_0$. We prove a global existence and uniqueness theorem for these systems. Then we discuss the relationship between these systems, the Bianchi equations, and the linearized Bianchi equations (the linearized Bianchi equations are obtained from the usual Bianchi equations by replacing the null Weyl tensor components with unknown tensorfields). We derive a novel algebraic constraint which must be satisfied, at every point in the spacetime, by tensorfields satisfying the linearized Bianchi equations.
\end{abstract}

\maketitle

\tableofcontents

\section{Introduction}

\subsection{Overview and main results} In this paper, we study the characteristic initial value problem for linear hyperbolic systems of equations on a spacetime $(M, g)$. The main motive for studying such systems is their applications to the Einstein equations and the Bianchi equations in a double null foliation, which are extremely important in general relativity \cites{CK, chr-bhf, kla_rod2012trapped_surfaces, bieri2009stability_extension}.

A Lorentzian manifold $(M, g)$ satisfies the {Einstein equations} if
\begin{equation}\label{eq:einstein_equations}
    \text{Ric}(g) - \frac{1}{2}R(g)g = \frac{8\pi G}{c^4}\mathbf{T}
\end{equation}
where $\text{Ric}(g)$ denotes the Ricci curvature of $(M, g)$, and $R(g)$ denotes the scalar curvature. The symmetric 2-tensor $\mathbf{T}$ is the {stress-energy tensor} and represents the matter distribution in the spacetime. We will work in four spacetime dimensions. We will be primarily interested in {vacuum spacetimes}, i.e. $\mathbf{T} = 0$. In this case \eqref{eq:einstein_equations} reduces to the \emph{vacuum Einstein equations}
\begin{equation}\label{eq:VEE}
    \text{Ric}(g) = 0.
\end{equation}
To capture the essential hyperbolicity of the Einstein equations, and in particular the manifestation of this hyperbolicity in the Bianchi equations in a double null foliation, we introduce a more general system of linear hyperbolic equations which we call \emph{double null hyperbolic systems} (see \eqref{eq:general_hyperbolic_system}). We prove a global existence and uniqueness result for these systems. The class of spacetimes we consider have a globally-defined double null foliation and are thus diffeomorphic to the product 
\begin{align*}
    M \cong [0, u_*] \times [0, \u_*] \times S^2
\end{align*}
for some $u_*, \u_* > 0$ and $S^2$ the standard 2-sphere (see Figure \ref{fig:double_null_foliation}).
\begin{figure}[ht]
    \centering
    \begin{tikzpicture}[scale=.9]    
        \fill[red!20, draw=red!50!black, opacity=.4]
            (-1,1) to (-2, 0)
            arc[x radius = 2.03, y radius = 0.5, start angle = 190, end angle = 350]
            to (1,1)    
            arc[x radius = 1.02, y radius = 0.2, start angle = 10, end angle = 170]
            ;
        \fill[red!15, opacity=1]
            (1,1) arc[x radius = 1.02, y radius = 0.2, start angle = 10, end angle = 170]
            arc [x radius = 1.02, y radius = 0.2, start angle = -170, end angle = -10]
            ;    
        \fill[red!20, draw=red!50!black, opacity=.3]
            (-2,2) to (-3, 1)
            arc[x radius = 3.05, y radius = 0.6, start angle = 190, end angle = 350]
            to (2,2)
            arc[x radius = 2.03, y radius = 0.25, start angle = 10, end angle = 170]    
            ;
        \fill[red!15, opacity=1]
            (2,2) arc[x radius = 2.03, y radius = 0.25, start angle = 10, end angle = 170]
            arc [x radius = 2.03, y radius = 0.25, start angle = -170, end angle = -10]
            ;
        \draw[black, opacity=1]
            (2,2) arc[x radius = 2.03, y radius = 0.25, start angle = -10, end angle = -170]
            ;
        \draw[black, opacity=.5]
            (2,2) arc[x radius = 2.03, y radius = 0.25, start angle = 10, end angle = 170];
        \fill[blue!20, draw=blue!50!black, opacity=.3]
            (-4,2) to (-2, 0)
            arc[x radius = 2.03, y radius = 0.5, start angle = 190, end angle = 350]
            to (4,2) 
            arc[x radius = 4.06, y radius = 0.7, start angle = -10, end angle = -170]
            ;
        \fill[blue!40, opacity=.3]
            (4,2) arc[x radius = 4.06, y radius = 0.7, start angle = 10, end angle = 170]
            arc [x radius = 4.06, y radius = 0.7, start angle = -170, end angle = -10]
            ;
        \fill[blue!20, draw=blue!50!black, opacity=.3]
            (-3,3) to (-1, 1)
            arc[x radius = 1.02, y radius = 0.2, start angle = 190, end angle = 350]
            to (3,3) 
            arc[x radius = 3.05, y radius = 0.6, start angle = -10, end angle = -170]
            ;
        \fill[blue!40, opacity=.3]
            (3,3) arc[x radius = 3.05, y radius = 0.6, start angle = 10, end angle = 170]
            arc [x radius = 3.05, y radius = 0.6, start angle = -170, end angle = -10]
            ;
        \draw[blue!20]
            (3,3) arc[x radius = 3.05, y radius = 0.6, start angle = 10, end angle = 170]
            ;
        \node[scale=1] at (4,1.5) {\small{$C_0$}};
        \node[scale=1] at (3.2,2.7) {\small{$C_u$}};
        \node[scale=1] at (-2.3,1.1) {\small{$\underline C_{\underline u}$}};
        \node[scale=1] at (-1.3,.1) {\small{$\underline C_0$}};
        \draw[->, >=stealth] (1.2 + 1.5, 1.84 - 1.5) -- (1.2,1.84); 
        \node[scale=1] at (1.2 + 1.5 + .40, 1.84 - 1.5 - .23) {\small{$S_{u,\underline u}$}};
    \end{tikzpicture}
    \caption{Basic setup of the double null foliation on $M$}
    \label{fig:double_null_foliation}
\end{figure}
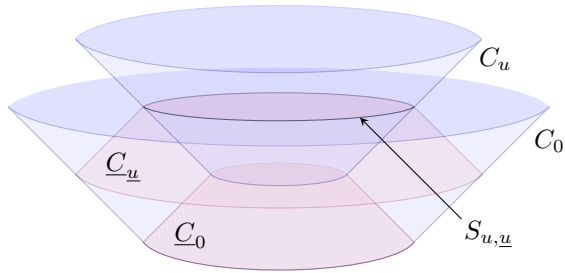
The manifold $M$ (with boundary) and the metric $g$ are assumed to be smooth, by which we mean $C^\infty$. The initial data for the hyperbolic systems we consider is posed on the ``initial'' hypersurfaces $\underline C_0$ and $C_0$. While $(M, g)$ is smooth, the initial data are required only to lie in certain Sobolev spaces. We require the data, restricted to a given sphere $S_{u,\u}$ in the initial hypersurfaces, to lie in $H^1(S_{u,\u})$. We note that while a general spacetime does not admit a global double null foliation, and indeed in the full Einstein equations (Bianchi equations coupled to the null structure equations; see Section \ref{subsec:geo_quantities_and_equations}) we only expect \emph{local} well-posedness, the double null hyperbolic systems introduced here are useful to investigate the structures of the Bianchi equations in double null foliations in these more general settings.
A rough version of our main result (Theorem \ref{thm:GWP_H^1}) for double null hyperbolic systems is the following. For the precise definition of a double null hyperbolic system, see Section \ref{subsec:hyperbolic_systems}.

\begin{thm}
Given initial data on $\C_0$ and $C_0$ for a double null hyperbolic system, there exists a unique global solution on $M$.
\end{thm}
Heuristically, a double null hyperbolic system takes the form 
\begin{equation}\label{eq:heuristic_DNH}
\begin{split}
    \sl\nabla_3 \Psi &= \sl\nabla\underline\Psi + \psi\cdot\Psi +  \psi\cdot\underline\Psi \\ 
    \sl\nabla_4 \underline\Psi &= \sl\nabla\Psi + \psi \cdot \Psi + \psi\cdot \underline\Psi,
\end{split}
\end{equation}
where $\Psi, \underline\Psi$ denote the unknowns (which are covariant tensorfields tangent to the spheres $S_{u,\u}$), $\psi$ denotes a Ricci coefficient of $(M, g)$, and $e_3, e_4$ are null vector fields which are tangent to the null hypersurfaces $\underline C_{\u}, C_u$, respectively. (See \eqref{eq:general_hyperbolic_system} for the full system.) The symbol $\sl\nabla$ denotes the spherical covariant derivative operator, i.e. the restriction of $\nabla$ to $S_{u,\u}$. The tensorfields $\sl\nabla_3\Psi$ and $\sl\nabla_4\Psi$ are the projections to the $S_{u,\u}$ of $\nabla_3\Psi$ and $\nabla_4\Psi$, respectively. See Section \ref{subsec:spacetime_and_notation} for our definitions and notation. We note that the particular structure of the principal terms $\sl\nabla\Psi, \sl\nabla\underline\Psi$ on the right-hand side is crucial and is discussed later in Section \ref{subsec:hyperbolic_systems}. 

We now briefly discuss the main ideas of the proof. We find it convenient and natural to pull back the equations \eqref{eq:heuristic_DNH} to the ``initial sphere'' $S_{0,0} = \underline C_0 \cap C_0$ to obtain a system of $(u, \u)$-dependent quantities on $S_{0,0}$. This is discussed in Sections \ref{subsec:null_flows} and \ref{subsec:pullback_equations_to_S}. Then we mollify the resulting system in the spherical variables in order to obtain an approximating system which can be viewed as a Sobolev space-valued system of ODE. For this we make use of and extend the results of \cite{fukuoka2006smoothing} regarding mollifiers of tensorfields on Riemannian manifolds (see Section \ref{sec:spherical_mollification_of_tensorfields}).  The required existence and uniqueness theory for the resulting system, which can be thought of as having two ``time'' variables $u$ and $\u$, is discussed in Appendix \ref{app:two-var_ode_theory}. We choose this method instead of a simpler one due to its robustness and, in particular, the possibility of adapting it to quasilinear hyperbolic systems.

At the heart of the proof are the energy estimates, which are a consequence of the essential structure of double null hyperbolic systems. These energy estimates follow the general structure of those used in the mathematical general relativity literature to derive estimates on null curvature components, cf. \cites{CK, chr-bhf, luk2011local_existence, kla_nic2003}. They take the form 
\begin{align*}
    \int_{C_u} |\Psi|^2 + \int_{\underline C_{\u}} |\underline\Psi|^2 \lesssim \text{Initial Data}.
\end{align*}
(See Sections \ref{subsec:energy_estimates_1} and \ref{subsec:energy_estimates_2}.) It is here that the hyperbolicity of the equations is used. These provide the required uniform bounds to extract a convergent subsequence of the approximating equations which is a global solution to \eqref{eq:heuristic_DNH}.

Afterwards we turn to the \emph{linearized Bianchi equations} on a vacuum spacetime $(M, g)$ in Section \ref{sec:linearized_bianchi_equations}. These are the equations which are obtained from replacing, in the usual Bianchi equations on $(M, g)$, the null Weyl tensor components with unknowns that are to be solved for (see \eqref{eq:linearized_bianchi_equations}). As the null Weyl tensor components themselves are one solution, it is of interest to determine if there are others, and to solve the characteristic initial value problem for this system for general initial data. 

One obstruction to setting up an initial value problem for the linearized Bianchi equations is that this system is overdetermined; for instance, $\beta$ satisfies the two equations 
\begin{align*}
    \sl\nabla_3 \beta &= \sl\nabla\rho + \prescript{*}{}{\sl\nabla\sigma} + \cdots \qquad \text{and} \qquad \sl\nabla_4 \beta = \sl\div\alpha + \cdots.
\end{align*}
One possible way of overcoming this issue is by choosing a subset of the equations to be \emph{constraints} which we must impose on the characteristic initial data, and then derive a propagation-of-constraints theorem. This is discussed in Sections \ref{subsec:LNB_preliminaries}-\ref{subsec:differential_constraints}, where we prove a partial result. Moreover, we discover an interesting set of purely \emph{algebraic} constraints which solutions to the linearized Bianchi equations must satisfy. This is the focus of Section \ref{subsec:algebraic_constraints}. The main theorem of this section can be stated heuristically as follows. 

\begin{thm} Let $(M, g)$ be a vacuum spacetime equipped with a double null foliation.
Let
\begin{align*}
    \mathcal{V}_p \coloneqq \Bigl\{\substack{\text{symmetric traceless} \\ \text{$2$-covariant tensors}}\Bigr\} \times T_p^* S_{u,\underline u} \times \mathbb{R} \times \mathbb{R} \times T_p^*S_{u,\underline u} \times \Bigl\{\substack{\text{symmetric traceless} \\ \text{$2$-covariant tensors}}\Bigr\}
\end{align*}
(this is the space to which solutions of the linearized Bianchi equations belong). There is a linear map $\mathfrak{L}_W|_p : \mathcal{V}_p \to \mathcal{V}_p$, depending only on the null components of the Weyl tensor $W$ of $(M, g)$, such that
\begin{enumerate}
    \item any solution of the linearized Bianchi equations lies in $\ker \mathfrak{L}_W|_p$, and
    \item  if $W|_p \neq 0$, then $\text{\normalfont dim} \,\text{\normalfont ker} \, \mathfrak{L}_W|_p <  \text{\normalfont dim} \mathcal{V}_p$.
\end{enumerate}
\end{thm}
One consequence of this is that solutions to the linearized Bianchi equations are more constrained than previously known. For instance, it is well known that in the characteristic initial value problem for the vacuum Einstein equations, the initial data for the metric must satisfy constraint equations which are ODEs along the null generators of the initial hypersurfaces. Rendall \cites{rendall1990characteristic_problem} gave the first mathematical proof of existence and uniqueness for the characteristic initial value problem for the Einstein equations, but this problem was also investigated earlier by Sachs \cites{sachs, sachs_2} and others.
These constraint equations have analogs in the linearized Bianchi equations, also taking the form of ODEs along the null generators. This theorem states that there are \emph{additional} constraints on the initial data for the linearized Bianchi equations, namely that they lie in the kernel of $\mathfrak{L}_W$. In contrast to the usual constraints, which are differential equations, these new constraints are purely algebraic. We remark that the proof is constructive, and an explicit formula for $\mathfrak{L}_W$ is written down in Section \ref{subsec:algebraic_constraints}. These constraints also must be confronted in a potential proof of well-posedness for the linearized Bianchi equations.

\subsection{Background} The larger goal motivating this work is to prove local well-posedness for the Einstein equations in a double null foliation in a manner carried out \emph{entirely within} the framework of the double null foliation, that is, using the null structure equations and the null Bianchi equations as the starting point. (We use the phrase ``null Bianchi equations'' to refer to the Bianchi equations decomposed with respect to a double null foliation.) We now briefly review the background of this subject.

The double null foliation (or double null gauge) is finding application in a wider and wider variety of problems in the general relativity community. A non-exhaustive list includes the stability of Minkowski spacetime for the Einstein and Einstein-Vlasov equations \cites{kla_nic2003, taylor2016stability}; trapped surface formation (\cites{chr-bhf} as well as \cites{an_athan2020scale_critical, kla_rod2012trapped_surfaces, pengyu_hyperplane}); and cosmic censorship \cite{daf_luk2017cauchy_horizon}. A crucial part of this setup is the null Bianchi equations and the null structure equations, satisfied by the spacetime Weyl curvature tensor and Ricci (connection) coefficients of the foliation, respectively. Even in works where a double null foliation is not used, the null Bianchi and null structure equations are essential to the analysis (for example \cite{CK} as well as \cites{zipser_stability,bieri2009stability_extension}). Despite its importance, there has not been an investigation of the local well-posedness of the Einstein equations in a double null foliation which has been carried out entirely within this framework.

Rendall \cite{rendall1990characteristic_problem} proved local existence and uniqueness for the characteristic initial value problem for the Einstein equations. The domain of existence is a neighborhood of the 2-surface of intersection of the initial null hypersurfaces. Luk \cite{luk2011local_existence} extended this existence and uniqueness result to hold in a neighborhood of the initial null hypersurfaces. In Luk's work, Rendall's theorem is used to start the local existence argument. Rendall's proof reduces the characteristic initial value problem to the Cauchy problem, making essential use of the Einstein equations in a wave gauge, rather than treating the Einstein equations purely geometrically, entirely within a double null foliation. 

As a brief recap, recall that the \emph{Cauchy problem} for the Einstein equations consists of solving the Einstein equations given spacelike initial data. The initial data consists of a Riemannian manifold $(\Sigma, \Bar{g})$ and a symmetric 2-covariant tensorfield $k$ on $\Sigma$.\footnote{If one is considering the Einstein equations coupled to a matter field, then one also must prescribe initial data for the matter fields. Here, to touch on the main ideas of the Cauchy problem, we focus on the vacuum Einstein equations.} The tensorfield $k$ represents how $(\Sigma, \bar g)$ is to be embedded in the yet-to-be-found ambient spacetime. One seeks a solution $(M, g)$ of the Einstein equations, with $M \cong [0, T] \times \Sigma$, such that 1) $g|_{t = 0} = \bar g$ and 2) the second fundamental form of $\{0\} \times \Sigma$ in $M$ is $k$  (see Figure \ref{fig:cauchy_problem_diagram}). 
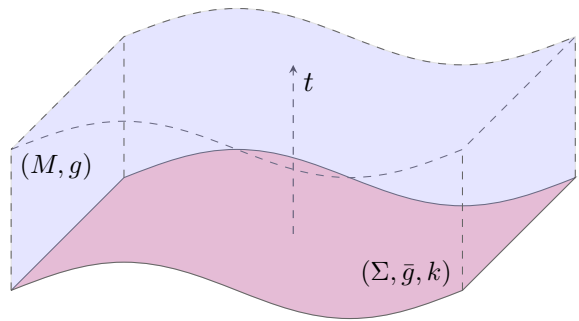
\begin{figure}[ht]
\def\a{2.5}
\centering
\begin{tikzpicture}[scale=.75]
    \fill[red!20, draw=black!60] (-4, -1) sin (-2, -.5) cos (0, -1) sin (2, -1.5) cos (4, -1) to (6, 1) sin (4, .5) cos (2, 1) sin (0, 1.5) cos (-2, 1) to (-4, -1);
    \draw[dashed, ->, >=stealth, black!70] (1, 0) to (1, 3);
    \draw[black!60, dashed] (-4, -1 + \a) sin (-2, -.5 + \a) cos (0, -1 + \a) sin (2, -1.5 + \a) cos (4, -1 + \a) to (6, 1 + \a) sin (4, .5 + \a) cos (2, 1 + \a) sin (0, 1.5 + \a) cos (-2, 1 + \a) to (-4, -1 + \a);
    \fill[blue!50, draw=black!60, opacity=.2] (-4, -1) to (-4, -1 + \a) to (-2, 1 + \a) sin (0, 1.5 + \a) cos (2, 1 + \a) sin (4, .5 + \a) cos (6, 1 + \a) to (6, 1) to (4, -1) sin (2, -1.5) cos (0, -1) sin (-2, -.5) cos (-4, -1);
    \draw[black!60, dashed] (-4, -1) to (-4, 1.5);
    \draw[black!60, dashed] (4, -1) to (4, 1.5);
    \draw[black!60, dashed] (-2, 1) to (-2, 3.5);
    \draw[black!60, dashed] (6, 1) to (6, 3.5);
    \node[scale=1] at (3, -.7) {\small $(\Sigma, \bar g, k)$};
    \node[scale=1] at (-3.2, 1.2) {\small $(M, g)$};
    \node[right] at (1, 2.7) {$t$};
\end{tikzpicture}
\caption{The Cauchy problem for the Einstein equations}
\label{fig:cauchy_problem_diagram}
\end{figure}
The local existence and uniqueness problem was solved in Choquet-Bruhat's fundamental work \cites{CB_local_existence}. Her work makes essential use of the choice of wave coordinates, in which Einstein's equations are equivalent to the \emph{reduced} Einstein equations, which are a system of quasilinear wave equations for the unknown metric $g$. In Rendall's proof \cite{rendall1990characteristic_problem}, by reducing the characteristic initial value problem to the Cauchy problem, the null structure and null Bianchi equations do not enter the picture. In Luk \cite{luk2011local_existence}, the null structure and null Bianchi equations are heavily used to obtain estimates on the Ricci coefficients and curvature components in the double null foliation to close a bootstrap argument and greatly extend the region of existence. However, we emphasize here that Luk's work uses the null structure and null Bianchi equations only for estimates; the local existence results which kickstart the proof and finish the ``last slice'' argument rely on wave coordinates and quasilinar wave equation theory, as in Rendall.

The goal of the present work is to begin the study of an eventual proof of local well-posedness for the characteristic initial value problem for the Einstein equations which is carried out entirely in the double null foliation, treating the null structure and null Bianchi equations as the primary system to be solved. We also hope that, by treating the problem entirely within the double null foliation, it may be used to analyze in greater detail the Einstein equations coupled to various null matter fields (for instance, Einstein-null dust, Einstein-massless Vlasov, or Einstein-Maxwell), as well as the propagation of gravitational radiation. Moreover, the new algebraic structures which we discover in this paper may find application in other important problems beyond those discussed here.

The linearization of the Einstein equations also has a very long history. We note here that \cite{dhr_linear} addresses (in particular) the issue of well-posedness of ``linearized gravity'', that is, the simultaneous linearization of the null Bianchi and null structure equations. This is a different system than we consider here; our use of the word ``linearization'' in reference to the linearized Bianchi equations refers to the fact that we decouple the spacetime geometry from the unknowns of the Bianchi equations (see Section \ref{sec:linearized_bianchi_equations}).

\subsection{Organization of the paper} In Section \ref{sec:basic_setup} we discuss the setup of the problem in more detail. We also introduce the concept of a double null hyperbolic system. Section \ref{sec:spherical_mollification_of_tensorfields} deals with spherical mollification of tensorfields and proves the results we require to apply them to the PDE we consider. In Section \ref{sec:dnh_lwp}, we discuss the pullbacks of double null hyperbolic systems to the sphere $S_{0,0}$. In this section the energy estimates and the main existence and uniqueness results are proven. Finally, Section \ref{sec:linearized_bianchi_equations} discusses the linearized Bianchi equations and the algebraic constraints $\mathfrak{L}_W$.

Appendix \ref{app:two-var_ode_theory} discusses and proves some basic results in what we call two-variable ODE theory, that is ODE systems which contain two independent variables. The primary example of such systems are the null structure and null Bianchi equations, in which the two independent variables are $u$ and $\u$. This theory has many similarities with ordinary ODE theory, but as we could not find a reference for these results, we prove what we need here. Appendix \ref{app:useful_formulae} provides some useful identities, including commutation formulae and a two-variable Gr\"onwall lemma. Appendix \ref{app:Bianchi_computations} provides the detailed computations of the proof of Theorem \ref{thm:prop_eq_for_diff_constraints}.

\subsection{Acknowledgements} The author would like to express his immense gratitude to Lydia Bieri for her generous support, as well as the inspiration to work on the present problem and helpful comments on the manuscript. The author is also thankful to Demetrios Christodoulou for several enlightening discussions concerning this project. The author would also like to thank Neel Patel and Phillip Lo for helpful discussions on the subject.

\section{Basic Setup}\label{sec:basic_setup}

\subsection{Spacetime and notation}\label{subsec:spacetime_and_notation}
Throughout this paper, $(M, g)$ will denote a Lorentzian spacetime (not always assumed to satisfy the Einstein equations) with $M$ a smooth manifold (with corners) and $g$ a time-orientable Lorentzian metric on $M$. Unless otherwise specified, $g$ is assumed to be smooth, by which we mean $C^\infty$. The manifold $M$ is assumed to be diffeomorphic to the product 
\begin{align*}
    [0, u_*] \times [0, \u_*] \times S^2,
\end{align*}
and the coordinate functions $u, \u$ which project onto the first and second coordinates (respectively) are optical functions.\footnote{That is, $du$ and $d\u$ are null 1-forms.} The numbers $u_*, \u_*$ are fixed positive numbers. We denote 
\begin{align*}
    D_{u_*, \u_*} = [0, u_*] \times [0, \u_*]
\end{align*}
and will write $D$ instead of $D_{u_*, \u_*}$ when it is clear, for brevity. We let $C_u$ and $\underline C_{\u}$ denote, respectively, the ``outgoing'' and ``incoming'' null hypersurfaces $\{u = \text{const}\}$ and $\{\u = \text{const}\}$. We also let $C_u^{\u_1}$ denote the submanifold of $C_u$ for which $\u \leq \u_1$, and $\underline C_{\u}^{u_1}$ the submanifold of $\underline C_{\u}$ for which $u \leq u_1$.
Also, denote by
\begin{align*}
    S_{u,\u} &= C_u \cap \underline C_{\u}.
\end{align*}
These are diffeomorphic to $S^2$, and they are assumed to be spacelike.
The spacetime metric $g$, when restricted to a given sphere $S_{u,\u}$, will be denoted $\gamma_{u,\u}$, and we will write $\gamma$ often for brevity, suppressing the $u, \u$ in our notation. The Levi-Civita connection of $(M, g)$ is denoted $\nabla$, and that of $(S_{u, \u}, \gamma_{u,\u})$ is denoted $\sl\nabla$. Similarly, $\sl\div$ and $\sl\curl$ represent the divergence and curl operators, respectively, associated to $(S_{u,\u}, \gamma_{u,\u})$. The volume form of $(M, g)$ is denoted $\epsilon$, and that of $(S_{u,\u}, \gamma)$ is denoted $\sl\epsilon$.

We define the following geodesic null pair: 
\begin{align*}
    \underline L'^\mu &= -2 g^{\mu\nu}\d_\nu \u \quad \text{and} \quad L'^\mu = -2 g^{\mu\nu}\d_\nu u
\end{align*}
and their null lapse $\Omega$ by
\begin{align*}
    -2\Omega^{-2} &= g(L', \underline L').
\end{align*}
We also define the following two null pairs:
\begin{align*}
    e_3 &= \Omega \underline L' &e_4 &= \Omega L', \\
    \underline L &= \Omega e_3  &L &= \Omega e_4.
\end{align*}
These have the property that 
\begin{align*}
    g(e_3, e_4) &= -2
\end{align*}
and that 
\begin{align*}
    Lu &= 0 & L\u &= 1 \\ 
    \L u &= 1 & \L \u &= 0.
\end{align*}
As in \cite[Chapter~1.2]{chr-bhf}, we refer to the formal operation of replacing $C_u$ with $\underline C_{\u}$ and $L$ with $\underline L$ as \emph{conjugation}, and we call two objects \emph{conjugate} if the definition of one is obtained from the other by conjugation.

We let $\Phi_{\u}$ denote the flow map of $L$ for parameter time $\u$, and $\underline\Phi_u$ denote the flow map of $\underline L$ for parameter time $u$. Thus, for a fixed $p \in M$,
\begin{align*}
    \frac{d}{d\u}\Big\rvert_{\u = 0} \Phi_{\u}(p) &= L_p, \quad \Phi_0(p) = p, \\
    \frac{d}{d u}\Big\rvert_{u = 0} \underline\Phi_u(p) &= \underline L_p, \quad \underline\Phi_0(p) = p.
\end{align*}
These are defined only on a subset of $M$. For example, for $\u \geq 0$, 
\begin{align*}
    \Phi_{\u} &: [0, u_*] \times [0, \u_* - \u] \times S^2 \to [0, u_*] \times [\u, \u_*] \times S^2 \\ 
    \Phi_{-\u} &: [0, u_*] \times [\u, \u_*] \times S^2 \to [0, u_*] \times [0, \u_* - \u] \times S^2,
\end{align*}
and similarly for $\underline\Phi_u, \underline\Phi_{-u}$. These flows do not commute in a general spacetime due to the presence of nonzero torsion $\zeta$ (see Section \ref{subsec:geo_quantities_and_equations}). They preserve the spherical foliation of the spacetime.

We call a tensorfield $\xi$ on $M$ an \emph{$S$ tensorfield} if it is everywhere tangent to $S_{u,\u}$.
For an arbitrary $S$ tensorfield $\xi$, we write $D\xi$ for the projection of the Lie derivative $\lie_L \xi$ to the $S_{u,\u}$, and $\underline D\xi$ for the projection of the Lie derivative $\lie_{\underline L} \xi$ to the $S_{u,\u}$. In general, for $\xi$ an arbitrary tensorfield on $M$ (which is not necessarily $S$ tangent), we write $\Pi \xi$ to denote the $S$ tensorfield which is, at every point $p \in M$, the projection of $\xi_p$ to $T_p S_{u(p), \u(p)}$. For example, 
\begin{align*}
    D\xi &= \Pi(\lie_L \xi).
\end{align*}

As $C_u$ and $\underline C_{\u}$ have no natural volume form, we define integration along these null hypersurfaces as
\begin{align*}
    \int_{C_u^{\u_1}} f &= \int_0^{\u_1} \int_{S_{u,\u}} f\, d\mu_{\gamma_{u,\u}}\, d\u \\
    \int_{\C_{\u}^{u_1}} f &= \int_0^{u_1} \int_{S_{u,\u}} f \, d\mu_{\gamma_{u,\u}}\, d u.
\end{align*}
The volume form $d\mu_g$ of the spacetime is 
\begin{align*}
    d\mu_g &= 2\Omega^2 d\mu_{\gamma} \wedge du \wedge d \underline u.
\end{align*}
Hence
\begin{align*}
    \int_M f \, d\mu_g &= \int_0^{u_*} \int_0^{\u_*} \int_{S_{u,\u}} 2\Omega^2 f \, d\mu_\gamma \, du \, d\u  \\ 
    &= \int_0^{u_*} \Big(\int_{C_u} 2\Omega^2 f\Big) \, d u \\ 
    &= \int_0^{\u_*} \Big(\int_{\C_{\u}} 2\Omega^2 f\Big) \, d \u.
\end{align*}
If $\xi$ is an $S$ 1-form, its Hodge dual $\prescript{*}{}{\xi}$ is the $S$ 1-form defined by 
\begin{equation}
    \prescript{*}{}{\xi}_A = \sl\epsilon_{AB}\xi^B.
\end{equation}
Note that $\prescript{*}{}{(\prescript{*}{}{\xi})} = -\xi$.
If $\xi$ is a symmetric traceless 2-covariant $S$ tensorfield, its Hodge dual $\prescript{*}{}{\xi}$ is the symmetric traceless 2-covariant $S$ tensorfield defined by 
\begin{equation}
    \prescript{*}{}{\xi}_{AB} = \sl\epsilon_{AC}\tensor{\xi}{_B ^C}.
\end{equation}
Again, note that $\prescript{*}{}{(\prescript{*}{}{\xi})} = -\xi$. Also, recall that
\begin{equation}
    \sl\epsilon^{AB}\sl\epsilon_{CD} = \tensor{\delta}{^A _C} \tensor{\delta}{^B _D} - \tensor{\delta}{^A _D}\tensor{\delta}{^B _C}
\end{equation}
and
\begin{equation}
    \sl\epsilon^{AB}\sl\epsilon_{AC} = \tensor{\delta}{^B _C}.
\end{equation}

Given a compact Riemannian manifold $(N, h)$, we denote the $k$th-order Sobolev space of $p$-covariant tensorfields on $N$ by $H^k_p(N, h)$. When the metric is understood, we write simply $H^k_p(N)$ or $H^k_p$. This space is equipped with the Sobolev norm
\begin{align*}
    \norm{\xi}_{H^k_p(N,h)} &= \Big(\sum_{i = 0}^k\int_N |\textbf{D}^i \xi|^2_h \, d\mu_{h}\Big)^{1/2},
\end{align*}
where $\textbf{D}$ denotes the Levi-Civita connection of $(N, h)$. In the case that $k = 0$ we also use the notation $L^2_p$. We emphasize here that the lower subscript refers to the rank of the tensorfields, rather than a weighting parameter. When the rank is understood or unimportant, we sometimes write simply $H^k$ or $L^2$.

\subsection{Canonical coordinates} \label{subsec:canonical_coordinates} (See \cite[Chapter~1.4]{chr-bhf}.)
One can introduce coordinates, called \emph{canonical coordinates}, on a subset $M_U \subset M$ as follows. Let $U \subset S_{0,0}$ be a coordinate chart with coordinates $(\theta^1, \theta^2)$. Any point $p \in C_0$ is assigned the coordinates $(0, \u, \theta)$ ($\theta = (\theta^1, \theta^2)$), where $p = \Phi_{\u}(p_0)$ for $p_0 \in S$, and the $\theta$-coordinates of $p_0$ are $(\theta^1, \theta^2)$. Any point $p \in M$, being in a unique $S_{u,\u}$, is the image of a unique point $q = \underline\Phi_{-u(p)} \in C_0$. If $(0, \u, \theta^1, \theta^2)$ are the coordinates of $q$, then the coordinates $(u, \u, \theta^1, \theta^2)$ are assigned to $p$. Such coordinates are defined on 
\begin{align*}
    M_U &= \bigcup_{(u, \u) \in D} \underline\Phi_u(\Phi_{\u}(U)).
\end{align*}
The $\theta$-coordinates satisfy the property 
\begin{align*}
    \underline L (\theta^A) = 0.
\end{align*}
For this reason, these coordinates will be referred to as \emph{$\underline L$-adapted canonical coordinates}. Note that in these coordinates, the map $\underline\Phi_{u_0}$ is represented by 
\begin{align*}
    (u, \u, \theta^1, \theta^2) &\mapsto (u + u_0, \u, \theta^1, \theta^2).
\end{align*}

Therefore the Jacobian matrix of $\underline\Phi_{u_0}$, represented in $\underline L$-adapted canonical coordinates, is the $4 \times 4$ identity matrix $\Id_{4 \times 4}$.
In these coordinates, the null vector fields have the following expressions:
\begin{align*}
    \underline L &= \frac{\d}{\d u} \\ 
    L &= \frac{\d}{\d \u} + b
\end{align*}
where $b$ is the vector field on $M$, tangential to $S_{u,\u}$, which solves the following ODE:
\begin{align*}
    \underline D b &= 4\Omega^2 \zeta^\sharp, \quad b|_{C_0} = 0.
\end{align*}
Here, $\zeta^\sharp$ denotes the $S$ vector field which is the metric dual to the 1-form $\zeta$.
Let $q \in M$ have coordinates $(u, \u, \theta)$.
Since the Jacobian matrix of $\underline\Phi_u$ in $\underline L$-adapted canonical coordinates is the identity matrix, it follows that
\begin{align*}
    d(\underline\Phi_{u_0})_{q}(L) &= d(\underline\Phi_{u_0})_{q}\big(\d_{\u}|_q + b^A(q) \d_A|_{q}\big) \\ 
    &= \d_{\u}|_{\underline\Phi_{u_0}(q)} + b^A(q) \d_A|_{\underline\Phi_{u_0}(q)}.
\end{align*}
Since $b^A(q) \neq b^A(\underline\Phi_{u_0}(q))$ in general, this is \emph{not} equal to $L_{\underline\Phi_{u_0}(q)}$. In particular, when $q \in C_0$, we have
\begin{align*}
    d(\underline\Phi_{u_0})_{q}(L) &= \d_{\u}|_{\underline\Phi_{u_0}(q)} = (L - b)|_{{\underline\Phi_{u_0}(q)}}.
\end{align*}

\subsection{Geometric quantities and equations}\label{subsec:geo_quantities_and_equations}

Let $W$ be the Weyl tensor of $(M, g)$. 
Let $(e_A)_{A = 1, 2}$ be an arbitrary (local) frame field on $S_{u, \u}$. Define the Ricci (connection) coefficients 
\begin{align*}
    \underline\chi_{AB} &= g(\nabla_A e_3, e_B) & \chi_{AB} &= g(\nabla_A e_4, e_B) \\
    \underline\zeta_A &= \frac{1}{2}g(\nabla_A e_3, e_4) & \zeta_A &= \frac{1}{2}g(\nabla_A e_4, e_3) \\ 
    \underline\omega &= -\frac{1}{4}g(\nabla_3 e_4, e_3) & \omega &= -\frac{1}{4}g(\nabla_4 e_3, e_4).
\end{align*}
Note that $\underline\zeta = -\zeta$.
Also, we define
\begin{align*}
    \underline\eta_A &= -\zeta_A + \sl\nabla_A \log\Omega & \eta_A &= \zeta_A + \sl\nabla_A \log\Omega.
\end{align*}
Define the null Weyl tensor components
\begin{align*}
    \underline\alpha_{AB}[W] &= W(e_A, e_3, e_B, e_3) & \alpha_{AB}[W] &= W(e_A, e_4, e_B, e_4) \\ 
    \underline\beta_A[W] &= \frac{1}{2}W(e_A, e_3, e_3, e_4) & \beta_A[W] &= \frac{1}{2}W(e_A, e_4, e_3, e_4) \\ 
    \sigma[W]&= \frac{1}{4}\prescript{*}{}{W}(e_3, e_4, e_3, e_4) & \rho[W] &= \frac{1}{4}W(e_3, e_4, e_3, e_4).
\end{align*}
Here, $\prescript{*}{}{W}$ denotes the spacetime Hodge dual of $W$: 
\begin{align*}
    \prescript{*}{}{W}_{\alpha\beta\gamma\delta} &\coloneqq \frac{1}{2}\epsilon_{\alpha\beta\mu\nu}\tensor{W}{^\mu ^\nu _\gamma _\delta}.
\end{align*}
Note that the choice of definition\footnote{If $\hat\omega, \hat{\underline\omega}$ denote the quantities defined as in Christodoulou \cite{chr-bhf}, note that $\hat\omega = -2\omega, \hat{\underline\omega} = -2\underline\omega$.} for $\underline\omega, \omega$ here follows Luk \cite{luk2011local_existence}. As the goal of this paper is to study the linearized Bianchi equations on a fixed background spacetime, the symbols $\alpha, \beta$, etc. without $[W]$ will be reserved to denote arbitrary $S$ tensorfields on $(M, g)$, and we will always write $\alpha[W], \beta[W]$, etc. when it is the null components of the Weyl tensor $W$ we are discussing. We will frequently write $\psi$ to denote an arbitrary connection coefficient. In this paper, we will use capital Latin letters for spherical indices taking values in $\{1, 2\}$ and Greek letters for spacetime indices taking values in $\{1, 2, 3, 4\}$. Also, we let $\hat\theta$ denote the traceless part of a symmetric 2-covariant $S$ tensorfield $\theta$.  

We now list the null structure and null Bianchi equations of a vacuum spacetime $(M, g)$, that is, a solution of the vacuum Einstein equations \eqref{eq:VEE}. The reader is referred to \cite{chr-bhf} and \cite{luk2011local_existence} for more details regarding the notation and equations here. Note that in a vacuum spacetime, the Weyl tensor is equal to the Riemann curvature tensor.

In a vacuum spacetime $(M, g)$, the Ricci coefficients satisfy the following propagation equations:
\begin{equation}\label{eq:null_structure_propagation}
\begin{split}
    \sl\nabla_4 \tr\chi + \frac{1}{2}(\tr\chi)^2 &= -|\hat\chi|^2 - 2\omega\tr\chi  \\ 
    \sl\nabla_4 \hat\chi + \tr\chi \hat \chi &= -2\omega\hat\chi - \alpha[W] \\ 
    \sl\nabla_3\tr\underline\chi + \frac{1}{2}(\tr\underline\chi)^2 &= -|\hat{\underline\chi}|^2 - 2\underline\omega \tr\underline\chi \\ 
    \sl\nabla_3 \hat{\underline\chi} + \tr\underline\chi \hat{\underline\chi} &= -2\underline\omega \hat{\underline\chi} - \underline\alpha[W] \\ 
    \sl\nabla_4 \tr\underline\chi + \frac{1}{2}\tr\chi\tr\underline\chi &= 2\omega \tr\underline\chi - \hat\chi\cdot\hat{\underline\chi} + 2\sl\div\underline\eta + 2|\underline\eta|^2 + 2\rho[W] \\
    \sl\nabla_4 \hat{\underline\chi} + \frac{1}{2}\tr\chi\hat{\underline\chi}  &= - \frac{1}{2}\tr\underline\chi \hat{\chi} + 2\omega \hat{\underline\chi} + \underline\eta \hat\otimes \underline\eta + \sl\nabla\hat\otimes \underline\eta \\
    \sl\nabla_3 \tr\chi + \frac{1}{2}\tr\underline\chi \tr\chi &= 2\underline\omega \tr\chi - \hat\chi \cdot\hat{\underline\chi} + 2\sl\div\eta + 2|\eta|^2 + 2\rho[W] \\ 
    \sl\nabla_3 \hat\chi + \frac{1}{2}\tr\underline\chi \hat\chi &= - \frac{1}{2}\tr\chi \hat{\underline\chi} + 2\underline\omega \hat\chi + \eta \hat\otimes \eta + \sl\nabla\hat\otimes \eta \\
    \sl\nabla_4 \eta &= - \hat\chi \cdot (\eta - \underline\eta) - \frac{1}{2}\tr\chi (\eta - \underline\eta) - \beta[W] \\ 
    \sl\nabla_3 \underline\eta &= \hat{\underline\chi} \cdot (\eta - \underline\eta) + \frac{1}{2}\tr\chi (\eta - \underline\eta) + \underline\beta[W] \\ 
    \sl\nabla_4 \underline\omega &= 2 \omega \underline\omega  + \frac{1}{2}|\eta|^2 - \eta \cdot\underline\eta + \frac{1}{2}\rho[W] \\ 
    \sl\nabla_3 \omega &= 2 \omega \underline\omega + \frac{1}{2}|\underline\eta|^2 - \eta \cdot\underline\eta + \frac{1}{2}\rho[W],
\end{split}
\end{equation}
as well as the following constraint equations:
\begin{equation}\label{eq:null_structure_constraint}
\begin{split}
    \sl\div \hat\chi &= \frac{1}{2}\sl\nabla \tr\chi - \frac{1}{2}(\eta - \underline\eta) \cdot (\hat\chi - \frac{1}{2}\tr\chi \gamma) - \beta[W] \\ 
    \sl\div \hat{\underline\chi} &= \frac{1}{2}\sl\nabla\tr\underline\chi + \frac{1}{2}(\eta - \underline\eta) \cdot (\hat{\underline\chi} - \frac{1}{2}\tr\underline\chi \gamma) + \underline\beta[W] \\ 
    \sl\curl \eta &= - \sl\curl \underline\eta = \sigma[W] + \frac{1}{2}\hat{\underline\chi} \wedge \hat\chi \\
    K &= \frac{1}{2}\hat\chi \cdot\hat{\underline\chi} - \frac{1}{4}\tr\chi \tr\underline\chi - \rho[W],
\end{split}
\end{equation}
where $K$ denotes the Gauss curvature of the spheres $(S_{u,\u},\gamma)$. 
Equations \eqref{eq:null_structure_propagation} and \eqref{eq:null_structure_constraint} are collectively referred to as the \emph{null structure} equations.

In a vacuum spacetime, the null Weyl tensor components satisfy the following null Bianchi equations: 
\begin{equation}\label{eq:null_Bianchi_equations}
\begin{split}
    \sl\nabla_3 \alpha[W] &= (4 \underline\omega - \frac{1}{2}\tr\underline\chi) \alpha[W] + \sl\nabla\hat\otimes\beta[W] + (4\eta + \zeta)\hat\otimes \beta[W] - 3\hat\chi\rho[W] - 3\prescript{*}{}{\hat{\chi}}\sigma[W] \\
    \sl\nabla_4 \underline\alpha[W] &= (4\omega - \frac{1}{2}\tr\chi)\underline\alpha[W] - \sl\nabla\hat\otimes \underline\beta[W] - (4\underline\eta - \zeta)\hat\otimes \underline\beta[W] - 3\hat{\underline\chi}\rho[W] + 3\prescript{*}{}{\hat{\underline\chi}}\sigma[W] \\
    \sl\nabla_4 \beta[W] &= -2(\tr\chi + \omega) \beta[W] + \sl\div \alpha[W] + \eta\cdot \alpha[W] \\
    \sl\nabla_3 \underline\beta[W] &= - 2(\tr\underline\chi + \underline\omega) \underline\beta[W] - \sl\div\underline\alpha[W] - \underline\eta\cdot \underline\alpha[W] \\
    \sl\nabla_3 \beta[W] &= (2\underline\omega - \tr \underline\chi) \beta[W] + \sl\nabla \rho[W] + \prescript{*}{}{\sl \nabla\sigma[W]}  + 2\hat\chi \cdot \underline\beta[W] + 3(\eta\rho[W] + \prescript{*}{}{\eta}\sigma[W]) \\ 
    \sl\nabla_4 \underline\beta[W] &= (2\omega -\tr\chi)\underline\beta[W] - \sl\nabla\rho[W] + \prescript{*}{}{\sl\nabla\sigma[W]} + 2\hat{\underline\chi}\cdot \beta[W] - 3(\underline\eta \rho[W] - \prescript{*}{}{\underline\eta}\sigma[W]) \\ 
    \sl\nabla_4 \rho[W] &= - \frac{3}{2}\tr\chi \rho[W] + \sl\div \beta[W] + (2\underline\eta + \zeta)\cdot \beta[W] - \frac{1}{2}\hat{\underline\chi}\cdot \alpha[W]  \\ 
    \sl\nabla_3 \rho[W] &= -\frac{3}{2}\tr\underline\chi \rho[W] - \sl\div \underline\beta[W] - (2\eta - \zeta)\cdot \underline\beta[W] -\frac{1}{2}\hat\chi\cdot \underline\alpha[W]  \\ 
    \sl\nabla_4 \sigma[W] &= -\frac{3}{2}\tr\chi \sigma[W] - \sl\curl \beta[W] - (2\underline\eta + \zeta) \cdot \prescript{*}{}{\beta[W]} + \frac{1}{2}\hat{\underline\chi}\cdot \prescript{*}{}{\alpha[W]} \\ 
    \sl\nabla_3 \sigma[W] &= -\frac{3}{2}\tr\underline\chi\sigma[W] - \sl\curl\underline\beta[W] + (\zeta - 2\eta) \cdot \prescript{*}{}{\underline\beta[W]} - \frac{1}{2}\hat\chi \cdot \prescript{*}{}{\underline\alpha[W]}.
\end{split}
\end{equation}

\subsection{Null flows and automorphisms of \texorpdfstring{$S_{0,0}$}{S0,0}}\label{subsec:null_flows} In this section we define the diffeomorphisms that will be essential throughout this paper, and in particular in pulling back the systems of equations we will consider to $S_{0,0}$. 

Define $\Theta_{u,\u}, \underline\Theta_{\u,u} : S_{0,0} \to S_{u,\u}$ by
\begin{equation}\label{defn:Theta_diffeos}
\begin{split}
    \underline\Theta_{\u, u} &= \Phi_{\u}\circ \underline\Phi_u \\
    \Theta_{u, \u} &= \underline\Phi_u \circ \Phi_{\u}.
\end{split}
\end{equation}
In general, due to the presence of nonzero torsion, $\Theta_{u,\u} \neq \underline\Theta_{\u,u}$.
Define also the automorphisms $A_{u, \u}, \underline A_{\u, u} : S_{0,0} \to S_{0,0}$ by 
\begin{equation}\label{defn:A_diffeos}
\begin{split}
    A_{u, \u} &= \underline\Theta_{\underline u, u}^{-1} \circ \Theta_{u, \u} \\ 
    \underline A_{\u,u} &= \Theta_{u,\u}^{-1} \circ \underline\Theta_{\u, u}.
\end{split}
\end{equation}

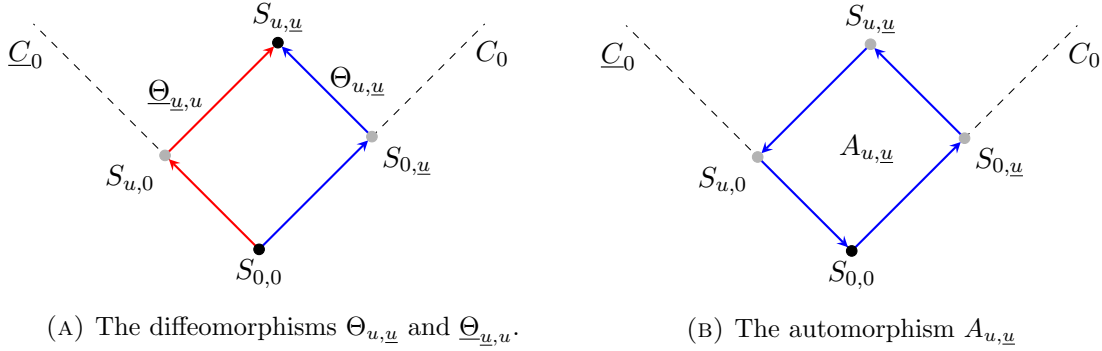
\begin{figure}
    \centering
    \begin{subfigure}{.5\textwidth}
    \begin{tikzpicture}

    \draw[dashed] (0,0) -- (3,3);
    \draw[dashed] (0,0) -- (-3,3);

    \node[below] at (3.1, 2.9) {$C_0$};
    \node[below] at (-3.1, 2.9) {$\underline C_0$};

    \draw[->, color = blue, thick, >=stealth] (0,0) -- (1.5 - .05 , 1.5 - .05);
    \draw[->, color=blue,  >=stealth, thick] (1.5, 1.5) -- (1.5 - 1.25 + .05, 1.5 + 1.25 - .05);
    \node[above right] at (.8, 1.9) {$\Theta_{u,\u}$};
    
    \draw[->, color=red, >=stealth, thick] (0,0) -- (-1.25 + .05, 1.25 - 0.05);
    \draw[->, color=red, >=stealth, thick] (-1.25, 1.25) -- (1.5 - 1.25 - .05, 1.5 + 1.25 - 0.05);
    \node[above left] at (-.6, 1.7) {$\underline\Theta_{\u,u}$};

    \draw[fill=black] (0,0) circle (2pt);
    \node[below] at (0,0) {$S_{0,0}$};
    
    \draw[fill = black] (1.5 - 1.25, 1.5 + 1.25) circle (2pt);
    \node[above] at (1.5 - 1.25, 1.5 + 1.25) {$S_{u,\u}$};
    
    \draw[fill=black!30, draw=black!30] (1.5, 1.5) circle (2pt);
    \node[below right] at (1.5,1.5) {$S_{0,\u}$};

    \draw[fill=black!30, draw=black!30] (-1.25, 1.25) circle (2pt);
    \node[below left] at (-1.25, 1.25) {$S_{u, 0}$};
    
    \end{tikzpicture}
    \caption{The diffeomorphisms $\Theta_{u,\u}$ and $\underline\Theta_{\u,u}$.}
    \label{fig:Theta_diagram}
    \end{subfigure}\hfill
    \begin{subfigure}{.5\textwidth}
    \centering
    \begin{tikzpicture}

    \draw[dashed] (0,0) -- (3,3);
    \draw[dashed] (0,0) -- (-3,3);

    \node[below] at (3.1, 2.9) {$C_0$};
    \node[below] at (-3.1, 2.9) {$\underline C_0$};

    \draw[->, color = blue, thick, >=stealth] (0,0) -- (1.5 - .05 , 1.5 - .05);
    \draw[->, color=blue,  >=stealth, thick] (1.5, 1.5) -- (1.5 - 1.25 + .05, 1.5 + 1.25 - .05);
    
    \draw[->, color=blue, >=stealth, thick] (-1.25, 1.25) -- (-0.05,0.05);
    \draw[->, color=blue, >=stealth, thick] (1.5 - 1.25, 1.5 + 1.25) -- (-1.25 + .05, 1.25 + .05);

    \node[] at (.2, 1.35) {$A_{u,\u}$};

    \draw[fill=black] (0,0) circle (2pt);
    \node[below] at (0,0) {$S_{0,0}$};
    
    \draw[fill = black!30, draw=black!30] (1.5 - 1.25, 1.5 + 1.25) circle (2pt);
    \node[above] at (1.5 - 1.25, 1.5 + 1.25) {$S_{u,\u}$};
    
    \draw[fill=black!30, draw=black!30] (1.5, 1.5) circle (2pt);
    \node[below right] at (1.5,1.5) {$S_{0,\u}$};

    \draw[fill=black!30, draw=black!30] (-1.25, 1.25) circle (2pt);
    \node[below left] at (-1.25, 1.25) {$S_{u, 0}$};
    \end{tikzpicture}
    \caption{The automorphism $A_{u,\u}$}
    \label{fig:A_diagram}
    \end{subfigure}
    
    \caption{Diagrams showing the diffeomorphisms $\Theta_{u,\u}$ and $\underline\Theta_{\u,u}$ and the automorphisms $A_{u,\u}$ and $\underline A_{\u,u}$. Note that $\underline A_{\u, u} = A_{u,\u}^{-1}$, and so can be visualized by reversing the arrows of $A_{u,\u}$.}
\end{figure}

\noindent
Note that 
\begin{align*}
    (A_{u, \u})^{-1} &= \Theta_{u, \u}^{-1} \circ \underline\Theta_{\u, u} = \underline A_{\u, u}.
\end{align*}
For brevity, we will frequently use the notation
\begin{equation}\label{eq:slg_notation}
    \slg_{u,\u} = \Theta_{u,\u}^* \gamma_{u,\u} \qquad \text{and} \qquad \underline{\slg}_{\u,u} = \underline\Theta_{\u,u}^*\gamma_{u,\u}.
\end{equation}
Thus $\slg_{u,\u}, \underline\slg_{\u,u}$ are metrics on the initial sphere $S_{0,0}$. When the particular $u, \u$ are unimportant or understood, we write simply $\slg, \underline\slg$. Given two tensorfields $\theta, \xi$ on the sphere $(S_{0,0}, h)$, where $h$ is any Riemannian metric on $S_{0,0}$,  we will let 
\begin{align*}
    h(\theta, \xi)
\end{align*}
denote the (potentially partial) contraction of $\theta$ with $\xi$. This notation will only be used when the precise form of the contraction is understood by context or not important to the argument. Furthermore, when the specific metric is either understood by context or is unimportant to the argument, we will let 
\begin{align*}
    \theta \cdot \xi
\end{align*}
denote the (potentially partial) contraction of $\theta$ with $\xi$. This notation, too, will only be used when the precise form of the contraction is either understood by context or not important to the argument.

\subsection{Double null hyperbolic systems}\label{subsec:hyperbolic_systems} In this section we discuss the general setup for the first-order linear hyperbolic systems in the double null foliation which we consider in this paper. We find it convenient to work with the Lie derivative formulation of the equations rather than the covariant derivative formulation; the two are related by 
\begin{align*}
    \lie_X \theta_{\mu_1 \ldots \mu_p} &= \nabla_X \theta_{\mu_1 \ldots \mu_p} + \theta_{\nu \mu_2 \ldots \mu_p} \nabla_{\mu_1} X^\nu + \cdots + \theta_{\mu_1 \ldots \mu_{p - 1}\nu} \nabla_{\mu_p} X^\nu
\end{align*}
for $\theta$ a $p$-covariant tensorfield on $M$ and $X$ a vector field on $M$. The reason for our preference of Lie derivatives lies in the method of solving the equations via pullbacks to the initial sphere $S_{0,0}$, which interact well with the Lie derivative; see Section \ref{subsec:pullback_equations_to_S}. 

Fix an integer $N > 0$. Let 
$$\{\Psi^{(i)}\}_{i = 1}^N \qquad \text{and} \qquad \{\underline\Psi^{(i)}\}_{i = 1}^N$$ 
be a collection of $p_i$-covariant $S$ tensorfields and a collection of $\underline p_i$-covariant $S$ tensorfields, respectively. These tensorfields are allowed to take values in $\R^{n_i}$ and $\R^{\underline n_i}$, respectively.\footnote{This is mainly for the convenience of grouping $\rho$ and $\sigma$ into a single $\R^2$-valued function; see \eqref{eq:rho_sigma_equations}.} These will denote the unknowns in our problem.

The setup and notation here is motivated by the hyperbolic structure of the null Bianchi equations. See especially Section 3.1 of \cite{taylor2016stability} for a discussion of the notion of Bianchi pairs and the importance of the anti-adjointness described below. Essentially, we write our system so that $(\Psi^{(i)}, \underline\Psi^{(i)})$ form a Bianchi pair. This is a particular way of pairing equations to obtain hyperbolicity and perform energy estimates; see the discussion after Remark \ref{rem:remark_5}.

By a first-order \emph{geometric} differential operator $\sld$ on $S_{u,\u}$, we mean that for an $S$ tensorfield $\theta$ on $S_{u, \u}$, $\sld\cdot \theta$ is a finite linear combination of contractions of $\sl\nabla\theta$ with $\gamma_{u, \u}$ or $\sl\epsilon_{u, \u}$.
Let $\sld_{\Psi^{(i)}}$ and $\sld_{\underline\Psi^{(i)}}$ denote first-order {geometric} differential operators on the $S_{u, \u}$ which have the property that they are \emph{anti-adjoint} with respect to $L^2(S_{u,\u}, \gamma)$. That is, for any $\underline p_i$-covariant $S$ tensorfield $\theta_1$ and any $p_i$-covariant $S$ tensorfield $\theta_2$, 
\begin{equation}\label{eq:anti-adjoint}
    \int_{S_{u,\u}} (\sl{\mathcal{D}}_{\Psi^{(i)}}\theta_1)\cdot \theta_2\, d\mu_\gamma = -\int_{S_{u,\u}}\theta_1 \cdot (\sl{\mathcal{D}}_{\underline{\Psi}^{(i)}}\theta_2)\, d\mu_\gamma.
\end{equation}
We require $\sl{\mathcal{D}}_{\Psi^{(i)}}$ map $\underline p_i$-covariant $S$ tensorfields to $p_i$-covariant $S$ tensorfields and $\sl{\mathcal{D}}_{\underline\Psi^{(i)}}$ to map $p_i$-covariant $S$ tensorfields to $\underline p_i$-covariant $S$ tensorfields. 

We consider systems of the following form, which we refer to as \emph{double null hyperbolic systems \eqref{eq:general_hyperbolic_system}}: 
\begin{equation}\label{eq:general_hyperbolic_system}\tag{DNH}
\begin{split}
&\begin{cases}
        \underline D \Psi^{(1)} &= \Omega \big(\sl{\mathcal{D}}_{\Psi^{(1)}}\underline\Psi^{(1)} + E^{(1)} \big) \\ 
        &\vdots \\ 
        \underline D \Psi^{(N)} &= \Omega \big(\sl{\mathcal{D}}_{\Psi^{(N)}}\underline\Psi^{(N)} + E^{(N)} \big) \\
\end{cases} \\ 
&\begin{cases}
        D \underline\Psi^{(1)} &= \Omega \big(\sl{\mathcal{D}}_{\underline\Psi^{(1)}}\Psi^{(1)} + \underline E^{(1)} \big) \\ 
        &\vdots \\ 
        D \underline\Psi^{(N)} &= \Omega \big(\sl{\mathcal{D}}_{\underline\Psi^{(N)}}\Psi^{(N)} + \underline E^{(N)} \big).
\end{cases}
\end{split}
\end{equation}
Note that there are $2N$ equations and $2N$ unknowns.
Here, $E^{(i)}$ and $\underline E^{(i)}$ denote ``lower-order'' terms, which we assume are linear combinations of terms of the following form, with coefficients depending only on the metric $\gamma$ and volume form $\sl\epsilon$:
\begin{align*}
    \psi \cdot \Psi^{(j)}, \quad \psi \cdot \underline\Psi^{(j)},
\end{align*}
where $\psi$ denotes an arbitrary Ricci coefficient and $\cdot$ denotes an arbitrary contraction of tensorfields.

\begin{remark}\label{rem:remark_10}
Note that in some applications (e.g. the Bianchi equations), in order to satisfy this anti-adjoint property, it is necessary to consider the action of one of these differential operators only on a linear subspace of the space of $p_i$-covariant tensorfields, for instance (in the case of $\alpha$ and $\underline\alpha$) symmetric tracless 2-covariant tensorfields. In this case additional argument is required to ensure that the unknowns remain within this subspace. This sometimes requires conditions on the lower-order terms $E^{(i)}$ and $\underline E^{(i)}$. This is not the focus of this paper; it will be addressed in future work. Our unknowns are not further restricted within the class of covariant $S$ tensorfields, except in Section \ref{sec:linearized_bianchi_equations}, where the theory of general double null hyperbolic systems developed earlier in the paper is not needed.
\end{remark}

\begin{example}[Bianchi equations I] Let $\rho,\sigma$ denote two scalar functions on $M$ and $\beta$ an $S$ 1-form. The system
\begin{align*}
    D(\rho,\sigma) &= \Omega\big( \sl\div\beta, -\sl\curl\beta \big) \\ 
    \underline D \beta &= \Omega(\sl\nabla\rho + \prescript{*}{}{\sl\nabla\sigma})
\end{align*}
is an example of \eqref{eq:general_hyperbolic_system}, since the operator $\sld_{(\rho,\sigma)} = (\sl\div, -\sl\curl)$ mapping $S$ 1-forms to pairs of scalar functions has as its adjoint $-\sld_\beta$, the operator mapping pairs of functions $(f_1, f_2)$ to the $S$ 1-form $\sl\nabla f_1 + \prescript{*}{}{\sl\nabla} f_2$. This example comes from the equations for $\sl\nabla_4\rho[W],\sl\nabla_4\sigma[W]$, and $\sl\nabla_3\beta[W]$ in the Bianchi equations \eqref{eq:null_Bianchi_equations}.
\end{example}

\begin{example}[Bianchi equations II]
Let $\alpha$ denote a symmetric traceless 2-covariant $S$ tensorfield and $\beta$ an $S$ 1-form. The system 
\begin{align*}
    \underline D \alpha &= \Omega\sl\nabla\hat\otimes \beta \\ 
    D \beta &= \Omega\sl \div\alpha 
\end{align*}
is an example of such a system, since the operator $\sl{\mathcal{D}}_\alpha = \sl\nabla\hat\otimes$ mapping $S$ 1-forms to symmetric traceless 2-covariant $S$ tensorfields has as its adjoint $-\sl\div = -\sl{\mathcal{D}}_\beta$. Note in this case that we need to restrict the unknown $\alpha$ to lie in, rather than the full space of 2-covariant $S$ tensorfields, the space of symmetric tracless 2-covariant $S$ tensorfields (see Remark \ref{rem:remark_10}). This example comes from the equations for $\sl\nabla_3\alpha[W]$ and $\sl\nabla_4\beta$ in the Bianchi equations \eqref{eq:null_Bianchi_equations}.
\end{example}

\begin{example}
An important non-example is the Bianchi equations in a double null foliation. These fail to be a system of the same type as \eqref{eq:general_hyperbolic_system} since they are overdetermined---see Section \ref{sec:linearized_bianchi_equations}. They are however the main inspiration for studying systems of \eqref{eq:general_hyperbolic_system}.
\end{example}

The natural initial value formulation of \eqref{eq:general_hyperbolic_system} is the characteristic initial value problem posed on two intersecting null hypersurfaces. In this paper, as we are concerned with the double null foliation, we pose initial data on $\underline C_0 \cup C_0$. We remark that \emph{initial data} for \eqref{eq:general_hyperbolic_system} consists of 
\begin{align*}
    \{\Psi^{(i)}_0\}_{i = 1}^N \qquad \text{and} \qquad \{\underline\Psi^{(i)}_0\}_{i = 1}^N
\end{align*}
with $\Psi^{(i)}_0$ a $p_i$-covariant $S$ tensorfield on $C_0$ and $\underline\Psi^{(i)}_0$ a $\underline p_i$-covariant $S$ tensorfield on $\underline C_0$. Global existence and uniqueness for these systems is shown in Section \ref{sec:dnh_lwp} (Theorem \ref{thm:GWP_H^1}).

As remarked above, the Bianchi equations in a double null foliation are the primary motivation for studying systems of this form. The hyperbolic structure of the Einstein equations is expressed in the double null foliation by the precise pairing of the principal terms between ``paired'' Bianchi equations; for example,
\begin{align*}
    \underline D \alpha &= \Omega \sl\nabla\hat\otimes \beta & & \text{is paired with} &  D\beta &= \Omega\sl\div\alpha,  \eqcount\label{eq:computation_35}\\ 
    \underline D \beta &= \Omega(\sl\nabla\rho + \prescript{*}{}{\sl\nabla\sigma}) & & \text{is paired with} &  D(\rho,\sigma) &= \big( \Omega\sl\div\beta, -\Omega\sl\curl\beta \big). \eqcount\label{eq:computation_36}
\end{align*} 
This has been known at least since the stability of Minkowski space (see for instance \cites[Proposition~7.3.2]{CK}), and this structure has been used extensively since then \cites{kla_nic2003, chr-bhf, taylor2016stability, bieri2009stability_extension}. We hope that a thorough study of systems \eqref{eq:general_hyperbolic_system} exhibiting this type of hyperbolicity will shed new light on the Einstein equations in a double null foliation and potentially uncover new structures. Already at the linear level we find new constraints on solutions of the linearized Bianchi equations (see Section \ref{subsec:algebraic_constraints}), and we are interested to see if these manifest in the full nonlinear problem when the Bianchi equations are coupled to the null structure equations.

One key difference between \eqref{eq:general_hyperbolic_system} and the linearized Bianchi equations is that in the former, every unknown has precisely one equation it satisfies. In the Bianchi equations, all but two unknowns ($\alpha$ and $\underline\alpha$) satisfy \emph{two} equations, and so the system is overdetermined. Another way of viewing this is that while $\alpha$ and $\underline\alpha$ can be viewed, respectively, as a $\Psi^{(i)}$ and a $\underline\Psi^{(i)}$, the curvature components $\beta, \rho, \sigma, \underline\beta$ can be thought of as \emph{both} $\Psi^{(i)}$ and $\underline\Psi^{(i)}$. For example, since $\beta$ satisfies a propagation equation in the $D$-direction in \eqref{eq:computation_35} above, we would like to think of $\beta$ as one of the $\underline\Psi^{(i)}$. But since $\beta$ satisfies also an equation in the $\underline D$-direction in \eqref{eq:computation_36}, this suggests we should think of $\beta$ instead as one of the $\Psi^{(i)}$.

One strategy to overcome this difficulty is to instead view some of the Bianchi equations as constraints along the null hypersurfaces $C_u$ and $\underline C_{\u}$, instead of viewing them as evolution equations. For instance, if the equation for $D\beta$ in \eqref{eq:computation_35} is viewed as an evolution equation, then the equation for $\underline D \beta$ in \eqref{eq:computation_36} is viewed as a constraint equation along the $\underline C_{\u}$. This is done in Section \ref{sec:linearized_bianchi_equations}. This also motivates the study of \eqref{eq:general_hyperbolic_system} coupled to a set of constraints for some subset of the unknowns. Then one can ask which types of constraints are propagated by \eqref{eq:general_hyperbolic_system}. Some of these ideas are discussed in Section \ref{sec:linearized_bianchi_equations}; we intend to pursue this further in a future paper. 

Finally, we note that while we prove \textit{global} existence and uniqueness for \eqref{eq:general_hyperbolic_system} on manifolds of the form described in Section \ref{subsec:spacetime_and_notation}, in general we do \emph{not} expect global existence for the full Bianchi equations coupled to the null structure equations. Indeed, as this system is nonlinear and captures the full Einstein equations, we only expect local well-posedness. In a general spacetime, one also does not expect the double null foliation itself to be globally defined (for instance due to the existence of focal points). It is then of interest to analyze the maximal development, for instance for singularity or trapped surface formation, stability results, or the analysis of gravitational radiation.

\section{Spherical mollification of tensorfields}\label{sec:spherical_mollification_of_tensorfields}

\subsection{Preliminaries} In order to prove that \eqref{eq:general_hyperbolic_system} is well-posed, we will mollify the system by a family of spherical mollifiers. Doing so allows us to recast \eqref{eq:general_hyperbolic_system} as a Banach space-valued ODE to which theory in Appendix \ref{app:two-var_ode_theory} applies. This idea comes from standard hyperbolic PDE theory (see e.g. \cite[Chapter~16]{taylor_pde_III}), which we briefly discuss here. Let $f : \R \times \R^n \to \R^m$, $b^j : \R \times \R^n \to \R^{m \times m}$ a symmetric $m \times m$ matrix for $j = 1, \ldots, n$, and $\phi_0 : \R^n \to \R^m$. In order to solve the equation 
\begin{align*}
    \d_t \phi + b^j \d_j \phi = f, \qquad \phi|_{t = 0} = \phi_0 \eqcount\label{eq:computation_32}
\end{align*}
for the unknown $\phi : \R \times \R^n \to \R^m$, 
one instead considers the family of equations for unknowns $\phi_\eps$
\begin{align*}
    \d_t \phi_\eps + J^\eps(b^j \d_j(J^\eps \phi_\eps)) = f, \qquad \phi|_{t = 0} = \phi_{0,\eps}, \eqcount\label{eq:computation_33}
\end{align*}
where the operator $J^\eps$ is mollification by an approximation to the identity on $\R^n$. The initial data $\phi_{0, \eps}$ is the mollification $J^\eps \phi_0$. Note that $J^\eps$ is a smoothing operator, i.e. maps $H^k(\R^n; \R^m)$ into $C^\infty(\R^n)$, where $H^k(\R^n; \R^m)$ is the standard $k$th-order $L^2$-based Sobolev space on $\R^n$. By calling
\begin{align*}
    F(t, x, \psi) &\coloneqq f(t, x) - J^\eps (b^j\d_j (J^\eps \psi))(t, x),
\end{align*}
we can view \eqref{eq:computation_33} as the $H^k(\R^n; \R^m)$-valued ODE 
\begin{align*}
    \d_t\phi_\eps(t)(\cdot) &= F(t, \cdot, \phi_\eps(t)(\cdot)),
\end{align*}
where the unknown $\phi_\eps$ is now viewed as a map $\phi_\eps : \R \to H^k(\R^n; \R^m)$. Under appropriate assumptions on $f$ and $b$, Banach space-valued ODE theory then can be applied to obtain a unique solution.

One then extracts the solution $\phi$ as an appropriate limit of $\phi_\eps$ as $\eps \to 0$. Proving uniform bounds (energy estimates) on the $\phi_\eps$ is where hyperbolicity, manifested in this example by the symmetry of the $b^j$, is used. These are usually of the form 
\begin{align*}
    \norm{\phi_\eps(t)}_{H^k(\R^n; \R^m)} &\leq C, \eqcount\label{eq:computation_34}
\end{align*}
which allows a weakly convergent subsequence to be found. The \textit{a priori} estimates which guide the intuition in proving these energy estimates come from the following: if $\phi = (\phi^1, \ldots, \phi^m)$ is a solution of \eqref{eq:computation_32} which decays  sufficiently rapidly at infinity, we can integrate by parts and use the symmetry of $b^j$ to obtain
\begin{align*}
    \Big|\int_{\R^n} \phi \cdot (b^j \d_j \phi) \, dx\Big| &= \Big|- \int_{\R^n} (\d_j b^j)_{ik}\phi^i \phi^k\, dx\Big| \leq \norm{\d b}_{L^\infty_{t, x}}\norm{\phi}_{L^2_x}^2,
\end{align*}
$\d b$ denoting the spatial gradient of $b$. Therefore
\begin{align*}
    \d_t \frac{1}{2} \int_{\R^n} |\phi(t)|^2\, dx &= \int_{\R^n} \phi\cdot \d_t\phi\, dx \\ 
    &= \int_{\R^n} \phi \cdot f - \phi \cdot(b^j\d_j \phi) \, dx  \\ 
    &\leq \norm{\phi}_{L^2_x}\big( \norm{f}_{L^2_x} + \norm{\d b}_{L^\infty_{t,x}}\norm{\phi}_{L^2_x}\big).
\end{align*}
Gr\"onwall's inequality plus control on $f$ and $\d b$ can then be used to obtain an inequality of the form \eqref{eq:computation_34}. One must show that a similar computation can be done for the mollified equation \eqref{eq:computation_33}, using properties of the mollification operators $J^\eps$.

Motivated by the above, we now discuss the appropriate mollification operators on the $S_{u,\u}$. Fukuoka \cite{fukuoka2006smoothing} defined mollification for tensorfields on a Riemannian manifold and proved several fundamental smoothness and approximation properties. In this section, we show that his definitions satisfy the definition of a Friedrichs mollifier. In particular, we extend his results by showing they define a self-adjoint family of operators whose commutator with any spherical first-order differential operator is a uniformly bounded linear operator on $H^1$. We furthermore define two 2-parameter families (parametrized by $(u, \u)$) of Friedrichs mollifiers with these properties (Definition \ref{defn:u_ubar_Friedrichs_mollifiers}).

In this subsection we consider an arbitrary compact Riemannian manifold $(S, h)$. We let $T^{(p, q)}_x S$ denote the space of type $(p, q)$ tensors at $T_x S$, that is tensors which are $p$-contravariant, $q$-covariant. The following definitions are adapted directly from Definition 4.3 in \cite{fukuoka2006smoothing}:

\begin{defn}
Let $(S, h)$ be a complete, closed Riemannian manifold. Let $d\mu$ denote the volume form on $(S, h)$. Let $\delta = \delta(S, h)$ denjote the injectivity radius of $(S, h)$. Define for each $\eps \in (0, \delta)$, a smooth function $\eta_\eps' : S \times S \to \R$ by 
\begin{align*}
    \eta_\eps' (x, y) &= \begin{cases}
        \exp\Big(\frac{1}{\eps^{-2}(\text{dist}_h(x, y))^2 - 1} \Big) & \text{if } \text{dist}_h(x, y) < \eps \\
        0 & \text{otherwise}.
    \end{cases}
\end{align*}
Then define the \emph{standard mollifier} $\eta_\eps : S \times S \to \R$ by 
\begin{align*}
    \eta_\eps(x, y) &= \frac{1}{\int_S \eta_\eps' (x, z)\, d\mu(z)} \eta_\eps'(x, y).
\end{align*} 
\end{defn}
Note that $\eta_\eps$ is smooth and $\eta_\eps(\cdot, y)$ is supported in the closed ball $\bar B(y; \eps)$. Furthermore, for any $y \in S$, 
\begin{align*}
    \int_{S} \eta_\eps(x, y) d\mu(x) &= 1.
\end{align*}

\begin{defn}\label{defn:tau}
For any $x, y \in S$ for which there is a unique minimizing geodesic joining $x$ and $y$, we let $\tau_{x,y} : T_x^{(p, q)} S \to T_y^{(p, q)} S$ denote parallel transport (with respect to $h$) from $x$ to $y$. Note that $\tau_{x,y}^{-1} = \tau_{y,x}$, and these are isometries.
\end{defn}

\begin{defn}
Let $T$ be a tensorfield of type $(p, q)$ on $S$. Define $J_\eps T$ by: for any $x \in S$ and $v_1, \ldots, v_p \in T_x S, \alpha_1, \ldots, \alpha_q \in T_x^* M$, 
\begin{multline}
    (J_\eps T)(x)(v_1, \ldots, v_p, \alpha_1, \ldots, \alpha_q) \\ = \int_S \eta_\eps (x, y) T(y)\big[ \tau_{x,y} v_1, \ldots, \tau_{x,y} v_p, \tau_{x,y}\alpha_1, \ldots, \tau_{x,y}\alpha_q \big]\, d\mu(y).
\end{multline}
\end{defn}

\begin{prop}(\cites[Theorems~4.5,4.6]{fukuoka2006smoothing})
Let $T$ be an $L^1(S, h)$ tensorfield of type $(p, q)$ on $S$. Then for every $\eps \in (0, \delta)$, $J_\eps T$ is a smooth tensorfield on $S$ of the same type as $T$. Also, 
\begin{enumerate}
    \item $J_\eps T \to T$ a.e. as $\eps \to 0$.
    \item If $T$ is continuous, then $J_\eps T \to T$ uniformly as $\eps \to 0$. 
    \item If $1 \leq p < \infty$ and $T \in L^p(S)$, then $J_\eps T \to T$ in $L^p(S)$ as $\eps \to 0$.
\end{enumerate}
\end{prop}

Furthermore, these operators are all bounded as linear operators $L^2(S) \to L^2(S)$, with operator norm independent of $\eps$ \cite[p. 20]{fukuoka2006smoothing}. We now show that the operators $J_\eps$ are self-adjoint in $L^2(S, h)$ and are uniformly bounded operators on $L^2(S, h)$. First we need a preliminary lemma.

\begin{lemma} Let $x, y \in S$ be points joined by a geodesic curve $c$.
Then parallel transport along $c$ commutes with the musical isomorphism; that is,
\begin{align*}
    \tau_{x, y}(T^\sharp) = (\tau_{x, y}T)^\sharp
\end{align*}
for any $T \in T^{(p, q)}_x S$. 

Also, if $T \in \bigotimes_p T_y^* S$ is any $p$-covariant tensor and $v \in \bigotimes_q T_x S$ is any $q$-contravariant tensor, then 
$$
(\tau_{y, x}(T))(v) = T(\tau_{x, y}(v)).
$$
\end{lemma}

\begin{proof}
We write the proof of these statements for $T$ a 1-form; the general proof follows the same lines. We let $D/ds$ denote the covariant derivative operator along $c$. 

Without loss of generality, we have $x = c(0), y = c(1)$. Let $T(s)$ denote the 1-form along $c$ which is the parallel transport of $T \in T_x^* S$, and let $v(s)$ be the parallel transport of $T^\sharp \in T_x S$ along $c$. To prove the first statement, we want to show that: 
\begin{align*}
    h(v(1), w) &= T(1)(w)
\end{align*}
for all $w \in T_y S$. Since parallel transport is surjective, it suffices to show this holds for all $w$ of the form $w = \tau_{x, y}(u), u \in T_x S$. Let $u(s)$ denote the parallel transport of $u$ along $c$. Since parallel transport is an isometry, we have that $h(v(s), u(s)) = h(v, u) = T(u)$. Now, by the fact that the connection commutes with contraction and the Leibniz rule, we have
\begin{align*}
    \frac{d}{ds}(T(s)(u(s))) &= \Big(\frac{D}{ds}T\Big)(u(s)) + T\Big(\frac{D}{ds}u(s)\Big) = 0
\end{align*}
since $T(s)$ and $u(s)$ are parallel transported. Therefore $T(s)(u(s)) = T(u) = h(v(s), u(s))$ for all $s$, in particular for $s = 1$. So $T(1)(u(1)) = h(v(1), u(1))$ for all $u \in T_x S$, proving the first statement.

The second statement follows quickly from the first. Let $T \in T_y^* S$ and $v \in T_x S$. We have, using the fact that parallel transport is an isometry as well as the first statement of this lemma,
\begin{align*}
    T(\tau_{x, y} v) &= h_x(\tau_{x,  y} v, T^{\sharp}) \\ 
    &= h_y(v, \tau_{y,x}(T^\sharp)) \\ 
    &= h_y(v, (\tau_{y, x}T)^\sharp) \\ 
    &= \tau_{y, x}T(v).
\end{align*}
This completes the proof.
\end{proof}

We are now able to prove the main proposition of this section.
\begin{prop}
The operators $J_\eps : L^2(S, h) \to L^2(S, h)$ are self-adjoint; that is, for all $p$-covariant tensorfields $U, T$ on $S$, 
\begin{align*}
    \int_S h(J_\eps U, T)\, d\mu_h &= \int_S h(U, J_\eps T)\, d\mu_h.
\end{align*}
Here, $h(\cdot, \cdot)$ denotes the natural extension of $h$ to a metric on $p$-covariant tensors.
\end{prop}
\begin{proof}
We write the proof of the proposition when $p = 1$; for higher-rank tensorfields the proof is analogous.  We have:
\begingroup
\allowdisplaybreaks
\begin{align*}
    \int_S h_x(J_\eps U, T)\, d\mu(x) &= \int_S (J_\eps U)(T^\sharp)(x)\,d \mu(x) \\ 
    &= \int_S \int_S \eta_\eps(x, y) U(y)[\tau_{x,y}(T^\sharp(x))]\, d\mu(y)\, d\mu(x) \\ 
    &= \int_S \int_S \eta_{\eps}(y, x)  U(y)[(\tau_{x, y}T)^\sharp(y)]d\mu(x) \, d\mu(y) \\ 
    &= \int_S \int_S \eta_{\eps}(y, x) \tau_{y,x}(U(y))(x)[T^\sharp(x)]  d\mu(x) \, d\mu(y) \\ 
    &= \int_S \int_S \eta_\eps (y, x) T(x)[\tau_{y,x}(U(y))^\sharp]d\mu(x) \, d\mu(y) \\ 
    &= \int_S (J_\eps T)(U^\sharp)(y)\, d\mu(y) = \int_S h(U, J_\eps T)\, d\mu.
\end{align*}
\endgroup
Thus $J_\eps$ is self-adjoint on $L^2(S)$.
\end{proof}

\subsection{Friedrichs mollifiers on \texorpdfstring{$S_{0,0}$}{S0,0}}\label{subsec:friedrichs_mollifiers}

In this section, we define two $(u, \u)$-dependent families of Friedrichs mollifiers on the initial sphere $S_{0,0}$ and prove various uniformity conditions on them. 

\begin{defn}\label{defn:delta_0}
Define $\delta_0 = \delta_0(u_*, \u_*)$ by
\begin{equation}
    \delta_0 = \inf_{(u, \u) \in D} \text{inj}(S_{u,\u}, \gamma_{u,\u}),
\end{equation}
where $\text{inj}(S, h)$ denotes the injectivity radius of the Riemannian manifold $(S, h)$.
\end{defn}

\begin{remark}
When $S$ is a compact manifold, it is known that $\text{inj} : R(S) \to \R_{> 0}$ is a continuous function of the Riemannian metric given the $C^2$ topology on $R(S)$ ($R(S)$ denoting the space of Riemannian metrics on $S$); see \cite[Section~8]{ehrlich1974injectivity_radius}. This implies that $\delta_0 > 0$ under appropriate conditions on $(M, g)$. In particular, in our current setting where $g$ is a smooth Lorentzian metric on $M = D \times S^2$, we have $\delta_0 > 0$.
\end{remark}

\begin{defn}\label{defn:u_ubar_Friedrichs_mollifiers}
Let $\delta_0 > 0$ and let $\eps \in (0, \delta_0)$. Let $(u, \u) \in D$. Define the operators 
\begin{align*}
    J^\eps_{u,\u} &: L^2(S_{0,0}, \slg_{u,\u}) \to L^2(S_{0,0}, \slg_{u,\u}) \\
    \underline J^\eps_{\u, u} &: L^2(S_{0,0}, \underline\slg_{\u,u}) \to L^2(S_{0,0}, \underline\slg_{\u,u})
\end{align*}
to be the operators $J_\eps$ for $(S_{0,0}, \slg_{u,\u})$ and $(S_{0,0}, \underline\slg_{\u,u})$, respectively (recall the notation \eqref{eq:slg_notation}). For completeness we explicitly write them down as follows. Define first the standard mollifiers
\begin{align*}
    \eta_{\eps}'(u, \u; x, y) &= \begin{cases} 
        \exp\Big(\frac{1}{\eps^{-2}(\text{dist}_{\slg}(x, y))^2 - 1} \Big) & \text{if } \text{dist}_\slg(x, y) < \eps \\
        0 & \text{otherwise}
        \end{cases} \\ 
    \eta_\eps(u, \u; x, y) &= \frac{1}{\int_S \eta_\eps'(u, \u; x, z)\, d\mu_{\slg}(z)}\eta_\eps'(u, \u; x, y)\\
    \underline\eta_{\eps}'(\u, u; x, y) &= \begin{cases} 
        \exp\Big(\frac{1}{\eps^{-2}(\text{dist}_{\underline\slg}(x, y))^2 - 1} \Big) & \text{if } \text{dist}_{\underline\slg}(x, y) < \eps \\
        0 & \text{otherwise}
        \end{cases} \\ 
    \underline\eta_\eps(\u, u; x, y) &= \frac{1}{\int_S \underline\eta_\eps'(\u, u; x, z)\, d\mu_{\underline\slg}(z)}\underline\eta_\eps'(\u, u; x, y).
\end{align*}
Also, let $\tau^{u,\u}_{x, y}$ denote parallel transport on $(S_{0,0}, \slg)$ from $x$ to $y$ and $\underline\tau_{x, y}^{\u, u}$ the same for $(S_{0,0}, \underline\slg)$.
Then define, for $T$ any $(p, q)$-type tensorfield on $S_{0,0}$, 
\begin{multline}
    (J^\eps_{u,\u} T)(x)(v_1, \ldots, v_p, \alpha_1, \ldots, \alpha_q) \\ = \int_{S_{0,0}} \eta_\eps (u, \u; x, y) T(y)\big[ \tau^{u,\u}_{x, y} v_1, \ldots, \tau^{u,\u}_{x, y} v_p, \tau^{u,\u}_{x, y}\alpha_1, \ldots, \tau^{u,\u}_{x, y}\alpha_q \big]\, d\mu_{\slg}(y).
\end{multline}
and 
\begin{multline}
    (\underline J^\eps_{\u,u} T)(x)(v_1, \ldots, v_p, \alpha_1, \ldots, \alpha_q) \\ = \int_{S_{0,0}} \underline\eta_\eps (\u, u; x, y) T(y)\big[ \underline\tau_{x, y}^{\u, u} v_1, \ldots, \underline\tau_{x, y}^{\u, u} v_p, \underline\tau_{x, y}^{\u, u}\alpha_1, \ldots, \underline\tau_{x, y}^{\u, u}\alpha_q \big]\, d\mu_{\underline \slg}(y).
\end{multline}
\end{defn}
Note that $\text{inj}({S_{0,0}}, \slg) = \text{inj}({S_{0,0}}, \underline\slg)$ since these manifolds are both isometric to $(S_{u,\u}, \gamma_{u,\u})$. Also, both families of operators $\{J^\eps_{u,\u}\}, \{\underline J^\eps_{\u,u}\}$ are families of linear operators which, for a fixed $u$ and $\u$, have the same properties as $J_\eps$ in the previous section (after all, they are just particular instances of such a $J_\eps$). In particular, they are self-adjoint. Our next goal is to prove that they are uniformly bounded. We first prove several helpful lemmas.

\begin{lemma}\label{lem:smoothness_of_tau}
The maps $\tau_{x,y}^{u,\u}$ and $\underline\tau_{x,y}^{\u,u}$ are smooth in $(u, \u) \in D$ as well as $x, y$.
\end{lemma}

\begin{proof}
The smoothness with respect to $x$ and $y$ is proven in \cite[Theorem~2.1]{fukuoka2006smoothing}. The smoothness in $(u, \u)$ follows by a slight modification of this argument where one adds $(u, \u)$ as parameters. Since $g$ is smooth, $\slg$ and $\underline \slg$ depend smoothly on $(u, \u)$, which allows the same inverse function theorem argument to be applied.
\end{proof}

\begin{lemma}\label{lem:div_tau_x_v}
Let $(S, h)$ be a smooth compact Riemannian manifold. 
Let $x \in S, v \in T_x S$, and consider the vector field $\tau_x v$ defined in $B(x; \text{inj}(S,h))$ by 
\begin{align*}
    (\tau_x v)_y &= \tau_{x,y}v,
\end{align*}
where $\tau_{x,y}$ is as defined in Definition \ref{defn:tau}.
For any $y \in B(x; \text{inj}(S,h))$, let $\gamma_{y} : [0, \text{dist}(x, y)]\to B(x; \text{inj}(S,h))$ be the unique minimizing unit-speed geodesic from $x$ to $y$. Then it holds that:
\begin{equation}\label{eq:divergence_formula}
    \sl\div\tau_x v (y) = -h(\dot\gamma_y(0), v)\int_0^{\text{dist}(x,y)} K(\gamma_y(s))\, ds,
\end{equation}
where $\sl\div$ is the divergence operator of $(S, h)$. Note that since the function $y \mapsto \text{dist}(x, y)$ is smooth for $y \in B(x; \text{inj}(S,h))$, this is a smooth function of $y$.
\end{lemma}

\begin{proof}
Let $y \in B(x; \text{inj}(S,h))$ be arbitrary and define $f(t) = \sl\div\tau_x v (\gamma_y(t))$. For purposes of the present computation, let $(e_A)_{A = 1, 2}$ be a frame along $\gamma_y$ with the properties
\begin{align*}
    e_1(t) &= \dot\gamma_y(t), \quad \sl\nabla_{e_1(t)}e_2(t) = 0,
\end{align*}
where $\sl\nabla$ denotes the Levi-Civita connection of $(S, h)$.
Extend these locally to a coordinate frame field, so that for some coordinates $\theta^A$ in a neighborhood of the image of $\gamma_y$, we have $\d_{\theta^A} = e_A$. Denote $V = \tau_x v$. Then we have:
\begin{align*}
    f'(t) &= \frac{d}{dt} d\theta^A(\sl\nabla_A V)|_{\gamma_y(t)} \\ 
    &= -\sl\Gamma_{1B}^A d\theta^B \sl\nabla_A V + R(e_1, e_A)V^A.
\end{align*}
Now, $\sl\Gamma_{1B}^A$ is the coefficient of $e_A$ in the basis expansion of $\sl\nabla_{e_1} e_B$, which is zero along $\gamma_y$ by construction, for any $B$. Meanwhile, expanding the Riemann tensor in terms of the Gauss curvature gives:
\begin{align*}
    f'(t) &=  R(e_1, e_A)V^A|_{\gamma_y(t)} \\ 
    &= K(\gamma_y(t))\big(\delta_1^A h_{AB} - 2h_{1B}\big)V^B \\ 
    &= -K(\gamma_y(t)) h(\tau_x v, e_1)|_{\gamma_y(t)}.
\end{align*}
Since $\tau_x v$ and $e_1$ are parallel transported along $\gamma_y$, we have $h(\tau_x v, e_1)|_{\gamma_y(t)} = h(v, \dot\gamma(0))$. Finally, note that $\sl\nabla_A(\tau_x v)|_{x} = 0$ for any $A$ since $\tau_x v$ is parallel transported along geodesics emanating from $x$. Integrating the previous formula gives the result.
\end{proof}

\begin{remark}
This lemma is useful for, in future work, obtaining more precise bounds on the commutator in terms of the number of derivatives on the Ricci coefficients and Gauss curvature required to bound it.
\end{remark}

We also note that the distance function on the sphere $(S_{0,0}, \slg_{u,\u})$, as well as the distance function on $(S_{0,0}, \underline\slg_{\u,u})$, are smooth in $(u, \u)$ in addition to smoothness on $S_{0,0}$.

\begin{lemma}\label{lem:smoothness_of_dist_on_(u,ubar)}
Let $\dist_{u,\u} : S_{0,0} \times S_{0,0} \to \R_{\geq 0}$ denote the Riemannian distance function with respect to the metric $\slg_{u,\u}$ on $S_{0,0}$. Then for any $x, y \in S_{0,0}$ such that for all $(u, \u) \in D$, $\dist_{u,\u}(x, y) < \delta_0$, the map 
\begin{align*}
    (u, \u) &\mapsto \dist_{u,\u}(x, y)
\end{align*}
is smooth. In addition, the same statement holds for the distance function of $(S_{0,0}, \underline\slg_{\u,u})$.
\end{lemma}

\begin{proof}
Let $\log(u,\u)_x(y)$ denote the unique $v \in T_x S_{0,0}$ such that $\exp(u,\u)_x(v) = y$, where $\exp(u,\u)$ is the Riemannian exponential map of $(S_{0,0}, \slg_{u,\u})$.By a similar argument as in \cite[Theorem 2.1]{fukuoka2006smoothing}, this logarithm is smooth with respect to all parameters. But then, by definition, 
\begin{align*}
    \dist_{u,\u}(x, y) &= |\log(u,\u)_x(y)|_{\slg_{u,\u}}.
\end{align*}
If $x = y$ this is zero for all $(u, \u) \in D$ and hence is smooth in $(u, \u)$. If $x \neq y$, then this is never zero for any $(u, \u) \in D$; hence the composition of the norm with the logarithm is smooth in $(u, \u)$, completing the proof. The case for the distance function of $(S_{0,0}, \underline\slg_{\u,u})$ is proved identically.
\end{proof}

\begin{prop}\label{prop:uniform_bound_1} Let $(M, g)$ be a smooth Lorentzian manifold of the form described in Section \ref{subsec:spacetime_and_notation}.
Let $1 \leq p < \infty$, and let $q \geq 0$ be an integer. Define
\begin{align*}
    \mathfrak{B}^p(q) = \sup_{0 < \eps < \delta_0} \sup_{(u,\u) \in D} \max\bigl\{ \norm{J^\eps_{u,\u}}_{L^p(S_{0,0}, \slg) \to L^p(S_{0,0},\slg)}, \norm{\underline J^\eps_{\u,u}}_{L^p(S_{0,0}, \underline\slg) \to L^p(S_{0,0},\underline\slg)} \bigr\}.
\end{align*}
Then 
\begin{align*}
    \mathfrak{B}^p(q) < \infty.
\end{align*}
\end{prop}

\begin{remark}
When $p = 2$, we write $\mathfrak{B}^2(q) \eqqcolon \mathfrak{B}(q)$.
\end{remark}

\begin{proof}
The positivity of $\delta_0$ follows from continuity of the injectivity radius and the fact that $g$ is $C^2$. We will write the proof of the second part in the case of $T$ a 1-form; tensorfields of other type are handled analogously. As the proof for $\underline J^\eps_{\u,u}$ is identical to the proof for $J^\eps_{u,\u}$, we only write the proof for $J^\eps_{u,\u}$. The approach follows the lines of Theorem 4.6 in \cite{fukuoka2006smoothing}.  

First, fix $x \in {S_{0,0}}$ and $v \in T_x {S_{0,0}}$ of $\slg$-norm at most 1. Then 
\begin{align*}\allowdisplaybreaks
    |(J^\eps_{u,\u} T)(x) (v)| &= \Big| \int_{S_{0,0}} \eta_\eps(u, \u; x, y) T(y)(\tau^{u,\u}_{x, y}v)\, d\mu_{\slg}(y) \Big| \\ 
    &\leq \int_{S_{0,0}} \eta_\eps(u, \u; x, y)^{ 1 - \frac{1}{p}} \eta_\eps(u, \u; x, y)^{1/p}|T(y)|_{\slg(y)}\, d\mu_{\slg}(y) \\ 
    &\leq \Big(\int_{S_{0,0}} \eta_\eps(u, \u; x, y)\, d\mu_{\slg}(y)\Big)^{\frac{p - 1}{p}} \Big(\int_{S_{0,0}} \eta_\eps(u, \u; x, y) |T(y)|_{\slg(y)}^p\, d\mu_{\slg}(y)\Big)^{1/p} \\ 
    &= \Big(\int_{S_{0,0}} \eta_\eps(u, \u; x, y) |T(y)|_{\slg(y)}^p\, d\mu_{\slg}(y)\Big)^{1/p}.
\end{align*}
By duality, taking the supremum over all $v \in T_x {S_{0,0}}$ such that $|v|_{\slg(x)} = 1$ yields the inequality 
\begin{align*}
    |(J^\eps_{u,\u} T)(x)|_{\slg(x)} &\leq \Big(\int_{S_{0,0}} \eta_\eps(u, \u; x, y) |T(y)|_{\slg(y)}^p\, d\mu_{\slg}(y)\Big)^{1/p}. 
\end{align*}
Now we integrate in $x$ and apply Fubini's theorem. For each $u, \u$, let $C^\eta(u, \u)$ be a number such that the $L^\infty({S_{0,0}} \times {S_{0,0}})$-norm of $\eta_\eps(u, \u; \cdot, \cdot)$ times $\eps^2$ is at most $C^\eta(u, \u)$. We compute:
\begingroup
\allowdisplaybreaks
\begin{align*}
    \norm{J^\eps_{u,\u} T}_{L^p({S_{0,0}}, \slg)}^p &= \int_{S_{0,0}} |(J^\eps_{u,\u} T)(x)|_{\slg(x)}^p\, d\mu_{\slg}(x) \\
    &\leq \int_{S_{0,0}} \int_{S_{0,0}} \eta_\eps(u, \u; x, y) |T(y)|_{\slg(y)}^p\, d\mu_{\slg}(y)\, d\mu_{\slg}(x) \\ 
    &= \int_{S_{0,0}} \int_{B(y; \eps)}\eta_\eps(u, \u; x, y) |T(y)|_{\slg(y)}^p\, d\mu_{\slg}(x)\, d\mu_{\slg}(y) \\ 
    &\leq \frac{C^\eta(u, \u)}{\eps^2} \int_{S_{0,0}}|T(y)|_{\slg(y)}^p \int_{B(y; \eps)}\, d\mu_{\slg(x)}\, d\mu_{\slg(y)} \\ 
    &\leq \frac{C^\eta(u, \u)}{\eps^2}\sup_{y \in {S_{0,0}}}\text{Area}(B(y; \eps)) \int_{S_{0,0}} |T(y)|_{\slg(y)}^p\, d\mu_{\slg}(y) \\ 
    &= \frac{C^\eta(u, \u)}{\eps^2}\sup_{y \in {S_{0,0}}}\text{Area}(B(y; \eps)) \norm{T}_{L^p({S_{0,0}}, \slg)}^p.
\end{align*}
\endgroup
It remains to show that $\text{Area}(B(y; \eps)) \lesssim \eps^2$ and $C^\eta(u, \u)$ is uniformly bounded in $(u, \u) \in D$.

Under the assumption $g \in C^2$, the Gaussian curvature of $(S_{u,\u},\gamma_{u,\u})$ is bounded above and below and hence there is a constant $C = C(u_*, \u_*, g)$, such that 
\begin{align*}
    \sup_{(u, \u) \in D}\sup_{y \in S_{u,\u}}\text{Area}(B(y; \eps)) \leq C\eps^2.
\end{align*}
Lastly, by the continuity of $g$ and the compactness of $M$, the numbers $C^\eta(u, \u)$ are all bounded by a constant $C(u_*, \u_*, g)$. Therefore, we obtain 
\begin{align*}
    \norm{J_{u, \u}^\eps T}_{L^p({S_{0,0}}, \slg)}^p &\leq C(u_*, \u_*, g) \norm{T}_{L^p({S_{0,0}}, \slg)}^p.
\end{align*}
This completes the proof.
\end{proof}

Another key property of Friedrichs mollifiers is that they almost commute with differential operators. That is, for a first-order differential operator $D$ on $\R^n$, if $J_\eps$ is a Friedrichs mollifier, then $[D, J_\eps]$ is a bounded linear operator on $L^2(\R^n)$, and the operator norm is bounded by a finite constant independent of $\eps$. We now prove the analogous property in our setting for angular derivatives. 

\begin{prop}\label{prop:mathfrak_C}
Let $(M, g)$ be a smooth Lorentzian manifold of the type described in Section \ref{subsec:spacetime_and_notation}.
Let $p \geq 0$ be an integer. Define 
\begin{multline}\label{eq:mathfrak_C}
    \mathfrak{C}_0(p) = \sup_{0 < \eps < \delta_0/2}\sup_{(u, \u) \in D} \max\bigl\{ \norm{[J^\eps_{u,\u}, \Theta_{u,\u}^*\sl\nabla]}_{L^2_p({S_{0,0}}, \slg) \to L^2_p({S_{0,0}}, \slg)}, \\ \norm{[\underline J^\eps_{u,\u}, \underline\Theta_{\u,u}^*\sl\nabla]}_{L^2_p({S_{0,0}}, \underline\slg) \to L^2_p({S_{0,0}}, \underline\slg)} \bigr\}.
\end{multline}
Then 
\begin{equation}
    \mathfrak{C}_0(p) < \infty.
\end{equation}
\end{prop}

\begin{remark}
The factor of $1/2$ in $\delta_0/2$ is to ensure we are lying compactly within the injectivity radius, so as to ensure smooth functions are everywhere bounded on the area under consideration.
\end{remark}

\begin{proof} Let $T$ be a $p$-covariant tensorfield on ${S_{0,0}}$. In this proof, for brevity, we write $\sl\nabla$ for the connection $\Theta_{u,\u}^*\sl\nabla$ and $\sl\div$ for $\Theta_{u,\u}^*\sl\div$. 

First consider the case $k = 0$. Let $x \in {S_{0,0}}$ and let
\begin{align*}
    v, X_1, \ldots, X_p \in T_x {S_{0,0}}
\end{align*}
have norm 1.
Let $\tilde X_i$ be local extensions of $X_i$ near $x$ with the property that $\sl\nabla_v \tilde X_i|_x = 0$. Also, let $\eta$ be a curve in ${S_{0,0}}$ with $\eta(0) = x$ and $\dot\eta(0) = v$. Then we have:
\begin{multline}\label{eq:computation_13}\allowdisplaybreaks
    \sl\nabla_v(J^\eps_{u,\u} T)(X_1, \ldots, X_p) - J^\eps_{u,\u}(\sl\nabla T)_v(X_1, \ldots, X_p) \\ =  \sum_{i = 1}^p\int_{S_{0,0}}  \eta_\eps(u, \u; x, y) T_y\Big(\tau_{x,y} X_1, \ldots, \frac{d}{dt}\Big|_{t = 0} \tau_{\eta(t),y}\tilde X_i, \ldots,\tau_{x,y}  X_p \Big)\, d\mu_\slg(y) \\ - \int_{S_{0,0}} \eta_\eps(u, \u; x, y)T(\tau_{x,y}X_1, \ldots, \tau_{x,y}X_p)\, d\mu_\slg(y)\cdot \int_{S_{0,0}}(\sl\div\tau_x v)(z) \eta_\eps(u, \u; x, z)\, d\mu_\slg(z) \\ 
    + \int_{S_{0,0}} \eta_\eps(u, \u; x, y)\sl\div(\tau_x v)(y)T(\tau_{x,y} X_1, \ldots, \tau_{x,y} X_p)\, d\mu_\slg(y) \\ 
    + \sum_{i = 1}^p\int_{S_{0,0}} \eta_\eps(u, \u; x,y) T(\tau_{x,y} X_1, \ldots, \sl\nabla_{\tau_{x,y}v}(\tau_x X_i), \ldots, \tau_{x,y} X_p)(y)\, d\mu_\slg(y).
\end{multline}
Note first that 
\begin{multline}
    - \int_{S_{0,0}} \eta_\eps(u, \u; x, y)T(\tau_{x,y}X_1, \ldots, \tau_{x,y}X_p)\, d\mu_\slg(y)\cdot \int_{S_{0,0}}(\sl\div\tau_x v)(z) \eta_\eps(u, \u; x, z)\, d\mu_\slg(z)  \\ 
    = J^\eps_{u,\u}T(x) \int_{S_{0,0}}(\sl\div\tau_x v)(z) \eta_\eps(u, \u; x, z)\, d\mu_\slg(z).
\end{multline}
Note also that $\tau_{x',y'} Y$ is smooth with respect to all parameters (Lemma \ref{lem:smoothness_of_tau}). The integrand in the first line on the right-hand side of \eqref{eq:computation_13} is therefore equal to 
\begin{align*}
    \eta_\eps(u, \u; x, y) T_y(\tau_{x,y} X_1, \ldots, e_A, \ldots, \tau_{x,y} X_p) \Big[\frac{d}{dt}\Big|_{t = 0} \tau_{\eta(t),y}\tilde X_i\Big]^A,
\end{align*}
where $e_A$ is a smooth orthonormal frame field near $y$. This is then bounded by
\begin{align*}
    C\eta_\eps(u, \u; x, y) |T_y|,
\end{align*}
where $C$ is a constant independent of $\eps$ and $(u, \u)$. Note that Lemma \ref{lem:smoothness_of_tau} is used to prove that this constant $C$ is independent of $(u, \u)$. A similar consideration shows that an expression of the same form bounds the integrand in the last line on the right-hand side of \eqref{eq:computation_13}. To bound the middle two lines, we note that 
\begin{align*}
    \abs*{\int_{S_{0,0}} (\sl\div\tau_x v)(z)\eta_\eps(u, \u; x,y)\, d\mu_\slg (z)} &\leq \norm{\eta_\eps}_{L^1({S_{0,0}})}\norm{\sl\div\tau_x v}_{L^\infty(B(x; \delta_0/2))} \leq C
\end{align*}
and, similarly,
\begin{align*}
    \abs*{\int_{S_{0,0}} \eta_\eps(u, \u; x, y)\sl\div(\tau_x v)(y)T(\tau_{x,y} X_1, \ldots, \tau_{x,y} X_p)\, d\mu_\slg(y)} &\leq C|J^\eps_{u,\u}|T||,
\end{align*}
where $C$ is a constant independent of $\eps$ and $(u, \u)$. This shows that the right-hand side of \eqref{eq:computation_13} is bounded in absolute value by
\begin{align*}
    C \big(|J^\eps_{u,\u}|T|| + |J^\eps_{u,\u} T|\big).
\end{align*}
Since this holds for arbitrary $v, X_i$ of norm 1, and since by duality
\begin{align*}
    \big| [J^\eps_{u,\u}, \sl\nabla] T \big| &\leq \sup_{\substack{v, X_i \in T_x {S_{0,0}} \\ |v| = |X_i| = 1}}
    \big| [J^\eps_{u,\u}, \sl\nabla] T (v, X_1, \ldots, X_p) \big|,
\end{align*}
this shows that the norm $\big| [J^\eps_{u,\u}, \sl\nabla] T \big|$ is bounded pointwise at $x \in {S_{0,0}}$ by $C \big(|J^\eps_{u,\u}|T|| + |J^\eps_{u,\u} T|\big)$.
Applying Proposition \ref{prop:uniform_bound_1}, this is bounded in $L^2({S_{0,0}})$ by 
\begin{align*}
    \mathfrak{B}(p)\norm{T}_{L^2({S_{0,0}}, \slg)},
\end{align*}
This shows, therefore, that for all $(u, \u) \in D$ and $0 < \eps < \delta_0/2$,
\begin{align*}
    \norm{[J^\eps_{u,\u}, \sl\nabla]T}_{L^2({S_{0,0}}, \slg)} &\leq C\norm{T}_{L^2({S_{0,0}}, \slg)},
\end{align*}
where $C$ is a constant independent of $\eps$ and $(u, \u)$.
This proves the proposition.
\end{proof}

\begin{prop}\label{prop:mathfrak_C_1}
Let $(M, g)$ be a smooth Lorentzian manifold of the type described in Section \ref{subsec:spacetime_and_notation}.
Let $p \geq 0$ be an integer. Define 
\begin{multline}\label{eq:mathfrak_C_1}
    \mathfrak{C}_1(p) = \sup_{0 < \eps < \delta_0/2}\sup_{(u, \u) \in D} \max\bigl\{ \norm{[J^\eps_{u,\u}, \Theta_{u,\u}^*\sl\nabla]}_{H^1_p({S_{0,0}}, \slg) \to H^1_p({S_{0,0}}, \slg)}, \\ \norm{[\underline J^\eps_{u,\u}, \underline\Theta_{\u,u}^*\sl\nabla]}_{H^1_p({S_{0,0}}, \underline\slg) \to H^1_p({S_{0,0}}, \underline\slg)} \bigr\}.
\end{multline}
Then 
\begin{equation}
    \mathfrak{C}_1(p) < \infty.
\end{equation}
\end{prop}

\begin{remark}
One can extend this result to show that the commutators are bounded as operators $H^k_p \to H^k_p$ for $k > 1$. However, we do not need this result for the purposes of this paper.
\end{remark}

\begin{proof}
We have already derived the expression \eqref{eq:computation_13} for the commutator. Lemmas \ref{lem:smoothness_of_tau}-\ref{lem:smoothness_of_dist_on_(u,ubar)}, as well as applying $\sl\nabla$ to \eqref{eq:computation_13} and repeating the argument of the proof of Proposition \ref{prop:mathfrak_C}, prove this proposition.
\end{proof}

\begin{defn}
For ease of notation, we define 
\begin{align*}
    \mathfrak{B} &\coloneqq \max_{1 \leq i \leq N}\{\mathfrak{B}(p_i), \mathfrak{B}(\underline p_i)\}
\end{align*}
and 
\begin{align*}
    \mathfrak{C} &\coloneqq \max_{1 \leq i \leq N}\{\mathfrak{C}_0(p_i), \mathfrak{C}_0(\underline p_i), \mathfrak{C}_1(p_i), \mathfrak{C}_1(\underline p_i)\}.
\end{align*}
\end{defn}

\begin{remark}
It is of interest for future work to precisely identify how $\mathfrak{B}$ and $\mathfrak{C}$ depend on the spacetime $(M, g)$, and in particular to obtain estimates on these constants depending on $C_0$ and $\sl C_m$ (see Definitions \ref{defn:C_0} and \ref{defn:sl_C_k}).
\end{remark}

\section{Global well-posedness for DNH}\label{sec:dnh_lwp}

\subsection{Setup. Pullbacks of the equations to \texorpdfstring{${S_{0,0}}$}{S0,0}}\label{subsec:pullback_equations_to_S} This section will complete the first half of this paper. We will prove global well-posedness for the system \eqref{eq:general_hyperbolic_system} by recasting it as a two-variable ODE system as discussed in Appendix \ref{app:two-var_ode_theory}. The essential idea is to pull everything back to the initial sphere ${S_{0,0}}$ by the diffeomorphisms $\Theta_{u, \u}$ and $\underline\Theta_{\u, u}$ (see \eqref{defn:Theta_diffeos} and Figure \ref{fig:Theta_diagram}).

Now, at the level of the spacetime geometry, there is no intrinsic way to choose which of these diffeomorphisms to pull back by. Therefore we let the equations guide our choice. That is, for the unknowns $\Psi^{(i)}$, which satisfy propagation equations in the $\underline L$ direction, we pull back by $\Theta_{u, \u}$; and for the unknowns $\underline\Psi^{(i)}$, which satisfy propagation equations in the $L$ direction, we pull back by $\underline\Theta_{\u, u}$. This gives $(u, \u)$-dependent tensorfields on $S_{0,0}$ which we call $\Psi_{(S)}^{(i)}$ and $\underline\Psi^{(i)}_{(S)}$. After solving the system for these unknowns we can then push them forward by $\Theta_{u, \u}$ and $\underline\Theta_{\u, u}$ to obtain the unknowns on the spacetime $M$. The choice of which diffeomorphism to apply to which unknowns comes from the fact that $\underline L$ is $\Theta_{u,\u}$-invariant and $L$ is $\underline\Theta_{\u,u}$-invariant (this can be seen for instance in canonical coordinates, see Section \ref{subsec:canonical_coordinates}).  

In this section we also establish some basic properties regarding the null flows $\Phi_{\u}, \underline\Phi_u$, $\Theta_{u, \u}, \underline\Theta_{\u, u}$, and the automorphisms $A_{u, \u}, \underline A_{\u, u}$. Most of these properties are proven by elementary properties of the pullback (in particular \cite[Proposition~12.36]{lee2013manifolds}). 

We restrict our attention at the moment to smooth $S$ tensorfields on $M$. By pulling back to the initial sphere $S_{0,0}$, these are in 1-1 correspondence with $(u,\u)$-dependent tensorfields on $S_{0,0}$. To be explicit, if $\Psi$ is a smooth $S$ tensorfield on $M$, then 
\begin{align*}
    \Psi_{(S)}[u, \u] \coloneqq \Theta_{u,\u}^*(\Psi|_{S_{u,\u}})
\end{align*}
is a smooth $(u, \u)$-dependent tensorfield on $S_{0,0}$ which is smooth in $(u, \u) \in D$. Furthermore, the map
\begin{align*}
    \bigl\{ \text{\small smooth $S$ tensorfields on $M$} \bigr\} &\to \bigl\{ \substack{ \text{smoothly $(u,\u)$-dependent} \\ \text{smooth tensorfields on $S_{0,0}$}  }  \bigr\} \\ 
    \Psi &\mapsto \Psi_{(S)}
\end{align*}
is a bijection. So too is the map $\Psi \mapsto \underline\Psi_{(S)}$ defined by 
\begin{align*}
    \underline\Psi_{(S)}[u,\u] &\coloneqq \underline\Theta_{\u,u}^*(\Psi|_{S_{u,\u}}).
\end{align*}

\begin{defn}\label{defn:pullback_differential_operators}
Let $\sl{\mathcal{D}}$ denote a geometric first-order differential operator\footnote{For applications to the Bianchi equations, one is primarily concerned with $\sld \in \{ \sl\nabla, \sl\nabla\hat\otimes, \sl\div, \sl\curl \}$.} on $(S_{u, \u}, \gamma_{u, \u})$; that is, for a tensorfield $\theta$ on $S_{u, \u}$, $\sld\cdot \theta$ is a finite linear combination of contractions of $\sl\nabla\theta$ with $\gamma_{u, \u}$ or $\sl\epsilon_{u, \u}$. Given a diffeomorphism $\Theta : S_{0,0} \to S_{u,\u}$, define the operator $\Theta^*\sl{\mathcal{D}}$ to be the corresponding operator on the Riemannian manifold $(S_{0,0}, \Theta^*\gamma_{u,\u})$. For instance, $\Theta_{u,\u}^*\sl\nabla$ is the Levi-Civita connection of the manifold $(S_{0,0}, \slg)$ and $\Theta_{u, \u}^*\sl\div$ is the divergence operator of $(S_{0,0}, \slg)$. 
\end{defn}

The following lemmas are fundamental, The first is a sort of ``almost-commutation'' statement of the automorphisms $A_{u, \u}$ and $\underline A_{\u, u}$ with a given differential operator $\sld$. Note that upon interchanging the order of the automorphism $A_{u,\u}^*$ or $\underline A_{\u,u}^*$ with the differential operator $\Theta_{u,\u}^*\sl{\mathcal{D}}$, the differential operator is conjugated, becoming $\underline\Theta_{\u,u}^*\sld$. This property of the $A_{u,\u}, \underline A_{\u,u}$ is essential to obtaining the proper structure when pulling back the equations \eqref{eq:general_hyperbolic_system} to $S_{0,0}$.

\begin{lemma}\label{lem:spherical_derivative_almost_commutation}
For any covariant tensorfield $\Psi$ on $S_{0,0}$, we have
\begin{equation}
\begin{split}
    A_{u,\u}^*([\underline\Theta_{\u, u}^*\sl{\mathcal{D}}]\Psi) &= [\Theta_{u,\u}^*\sl{\mathcal{D}}](A_{u,\u}^* \Psi) \\
    \underline A_{\u, u}^*([\Theta_{u,\u}^*\sl{\mathcal{D}}]\Psi) &= [\underline\Theta_{\u,u}^*\sl{\mathcal{D}}](\underline A_{\u,u}^*\Psi).
\end{split}
\end{equation}\end{lemma}

\begin{proof}
We prove the first formula; the second follows by conjugation. We have:
\begin{align*}\allowdisplaybreaks
    A_{u,\u}^* &= \Theta_{u,\u}^* (\underline\Theta_{\u,u}^{-1})^*, \\
    \underline A_{\u,u}^* &= \underline\Theta_{\u,u}^*(\Theta_{u,\u}^{-1})^*, \\
    \Psi &= (\underline A_{\u,u} \circ A_{u,\u})^*\Psi \\
    &= A_{u,\u}^* \underline A_{\u, u}^*\Psi \\ 
    &= \Theta_{u,\u}^* (\underline\Theta_{\u,u}^{-1})^*[\underline A_{\u, u}^*\Psi].
\end{align*}
Now, by the definition of $\underline\Theta_{\u,u}^*\sl{\mathcal{D}}$ as the operator $\sl{\mathcal{D}}$ on $S_{0,0}$ defined with respect to the pullback metric $\underline\Theta_{\u,u}^*\gamma_{u,\u}$, we have
\begin{align*}
    [\underline\Theta_{\u,u}^*\sl{\mathcal{D}}]\Psi &= \underline\Theta_{\u,u}^*(\sl{\mathcal{D}}_{S_{u,\u}}[(\underline\Theta_{\u,u}^{-1})^*\Psi]) \\ 
    &= \underline\Theta_{\u,u}^*(\Theta_{u,\u}^{-1})^*\Theta_{u,\u}^*(\sl{\mathcal{D}}_{S_{u,\u}}[(\underline\Theta_{\u,u}^{-1})^*\Psi]) \\ 
    &= \underline\Theta_{\u,u}^*(\Theta_{u,\u}^{-1})^*[(\Theta_{u,\u}^* \sl{\mathcal{D}})(\Theta_{u,\u}^*(\underline\Theta_{\u,u}^{-1})^*\Psi)] \\ 
    &= \underline A_{\u,u}^*[[\Theta_{u,\u}^*\sl{\mathcal{D}}](A_{u,\u}^*\Psi)]
\end{align*}
Since $A_{u,\u} = \underline A_{\u, u}^{-1}$, if we apply $A_{u,\u}^*$ to both sides we obtain
\begin{align*}
    A_{u,\u}^*([\underline\Theta_{\u,u}^*\sl{\mathcal{D}}]\Psi) &= [\Theta_{u,\u}^*\sl{\mathcal{D}}](A_{u,\u}^*\Psi)
\end{align*}
as desired.
\end{proof}

The next lemma is used in pulling back quantities appearing in \eqref{eq:general_hyperbolic_system} along the ``wrong'' flow. By this we mean the following. In the equation for $\underline D\Psi$ appear terms of the form $\psi\cdot\underline\Psi$. Since this equation is a propagation equation in the incoming direction $\underline L$, we will pull this equation back by $\Theta_{u,\u}$, as $\underline L$ is invariant under this diffeomorphism. However, for $\underline\Psi$, appearing on the right-hand side, it is more natural to pull back by $\underline\Theta_{\u,u}$. One can think of the $A_{u,\u}^*$ as ``correcting'' this mismatch. 

\begin{lemma}\label{lem:lower_order_pullback}
For any $(u,\u)$-dependent $p$-covariant tensorfield $\Psi_{(S)}$ on $S_{0,0}$, let $\Psi = (\Theta_{u, \u}^{-1})^*\Psi_{(S)}$ and $\underline\Psi = (\underline\Theta_{\u,u}^{-1})^*\Psi_{(S)}$. Then we have:
\begin{align*}
    \Theta_{u,\u}^*(\gamma_{u,\u}(\psi,\underline\Psi)) &= \slg (\Theta_{u,\u}^*\psi, A_{u,\u}^*\Psi_{(S)}) \\ 
    \underline\Theta_{\u,u}^*(\gamma_{u,\u}(\psi, \Psi)) &= \underline\slg (\underline\Theta_{\u,u}^*\psi, \underline A_{\u,u}^*\Psi_{(S)}).
\end{align*}
\end{lemma}

\begin{proof}
This follows from the commutation of the pullback with contraction. For concreteness we write the proof when $\psi$ and $\underline\Psi$ are 1-forms. Then:
\begin{align*}
    \Theta_{u,\u}^*(\gamma_{u,\u}(\psi,\underline\Psi)) &= \Theta_{u,\u}^* (\gamma^{BC}\psi_B \underline\Psi_C) \\ 
    &= (\Theta_{u,\u}^*\gamma)^{BC} (\Theta_{u,\u}^*\psi)_B (\Theta_{u,\u}^*\underline\Psi)_C \\ 
    &= \slg^{BC} (\Theta_{u,\u}^*\psi)_{B}(\Theta_{u,\u}^*(\underline\Theta_{\u,u}^{-1})^*\underline\Theta_{\u,u}^*\underline\Psi)_C \\ 
    &= \slg^{BC}(\Theta_{u,\u}^*\psi)_B (A_{u,\u}^*\Psi_{(S)})_C.
\end{align*}
This completes the proof.
\end{proof}

The next lemma concerns the pullbacks of the Lie derivatives $D$ and $\underline D$ appearing in \eqref{eq:general_hyperbolic_system} and relates these to the Fr\'echet derivatives of $(u,\u)$-dependent tensorfields on $S_{0,0}$. It also records natural consequences of the pullback on spherical covariant derivatives, as well as on contractions of $S$ tensorfields on $M$.

\begin{lemma}\label{lem:a_lemma_1}
For any $(u,\u)$-dependent $p$-covariant tensorfield $\Psi_{(S)}$ on $S_{0,0}$, let $\Psi = (\Theta_{u, \u}^{-1})^*\Psi_{(S)}$ and $\underline\Psi = (\underline\Theta_{\u,u}^{-1})^*\Psi_{(S)}$. Then we have
\begin{align*}
    \d_u \Psi_{(S)}[u, \u] &= \Theta_{u, \u}^*(\underline D \Psi) \\ 
    \d_{\u}\underline\Psi_{(S)}[u, \u] &= \underline\Theta_{\u,u}^*(D\underline\Psi).
\end{align*}
Also, 
\begin{align*}
    (\Theta_{u, \u}^*\sl\nabla)\Psi_{(S)}[u, \u] &= \Theta_{u,\u}^*(\sl\nabla\Psi) \\ 
    (\underline\Theta_{\u, u}^*\sl\nabla)\underline\Psi_{(S)}[u, \u] &= \underline\Theta_{\u,u}^*(\sl\nabla\underline\Psi).
\end{align*}
Also, let $\psi$ denote an arbitrary Ricci coefficient. Recall that we let $\gamma_{u, \u}(\psi, \theta)$ denote the (partial) contraction of $\psi$ and $\theta$, with the convention that if $\psi$ or $\theta$ is a scalar, this is ordinary multiplication. Similarly for $\slg(\psi, \theta)$ and $\underline\slg(\psi, \theta)$. Then we have: 
\begin{align*}
    \Theta_{u, \u}^*(\gamma_{u, \u}(\psi, \Psi)) &= \slg\big((\Theta_{u, \u}^*\psi), \Psi_{(S)}[u, \u]\big) \\ 
    \underline\Theta_{\u, u}^*(\gamma_{u,\u}(\psi, \underline\Psi)) &= \slg\big((\underline\Theta_{\u,u}^*\psi), \underline\Psi_{(S)} [u, \u]\big).
\end{align*}
\end{lemma}

\begin{remark}
The first pair of equations in this lemma essentially states that we pulled back by the ``correct'' diffeomorphisms. 
\end{remark}

\begin{proof} (c.f. \cite[Proposition~12.36]{lee2013manifolds}) By properties of pullbacks and Lie derivatives, we have:
\begin{align*}\allowdisplaybreaks
    \d_u \Psi_{(S)}[u, \u] &= \d_u (\Theta_{u,\u}^*\Psi) \\ 
    &= \d_u (\Phi_{\u}^* \underline\Phi_u^*\Psi) \\ 
    &= \Phi_{\u}^*\underline\Phi_u^*(\underline D\Psi) \\
    &= \Theta_{u,\u}^*(\underline D\Psi).
\end{align*}
The second formula follows by conjugation. 

The third formula is essentially expressing the covariance of $\sl\nabla$. We have:
\begin{align*}
    (\Theta_{u, \u}^*\sl\nabla)\Psi_{(S)}[u, \u] &= (\Theta_{u, \u}^*\sl\nabla)(\Theta_{u,\u}^*\Psi) \\ 
    &= \Theta_{u,\u}^*(\sl\nabla\Psi).
\end{align*}
The fourth formula follows by conjugation. 

The fifth and six formulae are immediate consequences of the fact that pullback commutes with contraction.
\end{proof}

We now discuss in more detail the structure of the lower-order terms $E^{(i)}$ and $\underline E^{(i)}$. The need to do so arises because in each equation there may occur terms involving $\Psi^{(i)}$ \emph{and} terms which involve $\underline\Psi^{(j)}$. As discussed above, we wish to pull the former back by $\Theta_{u, \u}$, and the latter back by $\underline\Theta_{\u,u}$. Therefore we will analyze such terms differently. We write
\begin{align*}
    E^{(i)} &= E^{(i)}[\Psi] + E^{(i)}[\underline\Psi],
\end{align*}
where $E^{(i)}[\Psi]$ denotes the terms in $E^{(i)}$ which are of the form $\psi \cdot \Psi^{(j)}$ for some $j$, and $E^{(i)}[\underline\Psi]$ denotes those terms which are of the form $\psi \cdot\underline\Psi^{(j)}$ for some $j$, where $\psi$ denotes an arbitrary Ricci coefficient. Similarly we write
\begin{align*}
    \underline E^{(i)} &= \underline E^{(i)}[\Psi] + \underline E^{(i)}[\underline\Psi].
\end{align*}
As discussed above, the automorphisms $A_{u,\u}$ and $\underline A_{\u, u}$ allow us to reconcile the fact that there appear terms $\underline\Psi^{(j)}$ (which we wish to pull back by $\underline\Theta_{\u,u}$) in the equation for $\underline D \Psi^{(i)}$ (which we wish to pull back by $\Theta_{u,\u}$). 

Pulling back the equation for $\underline D\Psi^{(i)}$ in \eqref{eq:general_hyperbolic_system} by $\Theta_{u,\u}$ and applying the above lemmas gives the following $\d_u$ propagation equation for $(u, \u)$-dependent tensorfields $\Psi_{(S)}^{(i)}$ on $S_{0,0}$. As above, we write $\PsiSi = \Theta_{u,\u}^*\Psi^{(i)}$ and $\PsiBarSi = \underline\Theta_{\u,u}^*\underline\Psi^{(i)}$. The unknown $\PsiSi$ satisfies the equation
\begin{align*}
    \d_u \Psi_{(S)}^{(i)}[u,\u] &= \Theta_{u,\u}^*\Omega \cdot \Big(A_{u,\u}^* \big(\big[\underline\Theta_{\u,u}^*\sl{\mathcal{D}}_{\Psi^{(i)}}\big]\underline\Psi_{(S)}^{(i)}[u,\u]\big) + E^{(i)}_\Theta[\Psi_{(S)}] + E^{(i)}_\Theta [A_{u,\u}^* \underline\Psi_{(S)}] \Big),
\end{align*}
where the lower-order terms have the following form:
\begin{itemize}
	\item $E_\Theta^{(i)}[\Psi_{(S)}]$ is a sum of terms of the form $\slg(\Theta_{u,\u}^*\psi , \Psi_{(S)}^{(j)} )$ with coefficients depending only on $\sl\epsilon$ and $\slg$
	\item $E_\Theta^{(i)}[A_{u,\u}^*\underline\Psi_{(S)}]$ is a sum of terms of the form $\slg(\Theta_{u,\u}^*\psi , A_{u,\u}^*\underline\Psi_{(S)}^{(j)} )$ with coefficients depending only on $\sl\epsilon$ and $\slg$.
\end{itemize}
Similarly, the unknown $\PsiBarSi$ satisfies the equation 
\begin{align*}
    \d_{\u}\PsiBarSi[u,\u] &= \underline\Theta_{\u,u}^*\Omega \cdot \Big( \underline A_{\u,u}^* \big(\big[\Theta_{u,\u}^*\sld_{\underline\Psi^{(i)}}\big]\PsiSi[u,\u]\big) + \underline E_{\underline\Theta}^{(i)}[\underline A_{\u,u}^*\Psi_{(S)}] + \underline E^{(i)}_{\underline\Theta}[\underline\Psi_{(S)}]\Big),
\end{align*}
where the lower-order terms have the following form:
\begin{itemize}
	\item $\underline E_{\underline\Theta}^{(i)}[\underline\Psi_{(S)}]$ is a sum of terms of the form $\underline\slg(\underline\Theta_{\u,u}^*\psi , \underline\Psi_{(S)}^{(j)} )$  with coefficients  depending only on $\sl\epsilon$ and $\slg$
	\item $\underline E_{\underline\Theta}^{(i)}[\underline A_{\u,u}^* \Psi_{(S)}]$ is a sum of terms of the form $\underline\slg(\underline\Theta_{\u,u}^*\psi , \underline A_{\u,u}^* \Psi_{(S)}^{(j)} )$  with coefficients depending only on $\sl\epsilon$ and $\slg$.
\end{itemize}
We thus are led to consider the following equivalent system of equations for $2N$ $(u,\u)$-dependent tensorfields $\{\PsiSi\}_{i = 1}^N$ and $\{\PsiBarSi\}_{i = 1}^N$ on $S_{0,0}$:
\begin{equation}\label{eq:general_hyperbolic_system_S}
\begin{split}
&\begin{cases}
        \d_u \Psi_{(S)}^{(1)}&= \Theta_{u,\u}^*\Omega \cdot \Big(A_{u,\u}^* \big(\big[\underline\Theta_{\u,u}^*\sl{\mathcal{D}}_{\Psi^{(1)}}\big]\underline\Psi_{(S)}^{(1)}\big) + E^{(1)}_\Theta[\Psi_{(S)}] + E^{(1)}_\Theta [A_{u,\u}^* \underline\Psi_{(S)}] \Big)\\
        &\vdots \\ 
        \d_u \Psi_{(S)}^{(N)}&= \Theta_{u,\u}^*\Omega \cdot \Big(A_{u,\u}^* \big(\big[\underline\Theta_{\u,u}^*\sl{\mathcal{D}}_{\Psi^{(N)}}\big]\underline\Psi_{(S)}^{(N)}\big) + E^{(N)}_\Theta[\Psi_{(S)}] + E^{(N)}_\Theta [A_{u,\u}^* \underline\Psi_{(S)}] \Big)\\
\end{cases} \\ 
&\begin{cases}
        \d_{\u}\underline\Psi_{(S)}^{(1)} &= \underline\Theta_{\u,u}^*\Omega \cdot \Big( \underline A_{\u,u}^* \big(\big[\Theta_{u,\u}^*\sld_{\underline\Psi^{(1)}}\big]\Psi_{(S)}^{(1)}\big) + \underline E_{\underline\Theta}^{(1)}[\underline A_{\u,u}^*\Psi_{(S)}] + \underline E^{(1)}_{\underline\Theta}[\underline\Psi_{(S)}]\Big) \\
        &\vdots \\ 
        \d_{\u}\underline\Psi_{(S)}^{(N)} &= \underline\Theta_{\u,u}^*\Omega \cdot \Big( \underline A_{\u,u}^* \big(\big[\Theta_{u,\u}^*\sld_{\underline\Psi^{(N)}}\big]\Psi_{(S)}^{(N)}\big) + \underline E_{\underline\Theta}^{(N)}[\underline A_{\u,u}^*\Psi_{(S)}] + \underline E^{(N)}_{\underline\Theta}[\underline\Psi_{(S)}]\Big).
\end{cases}
\end{split}
\end{equation}
This system is the one for which we will prove well-posedness. We note that for brevity, we have omitted writing the argument $[u, \u]$ to $\PsiSi$ and $\PsiBarSi$. We note that \eqref{eq:general_hyperbolic_system_S} is the pullback of \eqref{eq:general_hyperbolic_system} to $S_{0,0}$. To obtain the original system \eqref{eq:general_hyperbolic_system}, one pushes forward the $\d_u\Psi^{(i)}_{(S)}$ equations by $\Theta_{u,\u}$ and the $\d_{\u}\underline\Psi^{(i)}_{(S)}$ equations by $\underline\Theta_{\u,u}$.

\subsection{Energy estimates I}\label{subsec:energy_estimates_1}  In this section we define the relevant energies and derive energy estimates for the hyperbolic system \eqref{eq:general_hyperbolic_system_S} on $S_{0,0}$. The strategy of these estimates is standard (see for instance \cite[Sections~12.4-12.5]{chr-bhf} or \cite[Sections~1.10 and 9]{taylor2016stability}). We mainly perform these computations to provide an example of the energy estimates we perform later for the mollified system (see Section \ref{subsec:mollified_system}).

\begin{defn}\label{defn:energies} For $\Psi_{(S)}$ a $(u,\u)$-dependent tensorfield on ${S_{0,0}}$, define the energies 
\begin{equation}
\begin{split}
    \mathcal{E}[\Psi_{(S)}](u,\u) &= \frac{1}{2}\int_{S_{0,0}} |\Psi_{(S)}[u, \u]|^2_{\slg} \, d\mu_{\slg}  \\ 
    \underline{\mathcal{E}}[\Psi_{(S)}](u,\u) &= \frac{1}{2}\int_{S_{0,0}} |\Psi_{(S)}[u, \u]|^2_{\underline\slg} \, d\mu_{\underline\slg}.
\end{split}
\end{equation}
\end{defn}

\begin{remark}
We will use the $\mathcal{E}$ for the $\PsiSi$ and $\underline{\mathcal{E}}$ for the $\PsiBarSi$.
\end{remark}

\begin{defn}\label{defn:F_quantities}
For $\Psi_{(S)}$ a $(u,\u)$-dependent tensorfield on ${S_{0,0}}$, define the energies
\begin{align*}
    \mathcal{F}[\Psi_{(S)}](u, \u) &= \int_0^{\u} \mathcal{E}[\Psi_{(S)}](u, \u')\, d\u' \\ 
    \underline{\mathcal{F}}[\Psi_{(S)}](u, \u) &= \int_0^u \underline{\mathcal{E}}[\Psi_{(S)}](u', \u)\, du'.
\end{align*}
We will also write $\mathcal{F}^*[\Psi_{(S)}](u)$ to denote the quantity $\mathcal{F}[\Psi_{(S)}](u, \u_*)$, and similarly for $\underline{\mathcal{F}}^*[\Psi_{(S)}](\u)$. Finally, define the initial data energies
\begin{align*}
    \mathcal{F}^*_0[\Psi_{(S)}] &= \mathcal{F}^*[\Psi_{(S)}](0) \\ 
    \underline{\mathcal{F}}_0^*[\Psi_{(S)}] &= \underline{\mathcal{F}}^*[\Psi_{(S)}](0).
\end{align*}
\end{defn}

\begin{defn}\label{defn:C_0}
Let $C_0$ be a constant such that 
\begin{align*}
    \sup_{(u, \u) \in D}\norm{\Omega}_{L^\infty(S_{u,\u})} + \max\big(\norm{\Omega}_{L^\infty(S_{u,\u})}, 1\big)\sup_{\psi \in \Gamma}\norm{\psi}_{L^\infty(S_{u,\u})} \leq C_0.
\end{align*}
\end{defn}

\begin{remark}\label{rem:cauchy-schwarz_basic}
The following facts, which can be proved by applying the Cauchy-Schwarz inequality and the pullback-invariance of the integral, are useful in bounding most of the lower-order expressions in the energy estimates. 
\begin{itemize}
    \item Terms of the form {\small $\slg\big( \Theta_{u,\u}^* \psi, \PsiSi\big)$} are bounded in {\small $L^2({S_{0,0}}, \slg)$} by {\small $C_0 \norm{\PsiSi}_{L^2(S_{0,0}, \slg)}$}.
    \item Terms of the form {\small $\underline\slg \big( \underline\Theta_{\u,u}^*\psi, \PsiBarSi\big)$} are bounded in {\small $L^2({S_{0,0}}, \underline\slg)$} by {\small $C_0 \norm{\PsiBarSi}_{L^2(S_{0,0}, \underline\slg)}$}.
    \item Terms of the form {\small $\slg\big( \Theta_{u,\u}^*\psi, A_{u,\u}^*\PsiBarSi \big)$} are bounded in {\small $L^2(S_{0,0}, \slg)$} by {\small $C_0\norm{\PsiBarSi}_{L^2({S_{0,0}}, \underline\slg)}$}.
    \item Terms of the form {\small $\underline\slg \big( \underline\Theta_{\u,u}^*\psi, \underline A_{\u,u}^*\PsiSi \big)$} are bounded in {\small $L^2({S_{0,0}}, \underline\slg)$} by  {\small $C_0 \norm{\PsiSi}_{L^2(S_{0,0}, \slg)}$}.
\end{itemize}
Note that in the latter two items, the metric with respect to which the $L^2$-norm is taken changes during the bounding. This is because, by the pullback-invariance of the integral and since $\underline A_{\u,u} = A_{u,\u}^{-1}$,
\begin{align*}
    \int_{S_{0,0}} A_{u,\u}^* f\, d\mu_{\slg} &= \int_{S_{0,0}} f \, d\mu_{\underline A_{\u,u}^*\slg} = \int_{S_{0,0}} f\, d\mu_{\underline\Theta_{\u,u}^*\gamma_{u,\u}} = \int_{S_{0,0}} f\, d\mu_{\underline\slg}.
\end{align*}
This latter property is fundamental and so we record it as a lemma.
\end{remark}

\begin{lemma}\label{lem:fundamental_integration_property_of_A}
For any scalar function $f$ on ${S_{0,0}}$,
\begin{equation}
    \int_{S_{0,0}} A_{u,\u}^*f \, d\mu_{\slg} = \int_{S_{0,0}} f\, d\mu_{\underline\slg} 
\end{equation}
and
\begin{equation}
    \int_{S_{0,0}} \underline A_{\u,u}^*f\, d\mu_{\underline\slg} = \int_{S_{0,0}} f\, d\mu_\slg.
\end{equation}
\end{lemma}
We now state the basic energy estimate for \eqref{eq:general_hyperbolic_system_S}:

\begin{prop}\label{prop:basic_energy_estimates_1} There exists a constant $C = C(C_0)$ (depending in a continuous way on $C_0$) such that for all $(u, \u) \in D$, 
\begin{equation}
    \sum_{i = 1}^N \Big(\d_u \cale[\PsiSi](u, \u) + \d_{\u}\calebar[\PsiBarSi](u, \u)\Big) \leq C \sum_{i = 1}^N \Big(  \mathcal{E}[\Psi_{(S)}^{(i)}](u, \u) +  \calebar[\underline\Psi_{(S)}^{(i)}](u, \u)\Big).
\end{equation}
As a consequence, there exists a (potentially different) constant $C = C(C_0)$ (depending in a continuous way on $C_0$) such that, for all $(u, \u) \in D$, 
\begin{equation}
    \sum_{i = 1}^N \Big(\mathcal{F}^*[\PsiSi](u) + \underline{\mathcal{F}}^*[\PsiBarSi](\u)\Big) \leq C \sum_{i = 1}^N \Big(\mathcal{F}_0^*[\PsiSi] + \underline{\mathcal{F}}_0^*[\PsiBarSi]\Big).
\end{equation}
\end{prop}

\begin{proof}
The derivative of the energies is computed as follows. For brevity, we omit the argument $[u, \u]$ in $\Psi^{(i)}_{(S)}$ and $\underline \Psi^{(i)}_{(S)}$. We have:
\begin{align*}
    \d_u \mathcal{E}[\PsiSi](u, \u) &= \int_{S_{0,0}} \Big(\Psi_{(S)}^{(i)}\cdot \d_u \PsiSi + \frac{1}{2}|\PsiSi|^2\Theta_{u,\u}^*(\Omega\tr\underline\chi)\Big)\, d\mu_{\slg} \\ 
    &= \int_{S_{0,0}} (\Theta_{u,\u}^*\Omega)\Big[\PsiSi\cdot \big(A_{u,\u}^*([\underline\Theta_{\u,u}^*\sld_{\Psi^{(i)}}]\PsiBarSi)\big) + \PsiSi\cdot E_\Theta^{(1)}[\Psi_{(S)}] \\ 
    &\hspace{35mm} + \PsiSi\cdot E_\Theta^{(1)}[A_{u,\u}^*\underline\Psi_{(S)}] + \frac{1}{2}|\PsiSi|^2\Theta_{u,\u}^*(\tr\underline\chi)\Big]\, d\mu_{\slg}.
\end{align*}
The terms $\Theta_{u,\u}^*\Omega$ and $\Theta_{u,\u}^*\tr\underline\chi$ are bounded by $C_0$. The integral of the last term is thus bounded by $C_0^2\mathcal{E}[\PsiSi]$. To every other lower-order term we apply Cauchy-Schwarz: 
\begin{align*}
    \int_{S_{0,0}} \big|\PsiSi\cdot E_\Theta^{(1)}[\Psi_{(S)}]\big|\, d\mu_\slg &\leq \mathcal{E}[\PsiSi]^{1/2}\norm{E_\Theta^{(1)}[\Psi_{(S)}]}_{L^2({S_{0,0}}, \slg)} \\ 
    \int_{S_{0,0}} \big|\PsiSi\cdot \underline E_\Theta^{(1)}[A_{u,\u}^*\underline\Psi_{(S)}]\big|\, d\mu_\slg &\leq \mathcal{E}[\PsiSi]^{1/2}\norm{E_\Theta^{(1)}[A_{u,\u}^*\underline\Psi_{(S)}]}_{L^2({S_{0,0}}, \slg)}.
\end{align*}
By Remark \ref{rem:cauchy-schwarz_basic} the first term is bounded by 
\begin{align*}
    \mathcal{E}[\PsiSi]^{1/2} \cdot C_0\sum_{j = 1}^N \norm{\Psi_{(S)}^{(j)}}_{L^2(S, \slg)} \leq C \sum_{j = 1}^N \mathcal{E}[\Psi_{(S)}^{(j)}]
\end{align*}
and the second is bounded by 
\begin{align*}
    \mathcal{E}[\PsiSi]^{1/2} \cdot C_0 \sum_{j = 1}^N \norm{\underline\Psi_{(S)}^{(j)}}_{L^2(S, \underline\slg)} \leq \mathcal{E}[\PsiSi] + C \sum_{j = 1}^N \calebar[\underline\Psi_{(S)}^{(j)}].
\end{align*}
So far we have shown 
\begin{align*}
    \d_u \mathcal{E}[\PsiSi](u, \u) &= \int_{S_{0,0}} (\Theta_{u,\u}^*\Omega)\PsiSi\cdot A_{u,\u}^*\big([\underline\Theta_{\u,u}^*\sld_{\Psi^{(i)}}\PsiBarSi]\big) \, d\mu_{\slg} \\ 
    &\hspace{15mm} + \text{ terms bounded by } C\Big( \sum_{j = 1}^N \mathcal{E}[\Psi_{(S)}^{(j)}] + \sum_{j = 1}^N \calebar[\underline\Psi_{(S)}^{(j)}]\Big).
\end{align*}
It remains to bound the principal term. We apply Lemma \ref{lem:spherical_derivative_almost_commutation} to write this term as
\begin{multline*}
    \int_{S_{0,0}} (\Theta_{u,\u}^*\Omega)\PsiSi\cdot A_{u,\u}^*\big([\underline\Theta_{\u,u}^*\sld_{\Psi^{(i)}}\PsiBarSi]\big) \, d\mu_{\slg} \\ = \int_{S_{0,0}} (\Theta_{u,\u}^*\Omega)\PsiSi \cdot [\Theta_{u, \u}^* \sld_{\Psi^{(i)}} ] (A_{u, \u}^*\PsiBarSi)\, d\mu_\slg.
\end{multline*}
Then we use the pullback-invariance of the integral (we denote $\Psi^{(i)} = (\Theta_{u,\u}^{-1})^* \PsiSi$) to write this as 
\begin{align*}
    \int_{S_{u,\u}}\Omega \Psi^{(i)} \cdot \sld_{\Psi^{(i)}} \big((\underline\Theta_{\u,u}^{-1})^* \PsiBarSi\big)\, d\mu_{\gamma} .
\end{align*}
Then we use the fact that the $L^2(S_{u,\u}, \gamma)$-adjoint of $\sld_{\Psi^{(i)}}$ is $-\sld_{\underline\Psi^{(i)}}$ to write this as 
\begin{align*}
    -\int_{S_{u,\u}} \sld_{\underline\Psi^{(i)}}(\Omega\Psi^{(i)}) \cdot ((\underline\Theta_{\u,u}^{-1})^*\PsiBarSi)\, d\mu_\gamma.
\end{align*}
In the following we write $\sld_{\underline\Psi^{(i)}}\Omega \cdot \Psi^{(i)}$ to denote the term in the product rule when $\sld_{\underline\Psi^{(i)}}$ hits $\Omega$; for example, when $\sld_{\underline\Psi^{(i)}} = \sl\div$ on a 1-form, 
\begin{align*}
    \sld_{\underline\Psi^{(i)}} \Omega \cdot \Psi^{(i)} &= \sl\nabla\Omega \cdot \Psi^{(i)}.
\end{align*}
In general these will be a linear combination of terms of the form $\sl\nabla\Omega \cdot \Psi^{(i)}$, where $\cdot$ denotes a contraction, potentially with coefficients involving $\sl\epsilon$. Since $\sl\nabla\Omega = \frac{\Omega}{2}(\eta + \underline\eta)$, such terms are bounded in $L^2(S, \gamma)$ by $C(C_0)\norm{\PsiSi}_{L^2(S, \slg)}$.
Therefore, continuing from above, we have
\begin{align*}
    - \int_{S_{u,\u}} \big(\sld_{\underline\Psi^{(i)}}\Omega\cdot \Psi^{(i)} + \Omega \sld_{\underline\Psi^{(i)}}\Psi^{(i)} \big) \cdot ((\underline\Theta_{\u,u}^{-1})^*\PsiBarSi)\, d\mu_\gamma. 
\end{align*}
Combining this with the above and bounding the non-principal term using Cauchy-Schwarz, we have shown now that 
\begin{align*}
    \d_u \cale[\PsiSi](u, \u) &= -\int_{S_{u,\u}}\Omega \sld_{\underline\Psi^{(i)}}\Psi^{(i)} \cdot ((\underline\Theta_{\u,u}^{-1})^*\PsiBarSi)\, d\mu_\gamma \\
    &\hspace{15mm} + \text{ terms bounded by } C\Big( \sum_{j = 1}^N \mathcal{E}[\Psi_{(S)}^{(j)}] + \sum_{j = 1}^N \calebar[\underline\Psi_{(S)}^{(j)}]\Big).
\end{align*}
Call the first term on the right-hand side $I$.
We pull the integral back to ${S_{0,0}}$ by $\underline\Theta_{\u,u}$ (recall $\Psi^{(i)} = (\Theta_{u,\u}^{-1})^* \PsiSi$), apply Lemma \ref{lem:spherical_derivative_almost_commutation}, and then use the equation \eqref{eq:general_hyperbolic_system_S} to get 
\begin{align*}\allowdisplaybreaks
    I  &= -\int_{S_{0,0}} (\underline\Theta_{\u,u}^*\Omega) [\underline\Theta_{\u,u}^* \sld_{\underline\Psi^{(i)}}](\underline A_{\u, u}^*\PsiSi)\cdot \PsiBarSi\, d\mu_{\underline\slg} \\ 
    &= - \int_{S_{0,0}} (\underline\Theta_{\u,u}^*\Omega) \underline A_{\u,u}^*\big( [\Theta_{u,\u}^* \sld_{\underline\Psi^{(i)}}] \PsiSi\big)\cdot \PsiBarSi\, d\mu_{\underline\slg} \\ 
    &= - \int_{S_{0,0}} \PsiBarSi \cdot \d_{\u}\PsiBarSi\, d\mu_{\underline\slg} \\ 
    &\hspace{25mm} + \int_{S_{0,0}} (\underline\Theta_{\u,u}^*\Omega)\PsiBarSi\cdot\big( \underline E_{\underline\Theta}^{(i)}[\underline A_{\u,u}^*\Psi_{(S)}] + \underline E^{(i)}_{\underline\Theta}[\underline\Psi_{(S)}]\big)\, d\mu_{\underline\slg}.
\end{align*}
The second integral is exactly of the form previously considered, except conjugated, and therefore it can be bounded by 
\begin{align*}
    C\Big( \sum_{j = 1}^N \mathcal{E}[\Psi_{(S)}^{(j)}] + \sum_{j = 1}^N \calebar[\underline\Psi_{(S)}^{(j)}]\Big).
\end{align*}
The first integral, meanwhile, is equal to 
\begin{align*}
    - \d_{\u}\calebar[\PsiBarSi](u, \u) + \frac{1}{2}\int_{S_{0,0}} |\PsiBarSi|^2\underline\Theta_{\u,u}^*(\Omega\tr\chi)\, d\mu_{\underline\slg}.
\end{align*}
And again the latter term is bounded by the same quantity as above. This shows that 
\begin{equation}
    \d_u \cale[\PsiSi](u, \u) + \d_{\u}\calebar[\PsiBarSi](u, \u) \leq C \Big( \sum_{j = 1}^N \mathcal{E}[\Psi_{(S)}^{(j)}] + \sum_{j = 1}^N \calebar[\underline\Psi_{(S)}^{(j)}]\Big).
\end{equation}
Summing this over all $i = 1, \ldots, N$ gives the first result of the proposition. The second part is then a direct application of a form of Gr\"onwall's inequality, Proposition \ref{prop:gronwall_1}.
\end{proof}

\subsection{The mollified system \texorpdfstring{$\eps$-DNH}{epsilon-DNH}}\label{subsec:mollified_system}

In this section we use the Friedrichs mollifiers $J_{u,\u}^\eps, \underline J_{\u,u}^\eps$ from Section \ref{subsec:friedrichs_mollifiers} to construct a smooth version of \eqref{eq:general_hyperbolic_system_S} that can be viewed as a Banach space-valued ODE. The unknowns for the mollified system will be denoted $\Psi_{(S),\eps}^{(i)}$ and $\underline\Psi_{(S),\eps}^{(i)}$. The mollified equations are obtained from the original by the following procedure:
\begin{enumerate}
    \item Replace every occurrence of an unknown $\PsiSi$ or $\PsiBarSi$ by $\PsiSiEps$ or $\PsiBarSiEps$, respectively.
    \item Replace the principal term on the right-hand side of an equation for $\d_u \PsiSiEps$ with
    \begin{align*}
        J_{u,\u}^\eps \Big[ (\Theta_{u,\u}^*\Omega) \cdot A_{u,\u}^* \big(\big[\underline\Theta_{\u,u}^*\sl{\mathcal{D}}_{\Psi^{(i)}}\big](\underline J_{\u,u}^\eps \underline\Psi_{(S),\eps}^{(i)}) \big)\Big].
    \end{align*}
    \item Replace the principal term on the right-hand side of an equation for $\d_{\u}\PsiBarSiEps$ with
    \begin{align*}
        \underline J_{\u,u}^\eps \Big[ (\underline\Theta_{\u,u}^*\Omega) \cdot \underline A_{\u,u}^* \big(\big[\Theta_{u,\u}^*\sl{\mathcal{D}}_{\underline\Psi^{(i)}}\big](J_{u,\u}^\eps \Psi_{(S),\eps}^{(i)}) \big)\Big].
    \end{align*}
\end{enumerate}
Making these adjustments gives us the following system, for every $\eps > 0$:
\begin{equation}\label{eq:mollified_equations}
\begin{split}
&\begin{cases}
        \d_u \Psi_{(S),\eps}^{(1)}&= J_{u,\u}^\eps \Big[ (\Theta_{u,\u}^*\Omega) \cdot A_{u,\u}^* \big(\big[\underline\Theta_{\u,u}^*\sl{\mathcal{D}}_{\Psi^{(1)}}\big](\underline J_{\u,u}^\eps \underline\Psi_{(S),\eps}^{(1)})  \big)\Big] \\ 
        &\hspace{35mm} +\, \Theta_{u,\u}^*\Omega \cdot \Big(E^{(1)}_\Theta[\Psi_{(S),\eps}] + E^{(1)}_\Theta [A_{u,\u}^* \underline\Psi_{(S),\eps}] \Big)\\
        &\vdots \\ 
        \d_u \Psi_{(S),\eps}^{(N)}&= J_{u,\u}^\eps \Big[ (\Theta_{u,\u}^*\Omega) \cdot A_{u,\u}^* \big(\big[\underline\Theta_{\u,u}^*\sl{\mathcal{D}}_{\Psi^{(N)}}\big](\underline J_{\u,u}^\eps \underline\Psi_{(S),\eps}^{(N)})  \big)\Big] \\ 
        &\hspace{35mm} +\, \Theta_{u,\u}^*\Omega \cdot \Big(E^{(N)}_\Theta[\Psi_{(S),\eps}] + E^{(N)}_\Theta [A_{u,\u}^* \underline\Psi_{(S),\eps}] \Big)
\end{cases} \\ \\
&\begin{cases}
        \d_{\u}\underline\Psi_{(S),\eps}^{(1)} &= \underline J_{\u,u}^\eps \Big[ (\underline\Theta_{\u,u}^*\Omega) \cdot \underline A_{\u,u}^* \big(\big[\Theta_{u,\u}^*\sl{\mathcal{D}}_{\underline\Psi^{(1)}}\big](J_{u,\u}^\eps \Psi_{(S),\eps}^{(1)}) \big)\Big] \\
        &\hspace{35mm} +\, \underline\Theta_{\u,u}^*\Omega \cdot \Big(    \underline E_{\underline\Theta}^{(1)}[\underline A_{\u,u}^*\Psi_{(S),\eps}] + \underline E^{(1)}_{\underline\Theta}[\underline\Psi_{(S),\eps}]\Big) \\
        &\vdots \\ 
        \d_{\u}\underline\Psi_{(S),\eps}^{(N)} &= \underline J_{\u,u}^\eps \Big[ (\underline\Theta_{\u,u}^*\Omega) \cdot \underline A_{\u,u}^* \big(\big[\Theta_{u,\u}^*\sl{\mathcal{D}}_{\underline\Psi^{(N)}}\big](J_{u,\u}^\eps \Psi_{(S),\eps}^{(N)}) \big)\Big] \\
        &\hspace{35mm} +\, \underline\Theta_{\u,u}^*\Omega \cdot \Big(    \underline E_{\underline\Theta}^{(N)}[\underline A_{\u,u}^*\Psi_{(S),\eps}] + \underline E^{(N)}_{\underline\Theta}[\underline\Psi_{(S),\eps}]\Big).
\end{cases}
\end{split}
\end{equation}
The system \eqref{eq:mollified_equations} can be viewed as a Banach-space valued ODE as follows. Fix an integer $k_0 \geq 0$. Denote
\begin{align*}
    \mathbb{Y}^{k_0} &= \bigoplus_{i = 1}^N H^{k_0}_{p_i}(S_{0,0}) \qquad \text{and} \qquad \underline{\mathbb{Y}}^{k_0} = \bigoplus_{i = 1}^N H^{k_0}_{\underline p_i}(S_{0,0})
\end{align*}
as well as 
\begin{align*}
    \mathbb{X}^{k_0} &= \mathbb{Y}^{k_0} \oplus \underline{\mathbb{Y}}^{k_0}.
\end{align*}
Define 
\begin{align*}
    F_\eps &: D \times \mathbb{X}^{k_0} \to \mathbb{Y}^{k_0} \\ 
    \underline F_\eps &: D \times \mathbb{X}^{k_0} \to \underline{\mathbb{Y}}^{k_0}
\end{align*}
by 
\begingroup\allowdisplaybreaks
\begin{align*}
    F_\eps(u, \u, \Psi^{(1)}_{(S),\eps}, &\ldots, \Psi^{(N)}_{(S),\eps}, \underline\Psi^{(1)}_{(S),\eps}, \ldots, \underline\Psi^{(N)}_{(S),\eps}) \\ 
    &= \Bigg( J_{u,\u}^\eps \Big[ (\Theta_{u,\u}^*\Omega) \cdot A_{u,\u}^* \big(\big[\underline\Theta_{\u,u}^*\sl{\mathcal{D}}_{\Psi^{(1)}}\big](\underline J_{\u,u}^\eps \underline\Psi_{(S),\eps}^{(1)})  \big)\Big]\\ 
    &\qquad + \Theta_{u,\u}^*\Omega \cdot \Big(E^{(1)}_\Theta[\Psi_{(S),\eps}] + E^{(1)}_\Theta [A_{u,\u}^* \underline\Psi_{(S),\eps}] \Big) , 
    \ldots, \\  
    &\qquad J_{u,\u}^\eps \Big[ (\Theta_{u,\u}^*\Omega) \cdot A_{u,\u}^* \big(\big[\underline\Theta_{\u,u}^*\sl{\mathcal{D}}_{\Psi^{(N)}}\big](\underline J_{\u,u}^\eps \underline\Psi_{(S),\eps}^{(N)})  \big)\Big] \\ 
    &\qquad + \Theta_{u,\u}^*\Omega \cdot \Big(E^{(N)}_\Theta[\Psi_{(S),\eps}] + E^{(N)}_\Theta [A_{u,\u}^* \underline\Psi_{(S),\eps}] \Big) \Bigg)
\end{align*}
\endgroup
and 
\begin{align*}
    \underline F_\eps(u, \u, \Psi^{(1)}_{(S),\eps}, &\ldots, \Psi^{(N)}_{(S),\eps}, \underline\Psi^{(1)}_{(S),\eps}, \ldots, \underline\Psi^{(N)}_{(S),\eps}) \\ 
    &= \Bigg( \underline J_{\u,u}^\eps \Big[ (\underline\Theta_{\u,u}^*\Omega) \cdot \underline A_{\u,u}^* \big(\big[\Theta_{u,\u}^*\sl{\mathcal{D}}_{\underline\Psi^{(1)}}\big](J_{u,\u}^\eps \Psi_{(S),\eps}^{(1)}) \big)\Big] \\
    &\qquad + \underline\Theta_{\u,u}^*\Omega \cdot \Big(    \underline E_{\underline\Theta}^{(1)}[\underline A_{\u,u}^*\Psi_{(S),\eps}] + \underline E^{(1)}_{\underline\Theta}[\underline\Psi_{(S),\eps}]\Big), \ldots, \\ 
    &\qquad \underline J_{\u,u}^\eps \Big[ (\underline\Theta_{\u,u}^*\Omega) \cdot \underline A_{\u,u}^* \big(\big[\Theta_{u,\u}^*\sl{\mathcal{D}}_{\underline\Psi^{(N)}}\big](J_{u,\u}^\eps \Psi_{(S),\eps}^{(N)}) \big)\Big] \\
    &\qquad + \underline\Theta_{\u,u}^*\Omega \cdot \Big(    \underline E_{\underline\Theta}^{(N)}[\underline A_{\u,u}^*\Psi_{(S),\eps}] + \underline E^{(N)}_{\underline\Theta}[\underline\Psi_{(S),\eps}]\Big)
    \Bigg).
\end{align*}
These are just the right-hand sides of \eqref{eq:mollified_equations}. Denote the collections of unknowns 
\begin{align*}
    U_\eps &= (\Psi^{(1)}_{(S),\eps}, \ldots, \Psi^{(N)}_{(S),\eps}) \\ 
    \underline U_\eps &= (\underline\Psi^{(1)}_{(S),\eps}, \ldots, \underline\Psi^{(N)}_{(S),\eps}).
\end{align*}
Then we can write \eqref{eq:mollified_equations} as 
\begin{equation}\label{eq:abstracted_ode}
\begin{split}
    \frac{\d U_\eps}{\d u}(u, \u) &= F_\eps(u, \u, U_\eps(u, \u), \underline U_\eps (u, \u)) \\ 
    \frac{\d \underline U_\eps}{\d \underline u}(u, \u) &= \underline F_\eps(u, \u, U_\eps(u, \u), \underline U_\eps (u, \u)).
\end{split}
\end{equation}
This is a two-variable system of ODE of the form addressed in Appendix \ref{app:two-var_ode_theory}. Applying the existence and uniqueness result for such systems (Theorem \ref{thm:picard_1} and Corollary \ref{cor:smooth_ode_solutions}) gives the following. We note that the hypotheses in this proposition are much stronger than necessary; in the finite regularity setting one can assume much less on the spacetime and initial data.

\begin{prop}
Let $(M, g)$ be a smooth spacetime of the type described in Section \ref{subsec:spacetime_and_notation}. Let $\eps \in (0, \delta_0/2)$. Let $U_{0,\eps}, \underline U_{0,\eps}$ have the property that for any $k_0 \geq 0$,
$$
U_{0,\eps} \in C^\infty([0, \u_*]; \mathbb{Y}^{k_0}) \qquad \text{and} \qquad \underline U_{0, \eps} \in C^\infty([0, u_*]; \underline{\mathbb{Y}}^{k_0}).
$$
Then there exists a unique solution $(U_\eps, \underline U_\eps)$ which lies in $C^\infty(D; \mathbb{X}^{k_0})$ for any integer $k_0 \geq 0$. This solution may be identified with the collection of smooth $S$ tensorfields $(V_\eps, \underline V_\eps)$ on $M$, where
\begin{align*}
	(V_\eps)_x &= (\Theta_{u,\u}^{-1})^* (U_\eps[u,\u])|_x, & (\underline V_\eps)_x &= (\underline\Theta_{\u,u}^{-1})^*(\underline U_{\eps}[u,\u])|_x & \forall x \in M, x \in S_{u,\u}.
\end{align*}
\end{prop}

\begin{proof} Note that the initial data can be identified as smooth $S$ tensorfields on $C_0$ and $\underline C_0$, respectively.
Since $J^\eps_{u,\u}, \underline J^\eps_{\u,u}$ are smoothing operators, $F_\eps, \underline F_\eps$ map $D \times \mathbb{X}^{k_0}$ into $\mathbb{Y}^{k_0}, \underline{\mathbb{Y}}^{k_0}$, respectively, for any integer $k_0 \geq 0$. Furthermore, by the assumptions on $(M, g)$, the maps $F_\eps$ and $\underline F_\eps$ are smooth. Theorem \ref{thm:picard_1} and Corollary \ref{cor:smooth_ode_solutions} thus apply, which proves the proposition.
\end{proof}

\subsection{Energy estimates II}\label{subsec:energy_estimates_2}
We now derive energy estimates for the $\PsiSiEps, \PsiBarSiEps$, uniform in $\eps$. The proof follows similar lines as the proof of Proposition \ref{prop:basic_energy_estimates_1}.

\begin{prop}\label{prop:mollified_basic_energy_estimates}
For $\eps \in (0, \delta_0)$, let $\{\PsiSiEps\}_{i = 1}^N, \{\PsiBarSiEps\}_{i = 1}^N$ denote the solutions to the mollified equations \eqref{eq:mollified_equations} on $D$. There exists a constant $C = C(C_0, \mathfrak{B})$ (depending in a continuous way on $C_0$ and $\mathfrak{B}$) such that for all $(u, \u) \in D$, 
\begin{equation}
    \sum_{i = 1}^N \Big(\d_u \cale[\PsiSiEps](u, \u) + \d_{\u}\calebar[\PsiBarSiEps](u, \u)\Big) \leq C \sum_{i = 1}^N \Big(  \mathcal{E}[\Psi_{(S),\eps}^{(i)}](u, \u) +  \calebar[\underline\Psi_{(S),\eps}^{(i)}](u, \u)\Big).
\end{equation}
As a consequence, there exists a (potentially different) constant $C = C(C_0, \mathfrak{B})$ (depending in a continuous way on $C_0$ and $\mathfrak{B}$) such that for all $(u, \u) \in D$, 
\begin{equation}
    \sum_{i = 1}^N \Big(\mathcal{F}^*[\PsiSiEps](u) + \underline{\mathcal{F}}^*[\PsiBarSiEps](\u)\Big) \leq C \sum_{i = 1}^N \Big(\mathcal{F}_0^*[\PsiSiEps] + \underline{\mathcal{F}}_0^*[\PsiBarSiEps]\Big).
\end{equation}
\end{prop}

\begin{proof}
We compute as in the proof of Proposition \ref{prop:basic_energy_estimates_1}: 
\begin{align*}
    \d_u \cale[\PsiSiEps](u, \u) &= \int_{S_{0,0}} \PsiSiEps\cdot J_{u,\u}^\eps \Big[ (\Theta_{u,\u}^*\Omega) \cdot A_{u,\u}^* \big(\big[\underline\Theta_{\u,u}^*\sl{\mathcal{D}}_{\Psi^{(i)}}\big](\underline J_{\u,u}^\eps \underline\Psi_{(S),\eps}^{(i)})  \big)\Big] \\ 
    &\hspace{15mm} +\, \Theta_{u,\u}^*\Omega \cdot \Big(\PsiSiEps\cdot E^{(i)}_\Theta[\Psi_{(S),\eps}] + \PsiSiEps\cdot E^{(i)}_\Theta [A_{u,\u}^* \underline\Psi_{(S),\eps}] \\ 
    &\hspace{15mm} +\frac{1}{2}|\PsiSiEps|^2\Theta_{u,\u}^*\tr\underline\chi \Big) \, d\mu_{\slg}.
\end{align*}
As before, the integral of the lower-order terms can be bounded by 
\begin{align*}
    C\Big( \sum_{j = 1}^N \cale[\Psi^{(j)}_{(S),\eps}] + \sum_{j = 1}^N \calebar[\underline\Psi^{(j)}_{(S),\eps}]  \Big).
\end{align*}
It remains to consider the integral of the principal term, which we denote $I$. First, by the self-adjointness of $J^\eps_{u,\u}$, $J^\eps_{u,\u}$ can be moved to the $\PsiSiEps$. Then we can use Lemma \ref{lem:spherical_derivative_almost_commutation} to move the $A_{u,\u}^*$ inside the $\underline\Theta_{\u,u}^*\sld_{\Psi^{(i)}}$, changing the $\underline\Theta_{\u,u}$ to $\Theta_{u,\u}$ in the process. This gives
\begin{align*}
    I &= \int_{S_{0,0}} J^\eps_{u,\u}\PsiSiEps \cdot \Big[ (\Theta_{u,\u}^*\Omega)\cdot \big([\Theta_{u,\u}^*\sld_{\Psi^{(i)}}](A_{u,\u}^*(\underline J^\eps_{\u,u} \PsiBarSiEps))\big) \Big]\, d\mu_\slg.
\end{align*}
We can pull this back to an integral over $S_{u,\u}$ by $\Theta_{u,\u}^{-1}$ and then integrate by parts:
\begin{multline}\label{eq:computation_2}
    I = -\int_{S_{u,\u}} \Omega \sld_{\underline\Psi^{(i)}}\big[(\Theta_{u,\u}^{-1})^*( J^\eps_{u,\u}\PsiSiEps)\big] \cdot \big((\underline\Theta_{\u,u}^{-1})^*(\underline J^\eps_{\u,u} \PsiBarSiEps)\big)\, d\mu_\gamma \\ 
    - \int_{S_{u,\u}} \big( \sld_{\underline\Psi^{(i)}}\Omega \cdot (\Theta_{u,\u}^{-1})^*(J^\eps_{u,\u}\PsiSiEps)\big) \cdot \big((\underline\Theta_{\u,u}^{-1})^*(\underline J^\eps_{\u,u} \PsiBarSiEps)\big)\, d\mu_\gamma.
\end{multline}
By Remark \ref{rem:cauchy-schwarz_basic}, and since $J^\eps_{u,\u}, \underline J^\eps_{\u,u}$ are uniformly bounded operators on $L^2({S_{0,0}}$, $\slg)$, $L^2({S_{0,0}}, \underline\slg)$, respectively (Proposition \ref{prop:uniform_bound_1}), the last line is bounded by
\begin{align*}
    C \norm{J^\eps_{u,\u}\PsiSiEps}_{L^2({S_{0,0}}, \slg)}\norm{\underline J^\eps_{\u,u}\PsiBarSiEps}_{L^2({S_{0,0}}, \underline\slg)} \leq C\mathfrak{B}^2\norm{\PsiSiEps}_{L^2({S_{0,0}}, \slg)} \norm{\PsiBarSiEps}_{L^2({S_{0,0}}, \underline\slg)}.
\end{align*}
In performing this step, note that we pull the integral of $J^\eps_{u,\u}\PsiSiEps$ back to ${S_{0,0}}$ by $\Theta_{u,\u}$, and of $\underline J^\eps_{\u,u} \PsiBarSiEps$ to ${S_{0,0}}$ by $\underline\Theta_{\u,u}$.

Now, call the first line in \eqref{eq:computation_2} $II$. We pull this back to ${S_{0,0}}$ by $\underline\Theta_{\u,u}$ and then use the self-adjointness of $\underline J^\eps_{\u,u}$, and again Lemma \ref{lem:spherical_derivative_almost_commutation}, to get 
\begin{align*}
    II &=  -\int_{{S_{0,0}}} \underline\Theta_{\u,u}^* \Omega \big[\underline\Theta_{\u,u}^*\sld_{\underline\Psi^{(i)}}\big]\big[\underline A_{\u,u}^*( J^\eps_{u,\u}\PsiSiEps)\big] \cdot (\underline J^\eps_{\u,u} \PsiBarSiEps)\, d\mu_{\underline\slg} \\ 
    &= -\int_{{S_{0,0}}} \underline J^\eps_{\u,u} \Big[ \underline\Theta_{\u,u}^* \Omega \cdot \big(\underline A_{\u,u}^*\big[\Theta_{u,\u}^*\sld_{\underline\Psi^{(i)}}\big]( J^\eps_{u,\u}\PsiSiEps)\big)\Big] \cdot \PsiBarSiEps\, d\mu_{\underline\slg}.
\end{align*}
The first term in this product is exactly the principal part of the equation for $\d_{\u} \PsiBarSiEps$. We therefore obtain 
\begin{equation}
    II = -\int_{{S_{0,0}}} \PsiBarSiEps\cdot \Big[ \d_{\u}\PsiBarSiEps - \underline\Theta_{\u,u}^*\Omega \cdot \Big(    \underline E_{\underline\Theta}^{(i)}[\underline A_{\u,u}^*\Psi_{(S),\eps}] + \underline E^{(i)}_{\underline\Theta}[\underline\Psi_{(S),\eps}]\Big)  \Big]\, d\mu_{\underline\slg}.
\end{equation}
The lower-order terms here are of the same form as considered above (except conjugated), and so can be bounded by 
\begin{align*}
    C\Big( \sum_{j = 1}^N \cale[\Psi^{(j)}_{(S),\eps}] + \sum_{j = 1}^N \calebar[\underline\Psi^{(j)}_{(S),\eps}]  \Big).
\end{align*}
Therefore, we have
\begin{multline*}
    I = - \d_{\u}\calebar[\PsiBarSiEps](u,\u) + \frac{1}{2}\int_{S_{0,0}} |\PsiBarSiEps|^2\underline\Theta_{\u,u}^*(\Omega\tr\chi)\, d\mu_{\underline \slg} \\ 
    + \text{ terms bounded by }C\Big( \sum_{j = 1}^N \cale[\Psi^{(j)}_{(S),\eps}] + \sum_{j = 1}^N \calebar[\underline\Psi^{(j)}_{(S),\eps}]  \Big).
\end{multline*}
The second term on the first line can again be bounded by terms of the above form. Putting this all together, we have shown that for all $i = 1, \ldots, N$, 
\begin{align*}
    \d_u \cale[\PsiSiEps](u, \u) + \d_{\u}\calebar[\PsiBarSiEps](u,\u) \leq C\Big( \sum_{j = 1}^N \cale[\Psi^{(j)}_{(S),\eps}] + \sum_{j = 1}^N \calebar[\underline\Psi^{(j)}_{(S),\eps}]  \Big),
\end{align*}
where $C$ is a constant depending on $C_0$ and $\mathfrak{B}$.
Applying a form of Gr\"onwall's inequality (Proposition \ref{prop:gronwall_1}) proves the second part of this statement.
\end{proof}

In order to derive higher-order energy estimates, we need the following two commutation lemmas. These are standard (see for instance Lemmas 4.1-4.2 in \cite{chr-bhf} or Lemma 7.3.3 in \cite{CK}).

\begin{lemma}\label{lem:null_derivative_nabla_commutator}
For a $(u,\u)$-dependent $p$-covariant tensorfield $\Psi_{(S)}$ on ${S_{0,0}}$, we have:
\begin{align*}
    [\d_u, (\Theta_{u,\u}^*\sl\nabla)](\Psi_{(S)})_{A B_1 \cdots B_p} &= -\sum_{i = 1}^p \Theta_{u,\u}^*(\underline D\sl\Gamma)_{A B_i}^C (\Psi_{(S)})_{B_1 \cdots > \substack{C \\ B_i} < \cdots B_p} \\ 
    [\d_{\u}, (\underline\Theta_{\u,u}^*\sl\nabla)](\Psi_{(S)})_{A B_1 \cdots B_p} &= -\sum_{i = 1}^p \underline\Theta_{\u,u}^*(D\sl\Gamma)_{A B_i}^C (\Psi_{(S)})_{B_1 \cdots > \substack{C \\ B_i} < \cdots B_p}
\end{align*}
where $>\substack{C \\ B_i}<$ denotes that $B_i$ is replaced with $C$. Here $\underline D \sl\Gamma$ and $D \sl\Gamma$ denote the Lie derivative of the connection in the $\underline L$, respectively $L$, directions. They are given by:
\begin{align*}
    \gamma_{CD}(D\sl\Gamma)^D_{AB} &= \sl\nabla_A (\Omega\chi)_{BC} + \sl\nabla_B (\Omega\chi)_{AC} - \sl\nabla_{C}(\Omega\chi)_{AB} \\
    \gamma_{CD}(\underline D\sl\Gamma)^D_{AB} &= \sl\nabla_A (\Omega\underline\chi)_{BC} + \sl\nabla_B (\Omega\underline\chi)_{AC} - \sl\nabla_{C}(\Omega\underline\chi)_{AB}.
\end{align*}
\end{lemma}

\begin{proof}
To prove this, push these expression forward to $M$, and then apply Lemmas 4.1, 4.2 of \cite{chr-bhf}.
\end{proof}

\begin{lemma}\label{lem:spherical_derivative_nable_commutator}
Let $\Psi_{(S)}$ be a $(u,\u)$-dependent $p$-covariant tensorfield on ${S_{0,0}}$, and let $\sld$ be a geometric first-order differential operator on $(S_{u,\u}, \gamma)$ (as defined in Definition \ref{defn:pullback_differential_operators}). Then
\begin{align*}
    [\Theta_{u,\u}^*\sl\nabla]\big(A_{u,\u}^*([\underline\Theta_{\u,u}^*\sl{\mathcal{D}}]\Psi_{(S)})\big) &= A_{u,\u}^* \big([\underline\Theta_{\u,u}^* {\sl{\mathcal{D}}}]([\underline\Theta_{\u,u}^*\sl\nabla]\Psi_{(S)})\big) + f(\sl{\mathcal{D}})\cdot (A_{u,\u}^* \Psi_{(S)})
\end{align*}
where $f(\sld)$ is a linear operator whose $L^2({S_{0,0}}, \slg) \to L^2({S_{0,0}}, \slg)$ operator norm is bounded by 
\begin{align*}
    C \sup_{(u,\u)} \norm{K}_{L^\infty(S_{u,\u})},
\end{align*}
where $K$ is the Gauss curvature of $(S_{u,\u},\gamma)$ and
$C$ is a constant depending only the particular form of $\sld$.
\end{lemma}

\begin{remark}
In the following cases (which are important for applications to the Bianchi equations) $f(\sld)$ has an explicit form:

\begin{itemize}
    \item In the case $\sl{\mathcal{D}} = \sl\nabla$ acts on a scalar function, 
    \begin{align*}
        f(\sl{\mathcal{D}}) \cdot A_{u,\u}^* \Psi_{(S)} = 0.
    \end{align*}
    \item In the case $\sl{\mathcal{D}} = \sl\nabla\hat\otimes$ acts on a 1-form: 
    \begin{align*}
        \big(f(\sl{\mathcal{D}}) \cdot A_{u,\u}^* \Psi_{(S)}\big)_A &= \Theta_{u,\u}^* K \Big(\slg_{AC} A_{u,\u}^*(\Psi_{(S)})_B + \slg_{AB} A_{u,\u}^*(\Psi_{(S)})_C - \slg_{BC}A_{u,\u}^* (\Psi_{(S)})_A\Big).
    \end{align*}
    \item In the case $\sl{\mathcal{D}} = \sl\div$ acts on a symmetric traceless 2-covariant $S$ tensorfields:
    \begin{align*}
        f(\sl{\mathcal{D}})\cdot A_{u,\u}^* \Psi_{(S)} &= -2(\Theta_{u,\u}^* K) A_{u,\u}^* \Psi_{(S)}.
    \end{align*}
    \item In the case $\sl{\mathcal{D}} = \sl \div$ acts on a 1-form: 
    \begin{align*}
        f(\sl{\mathcal{D}}) \cdot A_{u,\u}^*\Psi_{(S)} &= -(\Theta_{u,\u}^* K) A_{u,\u}^* \Psi_{(S)}.
    \end{align*}  
\end{itemize}
\end{remark}

\begin{proof}
Recall that by definition, for a covariant tensorfield $\theta$ on $S_{u,\u}$, $\sld\cdot \theta$ is a finite linear combination of contractions of $\sl\nabla\theta$ with $\gamma_{u, \u}$ or $\sl\epsilon_{u, \u}$. By Lemma \ref{lem:spherical_derivative_almost_commutation}, 
\begin{align*}
    [\Theta_{u,\u}^*\sl\nabla]\big(A_{u,\u}^*([\underline\Theta_{\u,u}^*\sl{\mathcal{D}}]\Psi_{(S)})\big) &= A_{u,\u}^* \big([\underline\Theta_{\u,u}^*\sl\nabla]([\underline\Theta_{\u,u}^* {\sl{\mathcal{D}}}]\Psi_{(S)})\big). \eqcount\label{eq:computation_3}
\end{align*}
Consider $[\underline\Theta_{\u,u}^*\sl\nabla]([\underline\Theta_{\u,u}^* {\sl{\mathcal{D}}}]\Psi_{(S)})$.
The pullback of this expression back to $S_{u,\u}$ by $\underline\Theta_{\u,u}^{-1}$ is
\begin{align*}
    \sl\nabla(\sld\xi)
\end{align*}
where $\xi = (\underline\Theta_{\u,u}^{-1})^*\Psi_{(S)}$. By definition of $\sld$, this is a finite sum
\begin{align*}
    \sum_{i = 1}^n c_i^{BC} \cdot \sl\nabla_A \sl\nabla_B\xi_{C}
\end{align*}
where each $c_i$ is a finite (possibly empty) sum of tensor products of $\slg, \sl\epsilon$ and $\cdot $ denotes an arbitrary contraction; the coefficients $c_i$ are thus bounded by a constant $C$ depending only on how many factors of $\slg, \sl\epsilon$ appear. Note that the indices $(BC)$ here stand for the full set of indices that are contracted. Commuting derivatives gives
\begin{align*}
    \sl\nabla_A\sl\nabla_B \xi_{C_1 \cdots C_p} &= \sl\nabla_B\sl\nabla_A\xi_{C_1 \cdots C_p} - \sum_{i = 1}^p K(\gamma_{B C_i} \xi_{C_1 \cdots > \substack{A \\ C_i} < \cdots C_p} - \gamma_{A C_i}\xi_{C_1 \cdots > \substack{B \\ C_i} < \cdots C_p})
\end{align*}
where the notation $> \substack{A \\ C_i} <$ denotes that the index $C_i$ has been replaced with $A$. 
Recombining these terms, we see that 
\begin{align*}
    \sum_{i = 1}^n c_i^{BC} \cdot \sl\nabla_B \sl\nabla_A\xi_{C} &= \sld\sl\nabla_A \xi,
\end{align*}
and all other terms are bounded pointwise by 
\begin{align*}
    \sup_{(u,\u) \in D}\norm{K}_{L^\infty(S_{u,\u})}\cdot \xi.
\end{align*}
Pulling this back to $S_{0,0}$ by $\underline\Theta_{\u,u}$,  and continuing from \eqref{eq:computation_3}, we have
\begin{align*}
    A_{u,\u}^* \big([\underline\Theta_{\u,u}^*\sl\nabla]([\underline\Theta_{\u,u}^* {\sl{\mathcal{D}}}]\Psi_{(S)})\big)
    &= A_{u,\u}^*\big( [\underline\Theta_{\u,u}^* \sld]( [\underline\Theta_{\u,u}^*\sl\nabla] \Psi_{(S)}) \big) + f(\sld)\cdot A_{u,\u}^*\Psi_{(S)}
\end{align*}
where $f(\sld)$ has the claimed properties.
\end{proof}

For the higher-order estimates, it is convenient to introduce the following quantity:

\begin{defn}\label{defn:sl_C_k}
For $m \geq 1$ let $\sl{C}_m$ denote a constant such that 
\begin{align*}
    C_0 + \sup_{(u,\u) \in D} \Big[\sum_{k = 0}^m \max\big(\norm{\sl\nabla^k \Omega}_{L^\infty(S_{u,\u})}, 1\big)\Big] & \Big[ \sup_{\psi \in \Gamma}\sum_{k = 0}^m \max\big(\norm{\sl\nabla^k \psi}_{L^\infty(S_{u,\u})} , 1 \big)\Big]  \\ 
    &\qquad \cdot\Big[\sum_{k = 0}^{m - 1}\max\big(\norm{\sl\nabla^k K}_{L^\infty(S_{u,\u})}, 1\big)\Big] \leq \sl{C}_m.
\end{align*}
The supremum over $\psi \in \Gamma$ denotes a supremum over all Ricci coefficients $\psi$. 
\end{defn}

\begin{prop}\label{prop:higher_order_energy_estimates} For $\eps \in (0, \delta_0/2)$, let $\{\PsiSiEps\}_{i = 1}^N, \{\PsiBarSiEps\}_{i = 1}^N$ denote the solutions to the mollified equations \eqref{eq:mollified_equations} on $D$. There exists a constant $C = C(\sl C_1, \mathfrak{B}, \mathfrak{C})$ (depending in a continuous way on $\sl C_1, \mathfrak{B}$, and $\mathfrak{C}$) such that for all $(u, \u) \in D$, we have
\begin{equation*}
    \sum_{i = 1}^N \Big(\mathcal{F}^*[\sl\nabla\Psi_{(S),\eps}^{(i)}](u)+ \underline{\mathcal{F}}^*[\sl\nabla \underline\Psi_{(S),\eps}](\u)\Big) \\ 
    \leq C\sum_{i = 1}^N \Big( \mathcal{F}_0^*[\sl\nabla\Psi_{(S),\eps}^{(i)}] + \underline{\mathcal{F}}_0^*[\sl\nabla\underline\Psi_{(S),\eps}^{(i)}]\Big).
\end{equation*}
\end{prop}

\begin{remark}
One can prove a similar statement for higher-order estimates on $\sl\nabla^m \PsiSiEps$ and $\sl\nabla^m \PsiBarSiEps$, but we do not need this result for the purposes of this paper. 
\end{remark}

\begin{proof}
We apply $\Theta_{u,\u}^*\sl\nabla$ to the equation for $\d_u\PsiSiEps$ and $\underline\Theta_{\u,u}^*\sl\nabla$ to the equation for $\d_{\u}\PsiBarSiEps$. In this proof, we will let $\Phi^{(i)} = (\Theta_{u,\u}^*\sl\nabla)\PsiSiEps$ and $\underline\Phi^{(i)} = (\underline\Theta_{\u,u}^*\sl\nabla)\PsiBarSiEps$. We do not specify the specific form of several contractions that appear in the computation, as it does not matter for the estimate we apply.

By Lemma \ref{lem:null_derivative_nabla_commutator}, 
\begin{align*}
    \Theta_{u,\u}^*\sl\nabla\big( \d_u\PsiSiEps \big) &= \d_u\Phi^{(i)} - \sum_{i = 1}^{p_i}\Theta_{u,\u}^*(\underline D\sl\Gamma)\cdot \PsiSiEps.
\end{align*}
The second term is bounded in $L^2({S_{0,0}}, \slg)$ by $\sl C_1 \norm{\PsiSiEps}_{L^2({S_{0,0}}, \slg)}$. Turning to the right-hand side of the equation for $\d_u \PsiSiEps$, the lower-order terms are estimated as follows. Terms involving $\Psi_{(S),\eps}^{(j)}$ are equal to
\begin{align*}
    (\Theta_{u,\u}^*\sl\nabla)\big(\Theta_{u,\u}^*\Omega \cdot E_\Theta^{(i)}[\PsiSiEps]\big) &= \Theta_{u,\u}^*(\sl\nabla\Omega)E_\Theta^{(i)} + \Theta_{u,\u}^*\Omega \slg(\Theta_{u,\u}^*(\sl\nabla\psi), \Psi_{(S),\eps}^{(j)})\\ 
    &\hspace{20mm} + \Theta_{u,\u}^*\Omega \slg(\Theta_{u,\u}^*\psi, \Phi^{(j)}) 
\end{align*}
while terms involving $\underline\Psi_{(S),\eps}^{(j)}$ are equal to (see Lemma \ref{lem:spherical_derivative_nable_commutator}):
\begin{align*}
    (\Theta_{u,\u}^*\sl\nabla)\big(\Theta_{u,\u}^*\Omega \cdot E_\Theta^{(i)}[A_{u,\u}^*\PsiBarSiEps]\big) &= \Theta_{u,\u}^*(\sl\nabla\Omega) E_\Theta^{(i)} + \Theta_{u,\u}^*\Omega \slg ( \Theta_{u,\u}^*(\sl\nabla\psi), A_{u,\u}^* \underline\Psi_{(S),\eps}^{(j)} ) \\ 
    &\hspace{20mm} + \Theta_{u,\u}^*\Omega \slg(\Theta_{u,\u}^*\psi, A_{u,\u}^* \underline\Phi^{(j)}).
\end{align*}
All of the terms on the right-hand side of these equations can be bounded in $L^2({S_{0,0}}, \slg)$ by (see Remark \ref{rem:cauchy-schwarz_basic}):
\begin{align*}
    C_0 \sum_{i = 1}^N \Big( \norm{\Phi^{(j)}}_{L^2({S_{0,0}}, \slg)} + \norm{\underline\Phi^{(j)}}_{L^2({S_{0,0}}, \underline\slg)} \Big) + \sl C_1\sum_{i = 1}^N\Big( \norm{\Psi_{(S),\eps}^{(j)}}_{L^2({S_{0,0}}, \slg)} + \norm{\underline\Psi_{(S),\eps}^{(j)}}_{L^2({S_{0,0}}, \underline\slg)} \Big).
\end{align*}
It remains to consider when $\Theta_{u,\u}^*\sl\nabla$ hits the principal term, namely
\begin{align*}
    I \coloneqq [\Theta_{u,\u}^*\sl\nabla]\Big( J_{u,\u}^\eps \Big[ (\Theta_{u,\u}^*\Omega) \cdot A_{u,\u}^* \big(\big[\underline\Theta_{\u,u}^*\sl{\mathcal{D}}_{\Psi^{(i)}}\big](\underline J_{\u,u}^\eps \underline\Psi_{(S),\eps}^{(i)})  \big)\Big] \Big).
\end{align*}
By Proposition \ref{prop:mathfrak_C}, the commutator $[J^\eps_{u,\u}, \Theta_{u,\u}^*\sl\nabla]$ is a bounded linear operator on $L^2({S_{0,0}}, \slg)$ with  operator norm $\leq \mathfrak{C}$. We have
\begin{align*}
    I &= J^\eps_{u,\u}\Big[ [\Theta_{u,\u}^*\sl\nabla]\Big( (\Theta_{u,\u}^*\Omega) \cdot A_{u,\u}^* \big(\big[\underline\Theta_{\u,u}^*\sl{\mathcal{D}}_{\Psi^{(i)}}\big](\underline J_{\u,u}^\eps \underline\Psi_{(S),\eps}^{(i)})  \big)\Big) \Big]  \\ 
    &\hspace{20mm} + [J^\eps_{u,\u}, \Theta_{u,\u}^*\sl\nabla] \Big( (\Theta_{u,\u}^*\Omega) \cdot A_{u,\u}^* \big(\big[\underline\Theta_{\u,u}^*\sl{\mathcal{D}}_{\Psi^{(i)}}\big](\underline J_{\u,u}^\eps \underline\Psi_{(S),\eps}^{(i)})  \big)\Big) \\ 
    &= J^\eps_{u,\u}\Big[ \Big( (\Theta_{u,\u}^*\Omega)[\Theta_{u,\u}^*\sl\nabla] A_{u,\u}^* \big(\big[\underline\Theta_{\u,u}^*\sl{\mathcal{D}}_{\Psi^{(i)}}\big](\underline J_{\u,u}^\eps \underline\Psi_{(S),\eps}^{(i)})  \big)\Big) \Big] \\ 
    &\hspace{20mm} + J^\eps_{u,\u}\Big[ \Big( (\Theta_{u,\u}^*(\sl\nabla\Omega)) A_{u,\u}^* \big(\big[\underline\Theta_{\u,u}^*\sl{\mathcal{D}}_{\Psi^{(i)}}\big](\underline J_{\u,u}^\eps \underline\Psi_{(S),\eps}^{(i)})  \big)\Big) \Big] \\
    &\hspace{20mm} + [J^\eps_{u,\u}, \Theta_{u,\u}^*\sl\nabla] \Big( (\Theta_{u,\u}^*\Omega) \cdot A_{u,\u}^* \big(\big[\underline\Theta_{\u,u}^*\sl{\mathcal{D}}_{\Psi^{(i)}}\big](\underline J_{\u,u}^\eps \underline\Psi_{(S),\eps}^{(i)})  \big)\Big).
\end{align*}
The second and third lines are bounded in $L^2({S_{0,0}}, \slg)$ by 
\begin{align*}
    C (\sl C_1 + \mathfrak{C})\Big(\norm{\PsiBarSiEps}_{L^2({S_{0,0}}, \underline\slg)} + \norm{\underline\Phi^{(i)}}_{L^2({S_{0,0}}, \underline\slg)} \Big),
\end{align*}
since all geometric derivative operators $\sld$ can be bounded by a constant times $\sl\nabla$. Meanwhile, by Lemma \ref{lem:spherical_derivative_nable_commutator}, we can write the first line (which we denote by $II$) as
\begin{align*}
    II &= J^\eps_{u,\u}\Big[ (\Theta_{u,\u}^*\Omega) A_{u,\u}^* \big(\big[\underline\Theta_{\u,u}^*\sl{\mathcal{D}}_{\Psi^{(i)}}\big][\underline\Theta_{\u,u}^*\sl\nabla](\underline J_{\u,u}^\eps \underline\Psi_{(S),\eps}^{(i)})  \big) \Big] + f(\sld_{\Psi^{(i)}})\cdot A_{u,\u}^*(\underline J_{\u,u}^\eps \underline\Psi_{(S),\eps}^{(i)}).
\end{align*}
By Lemma \ref{lem:fundamental_integration_property_of_A} and Proposition \ref{prop:uniform_bound_1}, as well as Lemma \ref{lem:spherical_derivative_nable_commutator}, the last term  in $II$ is bounded by 
\begin{align*}
    C\sl C_1\norm{\PsiBarSiEps}_{L^2({S_{0,0}}, \underline\slg)}.
\end{align*}
The first term in $II$ is equal to 
\begin{align*}
    J^\eps_{u,\u}\Big[ (\Theta_{u,\u}^*\Omega) A_{u,\u}^* \big(\big[\underline\Theta_{\u,u}^*\sl{\mathcal{D}}_{\Psi^{(i)}}\big](\underline J_{\u,u}^\eps \underline\Phi^{(i)})  \big) \Big] + III,
\end{align*}
where 
\begin{align*}
    III &= J^\eps_{u,\u}\Big[ (\Theta_{u,\u}^*\Omega) A_{u,\u}^* \big(\big[\underline\Theta_{\u,u}^*\sl{\mathcal{D}}_{\Psi^{(i)}}\big]([\underline\Theta_{\u,u}^*\sl\nabla, \underline J_{\u,u}^\eps] \underline\Psi_{(S),\eps}^{(i)})  \big) \Big].
\end{align*}
Using Propositions \ref{prop:uniform_bound_1} and \ref{prop:mathfrak_C_1}, this is bounded in $L^2({S_{0,0}}, \slg)$ by
\begin{align*}
    CC_0\mathfrak{B} \norm{\underline\Theta_{\u,u}^*\sl\nabla &([\underline\Theta_{\u,u}^*\sl\nabla, \underline J_{\u,u}^\eps] \underline\Psi_{(S),\eps}^{(i)})}_{L^2({S_{0,0}}, \underline\slg)} \\ 
    &= C C_0\mathfrak{B} \norm{[\underline\Theta_{\u,u}^*\sl\nabla, \underline J_{\u,u}^\eps]\PsiBarSiEps}_{H^1_{\underline p_i}({S_{0,0}}, \underline\slg)} \\ 
    &\leq C C_0 \mathfrak{B} \mathfrak{C}\norm{\PsiBarSiEps}_{H^1_{\underline p_i}(S_{0,0},\underline\slg)} \\
    &\leq C C_0\mathfrak{B}\mathfrak{C} \norm{\underline\Phi^{(i)}}_{L^2({S_{0,0}}, \underline\slg)}
\end{align*}
Therefore, we have 
\begin{equation}
    \d_u\Phi^{(i)} = J^\eps_{u,\u}\Big[ (\Theta_{u,\u}^*\Omega) A_{u,\u}^* \big(\big[\underline\Theta_{\u,u}^*\sl{\mathcal{D}}_{\Psi^{(i)}}\big](\underline J_{\u,u}^\eps \underline\Phi^{(i)})  \big) \Big] + E^{(i), 1},
\end{equation}
where $E^{(i), 1}$ is bounded in $L^2({S_{0,0}}, \slg)$ by 
\begin{align*}
    C'\Big(\sum_{i = 1}^N  \norm{\Phi^{(j)}}_{L^2({S_{0,0}}, \slg)} + \norm{\underline\Phi^{(j)}}_{L^2({S_{0,0}}, \underline\slg)} + \norm{\Psi_{(S),\eps}^{(j)}}_{L^2({S_{0,0}}, \slg)} + \norm{\underline\Psi_{(S),\eps}^{(j)}}_{L^2({S_{0,0}}, \underline\slg)} \Big),
\end{align*}
where $C'$ is a constant depending on $\sl C_1, \mathfrak{B}$, and $\mathfrak{C}$.
This is of the same form as the original system \eqref{eq:mollified_equations} and hence we can apply Proposition \ref{prop:mollified_basic_energy_estimates}. This completes the proof.
\end{proof}

\subsection{Solutions of DNH} \label{subsec:solutions_of_DNH}
We first define the notion of a weak solution to \eqref{eq:general_hyperbolic_system}. 

\begin{defn}\label{defn:weak_solution} A \emph{weak solution} to \eqref{eq:general_hyperbolic_system} is a collection $\{\Psi^{(i)}\}_{i = 1}^N, \{ \underline\Psi^{(i)}\}_{i = 1}^N$ of covariant $L^2(M)$ $S$ tensorfields on $M$ such that: for all $i = 1, \ldots, N$, any smooth $p_i$-covariant $S$ tensorfield $\phi$, and any smooth $\underline p_i$-covariant $S$ tensorfield $\underline\phi$ such that
\begin{align*}
    \supp(\phi), \supp(\underline\phi) \subset (0, u_*) \times (0, \u_*) \times S^2,
\end{align*}
we have
\begin{equation}
    0 = -\int_M \underline D(\phi^\sharp)\cdot \Psi^{(i)} - \sld_{\underline\Psi^{(i)}}(\Omega\phi)\cdot \underline\Psi^{(i)} + \Omega\tr\underline\chi\phi \cdot \Psi^{(i)} + \Omega\phi\cdot E^{(i)}\, d\mu_g 
\end{equation}
and 
\begin{equation}
    0 = -\int_M  D(\underline\phi^\sharp)\cdot \underline\Psi^{(i)} - \sld_{\Psi^{(i)}}(\Omega\underline\phi)\cdot \Psi^{(i)} + \Omega\tr\chi\underline\phi \cdot \underline\Psi^{(i)} + \Omega\underline\phi\cdot \underline E^{(i)}\, d\mu_g.
\end{equation}
Here, $\phi^\sharp$ and $\underline\phi^\sharp$ denote the totally contravariant metric dual tensorfields to $\phi$ and $\underline\phi$, respectively.
\end{defn}

\begin{lemma}\label{lem:strong_is_weak}
A classical smooth solution to \eqref{eq:general_hyperbolic_system} is also a weak solution. 
\end{lemma}

\begin{proof} The proof is integration by parts. Let $\Psi^{(i)}, \underline\Psi^{(i)}$ be a classical smooth solution to \eqref{eq:general_hyperbolic_system}. Let $\phi$ be a smooth $p_i$-covariant $S$ tensorfield with support in $(0, u_*) \times (0, \u_*) \times S^2$. 
We have: 
\begin{align*}
    \int_M \underline D(\phi^\sharp) \cdot\Psi^{(i)}\, d\mu_g &= \int_0^{u_*}\int_0^{\u_*}\int_{S_{u,\u}}\underline D(\phi^\sharp) \cdot\Psi^{(i)}\, d\mu_{\gamma_{u,\u}}\, d\u\, du.
\end{align*}
Integrating by parts (by pulling back to $S_{0,0}$ and then using a partition of unity on $S_{0,0}$), this equals 
\begin{equation}
    -\int_0^{u_*} \int_0^{\u_*}\int_{S_{u,\u}} \phi \cdot \underline D\Psi^{(i)} + \Omega\tr\underline\chi \phi \cdot\Psi^{(i)}\, d\mu_{\gamma_{u,\u}}\, d\u \, du = -\int_M \phi \cdot \underline D\Psi^{(i)} + \Omega\tr\underline\chi \phi \cdot\Psi^{(i)}\, d\mu_g.
\end{equation}
Next, we use the anti-adjointness of $\sld_{\underline\Psi^{(i)}}$ and $\sld_{\Psi^{(i)}}$ to get
\begin{align*}
    \int_M \sld_{\underline\Psi^{(i)}}(\Omega\phi) \cdot\underline\Psi^{(i)}\, d\mu_g &= \int_0^{u_*}\int_0^{\u_*}\int_{S_{u,\u}} \sld_{\underline\Psi^{(i)}}(\Omega\phi)\cdot\underline\Psi^{(i)}\, d\mu_\gamma\, d\u\, du \\ 
    &= -\int_0^{u_*}\int_0^{\u_*}\int_{S_{u,\u}} \Omega\phi\cdot \sld_{\Psi^{(i)}}\underline\Psi^{(i)}\, d\mu_\gamma\, d\u\, du \\
    &= -\int_M \Omega\phi \cdot\sld_{\Psi^{(i)}}\underline\Psi^{(i)}\, d\mu_g.
\end{align*}
Therefore, we have
\begingroup
\allowdisplaybreaks
\begin{multline*}
    -\int_M \underline D(\phi^\sharp)\cdot \Psi^{(i)} - \sld_{\underline\Psi^{(i)}}(\Omega\phi)\cdot \underline\Psi^{(i)} + \Omega\tr\underline\chi\phi \cdot \Psi^{(i)} + \Omega\phi\cdot E^{(i)}\, d\mu_g \\ 
    = \int_M \phi \cdot \underbrace{\Big( \underline D\Psi^{(i)} - \Omega\Big[ \sld_{\Psi^{(i)}}\underline\Psi^{(i)} + E^{(i)} \Big]\Big)}_{= 0}\, d\mu_g  = 0.
\end{multline*}
\endgroup
Analogously we can show the second equation in Definition \ref{defn:weak_solution} holds. This proves the lemma.
\end{proof}

This proof motivates the following definition of weak derivatives. For completeness here we list the definitions of the spherical and null weak derivatives, however in this paper we shall only make use of weak spherical derivatives.

\begin{defn}[Weak derivatives] \label{defn:weak_derivatives}
Let $\Psi$ be an $L^2(M)$ $p$-covariant $S$ tensorfield on $M$. 
\begin{enumerate}
    \item An \emph{$L^2(M)$ weak $D$-derivative} of $\Psi$ is an $L^2(M)$ $p$-covariant $S$ tensorfield $\Phi$ on $M$ such that for all smooth $p$-covariant $S$ tensorfields $\phi$ on $M$ with
\begin{align*}
    \supp(\phi) &\subset (0, u_*) \times (0, \u_*) \times S^2,
\end{align*}
we have
\begin{equation}
    \int_M D(\phi^\sharp) \cdot \Psi \, d\mu_g = -\int_M \phi \cdot \big( \Phi + \Omega\tr\chi\Psi \big)\, d\mu_g.
\end{equation}
Note that such a $\Phi$, if it exists, is unique. In this case we write $\tilde D \Psi \coloneqq \Phi$. 

    \item  An \emph{$L^2(M)$ weak $\underline D$-derivative} of $\Psi$ is an $L^2(M)$ $p$-covariant $S$ tensorfield $\Phi$ on $M$ such that for all smooth $p$-covariant $S$ tensorfields $\phi$ on $M$ with
\begin{align*}
    \supp(\phi) &\subset (0, u_*) \times (0, \u_*) \times S^2,
\end{align*}
we have
\begin{equation*}
    \int_M \underline D(\phi^\sharp) \cdot \Psi \, d\mu_g = -\int_M \phi \cdot \big( \Phi + \Omega\tr\underline\chi\Psi \big)\, d\mu_g.
\end{equation*}
Note that such a $\Phi$, if it exists, is unique. In this case we write $\tilde{\underline D} \Psi \coloneqq \Phi$. 

    \item An \emph{$L^2(M)$ weak $\sl\nabla$-derivative} of $\Psi$ is an $L^2(M)$ $(p + 1)$-covariant $S$ tensorfield $\Phi$ such that for all  smooth $(p+1)$-covariant $S$ tensorfields $\phi$ on $M$ with
\begin{align*}
    \supp(\phi) &\subset (0, u_*) \times (0, \u_*) \times S^2,
\end{align*}
we have
\begin{align*}
    \int_M \sl\nabla^B\phi_{B A_1 \cdots A_p} \cdot \Psi^{A_1 \cdots A_p}\, d\mu_g &= -\int_M \phi_{B A_1 \cdots A_p}\Phi^{B A_1 \cdots A_p}\, d\mu_g.
\end{align*}
Note that such a $\Phi$, if it exists, is unique. In this case we write $\tilde{\sl\nabla}\Psi \coloneqq \Phi$.
\end{enumerate}
\end{defn}

Our existence argument will make use of the following form of Arzel\`a-Ascoli:

\begin{lemma}[Arzel\`a-Ascoli]\label{lem:A-A} Let $X$ be a compact Hausdorff space and $(Y, d)$ a metric space. A family $\mathscr{F}$ of continuous functions from $X$ to $Y$ is precompact in the compact-open topology if and only if it is $d$-equicontinuous and pointwise precompact.
\end{lemma}

In the following theorem, we will use the following notation. If $\Psi^{(i)}_{(S)}, \underline\Psi^{(i)}_{(S)}$ are $(u,\u)$-dependent tensorfields on ${S_{0,0}}$, we will denote by $\Psi^{(i)}, \underline\Psi^{(i)}$ the ${S_{0,0}}$ tensorfields on $M$ defined by
\begin{align*}
    (\Psi^{(i)})_x &= (\Theta_{u,\u}^{-1})^*(\Psi^{(i)}_{(S)}[u,\u])\big|_x, & (\underline\Psi^{(i)})_x &= (\underline\Theta_{\u,u}^{-1})^*(\underline\Psi^{(i)}_{(S)}[u,\u])\big|_x  & \forall x \in M, x \in S_{u,\u}.
\end{align*}
(See the discussion before Definition \ref{defn:pullback_differential_operators}.)
We will also denote by $\Psi^{(i)}_{(C)}$ the $u$-dependent $S$ tensorfield on $C_0$ defined by
\begin{align*}
    \Psi^{(i)}_{(C)}[u] &= \underline\Phi_u^*(\Psi^{(i)}|_{C_u}),
\end{align*}
and by $\underline\Psi^{(i)}_{(\underline C)}$ the $\u$-dependent $S$ tensorfield on $\underline C_0$ defined by 
\begin{align*}
    \underline\Psi^{(i)}_{(\underline C)}[\u] &= \Phi_{\u}^*(\underline\Psi^{(i)}|_{\underline C_{\u}}).
\end{align*}
These tensorfields will be important to formulating the notion of a weak solution having the correct initial value. 

We remark that since $C_0$ is compact, the spaces $L^p(C_0, d\mu_u)$ consist exactly of the same functions as  $L^p(C_0, d\mu_0)$, where 
\begin{align*}
    d\mu_u &\coloneqq d\mu_{\gamma_{u,\u}}\, d\u. 
\end{align*}
Similarly, the spaces $L^p(\underline C_0, d\underline\mu_{\u})$ consist exactly of the same functions as $L^p(\underline C_0, d\underline\mu_0)$, where
\begin{align*}
    d\underline\mu_{\u} &\coloneqq d\mu_{\gamma_{u,\u}}\, du.
\end{align*}

\begin{thm}[Global existence and uniqueness]\label{thm:GWP_H^1}
Let $(M, g)$ be a smooth Lorentzian manifold of the form described in Section \ref{subsec:spacetime_and_notation}, for any $u_*, \u_* > 0$. Let
\begin{align*}
    \Psi^{(i)}_{(S),0} &\in C^0([0, \u_*]; H^1_{p_i}({S_{0,0}})) \\ 
    \underline\Psi^{(i)}_{(S),0} &\in C^0([0, u_*]; H^1_{\underline p_i}({S_{0,0}})).
\end{align*}
Then there exists a unique $\Psi_{(S)}^{(i)}, \underline\Psi_{(S)}^{(i)} : D \to {H^1({S_{0,0}})}$, such that for all $(u, \u) \in D$:
\begin{enumerate}
    \item $\Psi^{(i)}, \underline\Psi^{(i)} \in L^2(M)$ are weak solutions of \eqref{eq:general_hyperbolic_system}.
    \item $\Psi^{(i)}_{(C)} \in C([0, u_*]; L^2(C_0))$ and $\underline\Psi^{(i)}_{(C)} \in C([0, \u_*]; L^2(\underline C_0))$. 
    \item The solutions agree with the initial data, that is, 
    \begin{align*}
        \Psi^{(i)}_{(S)}[0, \u] &= \Psi^{(i)}_{(S), 0}[\u] \qquad \text{and} \qquad \underline\Psi^{(i)}_{(S)}[u, 0] = \underline\Psi^{(i)}[u].
    \end{align*}
\end{enumerate}
\end{thm}

\begin{remark}
Note that item (2) in the conclusion of the theorem ensures that item (3) makes sense.
\end{remark}

\begin{proof}
Extend $\Psi^{(i)}_{(S),0}, \underline\Psi^{(i)}_{(S),0}$ to be zero outside of $[0, \u_*], [0, u_*]$, respectively. Let $\eta_\eps : \R \to \R$ be a standard mollifier on $\R$. Define the mollified initial data
\begin{align*}
    \Psi^{(i)}_{(S),0,\eps}[\u] &= \eta_\eps * \big(J^\eps_{0,\u}\Psi^{(i)}_{(S),0}[\u]\big) \\ 
    \underline\Psi^{(i)}_{(S),0,\eps}[u] &= \eta_\eps * \big(\underline J^\eps_{0, u}\underline\Psi^{(i)}_{(S),0}[u]\big).
\end{align*}
Note that $\Psi^{(i)}_{(S),0,\eps} \to \Psi^{(i)}_{(S),0}$ strongly in $L^2([0, \u_*]; H^1_{p_i}({S_{0,0}}))$ and $\underline\Psi^{(i)}_{(S),0,\eps} \to \underline\Psi^{(i)}_{(S),0}$ strongly in $L^2([0, u_*]; H^1_{\underline p_i}({S_{0,0}}))$ as $\eps \to 0$. By applying Theorem \ref{thm:picard_1} to \eqref{eq:mollified_equations}, we obtain for every $\eps \in (0,\delta_0/2)$ a unique smooth solution $\Psi^{(i)}_{(S),\eps}, \underline\Psi^{(i)}_{(S),\eps}$ to \eqref{eq:mollified_equations} with initial data $\Psi^{(i)}_{(S),0,\eps}, \underline\Psi^{(i)}_{(S),0,\eps}$. As in the proposition statement, we let $\Psi^{(i)}_{\eps}, \underline\Psi^{(i)}_{\eps}$ denote the pullbacks
\begin{align*}
    \Psi^{(i)}_{\eps} = (\Theta_{u,\u}^{-1})^* \Psi^{(i)}_{(S),\eps}, \qquad \underline\Psi^{(i)}_{\eps} = (\underline\Theta_{\u,u}^{-1})^*\underline\Psi^{(i)}_{(S),\eps}.
\end{align*}
Then we have, for $j = 0, 1$:
\begin{align*}
    \frac{1}{2}\int_M |\sl\nabla^j\Psi^{(i)}_\eps|^2\, d\mu_g &= \int_0^{u_*}\mathcal{F}^*[\sl\nabla^j\PsiSiEps](u)\, du \\ 
    \frac{1}{2}\int_M |\sl\nabla^j\underline\Psi^{(i)}_\eps|^2\, d\mu_g &= \int_0^{\u_*}\underline{\mathcal{F}}^*[\sl\nabla^j\PsiBarSiEps](\u)\, d\u.
\end{align*}
Note that, since $\eta_\eps$ is a standard mollifier, and by Proposition \ref{prop:uniform_bound_1}, the size of the mollified initial data is controlled by the size of the original initial data, i.e. 
\begin{equation*}
    \sum_{j = 0}^1\sum_{i = 1}^n \mathcal{F}^*_0[\sl\nabla^j\Psi^{(i)}_{(S),0,\eps}] + \mathcal{F}^*_0[\sl\nabla^j\underline\Psi^{(i)}_{(S),0,\eps}] \\ 
    \leq C\mathfrak{B}\sum_{j = 0}^1\sum_{i = 1}^n \mathcal{F}^*_0[\sl\nabla^j\Psi^{(i)}_{(S),0}] + \mathcal{F}^*_0[\sl\nabla^j\underline\Psi^{(i)}_{(S),0}].
\end{equation*}
Therefore, by Propositions \ref{prop:mollified_basic_energy_estimates} and \ref{prop:higher_order_energy_estimates}, there is a constant $\mathcal{C}$, depending on $\sl C_1, \mathfrak{B}, \mathfrak{C}$, and the size of the initial data 
\begin{align*}
    \sum_{j = 0}^1\sum_{i = 1}^n \mathcal{F}^*_0[\sl\nabla^j\Psi^{(i)}_{(S),0}] + \mathcal{F}^*_0[\sl\nabla^j\underline\Psi^{(i)}_{(S),0}],
\end{align*} 
such that for all $(u, \u) \in D$ and all $\eps \in (0, \delta_0/2)$,
\begin{align*}
    \sum_{j = 0}^1\sum_{i = 1}^N \frac{1}{2}\int_M  |\sl\nabla^j \Psi^{(i)}_\eps|^2 + |\sl\nabla^j\underline\Psi^{(i)}_\eps|^2 \, d\mu_g &\leq \mathcal{C}(u_* + \u_*). \eqcount\label{eq:computation_20}
\end{align*}
Therefore there exist $L^2(M)$ covariant $S$ tensorfields $\Psi^{(i)}, \underline\Psi^{(i)}$,  with weak $\sl\nabla$-derivatives $\tilde{\sl\nabla}\Psi^{(i)}, \tilde{\sl\nabla}\underline\Psi^{(i)}$ belonging to $L^2(M)$, and a subsequence $\eps_n$ tending to zero, such that
\begin{equation}\label{eq:computation_54}
\begin{split}
    \Psi^{(i)}_{\eps_n} &\rightharpoonup \Psi^{(i)}, \qquad \sl\nabla\Psi^{(i)}_{\eps_n} \rightharpoonup \tilde{\sl\nabla}\Psi^{(i)}  \\ 
    \underline\Psi^{(i)}_{\eps_n} &\rightharpoonup \underline\Psi^{(i)}, \qquad{\sl\nabla}\underline\Psi^{(i)}_{\eps_n} \rightharpoonup \tilde{\sl\nabla}\underline\Psi^{(i)},
\end{split}
\end{equation}
the convergence here being weak convergence in $L^2(M)$. Thus, for any smooth covariant $S$ tensorfield $\phi$ with support in $(0, u_*) \times (0, \u_*) \times S^2$, 
\begin{multline}\label{eq:computation_21}
    -\int_M \underline D(\phi^\sharp)\cdot \Psi^{(i)} - \sld_{\underline\Psi^{(i)}}(\Omega\phi)\cdot \underline\Psi^{(i)} + \Omega\tr\underline\chi\phi \cdot \Psi^{(i)} + \Omega\phi\cdot E^{(i)}\, d\mu_g \\
    = -\lim_{\eps_n \to 0} \int_M \underline D(\phi^\sharp)\cdot \Psi_{\eps_n}^{(i)} - \sld_{\underline\Psi^{(i)}}(\Omega\phi)\cdot \underline\Psi_{\eps_n}^{(i)} + \Omega\tr\underline\chi\phi \cdot \Psi_{\eps_n}^{(i)} + \Omega\phi\cdot E_{\eps_n}^{(i)}\, d\mu_g.
\end{multline}
By pulling back to ${S_{0,0}}$ and applying Lemmas \ref{lem:spherical_derivative_almost_commutation}-\ref{lem:lower_order_pullback}, and using the fact that $\PsiSiEps$ satisfies \eqref{eq:mollified_equations}, the integral on the right-hand side of \eqref{eq:computation_21} is equal to 
\begin{multline*}
    \int_0^{u_*} \int_0^{\u_*} \int_{S_{0,0}} \Theta_{u,\u}^*\phi \cdot \Big[J_{u,\u}^{\eps_n} \Big[ (\Theta_{u,\u}^*\Omega) \cdot A_{u,\u}^* \big(\big[\underline\Theta_{\u,u}^*\sl{\mathcal{D}}_{\Psi^{(i)}}\big](\underline J_{\u,u}^{\eps_n} \underline\Psi_{(S),\eps_n}^{(i)})  \big)\Big] \\ 
    - (\Theta_{u,\u}^*\Omega)\cdot A_{u,\u}^* (\underline\Theta_{\u,u}^*\sld_{\Psi^{(i)}}) \underline\Psi_{(S),\eps_n}^{(i)} \Big]  \, d\mu_{\slg_{u,\u}}\, d\u \, du.
\end{multline*}
By \eqref{eq:computation_54} it follows that  $\underline J^\eps_{\u,u} \underline\Psi^{(i)}_{\eps_n} \rightharpoonup \Psi^{(i)}$ and also that 
\begin{align*}
    J_{u,\u}^{\eps_n} \Big[ (\Theta_{u,\u}^*\Omega) \cdot A_{u,\u}^* \big(\big[\underline\Theta_{\u,u}^*\sl{\mathcal{D}}_{\Psi^{(i)}}\big](\underline J_{\u,u}^{\eps_n} \underline\Psi_{(S),\eps_n}^{(i)})  \big)\Big] \rightharpoonup (\Theta_{u,\u}^*\Omega) \cdot A_{u,\u}^* \big(\big[\underline\Theta_{\u,u}^*\sl{\mathcal{D}}_{\Psi^{(i)}}\big]( \underline\Psi_{(S)}^{(i)})  \big),
\end{align*}
where here the convergence is weak convergence in $L^2(D \times S^2)$. Therefore the right-hand side of \eqref{eq:computation_21} is zero. Similarly one shows the analogous statement for $\underline\Psi^{(i)}$; therefore $\Psi^{(i)}, \underline\Psi^{(i)}$ weakly solve \eqref{eq:general_hyperbolic_system}, which proves (1).

We now show (2). Let $\mathcal{C}_1$ be such that $\|\Psi^{(i)}_{(C),\eps}[u]\|_{L^2(C_0)}, \|\underline\Psi^{(i)}_{(\underline C),\eps}[\u]\|_{L^2(\underline C_0)} \leq \mathcal{C}_1$ for all $(u, \u) \in D$ (that this is possible follows from the energy estimates). Let $(Y, [\cdot])$ be the normed space
\begin{align*}
    Y &= \{h \in L^2(C_0) \ | \ \norm{h}_{L^2(C_0)} \leq \mathcal{C}_1\}, \\
    [h] &= \sum_{n = 1}^\infty \frac{1}{2^n}|\langle h, y_n \rangle_{L^2(C_0)}|,
\end{align*}
where $y_n$ is a countable dense subset of $Y$. This metrizes the weak topology on $Y$. Let 
\begin{align*}
    \mathscr{F} &= \{\Psi^{(i)}_{(C),\eps}\}_{\eps > 0} \subset C([0, u_*]; Y).
\end{align*}
Note that there is a constant $C > 0$ depending on $(M, g)$ such that for all $u \in [0, u_*]$,
\begin{align*}
    \norm{\Psi^{(i)}_{(C),\eps}[u]}_{L^2(C_0, d\mu_0)}^2 &\leq C\int_{C_0}|\Psi^{(i)}_{(C),\eps}[u]|^2\, d\mu_u \\ 
    &= C\mathcal{F}^*[\Psi^{(i)}_\eps](u) \\ 
    &\leq C \mathcal{C}(u_* + \u_*).
\end{align*}
Hence since bounded sets in $L^2(C_0, d\mu_0)$ are weakly precompact, $\mathscr{F}$ is pointwise precompact in $(Y, d)$. We also have
\begin{align*}
    \norm{\Psi^{(i)}_{(C),\eps}(u) - \Psi^{(i)}_{(C),\eps}(u')}_{L^2(C_0, d\mu_0)} &\leq \int_{u'}^u \norm{\d_u\Psi^{(i)}_{(C),\eps}(s)}_{L^2(C_0, d\mu_0)}\, ds \\
    &\leq C\int_{u'}^u \norm{\underline D\Psi^{(i)}_\eps(s)}_{L^2(C_s)}\, ds \\
    &\leq C\norm{\underline D\Psi^{(i)}_\eps}_{L^2(M)}|u - u'|^{1/2} \\ 
    &\leq C\mathcal{C}(u_* + \u_*)|u - u'|^{1/2}.
\end{align*}
The last line follows due to \eqref{eq:computation_20} and the mollified equations \eqref{eq:mollified_equations}.
Therefore $\mathscr{F}$ is equicontinuous. By Arzel\`a-Ascoli, there exists a subsequence $\eps_n \to 0$ such that, in addition to the weak convergence properties above, we have $\Psi^{(i)}_{(C)} \in C([0, u_*]; (Y, d))$, and
\begin{equation}\label{eq:computation_22}
\Psi^{(i)}_{(C),\eps_n} \xrightarrow[(Y,d)]{} \Psi^{(i)}_{(C)} \qquad \text{uniformly in $u \in [0, u_*]$}.
\end{equation} 
Now, by construction, the mollified initial data $\Psi^{(i)}_{\eps}|_{u = 0} \to \Psi^{(i)}_{0}$ converges strongly in $L^2(C_0, d\mu_0)$. Since weak limits are unique, and since in particular \eqref{eq:computation_22} implies $\Psi^{(i)}_{(C),\eps_n}[0]\rightharpoonup \Psi^{(i)}_{(C)}[0]$ weakly, we must in fact have 
\begin{align*}
    \Psi^{(i)}_{(C)}[0] &= \Psi^{(i)}_{0}.
\end{align*}
The analogous argument holds to show that $\underline\Psi^{(i)}_{(\underline C)} \in C([0, \u_*]; (\underline Y, \underline d))$, where 
\begin{align*}
    \underline Y &= \{\underline h \in L^2(\underline C_0) \ | \ \norm{\underline h}_{L^2(C_0)} \leq \mathcal{C}_1\}, \\
    [\underline h] &= \sum_{n = 1}^\infty \frac{1}{2^n}|\langle h, \underline y_n \rangle_{L^2(\underline C_0)}|,
\end{align*}
where $\underline  y_n$ is a countable dense set in $\underline Y$, and also that 
\begin{align*}
    \underline\Psi^{(i)}_{(\underline C)}[0] &= \underline\Psi^{(i)}_{0}.
\end{align*}
This shows (3). To upgrade the continuity from the ``weak'' spaces $(Y, d)$ and $(\underline Y, \underline d)$ to the continuity claimed in (2), we note that by the equations \eqref{eq:general_hyperbolic_system}, $\Psi^{(i)}_{(C)} \in H^1([0,u_*]; L^2(C_0))$ and$\underline\Psi^{(i)}_{(\underline C)} \in H^1([0,\u_*]; L^2(\underline C_0))$. These conditions imply (2) (see \cite[Theorem 5.9.2.2]{evans_pde}).

It remains to note that $\Psi^{(i)}, \underline\Psi^{(i)}$ are unique, which follows because the equations are linear and the energy estimates in Proposition \ref{prop:basic_energy_estimates_1}. 
\end{proof}

\section{The linearized Bianchi equations and the algebraic constraints}\label{sec:linearized_bianchi_equations}

\subsection{Preliminaries}\label{subsec:LNB_preliminaries} We now restrict our attention to the case when $(M, g)$ is a \emph{vacuum} spacetime. The main result of this section is the existence of \emph{algebraic constraints} on solutions of the linearized Bianchi equations and the explicit form of these constraints (see Theorem \ref{thm:algebraic_constraints_1}). Note also that these constraints are likely to enter a potential future proof of well-posedness for the characteristic initial value problem for the linearized Bianchi equations.

The {linearized Bianchi equations} are obtained from the Bianchi equations \eqref{eq:null_Bianchi_equations} by replacing the null curvature components $\alpha[W], \beta[W], \rho[W], \sigma[W], \underline\beta[W], \underline\alpha[W]$ with unknowns $\alpha, \beta, \rho, \sigma, \underline\beta, \underline\alpha$, where
\begin{itemize}
    \item $\alpha,\underline\alpha$ are symmetric traceless 2-covariant $S$ tensorfields,
    \item $\beta,\underline\beta$ are $S$ 1-forms, and
    \item $\rho, \sigma$ are scalars.
\end{itemize}
Explicitly, the linearized Bianchi equations are the following system of ten equations for the unknowns $\alpha, \beta, \rho, \sigma, \underline\beta, \underline\alpha$:
\begin{equation}\label{eq:linearized_bianchi_equations}
\begin{split}
    \sl\nabla_3 \alpha &= (4 \underline\omega - \frac{1}{2}\tr\underline\chi) \alpha + \sl\nabla\hat\otimes\beta + (4\eta + \zeta)\hat\otimes \beta - 3\hat\chi\rho - 3\prescript{*}{}{\hat{\chi}}\sigma \\
    \sl\nabla_4 \underline\alpha &= (4\omega - \frac{1}{2}\tr\chi)\underline\alpha - \sl\nabla\hat\otimes \underline\beta - (4\underline\eta - \zeta)\hat\otimes \underline\beta - 3\hat{\underline\chi}\rho + 3\prescript{*}{}{\hat{\underline\chi}}\sigma \\
    \sl\nabla_4 \beta &= -2(\tr\chi + \omega) \beta + \sl\div \alpha + \eta\cdot \alpha \\
    \sl\nabla_3 \underline\beta &= - 2(\tr\underline\chi + \underline\omega) \underline\beta - \sl\div\underline\alpha - \underline\eta\cdot \underline\alpha \\
    \sl\nabla_3 \beta &= (2\underline\omega - \tr \underline\chi) \beta + \sl\nabla \rho + \prescript{*}{}{\sl \nabla\sigma}  + 2\hat\chi \cdot \underline\beta + 3(\eta\rho + \prescript{*}{}{\eta}\sigma) \\ 
    \sl\nabla_4 \underline\beta &= (2\omega -\tr\chi)\underline\beta - \sl\nabla\rho + \prescript{*}{}{\sl\nabla\sigma} + 2\hat{\underline\chi}\cdot \beta - 3(\underline\eta \rho - \prescript{*}{}{\underline\eta}\sigma) \\ 
    \sl\nabla_4 \rho &= - \frac{3}{2}\tr\chi \rho + \sl\div \beta + (2\underline\eta + \zeta)\cdot \beta - \frac{1}{2}\hat{\underline\chi}\cdot \alpha  \\ 
    \sl\nabla_3 \rho &= -\frac{3}{2}\tr\underline\chi \rho - \sl\div \underline\beta - (2\eta - \zeta)\cdot \underline\beta -\frac{1}{2}\hat\chi\cdot \underline\alpha  \\ 
    \sl\nabla_4 \sigma &= -\frac{3}{2}\tr\chi \sigma - \sl\curl \beta - (2\underline\eta + \zeta) \cdot \prescript{*}{}{\beta} + \frac{1}{2}\hat{\underline\chi}\cdot \prescript{*}{}{\alpha} \\ 
    \sl\nabla_3 \sigma &= -\frac{3}{2}\tr\underline\chi\sigma - \sl\curl\underline\beta + (\zeta - 2\eta) \cdot \prescript{*}{}{\underline\beta} - \frac{1}{2}\hat\chi \cdot \prescript{*}{}{\underline\alpha}.  
\end{split}
\end{equation}
We note that we are not linearizing the full Einstein equations (which would include linearizing the null structure equations as well), but rather only the Bianchi equations on a fixed spacetime (in contrast to e.g. \cites{aretakis2021characteristicgluing2, dhr_linear}). 

\begin{remark}\label{rem:remark_5}
It is sometimes convenient to view $(\rho, \sigma)$ as an $\R^2$-valued unknown on $M$ rather than as two $\R$-valued unknowns. When this is done we think of the last four equations in \eqref{eq:linearized_bianchi_equations} as the two equations
\begin{equation}\label{eq:rho_sigma_equations}
\begin{split}
    \sl\nabla_4 (\rho, \sigma) &= \Big( - \frac{3}{2}\tr\chi \rho + \sl\div \beta + (2\underline\eta + \zeta)\cdot \beta - \frac{1}{2}\hat{\underline\chi}\cdot \alpha, \\ 
    &\hspace{35mm} -\frac{3}{2}\tr\chi \sigma - \sl\curl \beta - (2\underline\eta + \zeta) \cdot \prescript{*}{}{\beta} + \frac{1}{2}\hat{\underline\chi}\cdot \prescript{*}{}{\alpha}\Big)  \\ 
    \sl\nabla_3 (\rho, \sigma) &= \Big( -\frac{3}{2}\tr\underline\chi \rho - \sl\div \underline\beta - (2\eta - \zeta)\cdot \underline\beta -\frac{1}{2}\hat\chi\cdot \underline\alpha , \\ 
    &\hspace{35mm} -\frac{3}{2}\tr\underline\chi\sigma - \sl\curl\underline\beta + (\zeta - 2\eta) \cdot \prescript{*}{}{\underline\beta} - \frac{1}{2}\hat\chi \cdot \prescript{*}{}{\underline\alpha}\Big).
\end{split}
\end{equation}
\end{remark}

It is a natural question to ask whether the system of linearized Bianchi equations is well-posed, given initial data on $C_0 \cup \underline C_0$. As an initial observation, note that \eqref{eq:linearized_bianchi_equations} is an overdetermined system, as there are ten equations for six unknowns. To view it from the initial value viewpoint, four of the equations must be treated as constraints and the other six as evolution equations. Since $\alpha, \underline\alpha$ are the only unknowns which have exactly one equation, we must treat the equations for $\sl\nabla_3 \alpha$ and $\sl\nabla_4 \underline\alpha$ as evolution equations. Now, the Bianchi equations are traditionally paired in the following way:
\begin{align*}
    \sl\nabla_3\alpha  & & \text{paired with} & & \sl\nabla_4 \beta \\ 
    \sl\nabla_3\beta & & \text{paired with} & & \sl\nabla_4(\rho, \sigma) \\ 
    \sl\nabla_3 (\rho, \sigma) & & \text{paired with} & & \sl\nabla_4 \underline\beta \\ 
    \sl\nabla_3 \underline\beta & & \text{paired with} & & \sl\nabla_4 \underline\alpha.
\end{align*}
The pairing is what Taylor calls \emph{Bianchi pairing} (see also Section \ref{subsec:hyperbolic_systems}) and is essential to the hyperbolicity of the Bianchi equations in the double null foliation. By inspecting the system, however, one sees that it is impossible to choose four equations from the remaining eight (recall we have already chosen $\sl\nabla_3\alpha, \sl\nabla_4\underline\alpha$) as evolution equations in such a way that both
\begin{enumerate}
    \item every unknown has exactly one evolution equation, and
    \item if an equation has been chosen to be an evolution equation, so has the equation which is its Bianchi pair.
\end{enumerate}
This is evidently an obstruction to formulating a well-posed initial value problem for the linearized Bianchi equations \eqref{eq:linearized_bianchi_equations}, which we plan to address in future work. At the moment, we proceed by choosing one equation at a time so as to make the system of chosen evolution equations satisfy item (2) above with the exception of the equation for $\underline\alpha$.\footnote{It would be of interest to find any more natural conditions to make the choice of which equations should be considered evolution equations and which should be considered constraints.} In this way we arrive at the following system, which we call the \emph{partial Bianchi equations}:
\begin{equation}\label{eq:PBE}
\begin{split}
    \sl\nabla_3 \alpha &= (4 \underline\omega - \frac{1}{2}\tr\underline\chi) \alpha + \sl\nabla\hat\otimes\beta + (4\eta + \zeta)\hat\otimes \beta - 3\hat\chi\rho - 3\prescript{*}{}{\hat{\chi}}\sigma \\
    \sl\nabla_4 \underline\alpha &= (4\omega - \frac{1}{2}\tr\chi)\underline\alpha - \sl\nabla\hat\otimes \underline\beta - (4\underline\eta - \zeta)\hat\otimes \underline\beta - 3\hat{\underline\chi}\rho + 3\prescript{*}{}{\hat{\underline\chi}}\sigma \\
    \sl\nabla_4 \beta &= -2(\tr\chi + \omega) \beta + \sl\div \alpha + \eta\cdot \alpha \\
    \sl\nabla_4 \underline\beta &= (2\omega -\tr\chi)\underline\beta - \sl\nabla\rho + \prescript{*}{}{\sl\nabla\sigma} + 2\hat{\underline\chi}\cdot \beta - 3(\underline\eta \rho - \prescript{*}{}{\underline\eta}\sigma) \\ 
    \sl\nabla_3 \rho &= -\frac{3}{2}\tr\underline\chi \rho - \sl\div \underline\beta - (2\eta - \zeta)\cdot \underline\beta -\frac{1}{2}\hat\chi\cdot \underline\alpha  \\ 
    \sl\nabla_3 \sigma &= -\frac{3}{2}\tr\underline\chi\sigma - \sl\curl\underline\beta + (\zeta - 2\eta) \cdot \prescript{*}{}{\underline\beta} - \frac{1}{2}\hat\chi \cdot \prescript{*}{}{\underline\alpha}.  
\end{split}
\end{equation}
This system is no longer overdetermined. To keep track of the constraint equations, we define the following four \emph{differential constraints}:
\begin{equation}\label{eq:differential_constraints}
\begin{split}
    \underline B = \underline B[\underline\beta, \underline\alpha] &= \sl\nabla_3 \underline\beta + \sl\div\underline\alpha + \underline\eta \cdot \underline\alpha + 2(\tr\underline\chi + \underline\omega)\underline\beta \\ 
    \Xi = \Xi[\beta, \rho, \sigma, \underline\beta] &= \sl\nabla_3 \beta + \tr \underline\chi \beta - 2{\underline\omega}\beta - \sl\nabla \rho - \prescript{*}{}{\sl \nabla\sigma}  - 2\hat\chi \cdot \underline\beta - 3(\eta\rho + \prescript{*}{}{\eta}\sigma) \\ 
    P = P[\alpha, \beta, \rho] &= \sl\nabla_4 \rho + \frac{3}{2}\tr\chi \rho - \sl\div \beta - (2\underline\eta + \zeta, \beta) + \frac{1}{2}(\hat{\underline\chi}, \alpha)  \\ 
    Q = Q[\alpha, \beta, \sigma] &= \sl\nabla_4 \sigma + \frac{3}{2}\tr\chi \sigma + \sl\curl \beta + (2\underline\eta + \zeta) \wedge \beta - \frac{1}{2}\hat{\underline\chi}\wedge \alpha.
\end{split}
\end{equation}
\begin{remark}
Solutions $(\alpha, \beta, \rho, \sigma, \underline\beta, \underline\alpha)$ of the full linearized Bianchi equations \eqref{eq:linearized_bianchi_equations} are exactly those solutions of the partial Bianchi equations \eqref{eq:PBE} for which the differential constraints vanish.
\end{remark}

One immediate drawback is that \eqref{eq:PBE} is not hyperbolic as written, since $\sl\nabla_4\underline\alpha$ has no Bianchi pair within the system (it is normally paired with $\sl\nabla_3 \underline\beta$). For this reason, \eqref{eq:PBE} is not a double null hyperbolic system, and so the theory developed for these systems in this paper does not apply. 

As a preliminary to the future goal of addressing well-posedness of \eqref{eq:linearized_bianchi_equations}, we now investigate some necessary conditions for well-posedness from the point of view of the initial value problem of \eqref{eq:PBE}. One condition on the initial data for \eqref{eq:PBE} which is manifestly necessary to solve \eqref{eq:linearized_bianchi_equations} is that the differential constraints for the initial data vanish.

\subsection{Initial data for the partial Bianchi equations} \label{subsec:initial_data} We now discuss initial data for \eqref{eq:PBE}. The main concepts in this section are the \emph{full initial data set} and the \emph{seed initial data set}. 

\begin{defn}
A \emph{full initial data set} for \eqref{eq:PBE} consists of 
\begin{itemize}
    \item On $\underline C_0$: a symmetric traceless 2-covariant $S$ tensorfield $\underline\alpha_0$ and two $S$ 1-forms $\underline\beta_0, \beta_0$
    \item On $C_0$: a symmetric traceless 2-covariant $S$ tensorfield $\alpha_0$ and two scalar functions $\rho_0, \sigma_0$
\end{itemize}
such that:
\begin{enumerate}
    \item On $\underline C_0$, $\underline B[\underline\beta_0, \underline\alpha_0] = 0$ and $\Xi[\beta_0, \tilde\rho_0, \tilde\sigma_0, \underline\beta_0] = 0$, where the functions $\tilde\rho_0, \tilde\sigma_0 : \underline C_0 \to \R$ are defined as the unique solution to the ODE
    \begin{align*}
        \sl\nabla_3 (\tilde\rho_0, \tilde\sigma_0) &= \Big( -\frac{3}{2}\tr\underline\chi \tilde\rho_0 - \sl\div \underline\beta_0 - (2\eta - \zeta)\cdot \underline\beta_0 -\frac{1}{2}\hat\chi\cdot \underline\alpha_0 , \\ 
        &\qquad -\frac{3}{2}\tr\underline\chi\tilde\sigma_0 - \sl\curl\underline\beta_0 + (\zeta - 2\eta) \cdot \prescript{*}{}{\underline\beta_0} - \frac{1}{2}\hat\chi \cdot \prescript{*}{}{\underline\alpha_0}\Big)
    \end{align*}
    with initial data 
    \begin{align*}
        (\tilde\rho_0, \tilde\sigma_0)|_{{S_{0,0}}} &= (\rho_0, \sigma_0)|_{S_{0,0}}. 
    \end{align*}
    We remark that the ODE for $(\tilde\rho_0, \tilde\sigma_0)$ is a genuine ODE along the integral curves of $\underline L$ on $\underline C_0$, since $\underline\beta_0, \underline\alpha_0$, and the Ricci coefficients are already defined on $\underline C_0$, and the initial data $(\rho_0, \sigma_0)$ is already defined on ${S_{0,0}}$.
    \item On $C_0$, $P[\alpha_0, \tilde\beta_0, \rho_0] = 0$ and $Q[\alpha_0, \tilde\beta_0 ,\sigma_0] = 0$, where the $S$ 1-form $\tilde\beta_0$ on $C_0$ is defined as the unique solution to the ODE
    \begin{equation*}
        \sl\nabla_4 \tilde\beta_0 = -2(\tr\chi + \omega)\tilde\beta_0 + \sl\div\alpha_0 + \eta \cdot\alpha_0
    \end{equation*}
    with initial data 
    \begin{align*}
        (\tilde\beta_0)|_{S_{0,0}} = (\beta_0)|_{S_{0,0}}.
    \end{align*}
    We remark that the ODE for $\tilde\beta_0$ is a genuine ODE along the integral curves of $L$ on $C_0$, since $\alpha_0$ and the Ricci coefficients are already defined on $C_0$, and the initial data $(\beta_0)|_{S_{0,0}}$ is already defined on ${S_{0,0}}$.
\end{enumerate}
\end{defn}

\begin{defn}\label{defn:seed_initial_data}
A \emph{seed initial data set} for \eqref{eq:PBE} consists of
\begin{itemize}
    \item On $\underline C_0$: a symmetric traceless 2-covariant $S$ tensorfield $\underline\alpha_0$, 
    \item On $C_0$: a symmetric traceless 2-covariant $S$ tensorfield $\alpha_0$, 
    \item On ${S_{0,0}}$: two 1-forms $\beta_0, \underline\beta_0$ and two scalar functions $\rho_0, \sigma_0$. 
\end{itemize}
\end{defn}
Note that a seed initial data set entails no constraints on the initial data. As the name suggests, a full initial data set can be constructed from a seed initial data set. This can be done as follows. First, define $\tilde\beta_0$ on $C_0$ as described in (2) above; since the equation for $\sl\nabla_4 \tilde\beta_0$ depends only on $\alpha_0$, and in a seed initial data set $\alpha_0$ is prescribed freely on all of $C_0$, this can be done. Then extend $\rho_0, \sigma_0$ to $C_0$ by solving the ODEs 
\begin{align*}
P[\alpha_0, \tilde\beta_0, \rho_0] &= 0 & Q[\alpha_0, \tilde\beta_0, \sigma_0] &= 0 \\ 
(\rho_0)|_{S_{0,0}} &= \rho_0  & (\sigma_0)|_{S_{0,0}} &= \sigma_0
\end{align*}
(here we also denote by $\rho_0, \sigma_0$ the extension to $C_0$ of these functions). These are decoupled ODEs for $\rho_0, \sigma_0$. This defines the components of the full initial data set which are prescribed on $C_0$. We extend $\underline\beta_0$ to $\underline C_0$ by solving the ODE 
\begin{align*}
    \underline B[\underline\beta_0, \underline\alpha_0] &= 0 \\ 
    (\underline\beta_0)|_{S_{0,0}} &= \underline\beta_0
\end{align*}
(again letting $\underline\beta_0$ denote also the extension to $\underline C_0$). This is an ODE for $\underline\beta_0$, since $\underline\alpha_0$ has been prescribed on $\underline C_0$. Next define $\tilde\rho_0, \tilde\sigma_0$ on $\underline C_0$ as described in (1) above. Finally, we extend $\beta_0$ to $\underline C_0$ by solving the ODE 
\begin{align*}
    \Xi[\beta_0, \tilde\rho_0, \tilde\sigma_0, \underline\beta_0] &= 0 \\ 
    (\beta_0)|_{S_{0,0}} &= \beta_0
\end{align*}
(again letting $\beta_0$ denote also the extension to $\underline C_0$). This defines the full initial data set components lying on $\underline C_0$. By  construction, all assumptions of the definition of a full initial data set are fulfilled. 

\begin{remark}
The ODEs $\underline B = 0, \Xi = 0, P = 0$, and $Q = 0$ are analogous to the well-known vacuum Einstein constraints for characteristic initial data, which take the form of ODEs along the initial null hypersurfaces \cites{rendall1990characteristic_problem}.
\end{remark}

\begin{defn}
Let $\hat S_p$ denote the space of all symmetric traceless 2-covariant tensors on $T_p S_{u, \u}$ (with $p \in S_{u, \u}$). Define 
\begin{align*}
    \mathcal{V}_p &\coloneqq \hat S_p \times T_p^* S_{u, \u} \times \R \times \R \times T_p^* S_{u, \u} \times \hat S_p.
\end{align*}
Also, let $\mathcal{V}$ denote the vector bundle over $M$ whose fiber at every point $p \in M$ is $\mathcal{V}_p$. 
\end{defn}
Then note that the collection of unknowns $(\alpha, \beta, \rho, \sigma, \underline\beta,\underline\alpha)$ can be thought of as sections of $\mathcal{V}$. Note that $\text{rank}(\mathcal{V}) = 10$.

\subsection{Evolution of the differential constraints}\label{subsec:differential_constraints} In this section we derive the differential equations satisfied by the constraints \eqref{eq:differential_constraints} in null directions under the assumption that the system \eqref{eq:PBE} is satisfied. We will see as a consequence that the right-hand side of these differential equations is homogeneous in the differential constraints and their angular derivatives, plus a term that only vanishes if certain \emph{algebraic constraints} are satisfied by the unknowns. 

Before we begin, we recall the following formula, which is found in equations (2.2.2a -- 2.2.2d) in \cite{CK}: 

\begin{lemma}\label{lem:formula_1}
For an $S$ 1-form $\theta$, it holds that 
\begin{equation}
    \sl\div(\sl\nabla\hat\otimes \theta) + \prescript{*}{}{\sl\nabla\curl}\theta - \sl\nabla\sl\div\theta = 2K\theta.
\end{equation}
\end{lemma} 
The following theorem is the main result of this section.

\begin{thm}[Propagation equations for the differential constraints]\label{thm:prop_eq_for_diff_constraints} Let $(M, g)$ be a vacuum spacetime of the type described in Section \ref{subsec:spacetime_and_notation}. Suppose that $\alpha, \beta, \rho, \sigma, \underline\beta, \underline\alpha$ satisfy the partial Bianchi equations \eqref{eq:PBE}. Then the differential constraints obey the following ODEs: 
\begin{multline}
    \sl\nabla_4 \underline B = (4 \omega - \tr\chi)\underline B  + 2 \hat{\underline\chi} \cdot \Xi \\ 
    + 2(\underline\alpha \cdot \beta[W] - \underline\alpha[W]\cdot \beta) + 6(\sigma \prescript{*}{}{\underline\beta[W]} - \sigma[W]\prescript{*}{}{\underline\beta}) + 6(\rho[W]\underline\beta - \rho \underline\beta[W])) 
\end{multline}
\begin{multline}
    \sl\nabla_4 \Xi = -2\tr\chi\Xi - \sl\nabla P - \prescript{*}{}{\sl\nabla Q} - \big( \frac{7}{2}\eta + \frac{1}{2}\underline\eta \big) P - \big( \frac{7}{2}\prescript{*}{}{\eta} + \frac{1}{2}\prescript{*}{}{\underline\eta} \big) Q \\
    + 2(\alpha[W]\cdot \underline\beta - \alpha \cdot\underline\beta[W]) + 6 (\rho \beta[W] - \rho[W]\beta) + 6(\sigma\prescript{*}{}{\beta[W]} - \sigma[W]\prescript{*}{}{\beta})
\end{multline}
\begin{multline}
    \sl\nabla_3 P = (2\underline\omega - \frac{3}{2}\tr\underline\chi) P - (\eta + 2\underline\eta) \cdot \Xi - \sl\div \Xi \\ 
    + \frac{1}{2}\big( \alpha[W]\cdot \underline\alpha - \underline\alpha[W]\cdot\alpha \big) + 2\big( \underline\beta \cdot \beta[W] - \beta\cdot\underline\beta[W] \big)
\end{multline}
\begin{multline}
    \sl\nabla_3 Q = (2\underline\omega - \frac{3}{2}\tr\underline\chi)Q - \prescript{*}{}{(\eta + 2\underline\eta)}\cdot \Xi + \sl\curl\Xi  \\
    +   \frac{1}{2}\big(\prescript{*}{}{\underline\alpha}\cdot\alpha[W] - \prescript{*}{}{\underline\alpha}[W]\cdot\alpha \big) + 2\big( \prescript{*}{}{\underline\beta}\cdot\beta[W] - \prescript{*}{}{\underline\beta}[W]\cdot\beta \big).
\end{multline}
\end{thm}
The proof of this theorem proceeds by taking term-by-term null derivatives of the differential constraints. One uses the partial Bianchi equations \eqref{eq:PBE} and the differential constraints \eqref{eq:differential_constraints} to write null derivatives of the unknowns as a sum of terms of the form  $\psi \cdot\Psi$ or $\sl\nabla\Psi$ (for $\psi$ a Ricci coefficient and $\Psi$ an unknown); if a differential constraint is used, then this differential constraint must also be included in the resulting expression. One uses the null structure equations \eqref{eq:null_structure_propagation}-\eqref{eq:null_structure_constraint} to write null derivatives of the Ricci coefficients as a sum of terms of the form $\psi^2$, $\sl\nabla\psi$, or $\Psi[W]$ (for $\Psi[W]$ a null Weyl tensor component). Due to the large number of terms that appear, it is convenient to visualize the computation via a tree structure. Each node of the tree is a term in the final expression for the derivative of the given differential constraint. If this node can be expanded by the null structure equations, the partial Bianchi equations, or the differential constraints, then we draw an edge from this node to its expanded expression. If a node cannot be substituted further, it is called a \emph{leaf}. In this way we can collect all terms in a given expression in an orderly manner to discover cancellations and additional structure. All leaves are then added together in the end; all minus signs will be kept track of within each node.

Furthermore, it is convenient to consider separately terms of distinct orders. In this context, by ``order'', we are referring to the highest derivative of an unknown $\alpha, \beta, \rho, \sigma, \underline\beta, \underline\alpha$ appearing in a given term. For instance, 
\begin{align*}
    \text{ord}(\omega \eta \cdot \alpha) &= 0, \quad \text{ord}(\rho\sl\div\hat{\underline\chi}) = 0,  \quad \text{ord}(\eta \cdot \sl\div \underline\beta) = 1, \quad \text{ord}(\sl\nabla\sl\div \underline\beta) = 2. \eqcount\label{eq:computation_37}
\end{align*}
It is convenient to also consider the differential constraints $\underline B, \Xi, P, Q$ themselves as order 1 and any spherical derivative of them order 2. 

For example consider the differential constraint $\underline B$. Note that
\begin{align*}
    \sl\nabla_4 \underline B = \sl\nabla_4 \sl\nabla_3 \underline B + \sl\nabla_4 \sl\div\underline\alpha + \sl\nabla_4(\underline\eta \cdot \underline\alpha) + 2\sl\nabla_4\big((\tr\underline\chi + \underline\omega)\underline\beta\big).
\end{align*}
The order 1 and 2 tree expansion for $\sl\nabla_4 \underline B$ is shown in Figure \ref{fig:Bbar_tree_1_2}, and the order 0 tree expansion is shown in Figure \ref{fig:Bbar_tree_0}. These figures can be used to deduce the computations described in the proof of this theorem. 
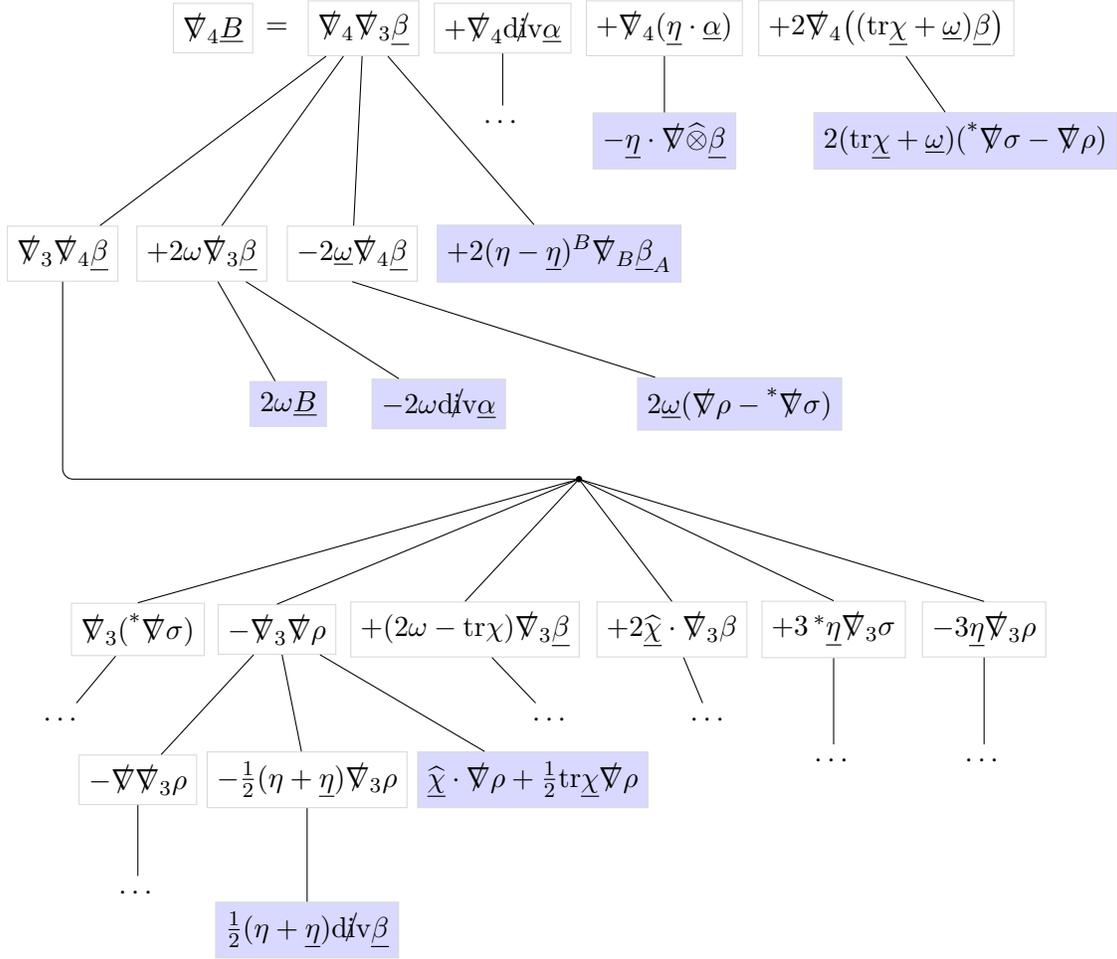
\begin{figure}[!ht]
    \centering
\begin{tikzpicture}
    
    \node [intermediate] at (-5, 0) {$\sl\nabla_4 \underline B$ };
    \node [no_dec] at (-4.15, 0) {$=$};
    \node (t11) [intermediate] at (-3.0, 0) {$\sl\nabla_4 \sl\nabla_3 \underline\beta$}; 
    \node (t12) [intermediate] at (-1.15, 0) {$+ \sl\nabla_4 \sl\div\underline\alpha$};
    \node (t12dots) [no_dec] at(-1.15, -1.25) {$\cdots$};
    \node (t13) [intermediate] at (1.0,0) {$+ \sl\nabla_4(\underline\eta \cdot\underline\alpha)$}; 
    \node (t14) [intermediate] at (3.95, 0) {$+ 2\sl\nabla_4\big((\tr\underline\chi + \underline\omega)\underline\beta\big)$}; 

    \node (t31) [finalterm] at (1, -1.5) {$-\underline\eta \cdot \sl\nabla\hat\otimes\underline\beta$}; 
    \node (t32) [finalterm] at (5, -1.5) {$2(\tr\underline\chi + \underline\omega)(\prescript{*}{}{\sl\nabla\sigma} - \sl\nabla\rho)$};

    \node (t21) [intermediate] at (-6 - 1, -3) {$\sl\nabla_3 \sl\nabla_4 \underline\beta$};
    \node (t22) [intermediate] at (-4.15 - 1, -3) {$+ 2\omega\sl\nabla_3\underline\beta$}; 
    \node (t23) [intermediate] at (-2.15 - 1, -3) {$-2\underline\omega \sl\nabla_4 \underline\beta$}; 
    \node (t24) [finalterm] at (.6 - 1, -3) {$ + 2(\eta - \underline\eta)^B \sl\nabla_B \underline\beta_A$};

    \node (t41) [finalterm] at (2, -5) {$2\underline\omega(\sl\nabla\rho - \prescript{*}{}{\sl\nabla\sigma})$};
    \node (t51) [finalterm] at (-4, -5) {$2\omega \underline B$};
    \node (t52) [finalterm] at (-2, -5) {$-2\omega\sl\div\underline\alpha$};

    \node (a) at (0, -6) {};

    \node (t61) [intermediate] at (-6, -8) {$\sl\nabla_3(\prescript{*}{}{\sl\nabla\sigma})$};
    \node (t61dots) [no_dec] at (-7, -9.2) {$\cdots$};
    \node (t62) [intermediate] at (-4.15, -8) {$ - \sl\nabla_3{\sl\nabla\rho}$};
    \node (t63) [intermediate] at (-1.65, -8) {$+(2\omega - \tr\chi)\sl\nabla_3\underline\beta$};
    \node (t63dots) [no_dec] at (-0.5, -9.2) {$\cdots$};
    \node (t64) [intermediate] at (1.1, -8) {$+ 2\hat{\underline\chi}\cdot\sl\nabla_3 \beta$};
    \node (t64dots) [no_dec] at (1.6, -9.2) {$\cdots$};
    \node (t65) [intermediate] at (3.25, -8) {$ + 3\prescript{*}{}{\underline\eta}\sl\nabla_3 \sigma$};
    \node (t65dots) [no_dec] at (3.25, -9.75) {$\cdots$};
    \node (t66) [intermediate] at (5.25, -8) {$-3\underline\eta  \sl\nabla_3 \rho$};
    \node (t66dots) [no_dec] at (5.25, -9.75) {$\cdots$};

    \node (t71) [intermediate] at (-6, -10) {$-\sl\nabla\sl\nabla_3 \rho$};
    \node (t71dots) [no_dec] at (-6, -11.5) {$\cdots$};
    \node (t72) [intermediate] at (-6 + 2.25, -10) {$- \frac{1}{2}(\eta + \underline\eta)\sl\nabla_3 \rho$};
    \node (t72_next) [finalterm] at (-6 + 2.25, -12) {$\frac{1}{2}(\eta + \underline\eta)\sl\div\underline\beta$};
    \node (t73) [finalterm] at (-6 + 5.25, -10) {$\hat{\underline\chi}\cdot\sl\nabla\rho + \frac{1}{2}\tr\underline\chi\sl\nabla\rho$};

    \draw[-] (t11) -- (t21);
    \draw[-] (t11) -- (t22);
    \draw[-] (t11) -- (t23);
    \draw[-] (t11) -- (t24);

    \draw[-] (t13) -- (t31); 
    \draw[-] (t14) -- (t32);

    \draw[-] (t23.south) -- (t41);
    \draw[-] (t22) -- (t51);
    \draw[-] (t22) -- (t52);

    \draw[-] [rounded corners]  (t21) -- (-7,-6) -- (a);
    \draw[fill=black] (a.west) circle (1pt);
    
    \draw[-] (a.west) -- (t61.north);
    \draw[-] (a.west) -- (t62.north); 
    \draw[-] (a.west) -- (t63.north); 
    \draw[-] (a.west) -- (t64.north); 
    \draw[-] (a.west) -- (t65.north); 
    \draw[-] (a.west) -- (t66.north); 

    \draw[-] (t62) -- (t71);
    \draw[-] (t62) -- (t72);
    \draw[-] (t62) -- (t73);

    \draw[-] (t72) -- (t72_next);

    \draw[-] (t12) -- (t12dots);
    \draw[-] (t61) -- (t61dots);
    \draw[-] (t63) -- (t63dots);
    \draw[-] (t64) -- (t64dots);
    \draw[-] (t65) -- (t65dots);
    \draw[-] (t66) -- (t66dots);
    \draw[-] (t71) -- (t71dots);
        
\end{tikzpicture}
    \caption{The order 1 and 2 tree expansion of $\sl\nabla_4 \underline B$. Terms marked $\cdots$ need to be further expanded and are omitted for diagram clarity.}
    \label{fig:Bbar_tree_1_2}
\end{figure}

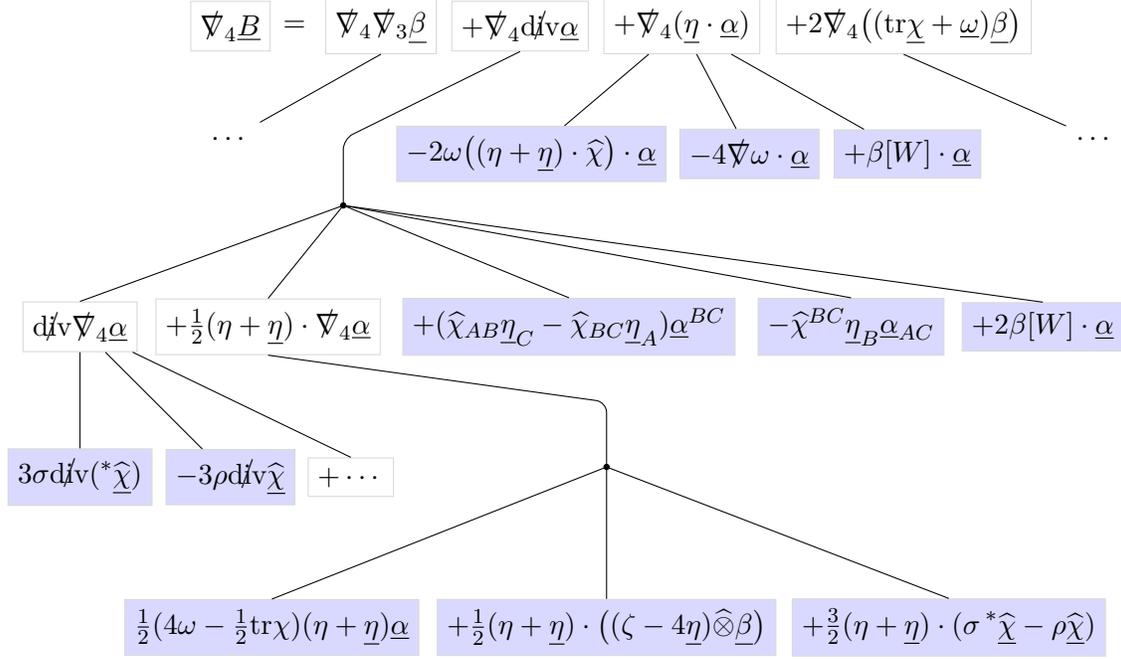
\begin{figure}[!ht]
    \centering
\begin{tikzpicture}
    
    \node [intermediate] at (-6, 0) {$\sl\nabla_4 \underline B$ };
    \node [no_dec] at (-4.15 - 1, 0) {$=$};
    \node (t11) [intermediate] at (-4.0, 0) {$\sl\nabla_4 \sl\nabla_3 \underline\beta$}; 
    \node (t11dots) [no_dec] at (-6, -1.5) {$\cdots$};
    \node (t12) [intermediate] at (-2.15, 0) {$+ \sl\nabla_4 \sl\div\underline\alpha$};
    \node (t13) [intermediate] at (0.0,0) {$+ \sl\nabla_4(\underline\eta \cdot\underline\alpha)$}; 
    \node (t14) [intermediate] at (2.95, 0) {$+ 2\sl\nabla_4\big((\tr\underline\chi + \underline\omega)\underline\beta\big)$}; 
    \node (t14dots) [no_dec] at (5.5, -1.5) {$\cdots$};

    \node (b) [no_dec] at (-4.5, -2.25) {};
    \draw[fill=black] (b.south) circle (1pt);
    
    \node (t21) [intermediate] at (-8, -4) {$\sl\div\sl\nabla_4 \underline\alpha$};
    \node (t22) [intermediate] at (-5.5, -4) {$+\frac{1}{2}(\eta + \underline\eta)\cdot\sl\nabla_4 \underline\alpha$};
    \node (t23) [finalterm] at (-1.5, -4) {$+ (\hat\chi_{AB}\underline\eta_C - \hat\chi_{BC}\underline\eta_A) \underline\alpha^{BC}$};
    \node (t24) [finalterm] at (2.25, -4) {$-\hat\chi^{BC}\underline\eta_B \underline\alpha_{AC}$}; 
    \node (t25) [finalterm] at (4.8, -4) {$+2\beta[W] \cdot\underline\alpha$};

    \node (s3) [finalterm] at (3, -1.7) {$+ \beta[W]\cdot\underline\alpha$}; 
    \node (s2) [finalterm] at (.9, -1.7) {$-4\sl\nabla\omega\cdot\underline\alpha$};
    \node (s1) [finalterm] at (-2, -1.7) {$-2\omega \big((\eta + \underline\eta)\cdot\hat\chi\big) \cdot\underline\alpha$};

    \node (t31) [finalterm] at (-8, -6) {$3 \sigma \sl\div(\prescript{*}{}{\hat{\underline\chi}})$}; 
    \node (t32) [finalterm] at (-6, -6) {$-3\rho \sl\div\hat{\underline\chi}$}; 
    \node (t33) [intermediate] at (-4.4, -6) {$ + \cdots$};

    \node (a) [no_dec] at (-1, -6) {};
    \node (t41) [finalterm] at (-5.45, -8) {$\frac{1}{2}(4\omega - \frac{1}{2}\tr\chi)(\eta + \underline\eta)\underline\alpha$};
    \node (t42) [finalterm] at (-1, -8) {$+\frac{1}{2}(\eta + \underline\eta) \cdot \big( (\zeta - 4\underline\eta)\hat\otimes \underline\beta \big)$};
    \node (t43) [finalterm] at (3.55, -8) {$+\frac{3}{2}(\eta + \underline\eta)\cdot (\sigma \prescript{*}{}{\hat{\underline \chi}} - \rho \hat{\underline\chi})$};

    \draw[-] [rounded corners] (t12.south) -- (-4.5, -1.5) -- (b.south);
    \draw[-] (b.south) -- (t21.north);
    \draw[-] (b.south) -- (t22.north);
    \draw[-] (b.south) -- (t23.north);
    \draw[-] (b.south) -- (t24.north);
    \draw[-] (b.south) -- (t25.north);

    \draw[-] (t21) -- (t31);
    \draw[-] (t21) -- (t32); 
    \draw[-] (t21) -- (t33.north);

    \draw[-] [rounded corners] (t22.south) -- (-1, -5) -- (a);
    \draw[fill=black] (a.north) circle (1pt);
    \draw[-] (a.north) -- (t41.north);
    \draw[-] (a.north) -- (t42);
    \draw[-] (a.north) -- (t43.north);

    \draw[-] (t13) -- (s1);
    \draw[-] (t13) -- (s2);
    \draw[-] (t13) -- (s3);

    \draw[-] (t11.south) -- (t11dots);
    \draw[-] (t14.south) -- (t14dots);
    
\end{tikzpicture}
    \caption{The order 0 tree expansion of $\sl\nabla_4 \underline B$. Terms marked $\cdots$ need to be further expanded and are omitted for diagram clarity.}
    \label{fig:Bbar_tree_0}
\end{figure}

\begin{proof}[Proof of Theorem \ref{thm:prop_eq_for_diff_constraints}] The detailed computations are provided in Appendix \ref{app:Bianchi_computations}. Here, we give the main ideas of the proof as well as a useful schematic overview. For each differential constraint, the proof proceeds broadly as follows. The null structure equations \eqref{eq:null_structure_propagation}-\eqref{eq:null_structure_constraint}, the partial Bianchi equations \eqref{eq:PBE}, and the differential constraints \eqref{eq:differential_constraints} can be used to eliminate all null derivatives which occur in the expansion of the null derivative of the given differential constraint. The null derivative of any differential constraint will thus be the sum of terms of the following form:
\begin{itemize}
    \item Order 2 terms which are second-order spherical derivative operators acting on the unknowns;
    \item Order 2 terms which are first-order spherical derivative operators acting on a differential constraint; 
    \item Order 1 terms which are a Ricci coefficient times a first-order spherical derivative operator acting on the unknowns; 
    \item Order 1 terms which are a Ricci coefficient times a differential constraint; 
    \item Order 0 terms which are of the form $\psi\psi' \cdot\Psi$ for $\psi, \psi'$ Ricci coefficients and $\Psi$ an unknown; 
    \item Order 0 terms which are of the form $\sl\nabla\psi \cdot \Psi$ for $\psi$ a Ricci coefficient and $\Psi$ an unknown; 
    \item Order 0 terms which are of the form $\Psi[W] \cdot \Psi$ for $\Psi[W]$ a null Weyl tensor component and $\Psi$ an unknown.
\end{itemize}
Note that quadratic terms in the unknowns do \emph{not} appear.
In the end all expressions cancel except for those which are homogeneous of degree 1 in either the differential constraints or a first-order spherical derivative operator applied to the differential constraints, \textit{and} those which are of the form $\Psi[W] \cdot \Psi$. Direct algebraic cancellation can evidently only occur between terms of the same order and with the same unknowns, and thus it is convenient in the computation to group terms according to the order of the term and which unknown appears in it. For expressions for the commutator of differential operators used here, as well as useful formulae involving the Hodge dual, see Propositions \ref{prop:commutation_formulae}-\ref{prop:proposition_5}.

\textbf{The expression for $ \sl\nabla_4\underline B$.} The terms in $\sl\nabla_4 \underline B$ which are order 2 are seen to be equal to 
\begin{align*}
    -\sl\div(\sl\nabla\hat\otimes \underline\beta) - \prescript{*}{}{\sl\nabla\curl}\underline\beta + \sl\nabla\sl\div\underline\beta.
\end{align*}
By Lemma \ref{lem:formula_1} this is equal to $-2 K \underline\beta$. This will cancel with an order 0 term of the form $\psi \psi' \cdot \underline\beta$ by using the Gauss equation. Thus there are no ``genuine'' order 2 terms in $\sl\nabla_4 \underline B$. The terms which are order 1 are seen to equal
\begin{align*}
    (4 \omega - \tr\chi)\underline B + 2\hat{\underline\chi}\cdot\Xi.
\end{align*}
The order 0 terms (grouped by unknown) are equal to: 
\begin{equation}
    2\underline\alpha \cdot \beta[W] + 6(\sigma \prescript{*}{}{\underline\beta[W]} - \rho \underline\beta[W])
    - 2\underline\alpha[W]\cdot\beta + 6(\rho[W]\underline\beta - \sigma[W]\prescript{*}{}{\underline\beta}).
\end{equation}

\textbf{The expression for $\sl\nabla_4 \Xi$.} The terms in $\sl\nabla_4 \Xi$ which are orders 1 and 2 equal
\begin{align*}
    -2\tr\chi \Xi - \sl\nabla P - \prescript{*}{}{\sl\nabla Q} - \big( \frac{7}{2}\eta + \frac{1}{2}\underline\eta \big) P - \big( \frac{7}{2}\prescript{*}{}{\eta} + \frac{1}{2}\prescript{*}{}{\underline\eta} \big) Q.
\end{align*}
The order 0 terms, grouped by unknown, are equal to 
\begin{equation}
    2\alpha[W]\cdot\underline\beta - 2 \alpha \cdot \underline\beta[W] 
    + 6(\rho \beta[W] + \sigma \prescript{*}{}{\beta[W]})  - 6 (\rho[W]\beta + \sigma[W]\prescript{*}{}{\beta}).
\end{equation}

\textbf{The expression for $\sl\nabla_3 P$.} The terms in $\sl\nabla_3 P$ which are orders 1 and 2 equal 
\begin{align*}
    (2\underline\omega - \frac{3}{2}\tr\underline\chi) P - (\eta + 2\underline\eta) \cdot \Xi - \sl\div \Xi.
\end{align*}
The order 0 terms are 
\begin{align*}
    \frac{1}{2}\big( \alpha[\mathcal{R}]\cdot \underline\alpha - \underline\alpha[\mathcal{R}]\cdot\alpha \big) + 2\big( \underline\beta\cdot\beta[W] - \beta\cdot\underline\beta[W] \big).
\end{align*}
\textbf{The expression for $\sl\nabla_3 Q$.} The terms in $\sl\nabla_3 Q$ which are orders 1 and 2 equal 
\begin{align*}
    (2 \underline\omega - \frac{3}{2}\tr\underline\chi) Q - \prescript{*}{}{(\eta + 2\underline\eta)}\cdot\Xi + \sl\curl\Xi.
\end{align*}
The order 0 terms are 
\begin{align*}
    \frac{1}{2}\big(\prescript{*}{}{\underline\alpha}\cdot\alpha[\mathcal{R}] - \prescript{*}{}{\underline\alpha}[\mathcal{R}]\cdot\alpha \big) + 2\big( \prescript{*}{}{\underline\beta}\cdot\beta[W] - \beta \cdot\prescript{*}{}{\underline\beta}[W] \big).
\end{align*}
This completes the proof.
\end{proof}

\subsection{The algebraic constraints}\label{subsec:algebraic_constraints} The goal of this section is to show that solutions of the linearized Bianchi equations are in fact constrained to lie, at every point $p$, in a subspace $\tilde{ \mathcal{V}}_p \subset \mathcal{V}_p$ which is determined solely by the spacetime $(M, g)$ and its null Weyl tensor components; and that at any point where the Weyl tensor is nonzero, this is a \emph{strict} subspace of $\mathcal{V}_p$. As an important consequence of this, we note that it is not possible to study the evolution of arbitrary perturbations to the initial data of the linearized Bianchi equations under the full system; the perturbations are required to lie within a codimension $\geq 1$ subspace. That is, in addition to the already-known differential constraints \eqref{eq:differential_constraints}, there are additional \emph{algebraic} constraints on the initial data for the linearized Bianchi equations.

\begin{thm}[Algebraic constraints]\label{thm:algebraic_constraints_1}
Let $T^*_S M$ denote the space of $S$ 1-forms on $M$. Let
$$
\mathfrak{L}_W : \mathcal{V} \to T^*_SM  \times T^*_SM \times \R \times \R 
$$
be the vector bundle morphism (i.e. fiber-wise linear map) defined by 
\begin{equation}\label{eq:L_W_defn}
\begin{split}
    \mathfrak{L}_W(\alpha, &\beta, \rho, \sigma, \underline\beta, \underline\alpha) \\ 
    &= \Big(2(\underline\alpha \cdot \beta[W] - \underline\alpha[W]\cdot \beta) + 6(\sigma \prescript{*}{}{\underline\beta[W]} - \sigma[W]\prescript{*}{}{\underline\beta}) + 6(\rho[W]\underline\beta - \rho \underline\beta[W])),  \\ 
    &\qquad 2(\alpha[W]\cdot \underline\beta - \alpha \cdot\underline\beta[W]) + 6 (\rho \beta[W] - \rho[W]\beta) + 6(\sigma\prescript{*}{}{\beta[W]} - \sigma[W]\prescript{*}{}{\beta}), \\ 
    &\qquad \frac{1}{2}\big( \alpha[W]\cdot \underline\alpha - \underline\alpha[W]\cdot\alpha \big) + 2\big( \underline\beta\cdot\beta[W] - \beta\cdot\underline\beta[W] \big), \\ 
    &\qquad \frac{1}{2}\big(\prescript{*}{}{\underline\alpha}\cdot\alpha[W] - \prescript{*}{}{\underline\alpha}[W]\cdot\alpha \big) + 2\big( \prescript{*}{}{\underline\beta}\cdot\beta[W] - \beta \cdot\prescript{*}{}{\underline\beta}[W] \big)\Big).
\end{split}
\end{equation}
If $(\alpha, \beta, \rho, \sigma, \underline\beta, \underline\alpha)$ solve the linearized Bianchi equations \eqref{eq:linearized_bianchi_equations}, then
\begin{align*}
    \mathfrak{L}_W(\alpha, \beta, \rho, \sigma, \underline\beta, \underline\alpha) &= 0.
\end{align*}
\end{thm}

\begin{proof}
By \eqref{eq:linearized_bianchi_equations}, the differential constraints $\underline B, \Xi, P$, and $Q$ vanish on $M$. Thus, by Theorem \ref{thm:prop_eq_for_diff_constraints}, we have
\begin{align*}
    0 &= 2(\underline\alpha \cdot \beta[W] - \underline\alpha[W]\cdot \beta) + 6(\sigma \prescript{*}{}{\underline\beta[W]} - \sigma[W]\prescript{*}{}{\underline\beta}) + 6(\rho[W]\underline\beta - \rho \underline\beta[W]) \\ 
    0 &= 2(\alpha[W]\cdot \underline\beta - \alpha \cdot\underline\beta[W]) + 6 (\rho \beta[W] - \rho[W]\beta) + 6(\sigma\prescript{*}{}{\beta[W]} - \sigma[W]\prescript{*}{}{\beta}) \\ 
    0 &= \frac{1}{2}\big( \alpha[W]\cdot \underline\alpha - \underline\alpha[W]\cdot\alpha \big) + 2\big( \underline\beta\cdot\beta[W] - \beta\cdot\underline\beta[W] \big)\\ 
    0 &= \frac{1}{2}\big(\prescript{*}{}{\underline\alpha}\cdot\alpha[W] - \prescript{*}{}{\underline\alpha}[W]\cdot\alpha \big) + 2\big( \prescript{*}{}{\underline\beta}\cdot\beta[W] - \beta \cdot\prescript{*}{}{\underline\beta}[W] \big),
\end{align*}
which is exactly the statement that $\mathfrak{L}_W(\alpha,\beta,\rho,\sigma,\underline\beta,\underline\alpha) = 0$.
\end{proof}
Theorem \ref{thm:algebraic_constraints_1} implies that solutions to \eqref{eq:linearized_bianchi_equations} are constrained to lie inside the kernel
\begin{align*}
    \ker\mathfrak{L}_W  \subset \mathcal{V}.
\end{align*}
We remark that the operator $\mathfrak{L}_W$ is determined by the null Weyl tensor components of $(M, g)$. This operator can be viewed as something intrinsic to $(M, g)$ which constrains solutions of \eqref{eq:linearized_bianchi_equations}. 

This motivates the question: when is this kernel not all of $\mathcal{V}$? The only linear map with full kernel is the zero map. As long as $(M, g)$ is \emph{not} isometric to an open subset of Minkowski spacetime, at least one null component of the Weyl tensor is nonzero; and if at least one null component of the Weyl tensor is nonzero, then indeed $\mathfrak{L}_W$ is not the zero map. 

\begin{lemma}\label{lem:non-full_kernel}
If at least one of the null components of the Weyl tensor of $(M, g)$ is nonzero, then $\mathfrak{L}_W$ is not the zero map.
\end{lemma}

\begin{proof}
It suffices to exhibit an element which $\mathfrak{L}_W$ does not map to 0. If $\alpha[W] \neq 0$, then $\Psi = (0, 0, 0, 0, 0, \alpha[W])$ has the property that the third component of $\mathfrak{L}_W\Psi$ is equal to 
$$\frac{1}{2}\alpha[W]\cdot\alpha[W] = \frac{1}{2}|\alpha[W]|^2 \neq 0.
$$
Similarly if $\underline\alpha[W] \neq 0$ we can choose $\Psi = (\underline\alpha[W],0,0,0,0,0)$ and obtain $\mathfrak{L}_W\Psi \neq 0$.

If $\beta[W] \neq 0$, then $\Psi = (0,0,0,0,\beta[W], 0)$ has the property that the third component of $\mathfrak{L}_W\Psi$ is equal to 
\begin{align*}
    2\beta[W]\cdot\beta[W] = 2|\beta[W]|^2 \neq 0.
\end{align*}
Similarly if $\underline\beta[W] \neq 0$, we can choose $\Psi = (0, \underline\beta[W], 0, 0, 0, 0)$ and obtain $\mathfrak{L}_W \Psi \neq 0$.

Now suppose $\rho[W] \neq 0$ or $\sigma[W] \neq 0$. Then for $\beta \in T^*_S M$ yet to be determined, $\Psi = (0,\beta, 0,0,0,0)$ has the property that the second component of $\mathfrak{L}_W$ is equal to 
\begin{align*}
    -6(\rho[W]\beta + \sigma[W]\prescript{*}{}{\beta}). \eqcount\label{eq:computation_53}
\end{align*}
Let $(e_A)_{A = 1,2}$ be an orthonormal basis of $T_p S_{u,\u}$, where $p$ is the point under consideration. In this basis, let $\beta$ be the 1-form with components $\beta_1 = 1, \beta_2 = 0$. Then $\prescript{*}{}{\beta}_1 = 0$ and $\prescript{*}{}{\beta}_2 = -1$, and so \eqref{eq:computation_53} is equal, in its component representation with respect to the basis $(e_A)_{A = 1,2}$, to 
\begin{align*}
    6\begin{pmatrix}
    -\rho[W] \\ \sigma[W]
    \end{pmatrix}.
\end{align*}
If $\rho[W] \neq 0$, this is nonzero, and thus $\mathfrak{L}_W \Psi \neq 0$. If $\sigma[W] \neq 0$, we see that still $\mathfrak{L}_W \Psi \neq 0$. This shows that if any null Weyl tensor component of $(M, g)$ is nonzero, then $\mathfrak{L}_W \neq 0$.
\end{proof}

Therefore, under the assumptions of the lemma, this implies  that 
\begin{align*}
    \dim\ker\mathfrak{L}_W &< \text{rank}(\mathcal{V}) = 10,
\end{align*}
that is, the kernel of $\mathfrak{L}_W$ has positive codimension.
This implies in particular the following corollary:

\begin{corollary}
Let $(M, g)$ be a vacuum spacetime of the type described in Section \ref{subsec:spacetime_and_notation}. Solutions of \eqref{eq:linearized_bianchi_equations} must lie in the kernel 
$$
\ker\mathfrak{L}_W.
$$
At points $p \in M$ where at least one of the null Weyl tensor components of $(M, g)$ is nonzero, $\ker\mathfrak{L}_W$ is a codimension $\geq 1$ subspace of $\mathcal{V}_p$.
\end{corollary}

\begin{remark}\label{rem:obstruction_1}
In the context of well-posedness for the linearized Bianchi equations, this implies in particular that, in order give rise to a solution of \eqref{eq:linearized_bianchi_equations},
seed initial data $\Psi_0$ for the partial Bianchi equations must satisfy the algebraic constraints $\mathfrak{L}_W\Psi_0 = 0$.

It is of interest to know if the converse holds: that is, if the seed data satisfies $\mathfrak{L}_W \Psi_0 = 0$, are the linearized Bianchi equations well-posed?\footnote{The issue of the lack of hyperbolicity in \eqref{eq:PBE} must also be dealt with to answer this, as discussed at the end of Section \ref{subsec:LNB_preliminaries}.} To answer this we first need to know whether or not the algebraic constraints propagate under \eqref{eq:PBE}. However, this is a nontrivial question to answer. To see why, note that one would like to take a null derivative of the right-hand side of \eqref{eq:L_W_defn} and derive a homogeneous ODE so as to apply Gr\"onwall's inequality. However, in the last two components of $\mathfrak{L}_W$, the following terms appear:
\begin{align*}
    \frac{1}{2}\big( \alpha[W] \cdot\underline\alpha - \underline\alpha[W] \cdot\alpha \big), \quad \frac{1}{2}\big(\prescript{*}{}{\underline\alpha}\cdot \alpha[W] - \prescript{*}{}{\underline\alpha}[W] \cdot\alpha \big).
\end{align*}
Both $\alpha$ and $\underline\alpha$ appear; yet $\alpha$ only satisfies a $\sl\nabla_3$ equation, and $\underline\alpha$ only satisfies a $\sl\nabla_4$ equation. Therefore it is not immediately clear which is the preferred null direction to differentiate in. 
\end{remark}

\begin{example}\label{ex:schwarzschild}
Consider the (exterior) Schwarzschild spacetime in Eddington-Finkelstein double null coordinates
\begin{align*}
    g = -4 \Omega \, du \, d\u + r^2\overset{\circ}{\gamma},
\end{align*}
where $\overset{\circ}{\gamma}$ is the standard  unit sphere metric and $\Omega$ is given below
(see \cite{dhr_linear}).
The only nonzero Weyl tensor component is 
\begin{align*}
    \rho[W] &= -\frac{2M}{r^3}
\end{align*}
and the only nonzero Ricci coefficients are
\begin{align*}
    \tr\chi &= \frac{2\Omega}{r}, &\tr\underline\chi &= -\frac{2\Omega}{r} \\ 
    {\omega} &= -\frac{M}{2\Omega r^2}, &{\underline\omega} &= \frac{M}{2\Omega r^2},
\end{align*}
where the null lapse is $\Omega =  \sqrt{1 - 2M/r}.$ Note that the null vector fields $e_3, e_4$ are given by $e_3 = \Omega^{-1}\d_u$ and $e_4 = \Omega^{-1} \d_{\u}$. In this case the algebraic constraints for unknowns $(\alpha,\beta,\rho,\sigma,\underline\beta,\underline\alpha)$ read
\begin{align*}
    -\frac{6M}{r^3}\beta &= 0, \qquad -\frac{6M}{r^3}\underline\beta = 0. \eqcount\label{eq:computation_23}
\end{align*}
In other words, if $(\alpha,\beta,\rho,\sigma,\underline\beta,\underline\alpha)$ solve \eqref{eq:linearized_bianchi_equations} on Schwarzschild spacetime, then $\beta = \underline\beta = 0$. The seed initial data for $\alpha,\underline\alpha, \rho$, and $\sigma$ can be prescribed freely. The linearized Bianchi equations, therefore, simplify in Schwarzschild to: 
\begin{align*}
    \sl\nabla_4 \underline\alpha + \frac{1}{2}\tr\chi \underline\alpha - 4\omega \underline\alpha &= 0  & \sl\nabla_3 \alpha + \frac{1}{2}\tr\underline\chi \alpha - 4{\underline\omega}\alpha &= 0 \\ 
    \sl\div \underline\alpha &= 0 & 
    \sl\div\alpha &= 0  \\ 
    \sl\nabla\rho - \prescript{*}{}{\sl\nabla\sigma} &= 0 & 
    \sl\nabla \rho + \prescript{*}{}{\sl\nabla \sigma} &= 0 \eqcount\label{eq:schwarzschild_lnb} \\ 
    e_4\rho + \frac{3}{2}\tr\chi \rho &= 0  & 
    e_3 \rho + \frac{3}{2}\tr\underline\chi \rho &= 0 \\ 
    e_4 \sigma + \frac{3}{2}\tr\chi \sigma &= 0 &
    e_3 \sigma + \frac{3}{2}\tr\underline\chi \sigma &= 0
\end{align*}
Note that this system is \textit{different}, and indeed simpler, than the linearized Bianchi equations on Schwarzschild that one would haved obtained without knowledge of the algebraic constraints $\mathfrak{L}_W$. Indeed they would have included, \emph{in addition}, the equations
\begin{align*}
    \sl\nabla_4\beta &= -2(\tr\chi + \omega)\beta + \sl\div\alpha & \sl\nabla_3 \beta &= (2\underline\omega - \tr\underline\chi)\beta + \sl\nabla\rho + \prescript{*}{}{\sl\nabla}\sigma \\
    \sl\nabla_4 \underline\beta &= (2\omega - \tr\chi)\underline\beta - \sl\nabla\rho + \prescript{*}{}{\sl\nabla}\sigma & \sl\nabla_3 \underline\beta &= -2(\tr\underline\chi + \underline\omega)\underline\beta - \sl\div\underline\alpha,
\end{align*}
\emph{as well as} additional complicating terms in the equations for $\alpha,\rho,\sigma$, and $\underline\alpha$ in \eqref{eq:schwarzschild_lnb}.
Thus, in addition to the terms which vanish due to vanishing of many of the Ricci coefficients in Schwarzschild, all terms involving $\beta, \underline\beta$ vanish due to the algebraic constraints \eqref{eq:computation_23}. 

This in fact allows us to explicitly solve the linearized Bianchi equations on Schwarzschild spacetime by integrating the resulting ODEs in \eqref{eq:schwarzschild_lnb}. Note that the equations involving the spherical derivatives of $\rho$ and $\sigma$ imply that $\rho$ and $\sigma$ are constant on the spheres $S_{u,\u}$. We can write down a solution by letting $(e_A)_{A = 1,2}$ be a local frame field on $U \subset S_{0,0}$. Propagate this to a local frame field on $[0, \u_*] \times U \subset C_0$ by parallel transport along the $\u$-curves from $S_{0,0}$. Then propagate $(e_A)_{A = 1, 2}$ to $[0, u_*] \times [0, \u_*] \times U$ by parallel transport along the $u$-curves from $C_0$. Note that then $\nabla_3 e_A = \nabla_4 e_A = 0$ everywhere $e_A$ is defined, since
\begin{align*}
    \nabla_3 \nabla_4 e_A &= \nabla_4\nabla_3 e_A + R(e_3, e_4)e_A \\ 
    &= 2\sigma\tensor{\sl\epsilon}{^B _A}e_B + \beta_A e_3 + \underline\beta_A e_4 = 0.
\end{align*}
Then, with respect to this frame, we have, for any $(u, \u) \in D$ and $\theta \in U$,
\begin{align*}\allowdisplaybreaks
    \alpha_{AB}(u, \u, \theta) &= \alpha_{AB}(0, \u, \theta)\exp\Big(\int_0^u \Omega\big(4\underline\omega - \frac{1}{2}\tr\underline\chi \big)(s, \u, \theta)\, ds\Big) \\ 
    \rho(u, \u, \theta) &= \rho(0, \u, \theta)\exp\Big(-\int_0^u\frac{3}{2}\Omega\tr\underline\chi(s, \u, \theta)\, ds \Big) \\ 
    \sigma(u, \u, \theta) &= \sigma(0, \u, \theta)\exp\Big(-\int_0^u\frac{3}{2}\Omega\tr\underline\chi(s, \u, \theta)\, ds \Big) \\ 
    \underline\alpha_{AB}(u,\u,\theta) &= \underline\alpha_{AB}(u, 0, \theta)\exp\Big( \int_0^{\u} \Omega\big(4\omega - \frac{1}{2}\tr\chi\big)(u, \underline s, \theta)\, d\underline s \Big).
\end{align*}
The fact that $\alpha,\underline\alpha$ remain divergence-free is a short computation using the commutation formulae in Proposition \eqref{prop:commutation_formulae}:
\begin{align*}
    \sl\nabla_3\sl\div\alpha &= - \sl\div(-4\underline\omega \alpha + \frac{1}{2}\tr\underline\chi\alpha) - \frac{1}{2}\tr\underline\chi \sl\div\alpha \\ 
    &= (4\underline\omega - \tr\underline\chi)\sl\div\alpha \\ 
    \implies \sl\nabla_3 |\sl\div\alpha|^2 &= 2(4\underline\omega - \tr\underline\chi)|\sl\div\alpha|^2,
\end{align*}
to which Gr\"onwall's inequality can be applied. Similarly we can verify that the $e_4$ equations for $\rho$ and $\sigma$ hold. Note that $\tr\chi = -\tr\underline\chi$. Also, one can verify that $\d_{u}r + \d_{\u} r = 0$. For $\rho$:
\begin{align*}
    e_3\big( e_4 \rho + \frac{3}{2}\tr\chi \rho  \big) &= -\frac{3}{2} e_4(\tr\underline\chi\rho) + \frac{3}{2}e_3(\tr\chi\rho) \\ 
    &= \frac{3}{2}(e_3 + e_4)(\tr\chi\rho) \\
    &= \frac{3}{2}\rho (e_3 + e_4)\tr\chi + \frac{3}{2}\tr\chi(e_4 \rho - \frac{3}{2}\tr\underline\chi\rho) \\ 
    &=  \frac{3}{2}\rho (e_3 + e_4)\tr\chi + \frac{3}{2}\tr\chi(e_4 \rho + \frac{3}{2}\tr\chi\rho).
\end{align*}
The first term vanishes since $\d_u r + \d_{\u} r = 0$ implies $(e_3 + e_4)\tr \chi = 0$. Then one can apply Gr\"onwall's inequality. A similar argument holds for $\sigma$.
\end{example}

\section{Conclusion}

In this paper, motivated by the null Bianchi equations, we introduce double null hyperbolic systems and prove a global existence and uniqueness result for such systems (Theorem \ref{thm:GWP_H^1}). We also derive a novel set of algebraic constraints that must be satisfied by solutions of the linearized Bianchi equations (Theorem \ref{thm:algebraic_constraints_1}). As the null Bianchi equations and their linearization are of primary interest, the next step is to correctly formulate and prove a well-posedness result for the linearized null Bianchi equations. As discussed in Remark \ref{rem:obstruction_1} and at the end of Section \ref{subsec:LNB_preliminaries}, there are several obstacles to doing so, among them the lack of hyperbolicity of the partial Bianchi equations \eqref{eq:PBE} and the propagation of the algebraic constraints. These constraints are also of interest in their own right as they constrain solutions of the linearized Bianchi equations more than previously known; see Example \ref{ex:schwarzschild} for the case of the linearized Bianchi equations on Schwarzschild spacetime. We are interested if there are any physical interpretations of the operator $\mathfrak{L}_W$, as well as if these constraints manifest in the nonlinear problem, i.e. the full Einstein equations. 

Another motivation for studying double null hyperbolic systems is that, as systems which are intrinsically adapted to null hypersurfaces and the propagation of quantities along them, we believe they will provide a powerful tool to analyze in further detail gravitational waves, and more generally, any phenomena which propagate along null hypersurfaces.

There are also various refinements we plan to discuss in future work. We would like to pose initial data on $\underline C_0 \cup C_0$ with $C_0$ a complete null hypersurface (non-compact) or past null infinity. We would also like to precisely track the dependence of the estimates (in particular the constants $\mathfrak{B}$ and $\mathfrak{C}$) on the $\sl C_m$, which controls the size of the Ricci coefficients, as well as to relax the regularity assumptions on the underlying spacetime. We also plan to derive more precise statements concerning solutions of \eqref{eq:general_hyperbolic_system}, such as quantitative decay statements. In addition to generalizing the results in this paper, this would be helpful in analyzing the nonlinear problem and the full Einstein equations when, instead of being fixed, the underlying spacetime is treated as a dynamic variable.

\appendix
\section{Two-variable ODE theory in Banach spaces}\label{app:two-var_ode_theory}

\subsection{Review of functional analysis}
For two Banach spaces $X, Y$, let $\mathcal{B}(X; Y)$ denote the space of bounded linear operators from $X$ to $Y$.
Recall that any map $g : X \to Y$ between Banach spaces is \emph{Fr\'echet differentiable} at $x \in X$ if there exists $L \in \mathcal{B}(X; Y)$ such that 
\begin{align*}
    \lim_{h \to 0} \frac{g(x + h) - g(x) - Lh}{\norm{h}_X} = 0.
\end{align*}
We call $Dg(x) \coloneqq L$ the \emph{Fr\'echet derivative of $g$ at $x$}. 

Let $D = \Pi_{i = 1}^n [a_i, b_i] \subset \mathbb{R}^n$ be a compact cube. Let $X$ be a separable Banach space. Let $f : D \to X$. We define the \emph{partial derivatives} of $f$ as follows. For $i = 1, \ldots, n$, let $e_i$ denote the standard basis vector in $\R^n$. For any $t \in D$, define $\phi_i : [a_i, b_i] \to X$ by 
\begin{align*}
    \phi_i(y) = f(t_1, \ldots, t_{i - 1}, y, t_{i + 1}, \ldots, _n).
\end{align*}
The partial derivative of $f$ in the direction $x^i$ is then defined to be
\begin{align*}
    \d_{x^i} f (t) \coloneqq \frac{\d f}{\d x^i}(t) &\coloneqq D\phi_i(t^i) (1) \in X, \qquad \forall i = 1, \ldots, n.
\end{align*}

\begin{remark}
Some authors define $\frac{\d f}{\d x^i}(t)$ as the linear map $D\phi_i(t^i)$. Since in our setting the domain of $D\phi_i(t^i)$ is $\R$, it is uniquely characterized by its action on the element 1. 
\end{remark}

Therefore, we can regard $\frac{\d f}{\d x^i}$ as a map $t \mapsto \frac{\d f}{\d x^i}(t)$ mapping $D$ into $X$, as in ordinary calculus. If for each $i$ this map is continuous, then we say $f \in C^1(D; X)$, and similar definitions can be made for $f \in C^k(D; X)$ and $f \in C^\infty(D; X)$.

We now specialize to the case $n = 2$ and $D = [0, u_*] \times [0, \u_*]$.  
By Pettis' theorem, strong measurability and weak measurability agree since $X$ is separable. If $f$ is continuous, then $t \mapsto \langle u^*, f(t) \rangle$ is continuous and hence measurable; hence $f$ is strongly measurable and also summable. If $f : D \to X$ is continuous, then for any $u \in [0, u_*]$, $f_u : \u \mapsto f(u, \u)$ is also continuous, and therefore summable on $[0, u_*]$; similarly for any fixed $\u \in [0, \u_*]$. 

\begin{prop}[Fundamental theorem of calculus]\label{prop:FTC_Banach_space}
Let $f : D \to X$ be continuous and let $a_0 \in X$. Define $F : D \to X$ by 
\begin{align*}
    F(u, \u) &= a_0 + \int_0^u f(u', \u)\, du'.
\end{align*}
Then
\begin{align*}
    \frac{\d F}{\d u}(u, \u) &= f(u, \u).
\end{align*}
\end{prop}

\begin{proof}
Let $\u \in [0, \u_*]$ be arbitrary and define $\phi(u) = F(u, \u)$. 
We compute:
\begin{align*}
    \frac{\phi(u + h) - \phi(u) - f(u, \u) h}{h} &= \frac{1}{h}\Big(\int_u^{u + h} f(u', \u)\, du'\Big) - f(u, \u).
\end{align*}
Since $f$ is continuous, the limit of this as $h \to 0$ is zero. This proves that $\phi$ is Fr\'echet differentiable at $u$ and that $D\phi(u)(h) = f(u, \u)(h)$. By definition then $\frac{\d F}{\d u}(u, \u) = f(u, \u)$. 
\end{proof}

\begin{prop}
Let $f : D \to X$ be continuous and let $a_0 : [0, \u_*] \to X$ be continuously differentiable. Also suppose that $\frac{\d f}{\d \u}$ exists and is continuous.
Then $F : D \to X$ defined by 
\begin{align*}
    F(u, \u) = a_0(\u) + \int_0^u f(u', \u)\, du'
\end{align*}
is in $C^1(D; X)$.
\end{prop}

\begin{proof}
The previous proposition shows that $\d f/ \d u$ is continuous. Since $\d f/\d \u$ is continuous by assumption and $[0, u_*]$ is compact, we have
\begin{align*}
    \frac{\d F}{\d \u}
    &= \frac{\d a_0}{\d \u}(\u) +  \int_0^u \frac{\d f}{\d \u}(u', \u)\, du'.
\end{align*}
The right-hand side is continuous in $(u, \u)$, completing the proof. 
\end{proof}

\begin{prop}
Let $Y$ be a finite-dimensional real vector space and let $f \in C^\infty(D \times S^2; Y)$. For any $k \geq 0$, let $F_k$ denote the function $F_k : D \to H^k(S^2)$ defined by $F_k(t) = f(t, \cdot)$. Then the partial Fr\'echet derivatives of $F_k$ are equal to the usual partial derivatives of $f$.
\end{prop}

\begin{proof}
Let $t = (u, \u) \in D$. Note that
\begin{equation}
    \frac{F_k(u + h, \u) - F_k(u, \u) - \d_u f(u, \u, \cdot)h}{|h|} = \frac{f(u + h, \u, \cdot) - f(u, \u, \cdot) - \d_u f(u, \u, \cdot)}{|h|}.
\end{equation}
Since $f$ is smooth, its difference quotients (as well as the difference quotients of all of its derivatives) converge uniformly as $h \to 0$, and furthermore its mixed partial derivatives commute. As $S^2$ is compact, this  implies, in particular, that the above quotient tends to 0 in $H^k(S^2; Y)$ as $h \to 0$.
\end{proof}

\subsection{Two-variable ODE theory in Banach spaces}
Let $u_*, \underline u_* > 0$ and write $D = [0, u_*] \times [0, \underline u_*]$. Let $\mathbb{Y}, \underline{\mathbb{Y}}$ be Banach spaces. Denote $\mathbb{X} = \mathbb{Y} \times \underline{\mathbb{Y}}$. Let 
\begin{align*}
    F &: D \times \mathbb{X} \to \mathbb{Y} \\ 
    \underline F &: D \times \mathbb{X} \to \underline{\mathbb{Y}}
\end{align*}
be continuous functions which are also Lipschitz in $\mathbb{X}$, uniformly in $(u, \underline u)$. That is, for all $x_1, x_2 \in \mathbb{X}$, there exists a constant $M$ such that for all $(u, \u) \in D$,
\begin{align*}
    \norm{F(u, \underline u, x_1) - F(u, \underline u, x_2)}_{\mathbb{Y}} &\leq M\norm{x_1 - x_2}_{\mathbb{X}} \\ 
    \norm{\underline F(u, \underline u, x_1) - \underline F(u, \underline u, x_2)}_{\underline{\mathbb{Y}}} &\leq M\norm{x_1 - x_2}_{\mathbb{X}}.
\end{align*}
Consider the system 
\begin{equation}\label{eq:ode_system}
\begin{split}
    \frac{\d U}{\d u}(u, \underline u) &= F(u, \underline u, U(u, \underline u), \underline U(u, \underline u)) \\
    \frac{\d \underline U}{\d \underline u}(u, \underline u) &= \underline F(u, \underline u, U(u, \underline u), \underline U(u, \underline u)).
\end{split}
\end{equation}
for unknowns $U : D \to \mathbb{Y}, \underline U : D \to \underline{\mathbb{Y}}$. 

\begin{thm}\label{thm:picard_1}
Let $F, \underline F$ as above be continuous and globally Lipschitz in $\mathbb{X}$, uniformly in $(u,\u)$. Then for any $U_0 \in C^0([0, \u_*]; \mathbb{Y}) \text{ and } \underline U_0 \in C^0([0, u_*]; \underline{\mathbb{Y}})$ there exist unique $U, \underline U$ with $U \in C^0(D; \mathbb{Y})$ and $\underline U \in C^0(D; \underline{\mathbb{Y}})$ such that both 
\begin{align*}
    \frac{\d U}{\d u} \qquad \text{and} \qquad \frac{\d \underline U}{\d \u}
\end{align*}
exist and are continuous, and $U, \underline U$ solve \eqref{eq:ode_system}.
\end{thm}

\begin{remark}
Note that any map $W \in C^0(D; \mathbb{X})$  can be identified as a pair of maps $(V, \underline V)$ with $V \in C^0(D; \mathbb{Y}), \underline V \in C^0(D; \underline{\mathbb{Y}})$.
\end{remark}

\begin{remark}
It is of interest to reduce the regularity in the transverse direction. That is, we would like to consider initial data lying instead in the Banach spaces $L^p([0, \u_*]; \mathbb{Y})$ and $L^p([0, u_*]; \underline{\mathbb{Y}})$, or even more singular.
\end{remark}

\begin{proof}
The proof is a slight variant of the standard Picard iteration scheme. Define a map 
\begin{align*}
    \Phi : C^0(D; \mathbb{X})  \to   C^0(D; {\mathbb{X}})
\end{align*}
by (using the Bochner integral)
\begin{multline*}
    \Phi(V, \underline V)(u, \underline u) = \Big( U_0(\underline u) + \int_0^u F(u', \underline u, V(u', \underline u), \underline V(u', \underline u))\, d u' , \\ 
    \underline U_0(u) + \int_0^{\u}\underline F(u, \underline u', V(u, \underline u'), \underline V(u, \underline u'))\, d\underline u'\Big).
\end{multline*}
Then the $\mathbb{X}$-norm of the difference between $\Phi(V, \underline V)(u, \underline u)$ and $\Phi(W, \underline W)(u, \underline u)$ is bounded by 
\begin{multline*}
    \int_0^u M\big( \norm{V - W}_{\mathbb{Y}}(u',\underline u) + \norm{\underline V - \underline W}_{\underline{\mathbb{Y}}}(u', \underline u)\big)\, du'  \\ 
    + \int_0^{\underline u} M \big( \norm{V - W}_{\mathbb{Y}}(u, \underline u') + \norm{\underline V - \underline W}_{\underline{\mathbb{Y}}}(u, \underline u')  \big)\, d\underline u'
\end{multline*}
which is less than or equal to 
\begin{align*}
    M(u + \underline u)&\big(\norm{V - W}_{C^0(D; \mathbb{Y})} + \norm{\underline V - \underline W}_{C^0(D; \underline{\mathbb{Y}})}  \big) \\ 
    &= M(u + \underline u)\norm{(V, \underline V) - (W, \underline W)}_{C^0(D; \mathbb{X})}.
\end{align*}
Therefore 
\begin{align*}
    \norm{\Phi(V, \underline V) - \Phi(W, \underline W)}_{C^0(D; \mathbb{X})} &\leq M(u + \underline u)\norm{(V, \underline V) - (W, \underline W)}_{C^0(D; \mathbb{X})}
\end{align*}
and hence if $u + \underline u < M^{-1}$, $\Phi$ is a contraction. There is therefore a unique fixed point 
$$
(U, \underline U) \in C^0(\big[0, \frac{1 - \delta/2}{2M}\big] \times \big[0, \frac{1 - \delta/2}{2M}\big]; \mathbb{X})
$$
for some $\delta > 0$ sufficiently small. 
We now discuss how to extend this solution to all of $D$. We first extend the solution in the $u$-direction. For this we take as our new initial data 
\begin{align*}
    U_0^{\text{new}}(u) &= U_0(u) &\forall u \in \big[ \frac{1 - \delta}{2M}, u_* \big] \\ 
    \underline U_0^{\text{new}}(\u) &= \underline U\big( \frac{1 - \delta}{2M}, \u \big) &\forall \u \in \big[ 0, \frac{1 - \delta}{2M} \big].
\end{align*}
Since the Lipschitz bounds on $F, \underline F$ are unchanged, we get a solution 
$$
(U^{\text{new}}, \underline U^{\text{new}}) \in C^0\big( \big[ \frac{1 - \delta}{2M}, \frac{2 - 3\delta/2}{2M} \big] \times \big[0, \frac{1 - \delta/2}{2M}\big] ; \mathbb{X} \big).
$$
To see that this agrees with the solution $U$ restricted to the domain $\big[ \frac{1 - \delta}{2M}, \frac{1 - \delta/2}{2M} \big] \times \big[ 0, \frac{1 - \delta}{2M} \big] \eqqcolon D'$, note that a map $\Phi'$ can be defined on $C^0(D'; \mathbb{X})$ in an analogous manner to how $\Phi$ was defined. It will still be a contraction on this smaller domain, and both $(U^{\text{new}},\underline U^{\text{new}})|_{D'}$ and $(U,\underline U)|_{D'}$ will be fixed points. By uniqueness of the fixed point, $(U^{\text{new}}, \underline U^{\text{new}})|_D' = (U, \underline U)|_{D'}$. We can therefore glue our two solutions to a new solution, which we also call $(U, \underline U)$,
\begin{align*}
    (U, \underline U) \in C^0\big( \big[ 0, \frac{2 - 3\delta/2}{2M} \big] \times \big[ 0, \frac{1 - \delta/2}{2M} \big] \big).
\end{align*}
Since $u_* < \infty$, we can repeat this process a finite number of times we arrive at a solution defined on $[0, u_*] \times \big[ 0, \frac{1 - \delta/2}{2M} \big]$. We can then repeat the process in the $\u$-direction to obtain a solution defined on 
\begin{align*}
    \Big([0, u_*] \times \big[ 0, \frac{1 - \delta/2}{2M} \big]\Big) \cup \Big(\big[ 0, \frac{1 - \delta/2}{2M} \big] \times [0, \u_*]\Big).
\end{align*}
We then repeat the argument from the beginning of the proof on the smaller domain $\big[ \frac{1 - \delta/2}{2M}, u_* \big] \times \big[ \frac{1 - \delta/2}{2M}, \u_* \big]$. Since the domain has shrunk and $\delta$ depends only on $M$, we can repeat this argument a finite number of times to obtain a solution defined on all of $D$.

The final statement of the proposition follows directly from the fact that $F, \underline F, U$, and $\underline U$ are continuous, as well as Proposition \ref{prop:FTC_Banach_space}.
\end{proof}

We also have the following higher-regularity statements.

\begin{thm}\label{thm:improved_regularity_1}
Let $F, \underline F$ be continuous and globally Lipschitz in $\mathbb{X}$, uniformly in $(u, \u)$. Suppose that 
\begin{align*}
    \frac{\d F}{\d \u}, \frac{\d F}{\d U}, \frac{\d F}{\d \underline U}, \frac{\d \underline F}{\d u}, \frac{\d \underline F}{\d U}, \frac{\d \underline F}{\d \underline U}
\end{align*}
exist and are continuous. Also, suppose that the initial data $U_0$ and $\underline U_0$ are continuously differentiable. Then the unique solution to \eqref{eq:ode_system} is in $C^1(D; \mathbb{X})$.
\end{thm}

\begin{proof}
The continuity of $\d U/ \d u$ and $\d \underline U/ \d \u$ follows immediately from Proposition \ref{prop:FTC_Banach_space} and Theorem \ref{thm:picard_1}. The existence and continuity of $\d U / \d \u$ follows by considering the linear ODE 
\begin{align*}
    \frac{\d V}{\d u}(u, \u) &=  \frac{\d F}{\d U}(u, \u, U(u, \u), \underline U(u, \u)) V(u, \u) + \frac{\d F}{\d \u}(u, \u, U(u, \u), \underline U(u, \u)) \\ 
    &\qquad + \frac{\d F}{\d \underline U}(u, \u, U(u, \u), \underline U(u, \u))\frac{\d \underline U}{\d \underline u}(u, \u).
\end{align*}
Here, $\frac{\d F}{\d U}(u, \u, U(u, \u), \underline U(u, \u))$ is to be interpreted in the Fr\'echet sense as a bounded linear map from $\mathbb{Y}$ to $\mathbb{Y}$. By the assumptions of the proposition and by applying\footnote{In fact only usual Banach space-valued ODE theory is required here.} Theorem \ref{thm:picard_1} we obtain a unique continuous solution to this ODE. One then shows that this solution is equal to $\frac{\d U}{\d \underline u}(u, \u)$. Similarly one argues for $\frac{\d \underline U}{\d u}$.
\end{proof}

\begin{thm}\label{thm:higher_regularity_2}
Let $k \geq 1$. Let $F, \underline F$  be continuous and globally Lipschitz in $\mathbb{X}$, uniformly in $(u, \u)$. Suppose that, in addition, they are continuously $k$-times Fr\'echet differentiable in all variables. Suppose that the initial data satisfies
\begin{align*}
    U_0 \in C^k([0, \u_*]; \mathbb{Y}) \qquad \text{and} \qquad \underline U_0 \in C^k([0, u_*]; \underline{\mathbb{Y}}).
\end{align*}
Then $(U, \underline U) \in C^k(D; \mathbb{X})$, and also 
\begin{align*}
    \frac{\d U}{\d u}, \frac{\d \underline U}{\d \u} \in C^k(D; \mathbb{X}).
\end{align*}
\end{thm}
\begin{remark}
This theorem illustrates that in the direction of propagation, the solutions are one degree more regular than in other directions (a phenomenon that occurs in usual ODE theory).
\end{remark}

\begin{proof}
The first statement follows by induction, following the same outline as in the proof of Theorem \ref{thm:improved_regularity_1}. The last statement follows then from the equations \eqref{eq:ode_system}.
\end{proof}

Additionally, we have the following statements for maps in the Sobolev space setting. Throughout the rest of this section, we consider an arbitrary smooth Riemannian metric $\gamma$ on $S^2$, and all Sobolev spaces on $S^2$ are defined with respect to $\gamma$.

\begin{thm}\label{thm:sobolev_ode_1}
Let $\mathbb{Y} = H^k(S^2; Z)$ and $\underline{\mathbb{Y}} = H^k(S^2; Z')$ for $Z, Z'$ finite-dimensional real vector bundles over $S^2$. Let the assumptions of Theorem \ref{thm:picard_1} hold (in particular note that the initial data $U_0, \underline U_0$ is assumed only to be continuous). Let $U, \underline U$ be the solution to \eqref{eq:ode_system} with initial data $U_0, \underline U_0$. Then : 
\begin{enumerate}
    \item If $k = 2$, then for any $\alpha \in (0, 1)$, we have $U, \underline U \in C^0(D; C^{0,\alpha}(S^2))$. Also, these may be identified with maps $D \times S^2 \to Z$ and $D \times S^2 \to Z'$, and under this identification, $U$ and $\underline U$ are continuous, $U$ is continuously differentiable in the variable $u$, and $\underline U$ is continuously differentiable in the variable $\u$.
    \item If $k \geq 2$, then for any $\alpha \in (0, 1)$, we have $U, \underline U \in C^0(D; C^{k - 2, \alpha}(S^2))$.
\end{enumerate}
\end{thm}

\begin{proof}
For brevity, we write $H^k(S^2)$ or simply $H^k$ to denote either $H^k(S^2; Z)$ or $H^k(S^2; Z')$. By Theorem \ref{thm:picard_1}, there is a unique solution $(U, \underline U) \in C^0(D; \mathbb{X})$, which implies that $U \in C^0(D; H^k), \underline U \in C^0(D; H^k)$. Therefore, by the Sobolev inequality, for every $(u, \u) \in D$, $U(u, \u)$ and $\underline U(u, \u)$ can be identified with continuous functions on the sphere, and for any $\alpha \in (0, 1)$, we have
\begin{multline*}
    \sup_{(u, \u) \in D}\big(\norm{U(u, \u)}_{C^{0,\alpha}(S^2)} + \norm{\underline U(u, \u)}_{C^{0,\alpha}(S^2)}\big) \\ \leq C \sup_{(u, \u) \in D}\big(\norm{U(u, \u)}_{H^2(S^2)} + \norm{\underline U(u, \u)}_{H^k(S^2)}\big)  \leq C\norm{(U, \underline U)}_{C^0(D; \mathbb{X})}.
\end{multline*}
We still need to prove the continuity statement. We prove that $U \in C^0(D; C^{0,\alpha}(S^2))$; the statement for $\underline U$ is proven analogously. First, note that
\begingroup
\allowdisplaybreaks
\begin{align*}
    \norm{U(u, \u) - U(u', \u)}_{C^{0,\alpha}(S^2)} &\leq \Big\|{\int_{u'}^u F(u'', \u, U(u'', \u), \underline U(u'', \u))\, du''}\Big\|_{C^{0,\alpha}(S^2)} \\ 
    &\leq {\int_{u'}^u \big\|F(u'', \u, U(u'', \u), \underline U(u'', \u))\big\|_{C^{0,\alpha}(S^2)}\, du''} \\ 
    &\leq C\int_{u'}^u\big\|F(u'', \u, U(u'', \u), \underline U(u'', \u))\big\|_{H^2(S^2)}\, du'' \\ 
    &\leq C|u - u'|,
\end{align*}
\endgroup
the last line since $U, \underline U$ are continuous maps into $H^k(S^2)$ and hence $(U(u'', \u), \underline U(u'', \u))$ lies in a compact subset of $\mathbb{X}$, and hence (since $F$ is continuous) the integrand is uniformly bounded.  

Next, note that (without loss of generality let $u' < u$)
\begin{align*}
    \norm{U(u, \u) - U(u', \u')}_{C^{0,\alpha}(S^2)} &\leq \norm{U_0(\u) - U_0(\u')}_{C^{0,\alpha}(S^2)} \\ 
    &\hspace{3mm} + \Big\|\int_0^u F(u'', \u, U, \underline U) \, du'' - \int_0^{u'} F(u'', \u', U, \underline U)\, du''\Big\|_{C^{0,\alpha}(S^2)} \\ 
    &\leq \underbrace{\norm{U_0(\u) - U_0(\u')}_{C^{0,\alpha}(S^2)}}_{\text{I}} \\ 
    &\quad+ \underbrace{\Big\|\int_{u'}^u F(u'', \u, U, \underline U) \, du''\Big\|_{C^{0,\alpha}(S^2)}}_{\text{II}} + \text{ III}
\end{align*}
where
\begin{align*}
    \text{III} &= \Big\|\int_0^{u'} F(u'', \u, U(u'', \u), \underline U(u'', \u) - F(u'', \u', U(u'', \u'), \underline U(u'', \underline u '))\, du''\Big\|_{C^{0,\alpha}(S^2)}.
\end{align*}
Now, 
\begin{align*}
    \text{I} &\leq C\norm{U_0(\u) - U_0(\u')}_{H^2(S^2)} \xrightarrow[u' \to u]{} 0,
\end{align*}
since $U_0$ is continuous. Also, by the same reasoning as above, the integrand in II is uniformly bounded and hence 
\begin{align*}
    \text{II} &\leq C|u' - u|.
\end{align*}
Finally, we have
\begin{align*}
    \text{III} &\leq C \int_0^{u'} \norm{\text{ integrand of III }}_{H^2(S^2)}\, du'' \\ 
    &\leq CM \int_0^{u'} \norm*{(U(u'', \u), \underline U(u'', \u)) - (U(u'', \u'), \underline U(u'', \u'))}_{H^2(S^2)} \, du''.
\end{align*}
Let $\eps > 0$ be arbitrary. Since $U, \underline U$ are continuous maps $D \to H^k(S^2)$, we can pick $\delta > 0$ such that whenever $|(u'', \u) - (u'', \u')| < \delta$, 
$$
\norm{U(u'', \u) - U(u'', \u')}_{H^k(S^2))} , \displaystyle \norm{\underline U(u'', \u) - \underline U(u'', \u')}_{H^k(S^2))} < \frac{\eps}{2 u_* CM}.
$$
Then for all such $\u, \u'$, we have
\begin{align*}
    \text{III} &\leq CM u' \big[ 2 \cdot \frac{\eps}{2 u_* CM} \big] \leq \eps.
\end{align*}
This shows that 
\begin{align*}
    \lim_{(u, \u) \to (u', \u')} \norm{U(u, \u) - U(u', \u')}_{C^{0,\alpha}(S^2)} &= 0
\end{align*}
and hence $U \in C^0(D; C^{0,\alpha}(S^2))$. By a similar proof, $\underline U \in C^0(D; C^{0,\alpha}(S^2))$.

We now show that $U$ may be identified as a continuous map on $D \times S^2$. Denote for $(u, \u, p) \in D \times S^2$
\begin{align*}
    \tilde U(u, \u, p) &= U(u, \u)(p) \qquad \text{and} \qquad \tilde{\underline U}(u, \u, p) = \underline U(u, \u)(p). \eqcount\label{eq:computation_55}
\end{align*}
Let $\eps > 0$ be arbitrary. By the previous part of this theorem, we can pick $\delta > 0$ such that 
\begin{enumerate}
    \item $\displaystyle \big|U(u, \u)(p) - U(u, \u)(q)\big| < \eps$ for any $q \in B(p; \delta)$, 
    \item $\displaystyle \big\| U(u, \u) - U(u', \u')\big\|_{C^0(S^2)} < \eps$ for any $|(u', \u') - (u, \u)| < \delta$.
\end{enumerate}
Then for any $q \in B(p; \delta)$ and $(u', \u') \in B((u, \u); \delta)$, we have
\begin{align*}
    \big|\tilde{U}(u, \u, p) - \tilde U(u', \u', q)\big| &\leq \big|\tilde{U}(u, \u, p) - \tilde U(u, \u, q)\big| + \big|\tilde{U}(u, \u, q) - \tilde U(u', \u', q)\big| \\
    &= \big|{U}(u, \u)(p) -  U(u, \u)( q)\big| + \big|{U}(u, \u)( q) -  U(u', \u')( q)\big| \\ 
    &\leq \eps + \big\|U(u, \u) - U(u', \u')\big\|_{C^0(S^2)} \\ 
    &\leq 2\eps.
\end{align*}
Therefore $\tilde U \in C^0(D \times S^2)$, and we may identify $\tilde U$ with $U$. Similarly for $\underline U$ and $\tilde{\underline U}$. Now by the equation \eqref{eq:ode_system}, $\d U/\d u$ and $\d \underline U/ \d \underline u$ are both equal to continuous functions on $S^2$, which proves the last two statements of (1). This proves the $k = 2$ part of this proposition. Item (2) is proven similarly. 
\end{proof}

In what follows, we do not distinguish in our notation between $\tilde U$ and $U$ or $\tilde{\underline U}$ and $\underline U$. That is, if $U \in C^0(D; H^k(S^2))$, we also write $U$ to denote the function defined a.e. on $D \times S^2$ as in \eqref{eq:computation_55} above (and similarly for $\underline U$).

\begin{thm}\label{thm:sobolev_ode_2} Let $k \geq 2$ and let $m \geq 1$.
Let $\mathbb{Y}, \underline{\mathbb{Y}}$ be as in the previous theorem. Let $F, \underline F$ be continuous and globally Lipschitz in $\mathbb{X}$, uniformly in $(u, \u)$, and suppose in addition that they are continuously $m$-times differentiable in all variables. Assume that the initial data satisfies
\begin{align*}
    U_0 \in C^m([0, \u_*]; \mathbb{Y}) \qquad \text{and} \qquad \underline U_0 \in C^m([0, u_*]; \underline{\mathbb{Y}}).
\end{align*}
Then:
\begin{align*}
    U, \underline U, \frac{\d U}{\d u}, \frac{\d \underline U}{\d \u} \in C^m_D C^{k - 2}_{S^2}(D \times S^2)  \subset C^{\min(m, k - 2)}(D \times S^2).
\end{align*}
We use the notation $C^a_X C^b_Y(X \times Y)$ to denote the space of functions $f$ on $X \times Y$ for which any partial derivative of the form 
\begin{align*}
    \frac{\d^\alpha}{\d x^\alpha} \frac{\d^\beta}{\d y^\beta} f, \qquad \alpha \leq a, \beta \leq b
\end{align*}
exists and is continuous.
\end{thm}

\begin{proof}
By Theorems \ref{thm:higher_regularity_2} and  \ref{thm:sobolev_ode_1},  we know that
\begin{align*}
    U, \underline U \in C^0(D; C^{k - 2, \alpha}(S^2)) \qquad \text{and} \qquad U, \underline U, \frac{\d U}{\d u}, \frac{\d \underline U}{\d \u} \in C^m(D; H^k(S^2)).
\end{align*}
Let $f \in C^m(D; H^k(S^2))$ and let $g : D \to H^k(S^2)$ be a mixed $(u, \u)$-partial derivative of $f$ of order at most $m$. Since $f \in C^m(D; H^k(S^2))$, we have $g \in C^0(D; H^k(S^2))$. By Sobolev embedding, for any $\alpha \in (0, 1)$,
\begin{align*}
    \norm{g(u, \u) - g(u', \u')}_{C^{k - 2, \alpha}(S^2)} &\leq C\norm{g(u, \u) - g(u', \u')}_{H^k(S^2)}
\end{align*}
which tends to 0 as $(u', \u') \to (u, \u)$, 
and hence $g \in C^0(D; C^{k - 2, \alpha}(S^2))$. 

Now, the goal is to show that $g \in C^0(D \times S^2)$. Let $\eps > 0$ be arbitrary. By the above, there is a $\delta > 0$ such that if $|(u, \u) - (u', \u')| < \delta$, then $\norm{g(u, \u) - g(u', \u')}_{C^{k - 2}(S^2)} < \frac{\eps}{2}$, i.e.
\begin{align*}
    \sup_{\theta \in S^2} |g(u, \u, \theta) - g(u', \u', \theta)| < \frac{\eps}{2}.
\end{align*}
Also, since $g\in C^0(D; C^{k - 2,\alpha}(S^2))$ and $k \geq 2$, there is a constant $C$ such that 
$$\norm{g(u, \u)}_{C^{0,\alpha}(S^2)} \leq C
$$
for all $(u, \u) \in D$. Therefore, for all $(u', \u', \theta') \in D \times S^2$ satisfying
\begin{align*}
    \text{dist}\big( (u, \u, \theta) , (u', \u', \theta') \big) \coloneqq |(u, \u) - (u', \u')| + \text{dist}_{S^2}(\theta, \theta') < \min\big(\frac{\eps^{1/\alpha}}{(2C)^{1/\alpha}}, \delta\big),
\end{align*}
we have:
\begin{align*}
    |g(u, \u, \theta) - g(u', \u', \theta')| &\leq |g(u, \u, \theta) - g(u, \u, \theta')| + |g(u, \u, \theta') - g(u', \u', \theta')| \\ 
    &\leq \norm{g(u, \u)}_{C^{0,\alpha}(S^2)}\text{dist}_{S^2}(\theta, \theta')^\alpha + \sup_{\theta'' \in S^2} |g(u, \u, \theta'') - g(u', \u', \theta'')| \\ 
    &\leq \eps.
\end{align*}
Hence $g$ is (uniformly) continuous on $D \times S^2$. This proves that $f$ is $m$-times continuously differentiable in the $u$ and $\u$ variables.

We now investigate the spherical regularity of $g$. Let $0 \leq j \leq k - 2$. Since $g$ maps $D$ continuously into $C^{k - 2, \alpha}(S^2)$, $\sl\nabla^j(g(u, \u)) \eqqcolon h$ exists, where $\sl\nabla$ denotes the Levi-Civita connection on $(S^2,\gamma)$, and $\theta\mapsto h(u, \u, \theta)$ is $\alpha$-H\"older continuous, with $C^{0, \alpha}(S^2)$ norm bounded by a constant $C$ independent of $(u, \u)$. Now, for any $(u', \u', \theta')$ we have
\begin{align*}
    |h(u, \u, \theta) - h(u', \u', \theta')| &\leq |h(u, \u, \theta) - h(u, \u, \theta')| + |h(u, \u, \theta') - h(u', \u', \theta')| \\ 
    &\leq \norm{h(u, \u)}_{C^{0,\alpha}(S^2)}\text{dist}_{S^2}(\theta, \theta')^\alpha + \norm{h(u, \u) - h(u', \u')}_{C^0(S^2)} \\ 
    &\leq C\text{dist}_{S^2}(\theta, \theta')^\alpha + \norm{h(u, \u) - h(u', \u')}_{C^0(S^2)}
\end{align*}
Therefore $h$ is continuous on $D \times S^2$. This shows that
\begin{align*}
    f \in C^m_D C^{k - 2}_{S^2}(D \times S^2).
\end{align*}
The result then follows by setting $f$ equal to $U, \underline U, \frac{\d U}{\d u}$, and $\frac{\d\underline U}{\d \u}$.
\end{proof}

\begin{corollary}[Smooth solutions]\label{cor:smooth_ode_solutions}
Let $F, \underline F$ be continuous and globally Lipschitz in $\mathbb{X}$, uniformly in $(u, \u)$. Suppose also that they are smooth in all variables. Assume also that the initial data is smooth. Then the solutions $U, \underline U$ to \eqref{eq:ode_system} satisfy $U, \underline U \in C^\infty(D \times S^2)$.
\end{corollary}

\begin{proof}
This follows from the fact that existence and uniqueness have been shown for every regularity level uniformly in $(u, \u)$. That is, for every $m \geq 1$ and $k \geq 2$, there is a unique solution $(U, \underline U)\in C^{\min(m, k - 2)}(D \times S^2)$, where $D = D_{u_*, \u_*}$. Apply this theorem with $m = k - 2 \geq 1$ to obtain, for every such $m$, a  unique solution $(U_m, \underline U_m) \in C^m(D \times S^2)$. Note that since uniqueness holds at the lowest regularity level (Theorem \ref{thm:picard_1}), $U_1 = U_m$ and $\underline U_1 = \underline U_m$ for all $m \geq 1$. Therefore $(U_1, \underline U_1) \in C^\infty(D \times S^2)$, and this is the desired smooth solution.
\end{proof}

\section{Useful formulae}\label{app:useful_formulae}

\begin{prop}\label{prop:commutation_formulae} Let $(M, g)$ be vacuum spacetime of the type described in Section \ref{subsec:spacetime_and_notation}. We record here the following commutation formulae. For a reference, the reader is directed to Lemma 7.3.3 in \cite{CK}. 

If $f$ is a scalar function on $(M, g)$, then the following hold:
\begin{align*}
    \sl\nabla_4\sl\nabla_A f &= \sl\nabla_A \sl\nabla_4 f + \frac{1}{2}(\eta + \underline\eta)_A\sl\nabla_4 f- \hat\chi_{AB}\sl\nabla^B f - \frac{1}{2}\tr\chi\sl\nabla_A f \\ 
    \sl\nabla_3\sl\nabla_A f &= \sl\nabla_A \sl\nabla_3 f + \frac{1}{2}(\eta + \underline\eta)_A\sl\nabla_3 f- \hat{\underline\chi}_{AB}\sl\nabla^B f - \frac{1}{2}\tr\underline\chi\sl\nabla_A f \\
    \sl\nabla_3\sl\nabla_4 f &= \sl\nabla_4 \sl\nabla_3 f - 2\omega \sl\nabla_3 f + 2\underline\omega \sl\nabla_4 f + 2(\eta - \underline\eta)\cdot\sl\nabla f. 
\end{align*}

If $\theta$ is an $S$ 1-form on $(M, g)$, then the following hold: 
\begin{align*}
    \sl\nabla_4\sl\nabla_B\theta_A &= \sl\nabla_B \sl\nabla_4\theta_A - \chi_{BC}\sl\nabla^C \theta_A + \frac{1}{2}(\eta + \underline\eta)_B \sl\nabla_4 \theta_A \\ 
    &\hspace{40mm} + \chi_{AB}\theta\cdot\underline\eta - \underline\eta_A \chi_{BC}\theta^C + \prescript{*}{}{\beta}_B[W] \prescript{*}{}{\theta}_A \\ 
    \sl\nabla_3\sl\nabla_B\theta_A &= \sl\nabla_B \sl\nabla_3\theta_A - \underline\chi_{BC}\sl\nabla^C \theta_A + \frac{1}{2}(\eta + \underline\eta)_B \sl\nabla_3 \theta_A \\ 
    &\hspace{40mm} + \underline\chi_{AB}\theta\cdot\eta - \eta_A \underline\chi_{BC}\theta^C - \prescript{*}{}{\underline\beta}_B[W] \prescript{*}{}{\theta}_A \\ 
    \sl\nabla_3 \sl\nabla_4 \theta_A &= \sl\nabla_4\sl\nabla_3\theta_A - 2\omega\sl\nabla_3\theta + 2\underline\omega \sl\nabla_4\theta + 2(\eta - \underline\eta)^B\sl\nabla_B \theta_A \\ 
    &\hspace{40mm} + 2(\eta\cdot \theta)\underline\eta_A - 2(\underline\eta\cdot\theta)\eta_A + 2\sigma[W] \prescript{*}{}{\theta}_A.
\end{align*}

If $\theta$ is a symmetric traceless 2-covariant $S$ tensorfield on $(M, g)$, then the following hold:
\begin{align*}
    \sl\nabla_4 \sl\nabla_B\theta_{A_1 A_2} &= \sl\nabla_B\sl\nabla_4\theta_{A_1 A_2} - \chi_{BC}\sl\nabla^C\theta_{A_1 A_2} + \frac{1}{2}(\eta + \underline\eta)_B\sl\nabla_4\theta_{A_1 A_2} \\ 
    &\hspace{15mm} + \big(\chi_{A_1 B}\underline\eta_C - \chi_{BC}\underline\eta_{A_1} + \epsilon_{A_1 C}\prescript{*}{}{\beta_B}\big)\tensor{\theta}{^C _{A_2}} \\ 
    &\hspace{15mm} + \big(\chi_{A_2 B}\underline\eta_C - \chi_{BC}\underline\eta_{A_2} + \epsilon_{A_2 C}\prescript{*}{}{\beta_B}\big)\tensor{\theta}{^C _{A_1}} \\ 
    \sl\nabla_3 \sl\nabla_B\theta_{A_1 A_2} &= \sl\nabla_B\sl\nabla_3\theta_{A_1 A_2} - \underline\chi_{BC}\sl\nabla^C\theta_{A_1 A_2} + \frac{1}{2}(\eta + \underline\eta)_B\sl\nabla_3\theta_{A_1 A_2} \\ 
    &\hspace{15mm} + \big(\underline\chi_{A_1 B}\eta_C - \underline\chi_{BC}\eta_{A_1} - \epsilon_{A_1 C}\prescript{*}{}{\underline\beta_B}\big)\tensor{\theta}{^C _{A_2}} \\ 
    &\hspace{15mm} + \big(\underline\chi_{A_2 B}\eta_C - \underline\chi_{BC}\eta_{A_2} - \epsilon_{A_2 C}\prescript{*}{}{\underline\beta_B}\big)\tensor{\theta}{^C _{A_1}}.
\end{align*}
\end{prop}
\begin{remark}
As a direct corollary to this proposition, we obtain the following useful formulae. For $\theta$ an $S$ 1-form on $(M, g)$, we have
\begin{align*}
    \sl\nabla_4\sl\div\theta &= \sl\div(\sl\nabla_4\theta) - \hat\chi\cdot\sl\nabla\theta - \frac{1}{2}\tr\chi\sl\div\theta + \frac{1}{2}(\eta + \underline\eta)\cdot(\sl\nabla_4 \theta) \\
    &\hspace{15mm} + \frac{1}{2}\tr\chi\theta\cdot\underline\eta - \underline\eta_A\hat\chi^{AB}\theta_B + \theta\cdot\beta \\ 
    \sl\nabla_4 \sl\curl\theta &= \sl\curl(\sl\nabla_4\theta) - \epsilon^{AB}\hat\chi_{AC}\sl\nabla^C\theta_B - \frac{1}{2}\tr\chi\sl\curl\theta+ \frac{1}{2}(\eta + \underline\eta) \wedge (\sl\nabla_4 \theta) \\ 
    &\hspace{15mm} - \hat\chi^{BC}\prescript{*}{}{\underline\eta}_B\theta_C - \frac{1}{2}\tr\chi \theta \wedge \underline\eta + \beta\wedge\theta,
\end{align*}
with similar formulae holding for $\sl\nabla_3$ which can be obtained by conjugation. 

For $\theta$ a symmetric traceless 2-covariant $S$ tensorfield on $(M, g)$, we have
\begin{align*}
    \sl\nabla_4\sl\div\theta_A &= \sl\div\sl\nabla_4\theta_A - \hat\chi_{BC}\sl\nabla^B\tensor{\theta}{^C _{A}} - \frac{1}{2}\tr\chi\sl\div\theta_A + \frac{1}{2}(\eta + \underline\eta)^B\sl\nabla_4\theta_{AB} \\ 
    &\hspace{15mm} + \big( \hat\chi_{AB}\underline\eta_C - \hat\chi_{BC}\underline\eta_A \big)\theta^{BC} - \hat\chi_{BC}\underline\eta^B \tensor{\theta}{_A ^C} + \tr\chi\underline\eta^C\theta_{AC} + 2\beta^B\theta_{AB}.
\end{align*}
Again a similar formula holds for $\sl\nabla_3$ which can be obtained by conjugation.
\end{remark}

Next we record several useful identities concerning $S$ tensorfields.

\begin{prop}\label{prop:proposition_3}
If $\theta$ and $\xi$ are symmetric traceless 2-covariant $S$ tensorfields, then
\begin{equation}
    \theta \wedge \xi = - \prescript{*}{}{\theta}\cdot \xi = \theta \cdot \prescript{*}{}{\xi}.
\end{equation}
If $\theta$ and $\xi$ are $S$ 1-forms, then 
\begin{equation}
    \theta \wedge \xi = - \prescript{*}{}{\theta} \cdot \xi = \theta \cdot \prescript{*}{}{\xi} 
\end{equation}
and 
\begin{equation}
    \prescript{*}{}{\theta}_A \prescript{*}{}{\xi}_B + \theta_A \xi_B = (\theta \cdot\xi) \gamma_{AB}.
\end{equation}
If $\theta$ is a symmetric traceless 2-covariant $S$ tensorfield and $\xi$ is an $S$ 1-form, then
\begin{equation}
    \theta \wedge \xi = -\prescript{*}{}{\theta}\cdot \xi = \theta \cdot \prescript{*}{}{\xi}.
\end{equation}
\end{prop}

\begin{prop}\label{prop:proposition_4}
If $\theta$ is a symmetric traceless 2-covariant $S$ tensorfield and $\xi$ any $S$ 1-form, then 
\begin{equation}
\begin{split}
    \sl\curl\theta &= \prescript{*}{}{(\sl\div}\theta) \\ 
    \sl\nabla_A \theta_{BC} - \sl\nabla_B \theta_{AC} &= \sl\epsilon_{AB}\sl\curl\theta_C.
\end{split}
\end{equation}
\end{prop}

\begin{prop}\label{prop:proposition_5}
If $\theta$ and $\xi$ are symmetric traceless 2-covariant $S$-tensor fields, then we have
\begin{equation}
    -\frac{1}{2}\theta_{AB}(\sl\div\xi)^B + \frac{1}{2}\theta^{BC}\sl\nabla_A \xi_{BC} - \frac{1}{2}\theta^{BC}\sl\nabla_B \xi_{AC} = 0.
\end{equation}
\end{prop}

Finally, we have a two-variable Gr\"onwall-type lemma. 

\begin{prop}\label{prop:gronwall_1}
Suppose nonnegative functions $f, g : D_{u_*, \u_*} \to \R$ satisfy the differential inequality
\begin{align*}
    \d_u f + \d_{\u} g \leq C(f + g). \eqcount\label{eq:computation_26}
\end{align*}
Then there exists a constant $C' = C'(u_*, \u_*, C)$ (depending continuously on $u_*, \u_*$, and $C$) such that for all $(u, \u) \in D_{u_*, \u_*}$, 
\begin{equation*}
    \int_0^{\u} f(u, \u')\, d\u' + \int_{0}^u g(u', \u)\, du' \leq C' \Big[ \int_0^{\u} f(0, \u')\, d\u' + \int_0^u g(u', 0)\, du' \Big].
\end{equation*}
\end{prop}

\begin{proof}
Integrating \eqref{eq:computation_26} in $u$ and $\u$, we obtain, for any $(u, \u) \in D$:
\begin{multline}\label{eq:computation_27}
    \int_0^{\u} f(u, \u')\, d\u' - \int_0^{\u} f(0, \u')\, d\u' + \int_0^u g(u', \u)\, du' - \int_0^u g(u', 0)\, du' \\ \leq C\int_0^{\u} \int_0^{u} f(u', \u') + g(u', \u') \, du'\, d\u'.
\end{multline}
Call 
\begin{align*}
    F_{\u}(u) = \int_0^{\u} f(u, \u')\, d\u' \qquad \text{and} \qquad G_u (\u) = \int_0^{u} g(u', \u)\, du'.
\end{align*}
Dropping the integral of $g(u', \u)$ from \eqref{eq:computation_27} yields 
\begin{align*}
    F_{\u}(u) &\leq  F_{\u}(0) + G_u (0) + C \int_0^{\u} G_{u}(\u') \, d\u' + C \int_0^{u} F_{\u}(u') \, du'. \eqcount\label{eq:computation_28}
\end{align*}
Applying Gr\"onwall's inequality gives
\begin{align*}
    F_{\u}(u) &\leq \Big(F_{\u}(0) + G_u (0) +  C\int_0^{\u} G_u (\u') \, d\u'\Big)e^{Cu_*}. \eqcount\label{eq:computation_29}
\end{align*}
Dropping instead the integral of $f(u, \u')$ from \eqref{eq:computation_27} and inserting the bound \eqref{eq:computation_29} in for the double integral of $f(u', \u')$ on the right-hand side yields
\begin{align*}
    G_u(\u) &\leq F_{\u}(0) + G_u(0) + C \int_0^{\u} G_u(\u') \, d\u' \\ 
    &\quad + C e^{Cu_*}\int_0^u  \Big[ \int_0^{\u} f(0, \u')\, d\u' + \int_0^{u'} g(u'', 0)\, du'' + C\int_0^{u'}\int_0^{\u} g(u'', \u')du''\, d\u' \Big]  \, du' \\ 
    &\leq F_{\u}(0) + G_u(0) + C \int_0^{\u} G_u(\u') \, d\u' \\ 
    &\quad + C e^{Cu_*}u_* \Big[ \int_0^{\u} f(0, \u')\, d\u' + \int_0^{u} g(u', 0)\, du' + C\int_0^{u}\int_0^{\u} g(u', \u')du'\, d\u' \Big]  \\ 
    &\leq (1 + Ce^{Cu_*} u_* )\big( F_{\u}(0) + G_u(0) \big) + C\big( 1 + Ce^{Cu_*}u_* \big)\int_0^{\u}G_u(\u')\, d\u'.
\end{align*}
Applying Gr\"onwall's inequality gives 
\begin{equation}\label{eq:computation_30}
    G_u(\u) \leq (1 + Ce^{Cu_*} u_* ) \exp\big( C ( 1 + Ce^{Cu_*}u_* )\u_* \big) \big( F_{\u}(0) + G_u(0) \big).
\end{equation}
This is half of the desired conclusion. Now, plugging \eqref{eq:computation_30} into \eqref{eq:computation_29}, we obtain:
\begin{align*}
    F_{\u}(u) &\leq e^{Cu_*}\big( F_{\u}(0) + G_u (0) \big) \\ 
    &\quad + Ce^{C u_*} \int_0^{\u} (1 + Ce^{Cu_*} u_* ) \exp\big( C ( 1 + Ce^{Cu_*}u_* ) \u_*\big) \big( F_{\u'}(0) + G_u(0) \big) \, d\u' \\ 
    &\leq e^{Cu_*}\big( F_{\u}(0) + G_u (0) \big) \\ 
    &\quad + Ce^{C u_*} (1 + Ce^{Cu_*} u_* ) \exp\big( C ( 1 + Ce^{Cu_*}u_* )\u_* \big) \u_*\big( F_{\u}(0) + G_u(0) \big) \\ 
    &= e^{C u_*}\big[ 1 + C(1 + Ce^{Cu_*} u_* ) \u_* \exp\big( C ( 1 + Ce^{Cu_*}u_* ) \u_* \big) \big]\big( F_{\u}(0) + G_u(0) \big).
\end{align*}
Setting 
\begin{multline}
    C' = (1 + Ce^{Cu_*} u_* ) \exp\big( C ( 1 + Ce^{Cu_*}u_* )\u_* \big) \\ 
    +  e^{C u_*}\big[ 1 + C(1 + Ce^{Cu_*} u_* ) \u_*\exp\big( C ( 1 + Ce^{Cu_*}u_* ) \u_* \big)  \big]
\end{multline}
completes the proof.
\end{proof}

\section{Computations for Theorem \ref{thm:prop_eq_for_diff_constraints}}\label{app:Bianchi_computations}

In this appendix, we use $\simeq_k$ to denote equality, modulo terms of order $< k$, where we use the term ``order'' in the sense discussed in \eqref{eq:computation_37}. For example, due to the Bianchi equation 
\begin{align*}
    \sl\nabla_3 \alpha &= (4\underline\omega - \frac{1}{2}\tr\underline\chi)\alpha + \sl\nabla\hat\otimes \beta + (4 \eta + \zeta) \hat\otimes \beta - 3\hat\chi\rho - 3\prescript{*}{}{\hat\chi}\sigma,
\end{align*}
we have
\begin{align*}
    \sl\nabla_3 \alpha \simeq_1 \sl\nabla\hat\otimes \beta.
\end{align*}
Note that $A \simeq_0 B$ if and only if $A = B$ exactly.
Using this notation in the context of the commutation formulae of Proposition \ref{prop:commutation_formulae}, we can simplify the formulae in a way which is helpful for breaking up computations into smaller pieces, for instance by noting that
\begin{align*}
    \sl\nabla_\mu \sl\nabla_\nu \theta &\simeq_2 \sl\nabla_\nu \sl\nabla_\mu \theta,
\end{align*}
for $\mu,\nu \in \{1, 2, 3, 4\}$. In addition, we have slightly simplified formulae at order 1, for instance:
\begin{align*}
    \sl\nabla_3 \sl\nabla_4 \theta_A &\simeq_1 \sl\nabla_4\sl\nabla_3\theta_A - 2\omega\sl\nabla_3\theta + 2\underline\omega \sl\nabla_4\theta + 2(\eta - \underline\eta)^B\sl\nabla_B \theta_A.
\end{align*}
Also, note that in the following computations, the identity 
\begin{align*}
    \zeta &= \frac{1}{2}(\eta - \underline\eta)
\end{align*}
is used.

\subsection{Computations for \texorpdfstring{$\sl\nabla_4\underline B$}{nabla4Bbar}}\label{subsec:Bbar} The only order 2 terms arise in the terms $\sl\nabla_4 \sl\nabla_3 \underline\beta$ and $\sl\nabla_4 \sl\div\underline\alpha$. These arise when commuting derivatives (applying the commutation formulae of Proposition \ref{prop:commutation_formulae}):
\begin{align*}
    \sl\nabla_4 \sl\nabla_3 \underline\beta &\simeq_2 \sl\nabla_3 \sl\nabla_4 \underline\beta \\
    &\simeq_2 \sl\nabla_3(\prescript{*}{}{\sl\nabla\sigma} - \sl\nabla\rho) \\ 
    &\simeq_2 -\prescript{*}{}{\sl\nabla}\sl\curl\underline\beta + \sl\nabla\sl\div\underline\beta,
\end{align*}
and
\begin{align*}
    \sl\nabla_4 \sl\div\underline\alpha  &\simeq_2 \sl\div\sl\nabla_4 \underline\alpha \\ 
    &\simeq_2 -\sl\div\sl\nabla\hat\otimes\underline\beta.
\end{align*}
This shows that 
\begin{align*}
    \sl\nabla_4 \underline B &\simeq_2 -\sl\div\sl\nabla\hat\otimes\underline\beta  -\prescript{*}{}{\sl\nabla}\sl\curl\underline\beta + \sl\nabla\sl\div\underline\beta. \eqcount\label{eq:computation_50}
\end{align*}
By Lemma \ref{lem:formula_1}, the right-hand side is equal to $-2K \underline\beta$, and hence we actually have 
\begin{align*}
    \sl\nabla_4 \underline B &\simeq_2 0.
\end{align*}
The order 1 terms arise in two distinct manners. The first is as order 1 ``error'' terms in the commutation formulae for $[\sl\nabla_4, \sl\nabla_3]\underline\beta$ and $[\sl\nabla_4, \sl\div]\underline\alpha$. The second is as the principal terms in the linearized Bianchi equations for $\sl\nabla_4\underline\alpha$ and $\sl\nabla_4 \underline\beta$. Grouping terms by which unknown appears, we obtain (before simplification):
\begin{multline}\label{eq:computation_38} 
    \sl\nabla_4 \underline B_A \simeq_1  2\omega\underline B_A + (2\omega - \tr\chi)\underline B_A + 2\hat{\underline\chi}_{AB}\Xi^B \\ 
     - 2\omega\sl\div\underline\alpha_A  - (2 \omega - \tr\chi)\sl\div\underline\alpha_A -\frac{1}{2}\hat{\chi}^{BC}  \prescript{*}{}{\sl\nabla}_A \prescript{*}{}{\underline\alpha}_{BC} + \frac{1}{2}\hat\chi^{BC}\sl\nabla_A \underline\alpha_{BC}  \\ 
    - \hat\chi^{BC}\sl\nabla_B \underline\alpha_{AC} - \frac{1}{2}\tr\chi \sl\div\underline\alpha_{A} + (4\omega - \frac{1}{2}\tr\chi)\sl\div\underline\alpha_A \\ 
    + 2(\underline\eta - \eta)^B\sl\nabla_B \underline\beta_A - \underline\eta^B \sl\nabla\hat\otimes \underline\beta_{AB} + (\frac{1}{2}\underline\eta - \frac{5}{2}\eta)^B \prescript{*}{}{\sl\nabla_A}\prescript{*}{}{\underline\beta}_B  \\ 
    + \frac{1}{2}(\eta + \underline\eta)_A \sl\div\underline\beta - (\frac{3}{2}\eta + \frac{1}{2}\underline\eta)^B \sl\nabla_A \underline\beta_B - \frac{1}{2}\prescript{*}{}{(\eta + \underline\eta)}_A \sl\curl\underline\beta - 3\prescript{*}{}{\underline\eta}_A\sl\curl\underline\beta \\ 
     + 3\underline\eta_A \sl\div\underline\beta - \frac{1}{2}(\eta + \underline\eta)^B \sl\nabla\hat\otimes\underline\beta_{AB} + (\frac{1}{2}\eta - \frac{9}{2}\underline\eta)_A \sl\div\underline\beta + (\frac{1}{2}\eta - \frac{9}{2}\underline\eta)^B \sl\nabla_B \underline\beta_A \\ 
      + (\frac{1}{2}\eta - \frac{9}{2}\underline\eta)_A \sl\div\underline\beta - (\frac{1}{2}\prescript{*}{}{\eta} - \frac{9}{2}\prescript{*}{}{\underline\eta})_A\sl\curl\underline\beta \\ + \frac{3}{2}\tr\underline\chi \sl\nabla\rho
    + 2\underline\omega \sl\nabla\rho  - 2(\tr\underline\chi + \underline\omega)\sl\nabla\rho + \hat{\underline\chi}_{AB}\sl\nabla^B \rho + \frac{1}{2}\tr\underline\chi\sl\nabla\rho  - 3\hat{\underline\chi}_{AB}\sl\nabla^B \rho \\ - 2\underline\omega \prescript{*}{}{\sl\nabla\sigma} + 2(\tr\underline\chi + \underline\omega)\prescript{*}{}{\sl\nabla\sigma}  - \prescript{*}{}{\hat{\underline\chi}}_{AB}\sl\nabla^B\sigma - \frac{1}{2}\tr\underline\chi \prescript{*}{}{\sl\nabla_A}\sigma
    \\ 
    - \frac{3}{2}\tr\underline\chi \prescript{*}{}{\sl\nabla_A\sigma}  + 3\prescript{*}{}{\hat{\underline\chi}}_{AB}\sl\nabla^B \sigma + 2\hat{\chi}_{AB}(\sl\nabla \rho + \prescript{*}{}{\sl\nabla}\sigma)^B
\end{multline}
In fact the terms after the first line all sum to zero. To see this for the terms involving $\underline\alpha$, we first note that
\begin{align*}
    \hat\chi^{BC}\prescript{*}{}{\sl\nabla}_A \prescript{*}{}{\underline\alpha}_{BC} &= \hat\chi^{BC}\sl\epsilon_{AD}\sl\epsilon_{BD'} \sl\nabla^D \tensor{\underline\alpha}{^{D'} _C} \\ 
    &= \hat\chi_{AB}\sl\div\underline\alpha^B - \hat\chi^{BC}\sl\nabla_B \underline\alpha_{AC}.
\end{align*}
Note also that the first two terms on the second line, namely $-2\omega\sl\div\underline\alpha_A - (2\omega - \tr\chi)\sl\div\underline\alpha_A$, cancel with the last two terms on the third line, namely $-\frac{1}{2}\tr\chi \sl\div\underline\alpha_A + (4\omega - \frac{1}{2}\tr\chi)\sl\div\underline\alpha_A$. Finally, one applies Proposition \ref{prop:proposition_5} to see the remaining $\underline\alpha$ terms vanish.
Next, one simplifies the terms involving $\underline\beta$ by applying Propositions \ref{prop:proposition_3} and \ref{prop:proposition_4} and expanding the definition of $\sl\nabla\hat\otimes \underline\beta$. Finally, one applies Proposition \ref{prop:proposition_3} to the $\rho$ and $\sigma$ terms and groups similar terms to show that these, too, vanish. Therefore we have
\begin{equation}
    \sl\nabla_4 \underline B_A \simeq_1 (4\omega - \tr\chi)\underline B_A + 2\hat{\underline\chi}_{AB}\Xi^B.
\end{equation}

We now proceed with the order 0 term analysis. Note that this is the first time we will encounter the null Weyl tensor components of $(M, g)$, since these are where the null structure equations \eqref{eq:null_structure_propagation}-\eqref{eq:null_structure_constraint} are used, and these equations are the only place that the null Weyl tensor components of $(M, g)$ appear. There are many more order 0 terms than order 1 and 2 terms, so we proceed more cautiously. By inspecting $\sl\nabla_4 \underline B$ one sees that the order 0 terms which appear are of the form 
\begin{align*}
    I[\underline\alpha] + I[\underline\beta] + I[\rho,\sigma] + I[\beta],
\end{align*}
where $I[\Psi]$ denotes a sum of terms, each of which contains exactly the unknown $\Psi \in \{\alpha, \beta, \rho, \sigma, \underline\beta, \underline\alpha\}$ and \emph{no others}. There can be no terms in the expression for $\sl\nabla_4 \underline B$ which are quadratic in the unknowns. Indeed the only way for this to happen would be in a term of the form $\sl\nabla_\mu \psi \cdot \Psi$ with $\mu \in \{3, 4\}$. However, the expression for $\sl\nabla_\mu \psi$ we obtain via the null structure equations \eqref{eq:null_structure_propagation} does not involve any unknowns, only Ricci coefficients and null Weyl tensor components. 

Gathering terms in the expression for $\sl\nabla_4 \underline B_A$, we obtain the following. For each $I[\Psi]$, we first write the unsimplified expression and then discuss its simplification.

For $I[\underline\alpha]$:
\begin{multline}\label{eq:computation:39}
    I[\underline\alpha]_A = \Big[\hat\chi_{BC}(\eta - \underline\eta)^C + \frac{1}{2}\tr\chi (\eta - \underline\eta)_B + \beta[W]_B - 4\sl\nabla_B\omega - 2\omega (\eta + \underline\eta)_B \\ 
    - \hat\chi_{BC}(\eta + \underline\eta)^C - \frac{1}{2}\tr\chi(\eta + \underline\eta)_B\Big]\tensor{\underline\alpha}{_A ^B} + \underline\eta^B\underline\alpha_{AB} \big(4\omega - \frac{1}{2}\tr\chi\big) - \frac{1}{2}\prescript{*}{}{\sl\nabla_A}\hat\chi^{BC}\prescript{*}{}{\underline\alpha}_{BC} \\
    - \frac{1}{4}\prescript{*}{}{(\eta + \underline\eta)}_A(\hat\chi\cdot\prescript{*}{}{\underline\alpha}) + \frac{1}{2}\sl\nabla_A\hat\chi^{BC}\underline\alpha_{BC} + \frac{1}{4}(\eta + \underline\eta)_A(\hat\chi\cdot\underline\alpha) - \frac{3}{2}\prescript{*}{}{\underline\eta}_A(\hat\chi\cdot\prescript{*}{}{\underline\alpha}) \\ 
    + \frac{3}{2}\underline\eta_A (\hat\chi\cdot\underline\alpha) + \big( \hat\chi_{AB}\underline\eta_C - \hat\chi_{BC}\underline\eta_A \big) \underline\alpha^{BC} - \hat\chi^{BC}\underline\eta_B \underline\alpha_{AC} + \tr\chi\underline\eta^C \underline\alpha_{AC} + 2\beta[W]^B\underline\alpha_{AB} \\
    + \frac{1}{2}(\eta + \underline\eta)^B\big[ -\frac{1}{2}\tr\chi\underline\alpha_{AB} + 4\omega\underline\alpha_{AB} \big] + \sl\nabla^B(4\omega - \frac{1}{2}\tr\chi)\underline\alpha_{AB} 
    - (4\omega - \tr\chi) \underline\eta^B \underline\alpha_{AB}.
\end{multline}
Grouping similar terms and simplifying, one obtains
\begin{multline*}
    I[\underline\alpha]_A = \underline\alpha_{AB}\Big(\sl\div\hat\chi_A - \frac{1}{2}\sl\nabla_A\tr\chi + \frac{1}{2}(\eta - \underline\eta)^B (\hat\chi - \frac{1}{2}\tr\chi \gamma)_{AB} + \beta[W]_A\Big) \\ 
    + \underline\alpha_{AB}\big[4\omega \underline\eta^B - 3\hat\chi^{BC}\underline\eta_C - \tr\chi\underline\eta^B] + \frac{3}{2} (\underline\alpha\cdot\hat\chi)\underline\eta_A  + 2\underline\alpha_{AB}\cdot\beta[W]^B \\ 
    - \frac{3}{2}(\underline\alpha\cdot\hat\chi)\underline\eta_A + 3\underline\alpha_{AB}\hat\chi^{BC}\underline\eta_C - (4\omega - \tr\chi)\underline\alpha\cdot \underline\eta.
\end{multline*}
Note that the parenthetical expression in the first line is zero due to the Codazzi equation \eqref{eq:null_structure_constraint}. Further simplification gives
\begin{equation}\label{eq:computation_43}
    I[\underline\alpha]_A = \underline\alpha_{AB}\cdot 2\beta[W]^B.
\end{equation}

For $I[\underline\beta]$:
\begin{multline}\label{eq:computation_40}\allowdisplaybreaks
    I[\underline\beta]_A = \underline\eta^B [(\zeta - 4\underline\eta)\hat\otimes \underline\beta]_{AB} + \Big[-\tr\chi\tr\underline\chi + 4\omega\tr\underline\chi + 4\rho[W] - 2\hat\chi\cdot\hat{\underline\chi} + 4\sl\div\underline\eta + 4|\underline\eta|^2 + 4\omega \underline\omega \\ 
    + \frac{3}{2}|\eta - \underline\eta|^2 - \frac{1}{2}(\eta - \underline\eta)\cdot(\eta + \underline\eta) - \frac{1}{4}|\eta + \underline\eta|^2 + \rho[W]\Big]\underline\beta_A + 2(\tr\underline\chi + \underline\omega)(2\omega - \tr\chi)\underline\beta_A \\ 
    - 2(\eta \cdot\underline\beta)\underline\eta_A + 2(\underline\eta\cdot\underline\beta)\eta_A - 2\sigma[W]\prescript{*}{}{\underline\beta}_A - 2\underline\omega (2\omega - \tr\chi)\underline\beta_A + \prescript{*}{}{\sl\nabla_A}(\frac{3}{2}\prescript{*}{}{\eta} + \frac{1}{2}\prescript{*}{}{\underline\eta})_B \underline\beta^B \\ 
    + \frac{1}{2}\prescript{*}{}{(\eta + \underline\eta)_A}[(\zeta - 2\eta) \cdot \prescript{*}{}{\underline\beta}] + \frac{1}{2}(\eta + \underline\eta)_A [(\frac{3}{2}\eta + \frac{1}{2}\underline\eta)\cdot\underline\beta] + \sl\nabla_A (\frac{3}{2}\eta + \frac{1}{2}\underline\eta)_B \underline\beta^B \\ 
    + \Big[ 4\omega\underline\omega + \frac{3}{2}|\eta - \underline\eta|^2 + \frac{1}{2}(\eta - \underline\eta)\cdot(\eta + \underline\eta) - \frac{1}{4}|\eta + \underline\eta|^2 + \rho[W] + \frac{1}{2}\tr\chi\tr\underline\chi - 2\underline\omega \tr\chi \\ 
    - 2\rho[W] + \hat\chi \cdot \hat{\underline\chi} - 2\sl\div\eta - 2|\eta|^2 \Big]\underline\beta_A + 3\prescript{*}{}{\underline\eta}_A [(\zeta - 2\eta)\cdot\prescript{*}{}{\underline\beta}] + 3\underline\eta_A [(\frac{3}{2}\eta + \frac{1}{2}\underline\eta)\cdot\underline\beta] \\ 
    + \frac{1}{2}(\eta + \underline\eta)^B[(\zeta - 4\underline\eta) \hat\otimes \underline\beta]_{AB} + \underline\beta^B \sl\nabla_B (\frac{1}{2}\eta - \frac{9}{2}\underline\eta)_A 
    + \underline\beta_A\sl\div(\frac{1}{2}\eta - \frac{9}{2}\underline\eta) \\ 
    - \sl\nabla_A (\frac{1}{2}\eta - \frac{9}{2}\underline
    \eta)_B \underline\beta^B - 2(\tr\underline\chi + \underline\omega)(4\omega - \tr\chi)\underline\beta_A + 4\hat{\underline\chi}_{AB}\hat\chi^{BC}\underline\beta_C - 2 K \underline\beta_A
\end{multline}
Grouping similar terms, we obtain
\begin{multline*}
    I[\underline\beta]_A = \underline\beta_A \big[ 8\omega \underline\omega - \frac{1}{2}\tr\chi\tr\underline\chi + 4\omega \tr\underline\chi - \hat\chi \cdot\hat{\underline\chi} - 2\tr\chi\tr\underline\chi + 4\rho[W] \big] \\ 
    + \prescript{*}{}{\underline\beta}_A[4\sl\curl\underline\eta - 2\sigma[W]] - 2(\tr\underline\chi + \underline\omega)(4\omega - \tr\chi)\underline\beta_A + 4\hat{\underline\chi}_{AB}\hat\chi^{BC}\underline\beta_C - 2 K \underline\beta_A.
\end{multline*}
Now, applying the Gauss equation and the equation for $\sl\curl\underline\eta$ (see \eqref{eq:null_structure_constraint}), the parenthetical in the first line is equal to 
\begin{align*}
    6\rho[W] - 2K + 8\omega \underline\omega - 2\hat\chi\cdot\hat{\underline\chi} + 4\omega \tr\underline\chi - 2\tr\chi\tr\underline\chi.
\end{align*}
Meanwhile the term $4\sl\curl\underline\eta - 2\sigma[W]$ in the second line is equal to 
\begin{align*}
    -6\sigma[W] - 2\hat{\underline\chi}\wedge \hat\chi.
\end{align*}
Note that by Proposition \ref{prop:proposition_3} and properties of the volume form, $\prescript{*}{}{\underline\beta}_A\hat{\underline\chi}\wedge \hat\chi = \hat{\underline\chi}_{AD}\hat\chi^{BD}\underline\beta_B - \hat\chi_{AD}\hat{\underline\chi}^{BD}\underline\beta_B$. Therefore we obtain in the end
\begin{equation}\label{eq:computation_45}
\begin{split}
    I[\underline\beta]_A &= 6(\rho[W]\underline\beta - \sigma[W]\prescript{*}{}{\underline\beta})_A + 2 K \underline\beta_A + \underline\beta_A (4\omega - \tr\chi)(2\underline\omega + 2\tr\underline\chi) - 4\hat{\underline\chi}_{AB}\hat\chi^{BC}\underline\beta_C\\ 
    &\qquad - 2(\tr\underline\chi + \underline\omega)(4\omega - \tr\chi)\underline\beta_A + 4\hat{\underline\chi}_{AB}\hat\chi^{BC}\underline\beta_C- 2 K \underline\beta_A\\ 
    &=  6(\rho[W]\underline\beta - \sigma[W]\prescript{*}{}{\underline\beta})_A.
\end{split}
\end{equation}

For $I[\rho,\sigma]_A$:
\begin{multline}\label{eq:computation_41}\allowdisplaybreaks
    I[\rho,\sigma] = \underline\eta^B[3\prescript{*}{}{\hat{\underline\chi}}_{AB}\sigma - 3\hat{\underline\chi}_{AB}\rho] + 2(\tr\underline\chi + \underline\omega)(-3\underline\eta_A \rho + 3\prescript{*}{}{\underline\eta}_A \sigma) + 6\underline\omega (\rho\underline\eta_A - \sigma\prescript{*}{}{\underline\eta}_A)  \\ 
    - \frac{3}{2}\prescript{*}{}{\sl\nabla}_A\tr\underline\chi \sigma - \frac{3}{4}\prescript{*}{}{(\eta + \underline\eta)}_A\tr\underline\chi \sigma + \frac{3}{2}\sl\nabla_A \tr\underline\chi \rho + \frac{3}{4}(\eta + \underline\eta)_A\tr\underline\chi \rho \\ 
    + 3\prescript{*}{}{(}\hat{\underline\chi} \cdot(\eta - \underline\eta) + \frac{1}{2}\tr\underline\chi(\eta - \underline\eta) + \underline\beta[W] )_A\sigma - \frac{9}{2}\prescript{*}{}{\underline\eta}_A\tr\underline\chi\sigma \\
    + 3(\hat{\underline\chi}\cdot(\underline\eta - \eta) + \frac{1}{2}\tr\underline\chi(\underline\eta - \eta) - \underline\beta[W] )_A\rho + \frac{9}{2}\underline\eta_A \tr\underline\chi \rho \\ 
    + \frac{1}{2}(\eta + \underline\eta)^B (3 \prescript{*}{}{\hat{\underline\chi}}_{AB}\sigma - 3\hat{\underline\chi}_{AB}\rho) + 3\sigma \sl\div(\prescript{*}{}{\hat{\underline\chi}})_A - 3\rho\sl\div\hat{\underline\chi}_A
    + 6\hat{\underline\chi}_{AB}(\rho\eta + \sigma\prescript{*}{}{\eta})^B.
\end{multline}
Grouping similar terms, we obtain
\begin{multline*}
    I[\rho,\sigma]_A = 3\rho \Big(\sl\div\hat\chi_A - \frac{1}{2}\sl\nabla_A\tr\chi + \frac{1}{2}(\eta - \underline\eta)^B (\hat\chi - \frac{1}{2}\tr\chi \gamma)_{AB} + \beta[W]_A - 2\underline\beta[W] - 2\hat{\underline\chi}_{AB}\eta^B\Big) \\ 
    - 3\sigma \Big(\prescript{*}{}{\Big[} \sl\div\hat\chi - \frac{1}{2}\sl\nabla\tr\chi + \frac{1}{2}(\eta - \underline\eta)\cdot(\hat\chi - \frac{1}{2}\tr\chi \gamma) + \beta[W]\Big]_A \\ 
    - 2\prescript{*}{}{\underline\beta}[W]_A - 2\prescript{*}{}{\hat{\underline\chi}}_{AB}\eta^B\Big) + 6\hat{\underline\chi}_{AB}(\rho\eta + \sigma\prescript{*}{}{\eta})^B.
\end{multline*}
Noting the appearance again of the Codazzi equation \eqref{eq:null_structure_constraint}, this simplifies to:
\begin{equation}\label{eq:computation_49}
\begin{split}
    I[\rho,\sigma]_A &= 6(\sigma\prescript{*}{}{\underline\beta}[W]_A - \rho\underline\beta[W]_A) - 6\hat{\underline\chi}_{AB}(\rho\eta^B + \sigma \prescript{*}{}{\eta}^B) + 6\hat{\underline\chi}_{AB}(\rho\eta + \sigma\prescript{*}{}{\eta})^B \\ 
    &=  6(\sigma\prescript{*}{}{\underline\beta}[W]_A - \rho\underline\beta[W]_A).
\end{split}
\end{equation}

Finally, for $I[\beta]$:
\begin{equation}\label{eq:computation_42}
\begin{split}
    I[\beta]_A &= 4(\tr\underline\chi + \underline\omega)\hat{\underline\chi}_{AB}\beta^B - 4\underline\omega \hat{\underline\chi}_{AB}\beta^B - [2\tr\underline\chi\hat{\underline\chi} + 4\underline\omega\hat{\underline\chi} + 2\underline\alpha[W]]_{AB}\beta^B \\ 
    &\quad + 2(2\underline\omega- \tr\underline\chi)\hat{\underline\chi}_{AB}\beta^B \\ 
    &= - 2\underline\alpha[W]_{AB}\beta^B.
\end{split}
\end{equation}
Combining \eqref{eq:computation_43}, \eqref{eq:computation_45}, \eqref{eq:computation_49}, and \eqref{eq:computation_42}, we have that the order 0 terms in $\sl\nabla_4 \underline B_A$ are equal to
\begin{multline}
    \underline\alpha_{AB} 2\beta[W]^B  
    + 6(\rho[W]\underline\beta - \sigma[W]\prescript{*}{}{\underline\beta})_A  
    + 6(\sigma\prescript{*}{}{\underline\beta}[W]_A - \rho\underline\beta[W]_A) 
     - 2\underline\alpha[W]_{AB}\beta^B \\ 
     = 2\big(\underline\alpha_{AB}\beta[W]^B - \underline\alpha[W]_{AB}\beta^B\big) + 6\big( \sigma\prescript{*}{}{\underline\beta}[W]_A - \sigma[W]\prescript{*}{}{\underline\beta}_A \big) + 6\big( \rho[W]\underline\beta_A - \rho\underline\beta[W]_A \big).
\end{multline}
This completes the computation for $\sl\nabla_4 \underline B_A$, showing that 
\begin{multline}
    \sl\nabla_4 \underline B_A = (4\omega - \tr\chi)\underline B_A + 2\hat{\underline\chi}_{AB}\Xi^B + 2\big(\underline\alpha_{AB}\beta[W]^B - \underline\alpha[W]_{AB}\beta^B\big) \\ 
    + 6(\sigma\prescript{*}{}{\underline\beta}[W]_A - \sigma[W]\prescript{*}{}{\underline\beta}_A) +  6(\rho[W]\underline\beta_A - \rho{\underline\beta[W]}_A).
\end{multline}

\subsection{Computations for \texorpdfstring{$\sl\nabla_4\Xi$}{nabla4Xi}} The only order 2 terms arise in the terms $\sl\nabla_4 \sl\nabla_3 \beta$, $\prescript{*}{}{\sl\nabla_4}\sl\nabla\sigma$, and $\sl\nabla_4\sl\nabla\rho$. These introduce derivatives of $P$ and $Q$. They are:
\begin{align*}
    \sl\nabla_4 \Xi &\simeq_2 \sl\div\sl\nabla\hat\otimes\beta + \prescript{*}{}{\sl\nabla}\sl\curl\beta - \sl\nabla\sl\div\beta - \sl\nabla P - \prescript{*}{}{\sl\nabla} Q \\ 
    &= 2 K \beta - \sl\nabla P - \prescript{*}{}{\sl\nabla} Q.
\end{align*}
Grouping terms by which unknown appears, the order $\geq 1$ terms are (before simplification): 
\begin{multline}
    \sl\nabla_4 \Xi_A \simeq_1 - \sl\nabla P - \prescript{*}{}{\sl\nabla} Q \\ 
    - \big(\frac{7}{2}\eta + \frac{1}{2}\underline\eta\big)_A P - \prescript{*}{}{\big(} \frac{7}{2}\eta + \frac{1}{2}\underline\eta \big)_A Q - 2\tr\chi\Xi_A   \\
    + (\tr\underline\chi - 2\underline\omega)\sl\div\alpha_A - 2\underline\omega\sl\div\alpha_A  - \hat{\chi}^{BC}\sl\nabla_B \alpha_{AC} - \frac{1}{2}\tr\underline\chi \sl\div\alpha_A + (4\underline\omega - \frac{1}{2}\tr\underline\chi)\sl\div\alpha_A \\ 
    + \frac{1}{2}\prescript{*}{}{\hat{\underline\chi}}^{BC}\prescript{*}{}{\sl\nabla}_A\alpha_{BC} + \frac{1}{2}\hat{\underline\chi}^{BC}\sl\nabla_A \alpha_{BC} \\ 
    + 2(\underline\eta - \eta)^B\sl\nabla_B \beta_A + \eta^B (\sl\nabla\hat\otimes\beta)_{AB} - 3\eta_A \sl\div\beta + 3\prescript{*}{}{\eta}_A\sl\curl\beta + \frac{1}{2}(\eta + \underline\eta)^B (\sl\nabla\hat\otimes\beta)_{AB} \\ 
    + (4\eta + \zeta)_A \sl\div\beta + (4\eta + \zeta)^B \sl\nabla_B \beta_A - (4\eta + \zeta)^B \sl\nabla_A \beta_B + \frac{1}{2}\prescript{*}{}{(\eta + \underline\eta)}_A \sl\curl\beta \\ 
    + (2\underline\eta + \zeta)^B \prescript{*}{}{\sl\nabla}_A \prescript{*}{}{\beta}_B - \frac{1}{2}(\eta + \underline\eta)_A \sl\div\beta - (2\underline\eta + \zeta)^B \sl\nabla_A \beta_B \\ 
    - 2\hat\chi_{AB}(\prescript{*}{}{\sl\nabla}^B \sigma - \sl\nabla^B \rho) - 3\hat\chi_{AB}\sl\nabla^B \rho - 3\prescript{*}{}{\hat\chi}_{AB}\sl\nabla^B \sigma + \prescript{*}{}{\hat\chi}_{AB}\sl\nabla^B \sigma + \frac{1}{2}\tr\chi\prescript{*}{}{\sl\nabla}_A \sigma \\ 
    + \hat\chi_{AB}\sl\nabla^B \rho + \frac{1}{2}\tr\chi\sl\nabla_A \rho + \frac{3}{2}\tr\chi\sl\nabla_A \rho + \frac{3}{2}\tr\chi\prescript{*}{}{\sl\nabla}_A \sigma + 2\tr\chi(\sl\nabla \rho + \prescript{*}{}{\sl\nabla}\sigma)_A. 
\end{multline}
The terms involving $\alpha$ immediately simplify to 
\begin{equation*}
    - \hat{\chi}^{BC}\sl\nabla_B \alpha_{AC} + \frac{1}{2}\prescript{*}{}{\hat{\underline\chi}}^{BC}\prescript{*}{}{\sl\nabla}_A\alpha_{BC} + \frac{1}{2}\hat{\underline\chi}^{BC}\sl\nabla_A \alpha_{BC}, 
\end{equation*}
which after applying Proposition \ref{prop:proposition_5} is seen to be zero. 

For the terms involving $\beta, \rho$, and $\sigma$, one simplifies using the properties of the Hodge dual and the volume form to obtain that these terms vanish. 
This implies that 
\begin{equation}
    \sl\nabla_4 \Xi_A \simeq_1 -\sl\nabla_A P - \prescript{*}{}{\sl\nabla}_A Q - \big(\frac{7}{2}\eta + \frac{1}{2}\underline\eta\big)_A P - \prescript{*}{}{\big(} \frac{7}{2}\eta + \frac{1}{2}\underline\eta \big)_A Q - 2\tr\chi\Xi_A.
\end{equation}
We write, in a similar manner as in Section \ref{subsec:Bbar}, the order 0 terms in $\sl\nabla_4 \Xi$ in the form 
\begin{align*}
    I[\alpha] + I[\beta] + I[\rho,\sigma]  + I[\underline\beta].
\end{align*}
We now discuss the computation of each of these terms. For $I[\alpha]$: 
\begin{multline}
    I[\alpha]_A = (\tr\underline\chi - 2\underline\omega) \eta^B\alpha_{AB} + \frac{3}{2}\eta_A (\hat{\underline\chi}\cdot\alpha) - \frac{3}{2}\prescript{*}{}{\eta}_A (\hat{\underline\chi}\cdot\prescript{*}{}{\alpha}) - 2\underline\omega (\underline\eta + 2\zeta)^{B}\alpha_{AB} \\ 
    - \Big[ \underline\beta[W]_B + 4\sl\nabla_B \omega + 2\omega (\eta + \underline\eta)_B + 2\eta^C\hat{\chi}_{BC} + \tr\underline\chi \eta_B \Big] \alpha_{AB} + \eta^B\alpha_{AB}(4\underline\omega - \frac{1}{2}\tr\underline\chi) \\ 
    + (\hat{\underline\chi}_{AB}\eta_C - \hat{\underline\chi}_{BC}\eta_A)\alpha^{BC} \alpha^{BC} - \hat{\underline\chi}^{BC}\eta_B \alpha_{AC} + \tr\underline\chi \eta^B \alpha_{AB} - 2\underline\beta[W]^B \alpha_{AB} \\ 
    + \frac{1}{2}(\eta + \underline\eta)^B \alpha_{AB}(4\underline\omega - \frac{1}{2}\tr\underline\chi) + \alpha_{AB}\sl\nabla^B(4\underline\omega - \frac{1}{2}\tr\underline\chi) - \frac{1}{2}\prescript{*}{}{\sl\nabla}_A \hat{\underline\chi}_{BC}\prescript{*}{}{\alpha}^{BC} \\ 
    - \frac{1}{4}\prescript{*}{}{(\eta + \underline\eta)}_A(\hat{\underline\chi}\cdot\prescript{*}{}{\alpha}) + \frac{1}{2}\sl\nabla_A \hat{\underline\chi}_{BC}\alpha^{BC} + \frac{1}{4}(\eta + \underline\eta)_A(\hat{\underline\chi}\cdot\alpha).
\end{multline}
Using again the properties of the Hodge dual and the volume form, and also Proposition \ref{prop:proposition_4}, we obtain 
\begin{equation}\label{eq:computation_51}
\begin{split}
    I[\alpha]_A &= \tensor{\alpha}{_A ^B}\Big(\sl\div\hat\chi_B - \frac{1}{2}\sl\nabla_B\tr\chi + \frac{1}{2}(\eta - \underline\eta)^C (\hat\chi - \frac{1}{2}\tr\chi \gamma)_{BC} + \beta[W]_B\Big) - 2\alpha_{AB}\underline\beta[W]^B \\ 
    &= - 2\alpha_{AB}\underline\beta[W]^B.
\end{split}
\end{equation}

For $I[\beta]$:
\begin{multline}
    I[\beta]_A = -2(\tr\underline\chi - 2\underline\omega)(\tr\chi + \omega) \beta_A + \Big[ -\frac{1}{2}\tr\chi\tr\underline\chi + 2\omega \tr\underline\chi + 2\rho[W] - \hat\chi \cdot\hat{\underline\chi} \\
    + 2\sl\div\underline\eta + 2|\underline\eta|^2 - 4\omega \underline\omega - |\eta|^2 + 2\eta \cdot \underline\eta - \rho[W] \Big]\beta_A - 4\hat{\chi}_{AB}\hat{\underline\chi}^{BC}\beta_C  \\ 
    - 3\eta_A(2\underline\eta \cdot \beta + \zeta \cdot\beta) + 3\prescript{*}{}{\eta}_A [(2\underline\eta + \zeta)\cdot\prescript{*}{}{\beta}] - 2(\eta \cdot\beta)\underline\eta_A + 2(\underline\eta\cdot\beta)\underline\eta_A \\
    - 2\sigma[W]\prescript{*}{}{\beta}_A 
    + 2\underline\omega[2(\tr\chi + \omega)\beta]_A + \eta^B[(\zeta + 4\eta)\hat\otimes \beta]_{AB} \\ 
    + \frac{1}{2}(\eta + \underline\eta)^B[(\zeta + 4\eta)\hat\otimes\beta]_{AB} + \sl\nabla^B(\frac{9}{2}\eta - \frac{1}{2}\underline\eta)_{A} \beta_B + (\frac{9}{2}\sl\div\eta - \frac{1}{2}\sl\div\underline\eta)\beta_A  \\ 
    + (\frac{1}{2}\sl\nabla_A \underline\eta_B - \frac{9}{2}\sl\nabla_A \eta_B)\beta^B - \frac{1}{2}(\eta + \underline\eta)_A[(2\underline\eta + \zeta)\cdot\beta] - \sl\nabla_A (\frac{1}{2}\eta + \frac{3}{2}\underline\eta)_B \beta^B  \\ 
    + \prescript{*}{}{\sl\nabla}_A(\frac{1}{2}\eta + \frac{3}{2}\underline\eta)_B \prescript{*}{}{\beta}^B + \frac{1}{2} \prescript{*}{}{(\eta + \underline\eta)}_A [(2\underline\eta + \zeta)\cdot \prescript{*}{}{\beta}] \\ 
    -2\Big[ -\frac{1}{2}\tr\chi\tr\underline\chi + 2\underline\omega \tr\chi + 2\rho[W] - \hat\chi\cdot\hat{\underline\chi} + 2\sl\div\eta + 2|\eta|^2 + 2\omega\underline\omega + \frac{1}{2}|\underline\eta|^2 \\
    - \eta \cdot\underline\eta + \frac{1}{2}\rho[W]\Big]\beta_A
    -2\tr\chi(2\underline\omega - \tr\underline\chi)\beta_A + 2 K\beta_A.
\end{multline}
Using properties of the Hodge dual as well as the null structure equations \eqref{eq:null_structure_constraint}, this simplifies to:
\begin{align*}
    I[\beta]_A &= -2\tr\chi(\tr\underline\chi - 2\underline\omega)\beta_A - 6\rho[W]\beta_A - 6\sigma[W]\prescript{*}{}{\beta}_A - 2 K \beta_A \\ 
    &\qquad -2\tr\chi(2\underline\omega - \tr\underline\chi)\beta_A + 2 K\beta_A \\ 
    &= - 6\rho[W]\beta_A - 6\sigma[W]\prescript{*}{}{\beta}_A.
\end{align*}

For $I[\rho,\sigma]$: 
\begin{multline}
    I[\rho,\sigma]_A = -6\hat\chi_{AB}(\sigma\prescript{*}{}{\underline\eta}^B - \rho\underline\eta^B) + 3\rho(\hat\chi_{AB}(\eta - \underline\eta)^B + \frac{1}{2}\tr\chi (\eta - \underline\eta)_A + \beta[W]_A ) \\ 
    + \frac{9}{2}\eta_A  \tr\chi\rho + \frac{9}{2}\prescript{*}{}{\eta}_A \tr\chi \sigma + 3\sigma(\prescript{*}{}{\hat\chi}_{AB}(\eta - \underline\eta)^B + \frac{1}{2}\tr\chi\prescript{*}{}{(\eta - \underline\eta)}_A + \prescript{*}{}{\beta}[W]_A ) \\ 
    - \eta^B (3\prescript{*}{}{\hat\chi}_{AB} \sigma + 3\hat\chi_{AB}\rho) - \frac{1}{2}(\eta + \underline\eta)^B (3\hat\chi_{AB}\rho + 3\prescript{*}{}{\hat\chi}_{AB}\sigma) - 3\sl\div\hat\chi \rho - 3\sl\div(\prescript{*}{}{\hat\chi})\sigma \\ 
    + \frac{3}{4}(\eta + \underline\eta)_A\tr\chi\rho + \frac{3}{2}\sl\nabla_A \tr\chi \rho + \frac{3}{4}\prescript{*}{}{(\eta + \underline\eta)}_A \tr\chi \sigma + \frac{3}{2}\prescript{*}{}{\sl\nabla}_A \tr\chi \sigma
    -6\tr\chi(\eta\rho + \sigma\prescript{*}{}{\eta}).
\end{multline}
Using again properties of the Hodge dual, we obtain 
\begin{align*}
    I[\rho,\sigma]_A &= -3\rho\Big(\sl\div\hat\chi_A - \frac{1}{2}\sl\nabla_A\tr\chi + \frac{1}{2}(\eta - \underline\eta)^B (\hat\chi - \frac{1}{2}\tr\chi \gamma)_{AB} + \beta[W]_A\Big) \\ 
    &\quad - 3\sigma\prescript{*}{}{\Big(}\sl\div\hat\chi - \frac{1}{2}\sl\nabla\tr\chi + \frac{1}{2}(\eta - \underline\eta)\cdot(\hat\chi - \frac{1}{2}\tr\chi \gamma) + \beta[W]\Big)_A \\ 
    &\quad + 2\tr\chi (3\eta \rho + 3\prescript{*}{}{\underline\eta}\sigma)_A + 6(\rho\beta[W] + \sigma\prescript{*}{}{\beta}[W])_A -6\tr\chi(\eta\rho + \sigma\prescript{*}{}{\eta}),
\end{align*}
which by the Codazzi equations \eqref{eq:null_structure_constraint} is equal to 
\begin{equation}
    I[\rho,\sigma]_A = 6(\rho\beta[W] + \sigma\prescript{*}{}{\beta}[W])_A.
\end{equation}

For $I[\underline\beta]$, the computation is much shorter:
\begin{align*}
    I[\underline\beta]_A &= -2\hat\chi_{AB}\underline\beta^B (2\omega - \tr\chi) + \underline\beta^B(2\alpha[W]_{AB} + 4\omega\hat\chi_{AB} + 2\tr\chi\hat\chi_{AB}) -4\tr\chi\hat\chi\cdot\underline\beta \\ 
    &= -2\tr\chi(-2\hat\chi_{AB}\underline\beta^B) + 2\alpha[W]_{AB}\underline\beta^B -4\tr\chi\hat\chi\cdot\underline\beta \\ 
    &= 2\alpha[W]_{AB}\underline\beta^B.
\end{align*}
Thus we see that the order 0 terms in $\sl\nabla_4 \Xi_A$ are equal to:
\begin{equation}
    - 2\alpha_{AB}\underline\beta[W]^B - 6\rho[W]\beta_A - 6\sigma[W]\prescript{*}{}{\beta}_A 
    + 6(\rho\beta[W] + \sigma\prescript{*}{}{\beta}[W])_A
    + 2\alpha[W]_{AB}\underline\beta^B.
\end{equation}
This completes the computation for $\sl\nabla_4 \Xi$, showing that 
\begin{multline}
    \sl\nabla_4 \Xi_A = - \sl\nabla_A P - \prescript{*}{}{\sl\nabla}_A Q - \big(\frac{7}{2}\eta + \frac{1}{2}\underline\eta\big)_A P - \prescript{*}{}{\big(} \frac{7}{2}\eta + \frac{1}{2}\underline\eta \big)_A Q - 2\tr\chi \Xi_A \\ 
    + 2(\alpha[W]_{AB}\underline\beta^B - \alpha_{AB}\underline\beta[W]^B) + 6(\rho\beta[W]_A - \rho[W]\beta_A) \\ 
    + 6(\sigma\prescript{*}{}{\beta}[W]_A - \sigma[W]\prescript{*}{}{\beta}_A).
\end{multline}

\subsection{Computations for \texorpdfstring{$\sl\nabla_3 P$}{nabla3 P}}
The only order 2 terms arise from $\sl\nabla_3 \sl\nabla_4 \rho$ and $\sl\nabla_3 \sl\div \beta$. These introduce divergences of $\Xi$, $\sl\nabla\rho$, and $\prescript{*}{}{\sl\nabla}\sigma$; note that $\sl\div\sl\nabla\rho = \sl\Delta \rho$ and $\sl\div\prescript{*}{}{\sl\nabla}\sigma = 0$. We therefore have:
\begin{align*}
    \sl\nabla_3 P &\simeq_2 
    -(\sl\div\sl\nabla_4 \underline\beta + \sl\div\sl\nabla_3 \beta)\\
    &\simeq_2 -(\sl\div[-\sl\nabla\rho + \prescript{*}{}{\sl\nabla\sigma}] + \sl\div[\sl\nabla\rho + \prescript{*}{}{\sl\nabla\sigma} + \Xi]) \\
    &= \sl\Delta \rho - \sl\Delta \rho - \sl\div\Xi \\
    &= -\sl\div\Xi.
\end{align*}
Next, grouping terms by which unknown appears, the order $\geq 1$ terms are (before simplification):
\begin{multline}
    \sl\nabla_3 P \simeq_1 -\sl\div\Xi - \frac{1}{2}(\eta + \underline\eta)\cdot\Xi - (2\underline\eta + \zeta)\cdot\Xi  + 2\underline\omega P - \frac{3}{2}\tr\underline\chi P + 2\underline\omega \sl\div\beta - \frac{3}{2}\tr\underline\chi\sl\div\beta \\ 
    + (\tr\underline\chi - 2\underline\omega)\sl\div\beta 
    + \frac{1}{2}\hat{\underline\chi}\cdot (\sl\nabla\hat\otimes\beta)  + \hat{\underline\chi}\cdot\sl\nabla\beta + \frac{1}{2}\tr\underline\chi\sl\div\beta  - 2\hat{\underline\chi}\cdot\sl\nabla\beta \\ - \frac{3}{2}\tr\chi \sl\div\underline\beta  + 2\omega\sl\div\underline\beta + \frac{1}{2}\hat\chi\cdot(\sl\nabla\hat\otimes\underline\beta)
    + \hat\chi\cdot\sl\nabla\underline\beta + \frac{1}{2}\tr\chi\sl\div\underline\beta \\ - (2\omega - \tr\chi)\sl\div\underline\beta - 2\hat\chi \cdot\sl\nabla\underline\beta \\ 
    - (2\underline\eta + \zeta)\cdot(\sl\nabla\rho + \prescript{*}{}{\sl\nabla}\sigma) + 2(\eta - \underline\eta)\cdot\sl\nabla\rho + (\zeta - 2\eta)\cdot (\prescript{*}{}{\sl\nabla}\sigma - \sl\nabla\rho)\\ 
    - \frac{1}{2}(\eta + \underline\eta) \cdot (\prescript{*}{}{\sl\nabla}\sigma - \sl\nabla\rho) - 3\prescript{*}{}{\underline\eta}\cdot\sl\nabla\sigma + 3\underline\eta\cdot\sl\nabla\rho - 3\eta \cdot\sl\nabla\rho  \\ 
    - 3\prescript{*}{}{\eta}\cdot\sl\nabla\sigma - \frac{1}{2}(\eta + \underline\eta)\cdot(\prescript{*}{}{\sl\nabla\sigma} + \sl\nabla\rho).
\end{multline}
Note that since $\hat{\chi}$ is traceless, $\hat\chi\cdot \sl\nabla\underline\beta = \frac{1}{2}\hat\chi\cdot(\sl\nabla\hat\otimes\underline\beta)$ (similarly for $\hat{\underline\chi}\cdot\sl\nabla\beta$). By using this identity and combining similar terms, one checks that all the terms involving $\beta$ and $\underline\beta$ cancel. Similarly, using properties of the Hodge dual, one checks that all terms involving $\rho$ and $\sigma$ cancel, and we are left with 
\begin{equation}
    \sl\nabla_3 P \simeq_1 -\sl\div\Xi - (\eta + 2\underline\eta) \cdot \Xi + \big(2\underline\omega - \frac{3}{2}\tr\underline\chi \big) P.
\end{equation}
We write the order 0 terms in $\sl\nabla_3 P$ in the form 
\begin{align*}
    I[\underline\alpha] + I[\underline\beta] + I[\rho,\sigma] + I[\beta] + I[\alpha]
\end{align*}
(note that all unknowns appear in this expression). We now discuss the computation of each of these terms, each of which is much shorter than in $\sl\nabla_4 \underline B$ or $\sl\nabla_4 \Xi$. 

For $I[\underline\alpha]$:
\begin{equation}
\begin{split}
    I[\underline\alpha] &= -\frac{3}{4}\tr\chi \hat\chi\cdot\underline\alpha + \frac{1}{2}((\tr\chi + 2\omega)\hat\chi + \alpha[W])\cdot\underline\alpha  - \frac{1}{2}\hat\chi \cdot[(4\omega - \frac{1}{2}\tr\chi)\underline\alpha] + \omega\hat\chi \cdot\underline\alpha \\ 
    &= \frac{1}{2}\alpha[W]\cdot\underline\alpha.
\end{split}
\end{equation}

For $I[\underline\beta]$: 
\begin{multline}
    I[\underline\beta] = \frac{3}{2}\tr\chi (\zeta - 2\eta) \cdot\underline\beta - (\eta + 3\underline\eta)_A \hat{\chi}^{AB} \underline\beta_B - \frac{1}{2}\tr\chi \underline\beta \cdot \underline\eta + \underline\eta_A \hat\chi^{AB}\underline\beta_B - \underline\beta\cdot\beta[W] \\ 
    - \frac{1}{2}\hat\chi \cdot[(\zeta - 4\underline\eta)\hat\otimes \underline\beta] - (2\omega - \tr\chi)(\frac{3}{2}\eta + \frac{1}{2}\underline\eta)\cdot\underline\beta + (\sl\nabla\tr\chi - 2\sl\nabla\omega)\cdot\underline\beta \\ 
    - \frac{1}{2}\big[ \beta[W] - 4\sl\nabla\omega - 2\omega (\eta + \underline\eta) - 2\underline\eta\cdot\hat \chi - \tr\chi\underline\eta \big]\cdot\underline\beta + \frac{3}{2}\big[ \hat\chi\cdot(\eta - \underline\eta) \\ 
    + \frac{1}{2}\tr\chi(\eta - \underline\eta)  
    + \beta[W] \big]\cdot\underline\beta - \frac{1}{2}(2\omega - \tr\chi)(\eta + \underline\eta)\cdot\underline\beta - 2\omega (\zeta - 2\eta) \cdot\underline\beta\\ 
    - 2(\sl\div\hat\chi)\cdot\underline\beta - (\eta + \underline\eta)_A \hat\chi^{AB}\underline\beta_B.
\end{multline}
By grouping similar terms, we simplify this to
\begin{align*}
    I[\underline\beta] &= -2 \underline\beta \cdot \Big(\sl\div\hat\chi_A - \frac{1}{2}\sl\nabla_A\tr\chi + \frac{1}{2}(\eta - \underline\eta)^B (\hat\chi - \frac{1}{2}\tr\chi \gamma)_{AB}\Big) \\
    &=   -2 \underline\beta \cdot \Big(\sl\div\hat\chi_A - \frac{1}{2}\sl\nabla_A\tr\chi + \frac{1}{2}(\eta - \underline\eta)^B (\hat\chi - \frac{1}{2}\tr\chi \gamma)_{AB} + \beta[W]_A - \beta[W]_A \Big) \\ 
    &= 2\underline\beta\cdot\beta[W].
\end{align*}

For $I[\rho,\sigma]$:
\begin{multline}
    I[\rho,\sigma] = \frac{3}{2}\rho \Big[ -\frac{1}{2}\tr\chi\tr\underline\chi + 2\underline\omega \tr\chi + 2\rho[W] - \hat\chi \cdot\hat{\underline\chi} + 2\sl\div\eta + 2|\eta|^2 \Big] - \frac{9}{4}\tr\chi\tr\underline\chi \rho \\ 
    - 3(\frac{1}{2}\eta + \frac{3}{2}\underline\eta)\cdot (\eta\rho + \prescript{*}{}{\eta} \sigma) - \frac{1}{2}\hat{\underline\chi}\cdot (3\hat\chi\rho + 3\prescript{*}{}{\hat\chi}\sigma)  - \frac{1}{2}\hat\chi \cdot (3\prescript{*}{}{\hat{\underline\chi}}\sigma - 3\hat{\underline\chi}\rho ) \\ 
    + \frac{9}{4}\tr\underline\chi\tr\chi \rho - \frac{3}{2}\rho\big[ -\frac{1}{2}\tr\chi\tr\underline\chi + 2\omega \tr\underline\chi + 2\rho[W] - \hat\chi\cdot\hat{\underline\chi} + 2\sl\div\underline\eta + 2|\underline\eta|^2 \big] \\ 
    - (\frac{3}{2}\eta + \frac{1}{2}\underline\eta)\cdot (-3\underline\eta\rho + 3\prescript{*}{}{\underline\eta}\sigma) + 3\rho\sl\div\underline\eta - 3\sigma \sl\div(\prescript{*}{}{\underline\eta}) - \frac{1}{2}(\eta + \underline\eta)\cdot (-3\underline\eta \rho + 3\prescript{*}{}{\underline\eta}\sigma) \\ 
    + 3\omega \tr\underline\chi\rho - 3\underline\omega \tr\chi\rho - 3\rho\sl\div\eta - 3\sigma\sl\div(\prescript{*}{}{\eta}) - \frac{3}{2}(\eta + \underline\eta)\cdot(\eta\rho + \prescript{*}{}{\eta}\sigma).
\end{multline}
By grouping similar terms together, one verifies directly that all terms cancel with each other and so 
\begin{equation}
    I[\rho,\sigma] = 0.
\end{equation}

For $I[\beta]$:
\begin{multline}
    I[\beta] = \frac{1}{2}\beta \cdot\big[ \underline\beta[W] + 4\sl\nabla\underline\omega + 2\underline\omega (\eta + \underline\eta) + 2\eta \cdot\hat{\underline\chi} + \tr\underline\chi\eta  \big] - \frac{3}{2}\beta\cdot\big[ \hat{\underline\chi}\cdot(\eta - \underline\eta) \\ 
    + \frac{1}{2}\tr\underline\chi(\eta - \underline\eta) + \underline\beta[W] \big] - (\frac{1}{2}\eta + \frac{3}{2}\underline\eta)\cdot \beta (2\underline\omega - \tr\underline\chi) + \frac{1}{2}\hat{\underline\chi}\cdot \big[ (4\eta + \zeta) \hat\otimes \beta \big] \\ 
    - ({3}\eta + \underline\eta)\cdot (\hat{\underline\chi}\cdot\beta) - \frac{3}{2}\tr\underline\chi(2\underline\eta + \zeta)\cdot \beta - (\eta + \underline\eta)\cdot (\hat{\underline\chi}\cdot\beta) - 2\beta\cdot\sl\div\hat{\underline\chi} \\ 
    + 2\underline\omega \beta \cdot (2\underline\eta + \zeta) - \frac{1}{2}\tr\underline\chi \beta \cdot \eta + \eta \cdot (\hat{\underline\chi}\cdot\beta) + \beta \cdot \underline\beta[W] \\ 
    + \beta \cdot \sl\nabla (\tr\underline\chi - 2\underline\omega) - \frac{1}{2}(2\underline\omega + \tr\underline\chi)(\eta + \underline\eta) \cdot \beta.
\end{multline}
Grouping similar terms and then applying the Codazzi equation \eqref{eq:null_structure_constraint}, one simplifies this to:
\begin{align*}
    I[\beta] &= -2\beta \cdot \Big(\sl\div\hat{\underline\chi} - \frac{1}{2}\sl\nabla \tr\underline\chi - \frac{1}{2}(\eta - \underline\eta) \cdot \big( \hat{\underline\chi} - \frac{1}{2}\tr\underline\chi\gamma \big) \Big) \\ 
    &= -2\beta \cdot \Big(\sl\div\hat{\underline\chi} - \frac{1}{2}\sl\nabla \tr\underline\chi - \frac{1}{2}(\eta - \underline\eta) \cdot \big( \hat{\underline\chi} - \frac{1}{2}\tr\underline\chi\gamma \big) - \underline\beta[W] + \underline\beta[W] \Big) \\ 
    &= -2\beta \cdot\underline\beta[W].
\end{align*}

For $I[\alpha]$:
\begin{equation}
\begin{split}
    I[\alpha] &= \frac{1}{2}(4\underline\omega - \frac{1}{2}\tr\underline\chi)(\hat{\underline\chi}\cdot\alpha) + \frac{1}{2}\big( -(\tr\underline\chi + 2\underline\omega) \hat{\underline\chi} - \underline\alpha[W] \big) \cdot\alpha + \frac{3}{4}\tr\underline\chi (\hat{\underline\chi}\cdot\alpha) - \underline\omega \hat{\underline\chi}\cdot\alpha \\ 
    &= -\frac{1}{2}\underline\alpha[W]\cdot\alpha.
\end{split}
\end{equation}
Therefore we see that the order 0 terms in $\sl\nabla_3 P$ are equal to: 
\begin{multline}
    \frac{1}{2}\alpha[W]\cdot\underline\alpha + 2\underline\beta \cdot\beta[W] - 2\beta\cdot\underline\beta[W] - \frac{1}{2}\underline\alpha[W] \cdot\alpha \\ 
    = \frac{1}{2}\big( \alpha[W]\cdot\underline\alpha - \underline\alpha[W]\cdot\alpha \big) + 2\big( \underline\beta \cdot\beta[W] - \beta \cdot \underline\beta[W] \big).
\end{multline}
This completes the computation for $\sl\nabla_3 P$, showing that 
\begin{multline}
    \sl\nabla_3 P = -\sl\div\Xi - (\eta + 2\underline\eta)\cdot\Xi + \big(2\underline\omega - \frac{3}{2}\tr\underline\chi\big)P \\ 
    + \frac{1}{2}\big( \alpha[W]\cdot\underline\alpha - \underline\alpha[W]\cdot\alpha \big) + 2\big( \underline\beta \cdot\beta[W] - \beta \cdot \underline\beta[W] \big).
\end{multline}

\subsection{Computations for \texorpdfstring{$\sl\nabla_3 Q$}{nabla3 Q}} The only order 2 terms arise in the terms $\sl\nabla_3 \sl\nabla_4 \sigma$ and $\sl\nabla_3 \sl\curl\beta$. These introduce curls of $\Xi$, $\sl\nabla\rho$, and $\prescript{*}{}{\sl\nabla\sigma}$. Note that $\sl\curl\sl\nabla\rho = 0$ and $\sl\curl\prescript{*}{}{\sl\nabla\sigma} = -\sl\Delta\sigma$. We therefore have:
\begin{align*}
    \sl\nabla_3 Q &\simeq_2 -\sl\curl\sl\nabla_4\underline\beta + \sl\curl\sl\nabla_3 \beta \\ 
    &\simeq_2 \sl\Delta\sigma  - \sl\Delta \sigma + \sl\curl\Xi = \sl\curl\Xi.
\end{align*}
Next, grouping terms by which unknown appears, the order $\geq 1$ terms are (before simplification):
\begin{multline}
    \sl\nabla_3 Q \simeq_1 \sl\curl\Xi - \prescript{*}{}{(\eta + 2\underline\eta)}\cdot\Xi + (2\underline\omega - \frac{3}{2}\tr\underline\chi)Q  \\ 
    + \frac{1}{2}\prescript{*}{}{\hat{\underline\chi}}\cdot\sl\nabla\hat\otimes\beta - 2\underline\omega\sl\curl\beta + \frac{3}{2}\tr\underline\chi \sl\curl\beta - 2\prescript{*}{}{\hat{\underline\chi}}\cdot \sl\nabla\beta - \sl\epsilon^{AB}\hat{\underline\chi}_{AC}\sl\nabla^C \beta_B - \frac{1}{2}\tr\underline\chi \sl\curl\beta \\ 
    + (2\underline\omega - \tr\underline\chi)\sl\curl\beta + (\tr\chi - 2\omega)\sl\curl\underline\beta + 2\omega \sl\curl\underline\beta + \frac{1}{2}\tr\chi\sl\curl\underline\beta - \frac{3}{2}\tr\chi\sl\curl\underline\beta \\ 
    - \frac{1}{2}\prescript{*}{}{\hat\chi}\cdot\sl\nabla\hat\otimes\underline\beta + 2\prescript{*}{}{\hat\chi}\cdot\sl\nabla\underline\beta + \sl\epsilon^{AB}\hat\chi_{AC}\sl\nabla^C \underline\beta_B \\ 
    - \prescript{*}{}{(2\underline\eta + \zeta)}\cdot\sl\nabla\rho - \frac{1}{2}\prescript{*}{}{(\eta + \underline\eta)}\cdot\sl\nabla\rho - \prescript{*}{}{(2\eta - \zeta)}\cdot\sl\nabla\rho + 3\prescript{*}{}{\underline\eta}\cdot\sl\nabla\rho \\ 
    + 3\prescript{*}{}{\eta}\cdot \sl\nabla\rho - \frac{1}{2}\prescript{*}{}{(\eta + \underline\eta)}\cdot \sl\nabla\rho - \prescript{*}{}{(2\underline\eta + \zeta)}\cdot\prescript{*}{}{\sl\nabla\sigma} + 2(\eta - \underline\eta)\cdot\sl\nabla\sigma + \prescript{*}{}{(2\eta - \zeta)}\cdot\prescript{*}{}{\sl\nabla\sigma} \\ 
    + 3\prescript{*}{}{\underline\eta}\cdot\prescript{*}{}{\sl\nabla\sigma} + \frac{1}{2}\prescript{*}{}{(\eta + \underline\eta)}\cdot\prescript{*}{}{\sl\nabla\sigma} - \frac{1}{2}\prescript{*}{}{(\eta + \underline\eta)}\cdot\prescript{*}{}{\sl\nabla\sigma} - 3\prescript{*}{}{\eta}\cdot \prescript{*}{}{\sl\nabla\sigma}.
\end{multline}
The terms involving $\sl\curl\beta$ are seen to immediately cancel with each other, as are the terms involving $\sl\curl\underline\beta$. The remaining terms involving $\beta$ cancel with each other after using the fact that $\prescript{*}{}{\hat{\underline\chi}}\cdot\sl\nabla\hat\otimes \beta = 2\prescript{*}{}{\hat{\underline\chi}}\cdot\sl\nabla\beta$, as well as the definition of the Hodge dual to deal with the terms involving $\sl\epsilon$. Similarly for the remaining terms involving $\underline\beta$. By grouping similar terms, the terms involving $\rho$ cancel completely, as do those involving $\sigma$, and so we obtain
\begin{equation}
    \sl\nabla_3 Q \simeq_1 \sl\curl\Xi - \prescript{*}{}{(\eta + 2\underline\eta)}\cdot\Xi + (2\underline\omega - \frac{3}{2}\tr\underline\chi)Q.
\end{equation}
We write the order 0 terms in $\sl\nabla_3 Q$ in the form 
\begin{align*}
    I[\underline\alpha] + I[\underline\beta] + I[\rho, \sigma] + I[\beta] + I[\alpha]
\end{align*}
(note that again, all unknowns appear in this expression). We now discuss the computation of each of these terms, each of which is much shorter than in $\sl\nabla_4 \underline B$ or $\sl\nabla_4 \Xi$. 

For $I[\underline\alpha]$:
\begin{align*}
    I[\underline\alpha] &= \frac{3}{4}\tr\chi \prescript{*}{}{\hat\chi}\cdot\underline\alpha - \omega \prescript{*}{}{\hat\chi}\cdot\underline\alpha - \frac{1}{2}\big( \tr\chi \prescript{*}{}{\hat\chi} + 2\omega \prescript{*}{}{\hat\chi} + \prescript{*}{}{\alpha}[W] \big) \cdot\underline\alpha + \frac{1}{2}\prescript{*}{}{\hat\chi}\cdot\underline\alpha (4\omega - \frac{1}{2}\tr\chi) \\ 
    &= - \frac{1}{2}\prescript{*}{}{\alpha}[W]\cdot\underline\alpha.
\end{align*}
For $I[\underline\beta]$:
\begin{multline}
    I[\underline\beta] = \frac{3}{2}\tr\chi \underline\beta\cdot\prescript{*}{}{(2\eta - \zeta)} - 2\hat\chi^{AB}\underline\beta_A \prescript{*}{}{(2\underline\eta + \zeta)}_B - 2\omega \prescript{*}{}{(2\eta - \zeta)}\cdot\underline\beta + \hat\chi^{AB}\prescript{*}{}{\underline\eta}_A \underline\beta_B \\ 
    + \frac{1}{2}\tr\chi \underline\beta \cdot\prescript{*}{}{\underline\eta} + \prescript{*}{}{\beta}[W]\cdot\underline\beta + \frac{1}{2}\prescript{*}{}{(\eta + \underline\eta)}\cdot\underline\beta (2\omega - \tr\chi) - \underline\beta_B\sl\epsilon^{AB}\sl\nabla_A(2\omega \tr\chi) \\ 
    - \prescript{*}{}{(\eta + \underline\eta)}_A\hat\chi^{AB}\underline\beta_B + 2\underline\beta \cdot \prescript{*}{}{\sl\div\hat\chi} + \frac{1}{2}\prescript{*}{}{\big[} \beta[W] - 4\sl\nabla\omega - 2\omega(\eta + \underline\eta) - 2\underline\eta \cdot\hat\chi - \tr\chi \underline\eta \big]\cdot\underline\beta  \\ 
    - \frac{3}{2}\prescript{*}{}{\big[} \hat\chi\cdot(\eta - \underline\eta) + \frac{1}{2}\tr\chi (\eta - \underline\eta) + \beta[W] \big]\cdot \underline\beta + \prescript{*}{}{(2\eta - \zeta)}\cdot\underline\beta (2\omega - \tr\chi) \\ 
    + \frac{1}{2}\prescript{*}{}{\hat\chi}\cdot [(\zeta - 4\underline\eta)\hat\otimes \underline\beta].
\end{multline}
By grouping similar terms we eventually arrive at the following:
\begin{align*}
    I[\underline\beta] &= 2\underline\beta \cdot \prescript{*}{}{\Big(} \sl\div\hat\chi_A - \frac{1}{2}\sl\nabla_A\tr\chi + \frac{1}{2}(\eta - \underline\eta)^B (\hat\chi - \frac{1}{2}\tr\chi \gamma)_{AB} \Big)
\end{align*}
which, by the Codazzi equation \eqref{eq:null_structure_constraint}, is equal to 
\begin{equation}
    I[\underline\beta] = -2\underline\beta\cdot \prescript{*}{}{\beta}[W].
\end{equation}

For $I[\rho,\sigma]$: 
\begin{multline}
    I[\rho,\sigma] = \frac{3}{2}\sigma\big[ - \frac{1}{2}\tr\chi\tr\underline\chi + 2\underline\omega \tr\chi + 2\rho[W] - \hat\chi\cdot\hat{\underline\chi} + 2\sl\div\eta + 2|\eta|^2 \big] - \frac{9}{4}\tr\chi\tr\underline\chi \sigma \\ 
    + 3\omega\tr\underline\chi\sigma  - 3\underline\omega \tr\chi\sigma + 3\sigma\sl\div\underline\eta - 3\sigma\sl\div\eta + \frac{9}{4}\tr\chi\tr\underline\chi\sigma - \frac{3}{2}\sigma\big[ -\frac{1}{2}\tr\chi\tr\underline
    \chi \\ 
    + 2\omega \tr\underline\chi + 2\rho[W] - \hat\chi\cdot\hat{\underline\chi} + 2\sl\div\underline\eta + 2|\underline\eta|^2 \big] + \frac{1}{2}\prescript{*}{}{\hat\chi}\cdot (3\sigma\prescript{*}{}{\hat{\underline\chi}} - 3\rho\hat{\underline\chi}) \\ 
    - \prescript{*}{}{(2\underline\eta + \zeta)}\cdot (3\rho\eta + 3\sigma\prescript{*}{}{\eta}) - \frac{1}{2}\prescript{*}{}{\hat{\underline\chi}}\cdot (3\rho\hat\chi + 3\sigma\prescript{*}{}{\hat\chi}) + \frac{1}{2}\prescript{*}{}{(\eta + \underline\eta)} \cdot(3\sigma\prescript{*}{}{\underline\eta} - 3\rho\underline\eta ) \\ 
    - \frac{1}{2}\prescript{*}{}{(\eta + \underline\eta)}\cdot (3\sigma\prescript{*}{}{\eta} + 3\rho\eta) + 3\rho\sl\curl\underline\eta + 3\rho\sl\curl\eta + \prescript{*}{}{(2\eta - \zeta)}\cdot(3\sigma\prescript{*}{}{\underline\eta} - 3\rho\underline\eta).
\end{multline}
By grouping terms one directly verifies all terms cancel with each other and 
\begin{equation}
    I[\rho,\sigma] = 0.
\end{equation}

For $I[\beta]$:
\begin{multline}
    I[\beta] = \frac{1}{2}\big[\prescript{*}{}{\underline\beta}[W] + 4\prescript{*}{}{\sl\nabla}\underline\omega + 2\underline\omega \prescript{*}{}{(\eta + \underline\eta)} + 2\prescript{*}{}{\hat{\underline\chi}}\cdot\eta + \tr\underline\chi \prescript{*}{}{\eta}\big]\cdot \beta - \frac{3}{2}\big[ \prescript{*}{}{\hat{\underline\chi}}\cdot(\eta - \underline\eta) \\ 
    + \frac{1}{2}\tr\underline\chi\prescript{*}{}{(\eta - \underline\eta)} + \prescript{*}{}{\underline\beta}[W] \big]\cdot\beta  - (2\underline\omega - \tr\underline\chi)\beta\cdot\prescript{*}{}{(2\underline\eta + \zeta)} + \prescript{*}{}{\hat{\underline\chi}}^{AB}\beta_A (4\eta + \zeta)_B \\
    + 2\underline\omega \beta \cdot \prescript{*}{}{(2\underline\eta + \zeta)} + \hat{\underline\chi}^{AB}\beta_A \prescript{*}{}{(\eta + \underline\eta)}_B - 2\beta \cdot \prescript{*}{}{\sl\div}\hat{\underline\chi} + \frac{1}{2}(\tr\underline\chi - 2\underline\omega) \beta \cdot \prescript{*}{}{(\eta + \underline\eta)} \\
    - \frac{3}{2}\tr\underline\chi \beta \cdot \prescript{*}{}{(2\underline\eta + \zeta)} + 2\hat{\underline
    \chi}^{AB}\beta_A \prescript{*}{}{(2\eta - \zeta)}_B - \beta \cdot \prescript{*}{}{\sl\nabla}(2\underline\omega - \tr\underline\chi).
\end{multline}
Again, grouping similar terms and applying the Codazzi equation, this simplifies to:
\begin{align*}
    I[\beta] &= -2\beta\cdot \prescript{*}{}{\Big(} \sl\div\hat{\underline\chi} - \frac{1}{2}\sl\nabla\tr\underline\chi - \frac{1}{2}\hat{\underline\chi}\cdot(\eta - \underline\eta) + \frac{1}{4}\tr\underline\chi (\eta - \underline\eta) \Big) \\ 
    &= -2\beta\cdot \prescript{*}{}{\Big(} \sl\div\hat{\underline\chi} - \frac{1}{2}\sl\nabla\tr\underline\chi - \frac{1}{2}(\hat{\underline\chi} - \frac{1}{2}\tr\underline\chi\gamma)\cdot(\eta - \underline\eta) - \underline\beta[W] + \underline\beta[W] \Big) \\ 
    &= -2\beta \cdot \prescript{*}{}{\underline\beta}[W].
\end{align*}

For $I[\alpha]$: 
\begin{align*}
    I[\alpha] &= - \frac{1}{2}\big[ (2\underline\omega + \tr\underline\chi)\prescript{*}{}{\hat{\underline\chi}} + \prescript{*}{}{\underline\alpha}[W] \big] \cdot\alpha  + \frac{1}{2}\prescript{*}{}{\hat{\underline\chi}}\cdot\alpha (4\underline\omega - \frac{1}{2}\tr\underline\chi) - \underline\omega \prescript{*}{}{\hat{\underline\chi}}\cdot\alpha + \frac{3}{4}\tr\underline\chi \prescript{*}{}{\hat{\underline\chi}}\cdot\alpha \\ 
    &= - \frac{1}{2}\prescript{*}{}{\underline\alpha}[W]\cdot\alpha.
\end{align*}
Therefore we see that the order 0 terms in $\sl\nabla_3 Q$ are equal to:
\begin{multline}
    -\frac{1}{2}\prescript{*}{}{\alpha[W]}\cdot\underline\alpha - 2\underline\beta \cdot \prescript{*}{}{\beta}[W] - 2\beta \cdot\prescript{*}{}{\underline\beta}[W] - \frac{1}{2}\prescript{*}{}{\underline\alpha}[W]\cdot\alpha  \\ 
    = \frac{1}{2}\big( \prescript{*}{}{\underline\alpha} \cdot\alpha[W] - \alpha\cdot\prescript{*}{}{\underline\alpha}[W] \big) + 2\big( \prescript{*}{}{\underline\beta}\cdot\beta[W] -  \beta \cdot\prescript{*}{}{\underline\beta}[W] \big).
\end{multline}
This completes the computation for $\sl\nabla_3 Q$, showing that 
\begin{multline}
    \sl\nabla_3 Q = \sl\curl\Xi - \prescript{*}{}{(\eta + 2\underline\eta)}\cdot\Xi + (2\underline\omega - \frac{3}{2}\tr\underline\chi)Q \\ 
    + \frac{1}{2}\big( \prescript{*}{}{\underline\alpha} \cdot\alpha[W] - \alpha\cdot\prescript{*}{}{\underline\alpha}[W] \big) + 2\big( \prescript{*}{}{\underline\beta}\cdot\beta[W] -  \beta \cdot\prescript{*}{}{\underline\beta}[W] \big).
\end{multline}

\newpage
\emergencystretch=2em
\section*{References}
\printbibliography[heading=none]

\end{document}